\documentclass[a4paper, 11pt, openright, twoside]{report}
\usepackage[utf8]{inputenc}
\usepackage[ngermanb, english]{babel}
\usepackage{amsmath}
\usepackage{amssymb}
\usepackage{amsthm}
\usepackage{lipsum}
\usepackage{tikz-cd}
\usepackage{tikz}
\usepackage{commath}
\usepackage{mathrsfs}
\usepackage{mathtools}
\usepackage{slashed}
\usepackage{tensor}
\usepackage{parskip}
\usepackage{epstopdf}
\usepackage{graphicx}
\usepackage[section,numbib]{tocbibind}
\usepackage{babelbib}
\usepackage[pagebackref, bookmarksnumbered=true]{hyperref}
\usepackage{cleveref}
\usepackage{fancyhdr}
\usepackage[top=2.5cm, left=2.5cm, right=2.5cm, bottom=2.5cm]{geometry}
\usepackage[a-1b]{pdfx}

\renewcommand*{\backref}[1]{}
\renewcommand*{\backrefalt}[4]{%
  \ifcase #1 %
    No citations.%
  \or
    (Cited on page #4.)%
  \else
    (Cited on pages #4.)%
  \fi%
}

\theoremstyle{plain}
\newtheorem{thm}{Theorem}[section]
\newtheorem{lem}[thm]{Lemma}
\newtheorem{prop}[thm]{Proposition}
\newtheorem{obs}[thm]{Observation}
\newtheorem{col}[thm]{Corollary}
\newtheorem{conj}[thm]{Conjecture}
\theoremstyle{definition}
\newtheorem{defn/}[thm]{Definition}
\newtheorem{con/}[thm]{Convention}
\newtheorem{ter}[thm]{Terminology}
\newtheorem{ass}[thm]{Assumption}
\newtheorem{nota/}[thm]{Notation}
\newtheorem{exmp}[thm]{Example}

\newtheorem{sol}[thm]{Solution}
\theoremstyle{remark}
\newtheorem{rem}[thm]{Remark}

\newtheorem{pro}[thm]{Problem}

\newenvironment{defn}
{%
	\pushQED{\qed}\begin{defn/}}
	{\popQED\end{defn/}}
\newenvironment{con}
{%
\pushQED{\qed}\begin{con/}}
{\popQED\end{con/}}

\newcommand{\deDonder}{{d \negmedspace D \mspace{-2mu}}}
\newcommand{\linLorenz}{{l \negmedspace L \mspace{-2mu}}}
\newcommand{\Lorenz}{L}
\newcommand{\order}[1]{O \left ( #1 \right )}
\newcommand{\one}{\mathbb{I}}
\newcommand{\coone}{\hat{\mathbb{I}}}
\newcommand{\T}{{\boldsymbol{\mathrm{T}}}}
\newcommand{\Q}{{\boldsymbol{\mathrm{Q}}}}
\newcommand{\HQ}{\mathcal{H}_{\Q}}
\newcommand{\EQ}{{\mathcal{E}_{\Q}}}
\newcommand{\RQ}{\mathcal{R}_{\Q}}
\newcommand{\RQO}{{\mathcal{R}_{\Q}^{[0]}}}
\newcommand{\RQI}{{\mathcal{R}_{\Q}^{[1]}}}
\newcommand{\AQ}{\mathcal{A}_{\Q}}
\newcommand{\SQ}[1]{\mathcal{S} \left ( #1 \right )}
\newcommand{\CQ}[1]{\mathcal{C} \left ( #1 \right )}

\newcommand{\DQ}[1]{\mathcal{D} \left ( #1 \right )}
\newcommand{\DQprime}[1]{\mathcal{D}^\prime \left ( #1 \right )}
\newcommand{\IQ}[1]{\mathcal{I} \left ( #1 \right )}
\newcommand{\IQnolim}[1]{\mathcal{I} ( #1 )}
\newcommand{\IQrv}{\mathcal{I}^{r, \mathbf{v}}}
\newcommand{\GQ}{\mathcal{G}_{\Q}}
\newcommand{\TQ}{\mathcal{T}_{\Q}}
\newcommand{\QQ}{\boldsymbol{Q}_{\Q}}
\newcommand{\qQ}{\boldsymbol{q}_{\Q}}
\newcommand{\val}[1]{\operatorname{Val} \left ( #1 \right )}

\newcommand{\sym}[1]{\operatorname{Sym} \left ( #1 \right )}
\newcommand{\res}[1]{\operatorname{Res} \left ( #1 \right )}
\newcommand{\resnolim}[1]{\operatorname{Res} ( #1 )}
\newcommand{\cpl}[1]{\operatorname{Cpl} \left ( #1 \right )}

\newcommand{\cplgrd}[1]{\operatorname{CplGrd} \left ( #1 \right )}
\newcommand{\vtxgrd}[1]{\operatorname{VtxGrd} \left ( #1 \right )}
\newcommand{\vtxgrdnolim}[1]{\operatorname{VtxGrd} ( #1 )}
\newcommand{\extcpl}[1]{\operatorname{ExtCpl} \left ( #1 \right )}
\newcommand{\intcpl}[1]{\operatorname{IntCpl} \left ( #1 \right )}
\newcommand{\extvtx}[1]{\operatorname{ExtVtx} \left ( #1 \right )}
\newcommand{\intvtx}[1]{\operatorname{IntVtx} \left ( #1 \right )}

\newcommand{\ins}[2]{\operatorname{Ins} \left ( #1 \rhd #2 \right )}
\newcommand{\insaut}[3]{\operatorname{Ins}_\textup{Aut} \left ( #1 \rhd #2; #3 \right )}
\newcommand{\insrr}[1]{\operatorname{Ins}^{r, \mathbf{v}} \left ( #1 \right )}
\newcommand{\isoemb}[2]{\operatorname{Iso}_\textup{Emb} \left ( #1 \hookrightarrow #2 \right )}
\newcommand{\dt}[1]{\operatorname{Det} \left ( #1 \right )}
\newcommand{\tr}[1]{\operatorname{Tr} \left ( #1 \right )}

\newcommand{\bettio}[1]{\kappa \left ( #1 \right )}
\newcommand{\bettii}[1]{\lambda \left ( #1 \right )}
\newcommand{\degp}[1]{\operatorname{Deg}_p \left ( #1 \right )}
\newcommand{\sdd}[1]{\omega \left ( #1 \right )}
\newcommand{\csdd}[1]{\varpi \left ( #1 \right )}
\newcommand{\rsdd}[1]{\rho \left ( #1 \right )}
\newcommand{\ssdd}[1]{\sigma \left ( #1 \right )}
\newcommand{\ssddn}[1]{\sigma_\text{n} \left ( #1 \right )}
\newcommand{\ssddr}[1]{\sigma_\text{r} \left ( #1 \right )}
\newcommand{\ssdds}[1]{\sigma_\text{s} \left ( #1 \right )}
\newcommand{\D}[1]{\Delta \left ( #1 \right )}
\newcommand{\antipode}[1]{S \left ( #1 \right )}
\newcommand{\mult}{m}
\newcommand{\ring}{\Bbbk}
\newcommand{\field}{\mathbb{K}}
\newcommand{\FR}[1]{\Phi \left ( #1 \right )}
\newcommand{\FRP}[1]{\left ( \Phi \circ \mathscr{P} \right ) \left ( #1 \right )}

\newcommand{\RFR}[1]{\Phi_{\mathscr{R}} \left ( #1 \right )}
\newcommand{\regFR}{{\Phi_\mathscr{E}^\varepsilon}}
\newcommand{\renFR}{{\Phi_\mathscr{R}}}
\newcommand{\countertermsymbol}{S_\mathscr{R}^\regFR}
\newcommand{\counterterm}[1]{S_\mathscr{R}^\regFR \left ( #1 \right )}
\newcommand{\renscheme}[1]{\mathscr{R} \left ( #1 \right )}
\newcommand{\textfrac}[2]{#1 / #2}
\newcommand{\imaginary}{\mathrm{i}}
\newcommand{\id}{\operatorname{Id}}
\newcommand{\precombgreen}{\mathfrak{x}}
\newcommand{\combgreen}{\mathfrak{X}}
\newcommand{\rescombgreen}{\mathfrak{X}}
\newcommand{\combcharge}{\mathfrak{Q}}
\newcommand{\iQ}{\mathfrak{i}_\Q}
\newcommand{\ZvQ}{\mathbb{Z}^{\mathfrak{v}_\Q}}
\newcommand{\ZqQ}{\mathbb{Z}^{\mathfrak{q}_\Q}}
\newcommand{\surject}{\to \!\!\!\!\! \to}
\newcommand{\cgreen}[1]{\vcenter{\hbox{\includegraphics[width=\cgreenlength]{#1}}}}
\newcommand{\tcgreen}[1]{\vcenter{\hbox{\includegraphics[width=0.75\cgreenlength]{#1}}}}
\newcommand{\enter}{\vspace{\baselineskip}}
\newcommand{\mathfrakit}[1]{{\mspace{-1mu} \italicbox{$\mathfrak{#1}$} \mspace{2mu}}}
\newcommand{\mathbbit}[1]{{\mspace{-1mu} \italicbox{$\mathbb{#1}$} \mspace{2mu}}}
\newcommand{\mathbbsit}[1]{{\mspace{-1mu} \italicbox{$\scriptstyle{\mathbb{#1}}$} \mspace{1.5mu}}}

\newcommand{\GCD}[2]{\operatorname{GCD} \left ( #1, #2 \right )}
\newcommand{\bbL}{\mathbbit{L}}
\newcommand{\bbI}{\mathbbit{I} \mspace{1mu}}
\newcommand{\bbT}{\mathbbit{T} \mspace{1mu}}
\newcommand{\bbsL}{\mathbbit{\scriptstyle{L}}}
\newcommand{\bbsI}{\mathbbit{\scriptstyle{I}}}
\newcommand{\bbsT}{\mathbbit{\scriptstyle{T}}}
\newcommand{\legnumber}[1]{\mspace{-12mu} \raisebox{0.2ex}{$\scriptscriptstyle{#1}$} \mspace{9mu}}
\newcommand{\legnumberlongitudinal}[1]{\mspace{-23mu} \raisebox{0.2ex}{$\scriptscriptstyle{#1}$} \mspace{15mu}}
\newcommand{\legnumberghost}[1]{\mspace{-13mu} \raisebox{0.2ex}{$\scriptscriptstyle{#1}$} \mspace{9mu}}
\newcommand{\legnumberexponent}[1]{\mspace{-16mu} \raisebox{0.05ex}{$\scriptscriptstyle{#1}$} \mspace{9mu}}
\newcommand{\legnumberexponentlong}[1]{\mspace{-17mu} \raisebox{0.05ex}{$\scriptscriptstyle{#1}$} \mspace{-3mu}}
\newcommand{\legnumberexponentlongitudinal}[1]{\mspace{-29mu} \raisebox{0.05ex}{$\scriptscriptstyle{#1}$} \mspace{23mu}}
\newcommand{\legnumberexponentghost}[1]{\mspace{-17mu} \raisebox{0.05ex}{$\scriptscriptstyle{#1}$} \mspace{9mu}}
\newcommand{\bbM}{\mathbbit{M}}
\newcommand{\sbbM}{\mathbbsit{M}}
\newcommand{\met}{\gamma}
\newcommand{\trivmap}{\tau}
\newcommand{\conext}{\widetilde{M}}
\newcommand{\conmet}{\tilde{\met}}
\newcommand{\scri}{\mathscr{I}}
\newcommand{\particlefield}{\varphi}
\newcommand{\sctn}[1]{\Gamma \left ( #1 \right )}
\newcommand{\sctnbig}[1]{\Gamma \big ( #1 \big )}

\newcommand{\vectc}[1]{\mathfrak{X}_\text{c} \left ( #1 \right )}

\newcommand{\gcoupling}{\varkappa}
\newcommand{\gravitonghost}{C}

\newcommand{\diff}{\operatorname{Diff}_0 \left ( M \right )}
\newcommand{\Lie}{\mathsterling}
\newcommand{\lie}{\ell}
\newcommand{\diffbbM}{\operatorname{Diff}_0 \left ( \bbM \right )}
\newcommand{\deDonderFR}{{\mathfrak{d} \negmedspace \mathfrak{D}}}
\newcommand{\deDonderpreFR}{{\mathfrak{d} \negmedspace \mathfrak{d}}}

\newcommand{\BQ}{\mathcal{B}_\Q}
\newcommand{\FQ}{\mathcal{F}_\Q}
\newcommand{\LT}{\mathcal{L}_\T}
\newcommand{\LQ}{\mathcal{L}_\Q}

\newcommand{\ZQ}{\mathcal{Z}_\Q}
\newcommand{\ZQGSQ}{\mathcal{Z}_{\text{QGS}_\Q}}
\newcommand{\tf}[1]{\Omega^\text{top} \left ( #1 \right)}
\newcommand{\ds}{d}
\newcommand{\first}{\mathbf{a}}
\newcommand{\second}{\mathbf{b}}
\newcommand{\third}{\mathbf{c}}
\newcommand{\fourth}{\mathbf{d}}
\newcommand{\mtx}{\mathbf{M}}
\newcommand{\gravfr}{\mathfrak{G}}
\newcommand{\pregravfr}{\mathfrak{g}}

\newcommand{\gravghostfr}{\mathfrak{C}}
\newcommand{\preghostfr}{\mathfrak{c}}

\newcommand{\gravprop}{\mathfrak{P}}
\newcommand{\gravghostprop}{\mathfrak{p}}

\newcommand{\matterfrk}{\tensor[_k]{\mathfrak{M}}{}}
\newcommand{\prematterfrk}{\tensor[_k]{\mathfrak{m}}{}}
\newcommand{\prematterfri}{\tensor[_1]{\mathfrak{m}}{}}
\newcommand{\prematterfrii}{\tensor[_2]{\mathfrak{m}}{}}
\newcommand{\prematterfriii}{\tensor[_3]{\mathfrak{m}}{}}
\newcommand{\prematterfriv}{\tensor[_4]{\mathfrak{m}}{}}
\newcommand{\prematterfrv}{\tensor[_5]{\mathfrak{m}}{}}
\newcommand{\prematterfrvi}{\tensor[_6]{\mathfrak{m}}{}}
\newcommand{\prematterfrvii}{\tensor[_7]{\mathfrak{m}}{}}
\newcommand{\prematterfrviii}{\tensor[_8]{\mathfrak{m}}{}}
\newcommand{\prematterfrix}{\tensor[_9]{\mathfrak{m}}{}}
\newcommand{\prematterfrx}{\tensor[_{10}]{\mathfrak{m}}{}}
\newcommand{\triplevert}{\vert\kern-0.25ex\vert\kern-0.25ex\vert}

\newcommand{\commutatorbig}[2]{\big [ #1 , #2 \big ]}
\newcommand{\setbig}[1]{\big \{ #1 \big \}}

\newcommand{\M}{\bbM}

\newcommand{\SymM}{S^2 T^*M}
\newcommand{\SymMM}{S^2 T^*\mathbbit{M}}

\newsavebox{\foobox}
\newcommand{\italicbox}[2][.25]
{%
	\mbox
	{%
		\sbox{\foobox}{#2}%
		\hskip\wd\foobox
		\pdfsave
		\pdfsetmatrix{1 0 #1 1}%
		\llap{\usebox{\foobox}}%
		\pdfrestore
	}%
}

\makeatletter
\newcommand{\subalign}[1]{%
  \vcenter{%
    \Let@ \restore@math@cr \default@tag
    \baselineskip\fontdimen10 \scriptfont\tw@
    \advance\baselineskip\fontdimen12 \scriptfont\tw@
    \lineskip\thr@@\fontdimen8 \scriptfont\thr@@
    \lineskiplimit\lineskip
    \ialign{\hfil$\m@th\scriptstyle##$&$\m@th\scriptstyle{}##$\crcr
      #1\crcr
    }%
  }
}
\makeatother

\providecommand{\chpref}[1]{Chapter~\ref{#1}}

\providecommand{\chpsaref}[2]{Chapters~\ref{#1} and \ref{#2}}
\providecommand{\sectionref}[1]{Section~\ref{#1}}
\providecommand{\sectionsaref}[2]{Sections~\ref{#1} and \ref{#2}}
\providecommand{\ssecref}[1]{Subsection~\ref{#1}}

\providecommand{\sssecref}[1]{Subsubsection~\ref{#1}}

\providecommand{\eqnref}[1]{Equation~\eqref{#1}}
\providecommand{\eqnsref}[1]{Equations~\eqref{#1}}
\providecommand{\eqnsaref}[2]{Equations~\eqref{#1} and \eqref{#2}}
\providecommand{\eqnssaref}[3]{Equations~\eqref{#1}, \eqref{#2} and \eqref{#3}}

\providecommand{\propsaref}[2]{Propositions~\ref{#1} and \ref{#2}}
\providecommand{\propssaref}[3]{Propositions~\ref{#1}, \ref{#2} and \ref{#3}}
\providecommand{\proref}[1]{Problem~\ref{#1}}
\providecommand{\solref}[1]{Solution~\ref{#1}}
\providecommand{\solsaref}[4]{Solutions~\ref{#1}, \ref{#2}, \ref{#3} and \ref{#4}}
\providecommand{\solsoref}[4]{Solutions~\ref{#1}, \ref{#2}, \ref{#3} or \ref{#4}}

\providecommand{\obsref}[1]{Observation~\ref{#1}}
\providecommand{\conref}[1]{Convention~\ref{#1}}

\providecommand{\asssnref}[3]{Assumptions~\ref{#1}, \ref{#2} and \ref{#3}}
\providecommand{\exmpsaref}[2]{Examples~\ref{#1} and \ref{#2}}
\providecommand{\defnsaref}[2]{Definitions~\ref{#1} and \ref{#2}}
\providecommand{\colsaref}[2]{Corollaries~\ref{#1} and \ref{#2}}
\providecommand{\colssaref}[3]{Corollaries~\ref{#1}, \ref{#2} and \ref{#3}}
\providecommand{\lemsaref}[2]{Lemmata~\ref{#1} and \ref{#2}}
\providecommand{\thmsaref}[2]{Theoremata~\ref{#1} and \ref{#2}}
\providecommand{\conjref}[1]{Conjecture~\ref{#1}}
\providecommand{\remsaref}[2]{Remarks~\ref{#1} and \ref{#2}}

\DeclareSymbolFont{extraitalic}      {U}{zavm}{m}{it}
\DeclareMathSymbol{\Qoppa}{\mathord}{extraitalic}{161}
\DeclareMathSymbol{\qoppa}{\mathord}{extraitalic}{162}
\DeclareMathSymbol{\Stigma}{\mathord}{extraitalic}{167}
\DeclareMathSymbol{\Sampi}{\mathord}{extraitalic}{165}
\DeclareMathSymbol{\sampi}{\mathord}{extraitalic}{166}
\DeclareMathSymbol{\stigma}{\mathord}{extraitalic}{168}

\newlength{\graphlength}
\newlength{\cgreenlength}
\title{\textsc{Renormalization of Gauge Theories and Gravity}}
\author{David Nicolas Prinz\footnote{Department of Mathematics and Department of Physics at Humboldt University of Berlin, Max Planck Institute for Gravitational Physics (Albert Einstein Institute) in Potsdam-Golm and Department of Mathematics at University of Potsdam; prinz@\{math.hu-berlin.de, physik.hu-berlin.de, aei.mpg.de, math.uni-potsdam.de\}}}

\date{August 31, 2022}

\begin{document}

\thispagestyle{empty}

\begin{center}
\newlength{\Title}
\settowidth{\Title}{\textsc{\LARGE{\hspace{0.5\baselineskip}}\huge{\textsc{Renormalization of Gauge Theories}}}\LARGE{\hspace{0.5\baselineskip}}}
\fbox{\parbox{\Title}{\LARGE{\vspace{0.5\baselineskip}} \begin{center} \huge{\textsc{Renormalization of Gauge Theories\\and Gravity}} \end{center} \LARGE{\vspace{0.5\baselineskip}}}}

\Large{\vspace{3\baselineskip}}

\LARGE{Dissertation zur Erlangung des akademischen Grades \\ \vspace{0.125\baselineskip} \emph{doctor rerum naturalium (Dr. rer. nat.)} \\ \vspace{0.125\baselineskip} im Fach Mathematik.}

\Large{\vspace{3\baselineskip}}

\Large{Eingereicht an der mathematisch-naturwissenschaftlichen Fakultät \\ der Humboldt-Universität zu Berlin von:} \\ \Large{\vspace{\baselineskip}} \LARGE{\textbf{\underline{M.Sc. M.Sc. David Nicolas Prinz}}} \\

\Large{\vspace{3.5\baselineskip}}

\Large{Kommissarischer Präsident der Humboldt-Universität zu Berlin:} \\ \vspace{0.25\baselineskip}
\underline{Prof. Dr. Peter Frensch} \\

\Large{\vspace{1.5\baselineskip}}

\Large{Dekanin der mathematisch-naturwissenschaftlichen Fakultät:} \\ \vspace{0.25\baselineskip}
\underline{Prof. Dr. Caren Tischendorf} \\

\Large{\vspace{1.5\baselineskip}}

\Large{Gutachter:} \\ \vspace{0.5\baselineskip}
\begin{tabular}{l l}
1. & \underline{Prof. Dr. Dirk Kreimer} \\[5pt]
2. & \underline{Prof. Dr. Walter van Suijlekom} \\[5pt]
3. & \underline{Prof. Dr. Karen Yeats}
\end{tabular}

\Large{\vspace{1.5\baselineskip}}

\Large{Tag der mündlichen Prüfung:} \\ \vspace{0.25\baselineskip}
\underline{31. August 2022}

\Large{\vspace{\baselineskip}}

\end{center}

\newpage
\thispagestyle{plain}
\quad
\newpage
\thispagestyle{plain}
\emph{``It is only with the heart that one can see rightly; what is essential is invisible to the eye.''} \\ \vspace{-0.5\baselineskip} \\ \rightline{--- The Little Prince // Antoine de Saint Exup\'{e}ry ---}
\newpage
\thispagestyle{plain}
\quad
\newpage
\thispagestyle{plain}

\section*{Acknowledgements}

It is my pleasure to thank the special people in my life that have made this dissertation possible:

First and foremost, I would like to thank Prof.\ Dr.\ Dirk Kreimer for his supervision and refereeing. Quantum gravitation has always fascinated me, as it is deeply linked to the fundamental and philosophical questions \emph{how our universe emerged} and thus \emph{how our existence started}. Therefore, I was lucky to find in Dirk a professor who was not only willing to supervise such a demanding topic, but also one with innovative ideas for the quest of finding the theory of quantum gravity. Additionally, it is my pleasure to thank Prof.\ Dr.\ Hermann Nicolai and Prof.\ Dr.\ Christian Bär for additional supervision as well as Prof.\ Dr.\ Walter van Suijlekom and Prof.\ Dr.\ Karen Yeats for additional refereeing. Further thanks goes also to Prof.\ Dr.\ Chris Wendl for chairing the defense and to Prof.\ Dr.\ Dorothee Schüth for being an additional member of the committee. In addition, I wish to thank Dr.\ Axel Kleinschmidt for organizing the research school and for being a supportive mentor when needed. As well, I wish to thank Sylvia Richter for fast and helpful support concerning organizational duties. I also greatly acknowledge the funding by the International Max Planck Research School for Mathematical and Physical Aspects of Gravitation, Cosmology and Quantum Field Theory (IMPRS), the Kolleg Mathematik Physik Berlin (KMPB) and the University of Potsdam via the research group of Prof.\ Dr.\ Sylvie Paycha.

Furthermore, I would like to thank all my friends for their companionship during the time of this dissertation: This includes mental support related to my studies as well as assistance in times of personal challenges, but also simply celebrating the good times. Specifically, I would like to thank Tom Klose for various helpful words of advice, proof-readings and for being a tremendous flatmate and friend. Additionally, I would like to thank Murat Bardak for countless exciting badminton, table tennis and tennis matches, relaxing walks and for being a superb friend. In addition, I would like to thank Elias Rüttinger for many enjoyable discussions on physics, life in general and for being an excellent friend. As well, I would like to thank Tobias Kaiser for many adventurous walks, thrilling climbing sessions and for being a great friend. Last but not least, I would like to thank Ben Ossamy-Dornemann for countless good times since we met in high school, for bringing a lot of laughter and lightness into my life and for being a wonderful friend. Further thanks goes to all my friends that have been unmentioned: It is a pleasure having you in my life!

Moreover, I would like to thank my colleagues Paul Balduf, Marko Berghoff, Henry Ki\ss{}ler, Ren\'{e} Klausen and Maximilian Mühlbauer for interesting and stimulating discussions as well as intense working sessions. Additionally, I am grateful for the shared memories at the summer school on structures in local quantum field theory in Les Houches 2018 and the IMPRS excursion to Mallorca 2019. In addition, I would like to thank Alexander Schmeding for fruitful collaborations, helpful discussions and general support.

Additionally, I wish to thank my grandmother Annemarie Prinz and my dad Roland Sorichter for their support during my studies, in particular financially, which allowed me to completely focus on my work. Further thanks goes also to my uncles Manfred Prinz and Erwin Prinz together with their families.

Finally, I wish to express my gratitude towards my mother Marianne Prinz and my grandfather Georg Prinz for their love and support, both of whom I keep in good memory.

\newpage
\thispagestyle{plain}
\quad
\newpage
\thispagestyle{plain}

\setcounter{tocdepth}{2}
\addtocontents{toc}{\protect\setcounter{tocdepth}{-1}}
\tableofcontents
\addtocontents{toc}{\protect\setcounter{tocdepth}{2}}

\newpage
\thispagestyle{plain}
\quad
\newpage

\chapter{Introduction} \label{chp:introduction}

We start this dissertation with the abstracts in English and German and a list of the affiliated articles. Then we provide a pictorial introduction to quantum gravity via the dice of physics and discuss the history and context of (generalized) quantum gauge theories, with a particular focus on quantum gravity. This includes a discussion on different approaches to quantum gravity as well as (generalized) Slavnov--Taylor identities. Finally, we provide a list with all original results of this dissertation.

\section{Abstract}

\subsection{English version}

We study the perturbative quantization of gauge theories and gravity. Our investigations start with the geometry of spacetimes and particle fields. Then we discuss the various Lagrange densities of (effective) Quantum General Relativity coupled to the Standard Model. In addition, we study the corresponding BRST double complex of diffeomorphisms and gauge transformations. Next we apply Connes--Kreimer renormalization theory to the perturbative Feynman graph expansion: In this framework, subdivergences are organized via the coproduct of a Hopf algebra and the renormalization operation is described as an algebraic Birkhoff decomposition. To this end, we generalize and improve known coproduct identities and a theorem of van Suijlekom (2007) that relates (generalized) gauge symmetries to Hopf ideals. In particular, our generalization applies to gravity, as was suggested by Kreimer (2008). In addition, our results are applicable to theories with multiple vertex residues, coupling constants and such with a transversal structure. Additionally, we also provide criteria for the compatibility of these Hopf ideals with Feynman rules and the chosen renormalization scheme. We proceed by calculating the corresponding gravity-matter Feynman rules for any valence and with a general gauge parameter. Then we display all propagator and three-valent vertex Feynman rules and calculate the respective cancellation identities. Finally, we propose planned follow-up projects: This includes a generalization of Wigner's classification of elementary particles to linearized gravity, the representation of cancellation identities via Feynman graph cohomology and an investigation on the equivalence of different definitions for the graviton field. In particular, we argue that the appropriate setup to study perturbative BRST cohomology is a differential-graded Hopf algebra: More precisely, we suggest a modified version of the Feynman graph cohomology of Kreimer et al.\ (2013) and consider its action on the renormalization Hopf algebra. We then argue that the compatibility of cancellation identities with the renormalization operation is reflected in the well-definedness of the differential-graded renormalization Hopf algebra.

\subsection{German version}

\begin{otherlanguage}{ngerman}

Wir studieren die perturbative Quantisierung von Eichtheorien und Gravitation. Unsere Untersuchungen beginnen mit der Geometrie von Raumzeiten und Teilchenfeldern. Danach diskutieren wir die verschiedenen Lagrangedichten in der Kopplung der (effektiven) Quanten-Allgemeinen-Relativitätstheorie zum Standardmodell. Desweiteren studieren wir den zugehörigen BRST-Doppelkomplex von Diffeomorphismen und Eichtransformationen. Danach wenden wir Connes--Kreimer-Renormierungstheorie auf die perturbative Feynmangraph-Entwicklung an: In dieser Formulierung werden Subdivergenzen mittels des Koprodukts einer Hopfalgebra strukturiert und die Renormierungsoperation mittels einer algebraischen Birkhoff-Zerlegung beschrieben. Dafür verallgemeinern und verbessern wir bekannte Koprodukt-Identitäten und ein Theorem von van Suijlekom (2007), das (verallgemeinerte) Eichsymmetrien mit Hopfidealen verbindet. Insbesondere lässt sich unsere Verallgemeinerung auf Gravitation anwenden, wie von Kreimer (2008) vorgeschlagen. Darüberhinaus sind unsere Resultate anwendbar auf Theorien mit mehreren Vertexresuiden, Kopplungskonstanten und ebensolchen mit einer transversalen Struktur. Zusätzlich zeigen wir Kriterien für die Kompatibilität dieser Hopfideale mit Feynmanregeln und dem gewählten Renormierungsschema. Als nächsten Schritt berechnen wir die entsprechenden Gravitations-Materie Feynmanregeln für alle Vertexvalenzen und mit einem allgemeinen Eichparameter. Danach listen wir alle Propagator- und dreivalenten Vertex-Feynmanregeln auf und berechnen die entsprechenden Kürzungsidentitäten. Abschließend stellen wir geplante Folgeprojekte vor: Diese schließen eine Verallgemeinerung von Wigners Klassifikation von Elementarteilchen für linearisierte Gravitation ein, ebenso wie die Darstellung von Kürzungsidentitäten mittels Feynmangraph-Kohomologie und eine Untersuchung der Äquivalenz verschiedener Definitionen des Gravitonfeldes. Insbesondere argumentieren wir, dass das richtige Setup um perturbative BRST-Kohomologie zu studieren eine differentialgraduierte Hopfalgebra ist: Konkret schlagen wir eine modifizierte Version der Feynmangraph-Kohomologie von Kreimer et al.\ (2013) vor und betrachten dessen Wirkung auf der Renormierungs-Hopfalgebra. Wir argumentieren dann, dass sich die Kompatibilität der Kürzungsidentitäten mit der Renormierungsoperation in der Wohldefiniertheit der differentialgraduierten Hopfalgebra widerspiegelt.

\end{otherlanguage}

\section{Corresponding articles}

This dissertation is based on the published journal articles \cite{Prinz_2,Prinz_3,Prinz_4} and the preprints in preparation \cite{Prinz_5,Prinz_7}. We also remark the planned follow-up projects that build directly on the results of this dissertation \cite{Prinz_8,Prinz_9,Prinz_10}, which are outlined in \sectionref{sec:outlook}. Concretely, these manuscripts contribute to this dissertation as follows:

\begin{itemize}
\item[\cite{Prinz_2}:] This article is based on the author's second master thesis \cite{Prinz_Master_2}. We have studied several geometric and renormalization related aspects of (effective) Quantum General Relativity coupled to spinor Quantum Electrodynamics. As such, it provides the foundation for (effective) Quantum General Relativity coupled to general Quantum Gauge Theories. In particular, in this dissertation we consider (effective) Quantum General Relativity coupled to the Standard Model. We remark that the material taken from \cite{Prinz_2} concerns only some introductory parts in \chpref{chp:gauge_theories_and_gravity} and the heavily reworked \sectionref{sec:the_associated_renormalization_hopf_algebra}. Thus, the text overlap of this dissertation with the author's second master thesis \cite{Prinz_Master_2} is neglectable. In particular, all main results of this dissertation are entirely new.
\item[\cite{Prinz_3}:] In this article, we have studied renormalization related aspects of Quantum Gauge Theories. The main result is an essential generalization of a theorem of van Suijlekom which relates (generalized) quantum gauge symmetries to Hopf ideals. This theorem --- formerly proven for renormalizable Quantum Gauge Theories --- has been generalized to super- and non-renormalizable Quantum Field Theories, theories with multiple coupling constants and such with longitudinal and transversal degrees of freedom. These results are presented in \chpref{chp:hopf_algebraic_renormalization}.
\item[\cite{Prinz_4}:] In this article, we have derived the Feynman rules for (effective) Quantum General Relativity coupled to the Standard Model. More precisely, our analysis applies to 4 dimensional spacetimes, the metric decomposition \(g_{\mu \nu} = \eta_{\mu \nu} + \varkappa h_{\mu \nu}\) and the linearized de Donder gauge fixing. In particular, our results generalize earlier works in that we present the Feynman rules for arbitrary vertex valence and the graviton propagator with general gauge parameter \(\zeta\). These results are presented in \chpsaref{chp:gauge_theories_and_gravity}{chp:linearized_gravity_and_feynman_rules}.
\item[\cite{Prinz_5}:] In this preprint, we study the BRST double complex for (effective) Quantum General Relativity coupled to the Standard Model. This double complex arises from the invariance of the theory under infinitesimal diffeomorphisms and infinitesimal gauge transformations. To this end, we recall the setup for (effective) Quantum General Relativity and Quantum Yang--Mills theory. In particular, we study the corresponding gauge fixing and ghost Lagrange densities via gauge fixing fermions. Then, the main results are as follows: First, we find that the two BRST operators anticommute and thus produce a double complex. Furthermore, this gives rise to the total BRST operator and we find that we can obtain the complete gauge fixing and ghost Lagrange densities via a total gauge fixing fermion. Moreover, we find that the graviton-ghosts decouple from matter of the Standard Model if and only if the gauge fixing fermion for the gauge theory is a tensor density of weight \(w = 1\). Finally, we also present the corresponding anti-BRST operators and show that all BRST and anti-BRST operators mutually anticommute. These results are presented in \sectionref{sec:diffeomorphism-gauge-brst-double-complex}.
\item[\cite{Prinz_7}:] In this preprint, we study the transversal structure of (effective) Quantum General Relativity coupled to the Standard Model. This includes several identities and decompositions of the respective longitudinal, identical and transversal projection operators. In particular, we recall known identities of Quantum Yang--Mills theory with a Lorenz gauge fixing and introduce their counterparts in (effective) Quantum General Relativity. Then we provide the corresponding symmetric (hermitian) ghost Lagrange densities together with their gauge fixing Lagrange densities. Finally, we compute the longitudinal projections of all three-valent gluon, gluon-matter, graviton and graviton-matter vertex Feynman rules and show that they vanish on-shell. These results are presented in \sectionsaref{sec:explicit_feynman_rules}{sec:longitudinal_and_transversal_projections}.
\end{itemize}

Finally, we also remark the two side projects, joined with Prof.\ Dr.\ Alexander Schmeding, that additionally emerged during the time of the author's doctoral studies \cite{Prinz_Schmeding_1,Prinz_Schmeding_2}.

\section{The dice of physics}

The dice of physics is a vivid graphical representation of physical theories in different regimes. The axes are identified with three of the four fundamental physical coupling constants with mass dimension 1:\footnote{The fourth being Boltzmann's constant \(k\), which is not relevant for the motivation of Quantum Gravity.} The inverse of the speed of light \(c^{-1}\), Newton's gravitation constant \(G\) and Planck's constant \(\hbar\). Furthermore, we choose Planck units, i.e.\ \(c^{-1} = G = \hbar = 1\), and thus its linear span
\begin{equation}
\operatorname{Span}_{\left [ 0, 1 \right ]} \left ( \left \{ \begin{pmatrix} c^{-1} \\ 0 \\ 0 \end{pmatrix}, \begin{pmatrix} 0 \\ G \\ 0 \end{pmatrix}, \begin{pmatrix} 0 \\ 0 \\ \hbar \end{pmatrix} \right \} \right ) \cong \left [ 0, 1 \right ]^3
\end{equation}
represents a cube, where each vertex represents a particular limit of the `theory of everything':
\begin{center}
\begin{tikzpicture}
\newcommand*\edge{5}
\path[scale=\edge]
  (0,0,0)  coordinate [label=left:{\((0,c^{-1},0)\): SR}]  (A)
  (1,0,0)  coordinate [label=right:{\((G,c^{-1},0)\): GR}] (B)
  (1,1,0)  coordinate [label=right:{\, \(\mathbf{1}\): QG}] (C)
  (0,1,0)  coordinate [label=left:{\((0,c^{-1},\hbar)\): QFT}]  (D)
  (0,0,-1) coordinate [label=left:{\color{gray} \(\mathbf{0}\): CM}]  (E)
  (1,0,-1) coordinate [label=right:{\((G,0,0)\): NG}] (F)
  (1,1,-1) coordinate [label=right:{\((G,0,\hbar)\): nrQG}] (G)
  (0,1,-1) coordinate [label=left:{\((0,0,\hbar)\): QM}]  (H)
;
\draw[font=\tiny]
  (A) -- (B)
      -- (C)
      -- (D) -- cycle
  (B) -- (F)
      -- (G) -- (H) --(D)
  (C) -- (G)
;
\path[dashed, very thin, font=\small] (E) edge node[pos=.6,sloped,above]{\color{gray} \(c^{-1}\) \color{black}} (A) edge node[above] {\color{gray} \(G\)} (F) edge node[sloped,above]{\color{gray} \(\hbar\)} (H);
\end{tikzpicture}
\end{center}
Here, CM stands for Classical Mechanics, SR for Special Relativity, NG for Newton Gravitation, QM for Quantum Mechanics, GR for General Relativity, QFT for Quantum Field Theory, nrQG for non-relativistic Quantum Gravitation and finally QG for Quantum Gravity. Up to now, we have a pretty good understanding of the theories on the bottom and left faces and lack any substantial understanding of the theories on the upper right edge. Thus, the aim of this dissertation is to shed some further light on this unknown region, in particular on the corner \(\mathbf{1} \equiv (G,c^{-1},\hbar)\) that represents Quantum Gravity, which we approach via a perturbative quantization of General Relativity.

\section{History and context} \label{sec:history-and-context}

General Relativity and Quantum Theory are both fundamental theories in modern physics. While some of their predictions agree with outstanding precision with the corresponding experimental data, there are still regimes where both theories break down conceptually. Notably, this is the case with models of the big bang or in the inside of black holes. In these situations, both theories are needed simultaneously to capture the entire physical reality: General Relativity is needed in order to describe the huge masses and energies that are involved and Quantum Theory is needed in order to describe the interactions of the respective particles in these very small spatial dimensions. Unfortunately, a combined theory of Quantum Gravitation has not been found yet: While, given the success of the Standard Model, a perturbative quantization seems to be the canonical choice, it comes with several problems, most notably its non-renormalizability. This fact has lead to various, more radical approaches to Quantum Gravity, such as Supergravity, String Theory or Loop Quantum Gravity. While any of these theories fixes conceptual problems of the perturbative approach, they create additional problems elsewhere due to further assumptions. Therefore, in this dissertation, we go back to the foundation of Quantum General Relativity via its (effective) perturbative approach using Feynman rules. Feynman rules are calculated from the Lagrangian by extracting the potentials for all classically allowed interactions. Then, scattering amplitudes are calculated by applying the fundamental principle of Quantum Theory, namely that the sum over all unobserved intermediate states needs to be considered. This leads to the Feynman diagram expansion, where each non-tree Feynman diagram corresponds to a Feynman integral over the unobserved momenta of the virtual particles. We refer to \cite{Weinberg_1,Weinberg_2,Weinberg_3} for a more detailed treatment on Feynman rules and to \cite{Jimenez} for the corresponding treatment of supersymmetric theories. Proceeding with this approach and again the fundamental principle of Quantum Theory to sum over all unobserved intermediate states, we find that each loop in a Feynman diagram comes with an unobserved momenta for the involved virtual particles. Thus, we need to additionally integrate over these unobserved momenta, which are called Feynman integrals. However, both the sum over the Feynman diagram expansion as well as the Feynman integrals over unobserved momenta lead to serious problems: Typically both these operations do not converge and are thus ill-defined. The first problem can be solved via resummation techniques, such as Borel resummation, whereas the second problem can be solved via renormalization theory, which we study in the Framework of Connes and Kreimer \cite{Connes_Kreimer_0,Connes_Kreimer_1,Connes_Kreimer_2,Connes_Kreimer_NG}. In this dissertation, we want to apply this renormalization framework to a perturbative expansion of (effective) Quantum General Relativity. We find, that it is in principle indeed possible to render the corresponding Feynman integrals finite, if the diffeomorphism invariance is properly incorporated into the renormalization scheme, as was originally suggested in \cite{Kreimer_QG1,Kreimer_QG2}. However, the obstacle is now to show that the actual Feynman rules are compatible with this procedure to all orders, in the sense of cancellation identities \cite{tHooft_Veltman,Citanovic,Sars_PhD,Kissler_Kreimer,Gracey_Kissler_Kreimer,Kissler}, which will be the subject for future work.

The attempt to perturbatively quantize General Relativity (GR) is rather old: In fact, the approach to define the graviton field \(h_{\mu \nu}\) with gravitational coupling constant \(\gcoupling\) as the fluctuation around a fixed background metric \(b_{\mu \nu}\), i.e.\
\begin{equation}
h_{\mu \nu} := \frac{1}{\gcoupling} \left ( g_{\mu \nu} - b_{\mu \nu} \right ) \iff g_{\mu \nu} \equiv b_{\mu \nu} + \gcoupling h_{\mu \nu} \, ,
\end{equation}
--- oftentimes, and in particular in this thesis, chosen as the Minkowski metric \(b_{\mu \nu} := \eta_{\mu \nu}\) --- goes back to M. Fierz, W. Pauli and L. Rosenfeld in the 1930s \cite{Rovelli}. Then, R. Feynman \cite{Feynman_Hatfield_Morinigo_Wagner} and B. DeWitt \cite{DeWitt_I,DeWitt_II,DeWitt_III,DeWitt_IV} started to calculate the corresponding Feynman rules in the 1960s. However, D. Boulware, S. Deser and P. van Nieuwenhuizen \cite{Boulware_Deser_Nieuwenhuizen}, G. 't Hooft \cite{tHooft_QG} and M. Veltman \cite{Veltman} discovered serious problems in the perturbative expansion due to the non-renormalizability of Quantum General Relativity (QGR) in the 1970s. We refer to \cite{Rovelli} for a historical treatment.

Despite its age, it is still very hard to find references properly deriving and displaying Feynman rules for QGR, given via the Lagrange density
\begin{subequations} \label{eqns:QGR_Lagrange_density_Introduction}
\begin{align}
	\mathcal{L}_\text{QGR} & := \mathcal{L}_\text{GR} + \mathcal{L}_\text{GF} + \mathcal{L}_\text{Ghost}
	\intertext{with}
	\mathcal{L}_\text{GR} & := - \frac{1}{2 \gcoupling^2} R \dif V_g \, , \\
	\mathcal{L}_\text{GF} & := - \frac{1}{4 \gcoupling^2 \zeta}  \eta^{\mu \nu} \deDonder^{(1)}_\mu \deDonder^{(1)}_\nu \dif V_\eta
	\intertext{and}
	\mathcal{L}_\text{Ghost} & := - \frac{1}{2 \zeta} \eta^{\rho \sigma} \overline{C}^\mu \left ( \partial_\rho \partial_\sigma C_\mu \right ) \dif V_\eta - \frac{1}{2} \eta^{\rho \sigma} \overline{C}^\mu \left ( \partial_\mu \big ( \tensor{\Gamma}{^\nu _\rho _\sigma} C_\nu \big ) - 2 \partial_\rho \big ( \tensor{\Gamma}{^\nu _\mu _\sigma} C_\nu \big ) \right ) \dif V_\eta \, ,
\end{align}
\end{subequations}
where \(R := g^{\nu \sigma} \tensor{R}{^\mu _\sigma _\mu _\nu}\) is the Ricci scalar and \(\deDonder^{(1)}_\mu := \eta^{\rho \sigma} \Gamma_{\rho \sigma \mu}\) is the linearized de Donder gauge fixing functional. Additionally, \(\gravitonghost \in \Gamma \left ( M, T^*[1,0] \right )\) and \(\overline{\gravitonghost} \in \Gamma \left ( M, T[-1,0] M \right )\) are the graviton-ghost and graviton-antighost, respectively. Finally, \(\dif V_g := \sqrt{- \dt{g}} \dif t \wedge \dif x \wedge \dif y \wedge \dif z\) and \(\dif V_\eta := \dif t \wedge \dif x \wedge \dif y \wedge \dif z\) are the Riemannian and Minkowskian volume forms, respectively. We refer to \chpref{chp:gauge_theories_and_gravity} of this dissertation and \cite{Prinz_2,Prinz_4} for detailed introductions to the perturbative expansion. The existing literature known to the author \cite{Veltman,Sannan,Donoghue,Choi_Shim_Song,tHooft_PQG,Hamber,Schuster,Rodigast_Schuster_1,Rodigast_Schuster_2} limits the vertex Feynman rules to valence five, directly sets the de Donder gauge fixing parameter to \(\zeta := 1\) and omits the ghost vertex Feynman rules completely. This dissertation together with the corresponding article \cite{Prinz_4} aim to fix this gap in the literature by properly deriving the Feynman rules for gravitons, their ghosts and for their interactions with matter from the Standard Model: The analysis is carried out for the metric decomposition \(g_{\mu \nu} = \eta_{\mu \nu} + \gcoupling h_{\mu \nu}\), arbitrary vertex valence, a linearized de Donder gauge fixing with general gauge parameter \(\zeta\) and in four dimensions of spacetime. Moreover, the gravitational interactions with matter from the Standard Model are then classified into 10 different types and their vertex Feynman rules are also properly derived and presented for any valence.

This research is intended as the starting point for several related approaches to the perturbative renormalization of (effective) Quantum General Relativity: It is generally possible to render any Feynman integral finite by applying an appropriate subtraction for each divergent (sub-) integral. This treatment of (sub-)divergences has been studied extensively in the Hopf algebra approach of Connes and Kreimer: Here, the subdivergences are treated via the corresponding coproduct \cite{Kreimer_Hopf_Algebra} and the renormalized Feynman rules are then obtained via an algebraic Birkhoff decomposition \cite{Connes_Kreimer_0}. Then, this reasoning was soon applied to gauge theories \cite{Kreimer_Anatomy}, which lead to the identification of Ward--Takahashi and Slavnov--Taylor identities with Hopf ideals in the corresponding Connes--Kreimer renormalization Hopf algebra \cite{vSuijlekom_QED,vSuijlekom_QCD,vSuijlekom_BV,Prinz_3}. Following this route, it was then suggested by Kreimer to apply this duality to General Relativity \cite{Kreimer_QG1}, which was motivated via a scalar toy model \cite{Kreimer_vSuijlekom} and then studied in detail by the author \cite{Prinz_2,Prinz_3}. In this approach, the non-renormalizability of General Relativity manifests itself by the necessity to introduce infinitely many counterterms. The aim is now to relate these counterterms by generalized Slavnov--Taylor identities, which correspond to the diffeomorphism invariance of the theory. A first step in this direction is the construction of tree-level cancellation identities, which requires the longitudinal and transversal decomposition of the graviton propagator via the general gauge parameter \(\zeta\) as variable. This approach was supported by recent calculations for the metric density decomposition of Goldberg and Capper et al.\ (\cite{Goldberg,Capper_Leibbrandt_Ramon-Medrano,Capper_Medrano,Capper_Namazie}) up to valence six \cite{Kissler}. With the present work, we provide a foundation for a purely combinatorial argument, which will be valid for all vertex valences. This will be studied in future work via the diffeomorphism-gauge BRST double complex \cite{Prinz_5} and the longitudinal and transversal structure of the gravitational Feynman rules \cite{Prinz_7}, cf.\ \sectionsaref{sec:diffeomorphism-gauge-brst-double-complex}{sec:longitudinal_and_transversal_projections}. Additionally, we remark that this reasoning is implicit in the construction of Kreimer's Corolla polynomial \cite{Kreimer_Yeats,Kreimer_Sars_vSuijlekom,Kreimer_Corolla}. This graph polynomial, which is based on half-edges, relates the amplitudes of Quantum Yang--Mills theories to the amplitudes of the scalar \(\phi^3_4\)-theory, by means of the parametric representation of Feynman integrals \cite{Sars_PhD,Golz_PhD}. More precisely, in this approach the cancellation identities are encoded into amplitudes by means of Feynman graph cohomology \cite{Berghoff_Knispel}. In particular, this approach has been successfully generalized to spontaneously broken gauge theories and thus to the complete bosonic part of the Standard Model \cite{Prinz_1}. The possibility to apply this construction also to (effective) Quantum General Relativity will be checked in future work. Finally, we believe that the results of this dissertation are also of intellectual interest, as Feynman rules are an essential ingredient to perturbative Quantum Field Theories.

We remark the development of more concise formulations, aimed in particular for practical calculations: There are the KLT relations \cite{Kawai_Lewellen_Tye,Bern_Carrasco_Johansson_1,Bern_Carrasco_Johansson_2,Elvang_Huang}, which relate on-shell gravitational amplitudes with the amplitudes of the `double-copy' of a gauge theory, and are applied e.g.\ in \cite{Bern_et-al}. Furthermore, it is also possible to simplify the gravitational Feynman rules by a reformulation with different (possibly auxiliary) fields \cite{Cheung_Remmen,Tomboulis}, even on a de Sitter background \cite{Tsamis_Woodard}. Moreover, we remark the use of computer algebra programs, such as `XACT' \cite{Abreu_et-al} and `QGRAF' \cite{Blumlein_et-al}. For the projects mentioned in the previous paragraph, however, the original Feynman rules are needed to arbitrary vertex valence and with general gauge parameter \(\zeta\): This is because the KLT relations are only valid on-shell and thus rely on Cutkosky's Theorem \cite{Cutkosky}, cf.\ \cite{Bloch_Kreimer,Kreimer_Cutkosky}. Also, we want to study the direct relationship between combinatorial Green's functions and their counterterms, which becomes more complicated in the aforementioned reformulations with auxiliary fields. And finally, we are interested in a combinatorial proof that is valid to all vertex valences and thus excludes the use of computer algebra programs.

Going back to the aforementioned approach for a perturbative renormalization of (effective) Quantum General Relativity, we now discuss the renormalization-related aspects: In classical physics, Noether's Theorem relates symmetries to conserved quantities \cite{Noether}. In the context of classical gauge theories it states that gauge symmetries correspond to charge conservation. Thus, gauge symmetries are a fundamental ingredient in physical theories, such as the Standard Model and General Relativity. When considering their quantizations, identities related to their gauge invariance ensure that the gauge fields are indeed transversal: More precisely, the Ward--Takahashi and Slavnov--Taylor identities ensure that photons and gluons do only possess their two experimentally verified transversal degrees of freedom \cite{Ward,Takahashi,Taylor,Slavnov}.\footnote{Actually, Slavnov--Taylor identities were first discovered diagrammatically by Gerard 't Hooft in \cite{tHooft}.} As we will see in the upcoming examples, these identities relate the prefactors of different monomials in the Lagrange density that are linked via gauge transformations. In the corresponding Quantum Field Theories, however, each of these monomials might obtain a different energy dependence through its \(Z\)-factor, which could spoil this symmetry. Additionally, in order to calculate the propagator of the gauge field, a gauge fixing needs to be chosen. This gauge fixing then needs to be accompanied by its corresponding ghost and antighost fields, which have Grassmannian parity, i.e.\ obey Fermi--Dirac statistics. More precisely, the ghost fields are designed to satisfy the residual gauge transformations as equations of motion, whereas the antighost fields are designed to be constant and thus act as Lagrange multipliers.\footnote{We remark that this is the case for the Faddeev--Popov ghost construction. It is possible to construct a more general setup, which mixes or even reverses their roles \cite{Baulieu_Thierry-Mieg}.} This then ensures transversality of the gauge bosons, if the \(Z\)-factors fulfill certain identities. These identities, in turn, depend on the Feynman rules and the chosen renormalization scheme.

As a first example, consider Quantum Yang--Mills theory with a Lorenz gauge fixing, given via the Lagrange density
\begin{equation} \label{eqn:qym-lagrange-density-introduction}
\begin{split}
	\mathcal{L}_\text{QYM} & := \mathcal{L}_\text{YM} + \mathcal{L}_\text{GF} + \mathcal{L}_\text{Ghost} \\ & \phantom{:} = \eta^{\mu \nu} \eta^{\rho \sigma} \delta_{a b} \left ( - \frac{1}{4 \mathrm{g}^2} F^a_{\mu \rho} F^b_{\nu \sigma} - \frac{1}{2 \xi} \big ( \partial_\mu A^a_\nu \big ) \big ( \partial_\rho A^b_\sigma \big ) \right ) \dif V_\eta \\
	& \phantom{:=} + \eta^{\mu \nu} \left ( \frac{1}{\xi} \overline{c}_a \left ( \partial_\mu \partial_\nu c^a \right ) + \mathrm{g} \tensor{f}{^a _b _c} \overline{c}_a \left ( \partial_\mu \big ( c^b A^c_\nu \big ) \right ) \right ) \dif V_\eta \, ,
\end{split}
\end{equation}
where \(F^a_{\mu \nu} := \mathrm{g} \big ( \partial_\mu A^a_\nu - \partial_\nu A^a_\mu \big ) - \mathrm{g}^2 \tensor{f}{^a _b _c} A^b_\mu A^c_\nu\) is the local curvature form of the gauge boson \(A^a_\mu\). Furthermore, \(\dif V_\eta := \dif t \wedge \dif x \wedge \dif y \wedge \dif z\) denotes the Minkowskian volume form. Additionally, \(\eta^{\mu \nu} \partial_\mu A^a_\nu \equiv 0\) is the Lorenz gauge fixing functional and \(\xi\) the gauge fixing parameter. Finally, \(c^a\) and \(\overline{c}_a\) are the gauge ghost and gauge antighost, respectively. Then, the decomposition into its monomials is given via 
\begin{equation} \label{eqn:qym-lagrange-density-introduction-monomials}
\begin{split}
	\mathcal{L}_\text{QYM} & \equiv - \frac{1}{2} \eta^{\mu \nu} \eta^{\rho \sigma} \delta_{ab} \left ( \big ( \partial_\mu A^a_\rho \big ) \big ( \partial_\nu A^b_\sigma \big ) - \left ( 1 - \frac{1}{\xi} \right ) \big ( \partial_\mu A^a_\nu \big ) \big ( \partial_\rho A^b_\sigma \big ) \right ) \dif V_\eta \\
	& \phantom{\equiv} + \frac{1}{2} \mathrm{g} \eta^{\mu \nu} \eta^{\rho \sigma} f_{abc} \left ( \big ( \partial_\mu A^a_\rho \big ) \big ( A^b_\nu A^c_\sigma \big ) \right ) \dif V_\eta \\
	& \phantom{\equiv} - \frac{1}{4} \mathrm{g}^2 \eta^{\mu \nu} \eta^{\rho \sigma} \tensor{f}{^a _b _c} f_{ade} \big ( A^b_\mu A^c_\rho A^d_\nu A^e_\sigma \big ) \dif V_\eta \\
	& \phantom{\equiv} + \frac{1}{\xi} \eta^{\mu \nu} \overline{c}_a \left ( \partial_\mu \partial_\nu c^a \right ) \dif V_\eta + \mathrm{g} \eta^{\mu \nu} \tensor{f}{^a _b _c} \overline{c}_a \left ( \partial_\mu \big ( c^b A^c_\nu \big ) \right ) \dif V_\eta \, ,
\end{split}
\end{equation}
each of which contributing to a different Feynman rule. To absorb the upcoming divergences in a multiplicative manner, we multiply each monomial with an individual function \(Z^r \! \left ( \varepsilon \right )\) in the regulator \(\varepsilon \in \mathbb{R}\):
\begin{equation} \label{eqn:qym-lagrange-density-introduction-monomials-renormalized}
\begin{split}
	\mathcal{L}_\text{QYM}^\text{R} \left ( \varepsilon \right ) & := - \frac{1}{2} \eta^{\mu \nu} \eta^{\rho \sigma} \delta_{ab} \, Z^{A^2} \! \! \left ( \varepsilon \right ) \! \left ( \! \big ( \partial_\mu A^a_\rho \big ) \big ( \partial_\nu A^b_\sigma \big ) - \left ( 1 - \frac{Z^{1 / \xi} \! \left ( \varepsilon \right )}{\xi} \right ) \! \big ( \partial_\mu A^a_\nu \big ) \big ( \partial_\rho A^b_\sigma \big ) \! \right ) \dif V_\eta \\
	& \phantom{:=} + \frac{1}{2} \mathrm{g} \eta^{\mu \nu} \eta^{\rho \sigma} f_{abc} \, Z^{A^3} \! \! \left ( \varepsilon \right ) \left ( \big ( \partial_\mu A^a_\rho \big ) \big ( A^b_\nu A^c_\sigma \big ) \right ) \dif V_\eta \\
	& \phantom{:=} - \frac{1}{4} \mathrm{g}^2 \eta^{\mu \nu} \eta^{\rho \sigma} \tensor{f}{^a _b _c} f_{ade} \, Z^{A^4} \! \! \left ( \varepsilon \right ) \big ( A^b_\mu A^c_\rho A^d_\nu A^e_\sigma \big ) \dif V_\eta \\
	& \phantom{:=} + \frac{1}{\xi} \eta^{\mu \nu} Z^{\overline{c} c / \xi} \! \left ( \varepsilon \right ) \left ( \overline{c}_a \left ( \partial_\mu \partial_\nu c^a \right ) \right ) \dif V_\eta + \mathrm{g} \eta^{\mu \nu} \tensor{f}{^a _b _c} \,  Z^{\overline{c} c A} \! \left ( \varepsilon \right ) \left ( \overline{c}_a \left ( \partial_\mu \big ( c^b A^c_\nu \big ) \right ) \! \right ) \dif V_\eta \, ,
\end{split}
\end{equation}
where the regulator \(\varepsilon\) is related to the energy scale through the choice of a regularization scheme. Then, the invariance of \(\mathcal{L}_\text{QYM}^\text{R} \left ( \varepsilon \right )\) under (residual) gauge transformations away from the reference point (where all \(Z\)-factors fulfill \(Z^r \! \left ( \varepsilon_0 \right ) = 1\)) depends on the following identities:
\begin{subequations} \label{eqns:z-factor_identities_qym}
\begin{align}
	\frac{\big ( Z^{A^3} \! \! \left ( \varepsilon \right ) \! \big )^2}{Z^{A^2} \! \! \left ( \varepsilon \right )} & \equiv Z^{A^4} \! \! \left ( \varepsilon \right )
	\intertext{and}
	\frac{Z^{A^3} \! \! \left ( \varepsilon \right )}{Z^{A^2 / \xi} \! \left ( \varepsilon \right )} & \equiv \frac{Z^{\overline{c} c A} \! \left ( \varepsilon \right )}{Z^{\overline{c} c / \xi} \! \left ( \varepsilon \right )}
\end{align}
\end{subequations}
for all \(\varepsilon\) in the domain of the regularization scheme and with \(Z^{A^2 / \xi} \! \left ( \varepsilon \right ) := Z^{A^2} \! \! \left ( \varepsilon \right ) Z^{1 / \xi} \! \left ( \varepsilon \right )\). This example is continued, using the terminology introduced throughout this dissertation, in \exref{exmp:qym}, where we will reencounter these identities in a different language. We remark that these identities are essential for the Faddeev--Popov ghost construction to work, such that physical gluons are indeed transversal.

As a second example, consider (effective) Quantum General Relativity with the metric decomposition \(g_{\mu \nu} \equiv \eta_{\mu \nu} + \varkappa h_{\mu \nu}\), where \(h_{\mu \nu}\) is the graviton field and \(\varkappa := \sqrt{\kappa}\) the graviton coupling constant (with \(\kappa := 8 \pi G\) the Einstein gravitational constant), and a linearized de Donder gauge fixing, given via the Lagrange density
\begin{equation} \label{eqn:qgr-lagrange-density-introduction}
\begin{split}
	\mathcal{L}_\text{QGR} & := \mathcal{L}_\text{GR} + \mathcal{L}_\text{GF} + \mathcal{L}_\text{Ghost} \\ & \phantom{:} = - \frac{1}{2 \varkappa^2} \left ( \sqrt{- \dt{g}} R + \frac{1}{2 \zeta} \eta^{\mu \nu} \deDonder^{(1)}_\mu \deDonder^{(1)}_\nu \right ) \dif V_\eta \\ & \phantom{:=} - \frac{1}{2} \eta^{\rho \sigma} \left ( \frac{1}{\zeta} \overline{C}^\mu \left ( \partial_\rho \partial_\sigma C_\mu \right ) + \overline{C}^\mu \left ( \partial_\mu \big ( \tensor{\Gamma}{^\nu _\rho _\sigma} C_\nu \big ) - 2 \partial_\rho \big ( \tensor{\Gamma}{^\nu _\mu _\sigma} C_\nu \big ) \right ) \right ) \dif V_\eta \, ,
\end{split}
\end{equation}
where \(R := g^{\nu \sigma} \tensor{R}{^\mu _\nu _\mu _\sigma}\) is the Ricci scalar (with \(\tensor{R}{^\rho _\sigma _\mu _\nu} := \partial_\mu \tensor{\Gamma}{^\rho _\nu _\sigma} - \partial_\nu \tensor{\Gamma}{^\rho _\mu _\sigma} + \tensor{\Gamma}{^\rho _\mu _\lambda} \tensor{\Gamma}{^\lambda _\nu _\sigma} - \tensor{\Gamma}{^\rho _\nu _\lambda} \tensor{\Gamma}{^\lambda _\mu _\sigma}\) the Riemann tensor). Again, \(\dif V_\eta := \dif t \wedge \dif x \wedge \dif y \wedge \dif z\) denotes the Minkowskian volume form, which is related to the Riemannian volume form \(\dif V_g\) via \(\dif V_g \equiv \sqrt{- \dt{g}} \dif V_\eta\). Additionally, \(\deDonder^{(1)}_\mu := \eta^{\rho \sigma} \Gamma_{\mu \rho \sigma} \equiv 0\) is the linearized de Donder gauge fixing functional and \(\zeta\) the gauge fixing parameter. Finally, \(C_\mu\) and \(\overline{C}^\mu\) are the graviton-ghost and graviton-antighost, respectively. We refer to \cite{Prinz_2,Prinz_4} for more detailed introductions and further comments on the chosen conventions. Then, we decompose \(\mathcal{L}_\text{QGR}\) with respect to its powers in the gravitational coupling constant \(\varkappa\), the gauge fixing parameter \(\zeta\) and the ghost field \(C\) as follows\footnote{We omit the term \(\mathcal{L}_\text{QGR}^{-1,0,0}\) as it is given by a total derivative.}
\begin{equation} \label{eqn:qgr-lagrange-density-introduction-monomials}
	\mathcal{L}_\text{QGR} \equiv \sum_{i = 0}^\infty \sum_{j = -1}^0 \sum_{k = 0}^1 \mathcal{L}_\text{QGR}^{i,j,k} \, ,
\end{equation}
where we have set \(\mathcal{L}_\text{QGR}^{i,j,k} := \eval[1]{\left ( \mathcal{L}_\text{QGR} \right )}_{O (\varkappa^i \zeta^j C^k)}\), cf.\ \cite[Section 3]{Prinz_4}. Again, to absorb the upcoming divergences in a multiplicative manner, we multiply each monomial with an individual function \(Z^r \! \left ( \varepsilon \right )\) in the regulator \(\varepsilon \in \mathbb{R}\):
\begin{equation} \label{eqn:qgr-lagrange-density-introduction-monomials-renormalized}
	\mathcal{L}_\text{QGR}^\text{R} \left ( \varepsilon \right ) := \sum_{i = 0}^\infty \sum_{j = -1}^0 \sum_{k = 0}^1 Z^{i,j,k} \! \left ( \varepsilon \right ) \mathcal{L}_\text{QGR}^{i,j,k} \, ,
\end{equation}
where the regulator \(\varepsilon\) is related to the energy scale through the choice of a regularization scheme. Then, the invariance of \(\mathcal{L}_\text{QGR}^\text{R} \left ( \varepsilon \right )\) under (residual) diffeomorphisms away from the reference point (where all \(Z\)-factors fulfill \(Z^r \! \left ( \varepsilon_0 \right ) = 1\)) depends on the following identities:
\begin{subequations} \label{eqns:z-factor_identities_qgr}
	\begin{align}
	\frac{Z^{i,0,0} \! \left ( \varepsilon \right ) Z^{1,0,0} \! \left ( \varepsilon \right )}{Z^{0,0,0} \! \left ( \varepsilon \right )} & \equiv Z^{(i+1),0,0} \! \left ( \varepsilon \right )
	\intertext{and}
	\frac{Z^{i,0,0} \! \left ( \varepsilon \right )}{Z^{0,-1,0} \! \left ( \varepsilon \right )} & \equiv \frac{Z^{i,0,1} \! \left ( \varepsilon \right )}{Z^{0,-1,1} \! \left ( \varepsilon \right )}
	\end{align}
\end{subequations}
for all \(i \in \mathbb{N}_+\) and \(\varepsilon\) in the domain of the regularization scheme. This example is continued, using the terminology introduced throughout this dissertation, in \exref{exmp:qgr}, where we will reencounter these identities in a different language. Again, we remark that these identities are essential for the Faddeev--Popov ghost construction to work, such that physical gravitons are indeed transversal.

Identities for \(Z\)-factors, such as \eqnsaref{eqns:z-factor_identities_qym}{eqns:z-factor_identities_qgr}, are known in the literature as Ward--Takahashi identities in the realm of Quantum Electrodynamics and Slavnov--Taylor identities in the realm of Quantum Chromodynamics \cite{Ward,Takahashi,Taylor,Slavnov,tHooft}. We will study these identities on a general level and thus call them `quantum gauge symmetries (QGS)'. In particular, our results are directly applicable to (effective) Quantum General Relativity in the sense of \eqnref{eqns:z-factor_identities_qgr}, as was first suggested in \cite{Kreimer_QG1} and then studied for a scalar toy model in \cite{Kreimer_vSuijlekom}. We also refer to \cite{Prinz_2} for a more detailed introduction to (effective) Quantum General Relativity coupled to Quantum Electrodynamics and to \cite{Prinz_4} for the Feynman rules of (effective) Quantum General Relativity coupled to the Standard Model. In the present dissertation we study the renormalization-related properties of quantum gauge symmetries, such as \eqnsaref{eqns:z-factor_identities_qym}{eqns:z-factor_identities_qgr}. This is done in the framework of the Connes--Kreimer renormalization Hopf algebra: In this setup, the organization of subdivergences of Feynman graphs is encoded into the coproduct of a Hopf algebra \cite{Kreimer_Hopf_Algebra} and the renormalization of Feynman rules is described via an algebraic Birkhoff decomposition \cite{Connes_Kreimer_0}. Then, the aforementioned identities induce symmetries inside the renormalization Hopf algebra, which was first studied via Hochschild cohomology \cite{Kreimer_Anatomy} and then shown to be Hopf ideals \cite{vSuijlekom_QED,vSuijlekom_QCD,vSuijlekom_BV,Kreimer_vSuijlekom}. The aim of this dissertation is to generalize these results in several directions: We first introduce an additional coupling-grading in \defnref{defn:connectedness_gradings_renormalization_hopf_algebra}, which allows us to study theories with multiple coupling constants, like Quantum General Relativity coupled to the Standard Model. Furthermore, and more substantially, this allows us to discuss the transversality of such (generalized) quantum gauge theories, cf.\ \defnref{defn:transversal_structure} and \defnref{defn:quantum_gauge_symmetries}. Moreover, we generalize known coproduct and antipode identities to super- and non-renormalizable Quantum Field Theories (QFTs) in \propref{prop:proj_div_graphs_coprod} and \sectionref{sec:coproduct_and_antipode_identities}. This requires a detailed study of combinatorial properties of the superficial degree of divergence, as is presented in \sectionref{sec:a_superficial_argument}. The analysis then culminates in \thmref{thm:quantum_gauge_symmetries_induce_hopf_ideals}, stating that quantum gauge symmetries correspond to Hopf ideals also in this generalized context. Finally, we provide criteria for the validity of quantum gauge symmetries in terms of Feynman rules and renormalization schemes: Technically speaking, this corresponds to the situation that the aforementioned Hopf ideals are in the kernel of the counterterm map or even the renormalized Feynman rules. The result is then presented in \thmref{thm:criterion_ren-hopf-mod}. A consequence thereof is that the Corolla polynomial for Quantum Yang--Mills theory is well-defined without reference to a particular renormalization scheme if the renormalization scheme is proper, cf.\ \defnref{defn:proper_renormalization_scheme} and \remref{rem:corolla_polynomial}. The Corolla polynomial is a graph polynomial in half-edges that relates amplitudes in Quantum Yang--Mills theory to amplitudes in \(\phi^3_4\)-theory \cite{Kreimer_Yeats,Kreimer_Sars_vSuijlekom,Kreimer_Corolla}. More precisely, this graph polynomial is used for the construction of a so-called Corolla differential that acts on the parametric representation of Feynman integrals \cite{Sars_PhD,Golz_PhD}. Thereby, the corresponding cancellation-identities \cite{tHooft_Veltman,Citanovic,Sars_PhD,Kissler_Kreimer,Gracey_Kissler_Kreimer,Kissler} are encoded into a double complex of Feynman graphs, leading to Feynman graph cohomology \cite{Kreimer_Sars_vSuijlekom,Berghoff_Knispel}. We remark that this construction has been successfully generalized to Quantum Yang--Mills theories with spontaneous symmetry breaking \cite{Prinz_1} and Quantum Electrodynamics with spinors \cite{Golz_1,Golz_2,Golz_3}. The application of this formulation to (effective) Quantum General Relativity is a topic of ongoing research.

\section{Original results}

To the author's best knowledge, this dissertation provides the following new results:

\begin{itemize}
	\item[\secref{chp:introduction}:]
		\begin{itemize}
			\item \sectionref{sec:history-and-context}: We first recall the multiplicative renormalization of Quantum Yang--Mills theory starting with \eqnref{eqn:qym-lagrange-density-introduction} together with the Slavnov--Taylor identities in \eqnref{eqns:z-factor_identities_qym}. Then we discuss the corresponding generalizations to (effective) Quantum General Relativity around \eqnsaref{eqn:qgr-lagrange-density-introduction}{eqns:z-factor_identities_qgr}, respectively.
		\end{itemize}
	\item[\secref{chp:gauge_theories_and_gravity}:]
		\begin{itemize}
			\item \sectionref{sec:the-geometry-of-spacetimes}: We introduce the notion of a \emph{simple spacetime} in \defnref{defn:simple_spacetime} inspired from Penrose's notion of an \emph{asymptotically simple spacetime}. This then enables us to provide a global definition of the graviton field together with a Fourier transformation for particle fields in \defnsaref{defn:md_and_gf}{defn:fourier_transform}, respectively. Additionally, we comment on the diffeomorphism invariance of General Relativity as a generalized gauge symmetry from the viewpoint of Lie groupoids and Lie algebroids in \remref{rem:Lie_groupoid_and_algebroid}.
			\item \sectionref{sec:the-geometry-of-particle-fields}: We introduce the \emph{spacetime-matter bundle} in \defnref{defn:spacetime-matter_bundle} as the \(\mathbb{Z}^2\)-graded super vector bundle whose sheaf of sections describe the particle fields of (effective) Quantum General Relativity coupled to the Standard Model in \defnref{defn:sheaf-of-particle-fields}.
			\item \sectionref{sec:lagrange-densities}: We classify the gravitational couplings to matter from the Standard Model in \lemref{lem:matter-model-Lagrange-densities}. Then we discuss the respective Lagrange densities in detail.
			\item \sectionref{sec:diffeomorphism-gauge-brst-double-complex}: We state the diffeomorphism and gauge BRST operators in \defnsaref{defn:diffeomorphism_brst_operator}{defn:gauge_brst_operator} together with the gauge fixing fermions for the de Donder and Lorenz gauge fixings in \propsaref{prop:de_donder_gauge_fixing_fermion}{prop:lorenz_gauge_fixing_fermion}, respectively. We have reworked the conventions such that the propagators of longitudinal gravitons and gauge bosons as well as their ghosts are scaled by the corresponding gauge fixing parameters. Additionally, we have extended the diffeomorphism BRST operator to matter from the Standard Model. Furthermore, we classify all non-constant Lagrange densities that are essentially closed with respect to the diffeomorphism BRST operator as scalar tensor densities of weight \(w = 1\) in \lemref{lem:p_tensor_densities}. Moreover, we show that the two BRST operators anticommute and thus give rise to a double complex and a \emph{total BRST operator} in \thmref{thm:total_brst_operator}. In addition, we show that the matter fields decouple from the graviton-ghost if the gauge theory gauge fixing fermion is a tensor density of weight \(w = 1\) in \thmref{thm:no-couplings-grav-ghost-matter-sm}. Then we show that under said conditions both gauge fixing fermions can be added to construct the complete gauge fixing and ghost Lagrange densities via the action of the total BRST operator on this \emph{total gauge fixing fermion} in \thmref{thm:total_gauge_fixing_fermion}. Additionally, we also recall and improve the corresponding anti-BRST operators in \defnsaref{defn:diffeomorphism_anti-brst_operator}{defn:gauge_anti-brst_operator}. Finally, we show that all BRST and anti-BRST operators mutually anticommute in \colssaref{col:anti-diffeomorphism_brst_operator}{col:anti-gauge_brst_operator}{col:total_anti-brst_operator} together with the aforementioned \thmref{thm:total_brst_operator}.
		\end{itemize}
	\item[\secref{chp:hopf_algebraic_renormalization}:]
		\begin{itemize}
			\item \sectionref{sec:preliminaries_of_hopf_algebraic_renormalization}: We introduce several notions that allow us to study the perturbative renormalization of quantum gauge theories that are possibly super- or non-renormalizable, involve multiple coupling constants or longitudinal and transversal degrees of freedom: First, we introduce the notion of a \emph{transversal structure} as the set of all longitudinal, identical and transversal projection operators of a given gauge theory in \defnref{defn:transversal_structure}. Then we add the \emph{coupling-grading} as the physically relevant grading to the two so far established gradings for renormalization Hopf algebras, the loop-grading and vertex-grading, in \defnref{defn:connectedness_gradings_renormalization_hopf_algebra}. Next we introduce a projector to superficially divergent graphs in \defnref{defn:projection_divergent_graphs}. In addition, we introduce the notion of a \emph{superficially compatible grading} in \defnref{defn:superficially_compatible_grading}. Also, we introduce the \emph{algebra of (formal) Feynman integral expressions} in \defnref{defn:formal_feynman_integral_expressions} as the target algebra of Feynman rules, cf.\ \defnref{defn:fr_reg_ren_counterterm}. Finally, we characterize \emph{finite renormalization schemes} in \lemref{lem:finite_renormalization_schemes} and introduce the notion of a \emph{proper renormalization scheme} in \defnref{defn:proper_renormalization_scheme}.
			\item \sectionref{sec:the_associated_renormalization_hopf_algebra}: We discuss obstructions for quantum gauge theories that lead to ill-defined renormalization Hopf algebras in \proref{pro:problem}. Then we present four different possibilities to overcome them in \solsaref{sol:solution_1}{sol:solution_2}{sol:solution_3}{sol:solution_4} and discuss their physical difference in \remref{rem:physical_interpretation}. This includes in particular a modified Feynman graph set in \defnref{defn:modified_feynman_graph_set} or modified sets of divergent subgraphs in \defnref{defn:modified_sets_of_superficially_divergent_subgraphs}.
			\item \sectionref{sec:predictive-quantum-field-theories}: We first recall the established classification of Quantum Field Theories in super-renormalizable, renormalizable and non-renormalizable in \defnref{defn:classification_of_qfts}. Then we provide the connection between the mass-dimension criteria and the grading criteria in \lemref{lem:classification-qfts}. Next we introduce the notion of a \emph{predictive Quantum Field Theory} in \defnref{defn:predictive_quantum_field_theory} to address non-renormalizable Quantum Field Theories that still allow for a well-defined perturbative expansion. Finally, we comment that super-renormalizable and renormalizable Quantum Field Theories are predictive in \obsref{obs:s-ren-and-ren-qfts-are-predictive}. In addition, we claim that (effective) Quantum General Relativity, possibly coupled to matter from the Standard Model, is predictive in \conjref{conj:qgr-predictive}.
			\item \sectionref{sec:a_superficial_argument}: We study combinatorial properties of the superficial degree of divergence to generalize important identities known for renormalizable Quantum Field Theories to the super- and non-renormalizable cases: We start with an alternative definition of the superficial degree of divergence via \emph{weights of corollas} in \defnref{defn:asdd} and \lemref{lem:asdd}. This allows us to state that the superficial degree of divergence depends affine-linearly on the vertex-grading (or any coarser superficially-compatible grading) in \thmref{thm:asdd}. With that, we obtain a further criterion for the renormalization-classification of Quantum Field Theories in \colref{col:weights_of_corollas_and_renormalizability}. Additionally, this allows us to analyze the vertex-grading dependence of the superficial degree of divergence even further in \colref{col:ssdd_decomposition}. Then we introduce the notion of a \emph{cograph-divergent Quantum Field Theory} in \defnref{defn:cograph-divergent}, which we show to satisfy the known simple coproduct and antipode identities, even in the super- and non-renormalizable cases, in \propref{prop:proj_div_graphs_coprod}. Finally, we provide simple criteria for either cograph-divergenceness and superficial grade compatibility in \lemref{lem:cograph-divergence_criterion} and \propref{prop:superficial_grade_compatibility}, respectively. Then we conclude that (effective) Quantum General Relativity coupled to the Standard Model satisfies both notions in \colsaref{col:qgr-sm_is_cograph-divergent}{col:qgr-sm_is_sqgsc}, respectively.
			\item \sectionref{sec:coproduct_and_antipode_identities}: We provide an explicit equivalence between coproduct and antipode identities in \lemref{lem:coproduct_and_antipode_identities}. Then we generalize known coproduct identities to super- and non-renormalizable Quantum Field Theories, theories with multiple coupling constants and such with longitudinal and transversal degrees of freedom in \propssaref{prop:coproduct_greensfunctions}{prop:coproduct_charges}{prop:coproduct_exponentiated_combinatorial_charges}.
			\item \sectionref{sec:quantum_gauge_symmetries_and_subdivergences}: We introduce the notion of \emph{quantum gauge symmetries} in \defnref{defn:quantum_gauge_symmetries}. Then we discuss the respective cases in Quantum Yang--Mills theory and (effective) Quantum General Relativity in \exmpsaref{exmp:qym}{exmp:qgr}, respectively. Furthermore, we define the corresponding \emph{quantum gauge symmetry ideal} in \defnref{defn:qgs_ideal} and prove that it is a Hopf ideal in \thmref{thm:quantum_gauge_symmetries_induce_hopf_ideals}. This generalizes the earlier result of van Suijlekom \cite[Theorem 15]{vSuijlekom_BV} to super- and non-renormalizable Quantum Field Theories, theories with multiple coupling constants and such with longitudinal and transversal degrees of freedom. Finally, we show that the quantum gauge symmetry ideal is the smallest ideal such that in the quotient Hopf algebra the coupling-grading and vertex-grading become equivalent in \colref{col:equivalence_vtx-grd_cpl-grd}.
			\item \sectionref{sec:quantum_gauge_symmetries_and_renormalized_feynman_rules}: We provide an explicit definition of the \emph{renormalization Hopf module} in \defnref{defn:renormalization_hopf_module}. Then we show that it satisfies the coproduct identities of \propssaref{prop:coproduct_greensfunctions}{prop:coproduct_charges}{prop:coproduct_exponentiated_combinatorial_charges} also with respect to the coupling-grading in \colref{col:hopf_subalgebras_coupling-grading}. Additionally, we provide criteria that ensure the compatibility of the quantum gauge symmetry ideal with the counterterm in \lemref{lem:criterion_ren-hopf-mod} and the renormalized Feynman rules in \thmref{thm:criterion_ren-hopf-mod} and \colref{col:qgs_and_rfr}. Finally, we discuss the well-definedness of the Corolla polynomial with respect to the choice of a renormalization scheme as a consequence of our findings in \remref{rem:corolla_polynomial}.
		\end{itemize}
	\item[\secref{chp:linearized_gravity_and_feynman_rules}:]
		\begin{itemize}
			\item \sectionref{sec:linearized_gravity_and_further_preparations}: We study the expansion of the gravity-matter Lagrange densities with respect to the graviton coupling constant. To this end, we state the expansions of the inverse metric in \lemref{lem:inverse_metric_series} and vielbeins and inverse vielbeins in \lemref{lem:vielbeins_series}. Then we calculate the metric expression of the Ricci scalar with respect to the Levi-Civita connection in \propref{prop:ricci_scalar_for_the_levi_civita_connection} and its expansion in \colref{col:ricci_scalar_for_the_levi_civita_connection_restriction}. Next we calculate the metric expression of the de Donder gauge fixing in \propref{prop:metric_expression_for_de_donder_gauge_fixing} and its expansion in \colref{col:metric_expression_for_de_donder_gauge_fixing_restriction}. Finally, we calculate the metric expression of the determinant of the metric in \propref{prop:determinant_metric} and the expansion of the negative of its square-root in \colref{col:determinant_metric_restriction}.
			\item \sectionref{sec:feynman_rules_for_any_valence}: We introduce our notation for gravity-matter Feynman rules in \defnref{defn:notation_qgr-fr}. Then we provide the graviton, graviton-ghost and graviton-matter Feynman rules as nested sums of products. This generalizes previous works in that we provide the Feynman rules for any valence, with general gauge parameter and the graviton-ghost vertex Feynman rules at all. Explicitly, we provide the graviton Feynman rules for vertices in \thmref{thm:grav-fr} and for the propagator in \thmref{thm:grav-prop}. In addition, we provide the graviton-ghost Feynman rules for vertices in \thmref{thm:ghost-fr} and for the propagator in \thmref{thm:ghost-prop}. Additionally, we state the graviton and graviton-ghost vertex Feynman rules up to valence four in \exref{exmp:FR}. Finally, we provide the gravity-matter Feynman rules for matter from the Standard Model with respect to the classification in \lemref{lem:matter-model-Lagrange-densities} in \thmref{thm:matter-fr}.
			\item \sectionref{sec:explicit_feynman_rules}: We first propose the symmetric (hermitian) ghost Lagrange densities associated to the de Donder and Lorenz gauge fixing conditions for (effective) Quantum General Relativity coupled to the Standard Model. Then we display the Feynman rules for (effective) Quantum General Relativity coupled to the Standard Model explicitly. This includes all gravity-matter propagators in \ssecref{ssec:gravity-matter-propagators} and all gravity-matter vertex Feynman rules for valence three in \ssecref{ssec:gravity-matter-vertices}. This generalizes previous works in that we provide the graviton-gluon vertex Feynman rule with general gluon gauge parameter, the graviton-gluon-ghost vertex Feynman rule at all and the graviton-ghost vertex Feynman rule at all.
			\item \sectionref{sec:longitudinal_and_transversal_projections}: We introduce the notion of an \emph{optimal gauge fixing} in \defnref{defn:optimal-gauge-fixing} as a gauge fixing that acts only to the vertical (i.e.\ gauge) degrees of freedom. In particular, we then show that the Lorenz gauge fixing for Quantum Yang--Mills theory as well as the de Donder gauge fixing for (effective) Quantum General Relativity are both optimal in \thmsaref{thm:feynman-rule-gluon-propagator-lt-decomposition}{thm:feynman-rule-graviton-propagator-lt-decomposition}, respectively. Then we recall the transversal structure of Quantum Yang--Mills theory in \defnref{defn:qym-transversal-structure} and introduce the transversal structure of (effective) Quantum General Relativity in \defnref{defn:qgr-transversal-structure}. This includes the respective longitudinal, identical and transversal projection operators, as well as corresponding metrics to raise and lower the appearing Lorentz indices. We then recall several immediate identities for the transversal structure of Quantum Yang--Mills theory and introduce their involved counterparts in (effective) Quantum General Relativity: In particular, this includes the decomposition of the longitudinal projection operators in \lemsaref{lem:g_and_l_inverse_decomposition_gl_qym}{lem:g_and_l_inverse_decomposition_gl_qgr}. Then we show that the provided longitudinal, identical and transversal projection operators are indeed projectors in \propsaref{prop:qym-transversal-structure}{prop:qgr-transversal-structure}. Next we show that the decomposition tensors of the longitudinal projection operators are furthermore eigentensors of the longitudinal, identical and transversal projection operators in \colsaref{col:gl-eigenvectors-lit-qym}{col:gl-eigentensors-lit-qgr}. Thereafter, we study the action of the corresponding metrics on said tensors in \lemsaref{lem:identities_tensors_qym}{lem:identities_tensors_qgr} and \colsaref{col:l-tensor-gg-ll-qym}{col:l-tensor-gg-ll-qgr}. This allows us then finally to simplify the gluon and graviton propagator and relate them to their ghosts via a gauge fixing projection in \thmsaref{thm:feynman-rule-gluon-propagator-lt-decomposition}{thm:feynman-rule-graviton-propagator-lt-decomposition}. Additionally, we provide all cancellation identities for the gluon and graviton vertex Feynman rules of valence three in \thmsaref{thm:three-valent-contraction-identities-qym}{thm:three-valent-contraction-identities-qgr}, as well as for the corresponding couplings to matter from the Standard Model in \thmsaref{thm:three-valent-contraction-identities-qym-with-matter}{thm:three-valent-contraction-identities-qgr-with-matter}. Finally, we comment on the differences of the two most prominent definitions of the graviton field: The metric decomposition and the metric density decomposition in \remref{rem:md_vs_mdd}.
		\end{itemize}
	\item[\secref{chp:conclusion}:]
		\begin{itemize}
			\item \sectionref{sec:outlook}: We propose three follow-up projects where we have already achieved first promising results: First we discuss a generalization of Wigner's classification of elementary particles to linearized gravity. We discuss this construction for simple spacetimes with a fixed trivialization map in \defnsaref{defn:linear-structure}{defn:pullback-poincare-group}. With that, it is then possible to classify gravitons and matter particles via their mass and helicity or spin. This is especially useful in the operator-based approach to Quantum Field Theory, as it allows to construct the corresponding Fock space of particle configurations. In particular, we show that the choice of the trivialization map has no physical relevance for diffeomorphism-invariant theories in \lemref{lem:trivialization-map-diffeomorphism-invariant-theory}. However, we also comment that this is no longer true for the respective quantized theory. As the second project we discuss how cancellation identities can be implemented via a novel version of Feynman graph cohomology on the algebra of Feynman graphs. In addition, we argue that the compatibility of these cancellation identities with renormalization is guaranteed if the Feynman graph differential turns the renormalization Hopf algebra into a differential-graded Hopf algebra, cf.\ \eqnsref{eqns:dg-hopf-algebra}. In particular, we introduce the \emph{differential-graded renormalization Hopf algebra} in \defnref{defn:differential-graded_renormalization_hopf_algebra}. This also resembles BRST cohomology, where the BRST differential turns both the spacetime-matter bundle into a differential-graded supermanifold and the sheaf of particle fields into a differential-graded superalgebra. We then show that the Feynman graph cohomology of Kreimer et al.\ does not satisfy the mentioned criteria out of the box in \lemref{lem:edge-and-cycle-markings-not-compatible}. Nevertheless, we propose two solutions to establish said compatibility. Additionally, we also introduce the \emph{quantum gauge symmetry differential} in \defnref{defn:qgs_differential} and argue that this differential is the perturbative version of the BRST operator. We then claim that the compatibility of the Feynman rules with the action of this differential on the quantum gauge symmetry ideal is equivalent to the transversality of the theory in \conjref{conj:transversality-via-qgs-differential}. As the third and last project, we propose a study on the physical difference between the two most prominent definitions of the graviton field: The metric decomposition and the metric density decomposition. To this end, we suggest the definition of a \emph{graviton density family} as a homotopy in the tensor density weight in \defnref{defn:graviton-density-family}. The aim is then to calculate the corresponding Feynman rules and study the dependence of the resulting Feynman integrals on said homotopy parameter. First calculations suggest that the perturbative expansion depends only on the tensor density weight of the external legs, which we claim in \conjref{conj:perturbative-expansion-essentially-independent-of-tensordensityweight}. In particular, this continues the discussion started in \remref{rem:md_vs_mdd}. Additionally, we suggest a third definition of the graviton field as a special case thereof: We propose a graviton field with tensor density weight \(w = \textfrac{1}{2}\). This then leads to symmetric longitudinal and transversal projection operators for the graviton propagator. Thus, this definition of the graviton field leads to a symmetric transversal structure, as it is the case in Quantum Yang--Mills theory, cf.\ \defnsaref{defn:qym-transversal-structure}{defn:qgr-transversal-structure}. Finally, we conclude all three projects with open questions which we aim to address.
		\end{itemize}
\end{itemize}

\chapter{Gauge theories and gravity} \label{chp:gauge_theories_and_gravity}

In this chapter, we discuss the geometry of gauge theories and gravity. This includes the geometry of spacetimes and particle fields as well as the Lagrange densities of (effective) Quantum General Relativity coupled to the Standard Model together with the BRST double complex of diffeomorphisms and gauge transformations.

\section{The geometry of spacetimes} \label{sec:the-geometry-of-spacetimes}

In this section, we describe the geometry of spacetimes. To this end, we first introduce our sign choices and further conventions. Then we recall Penrose's notion of an asymptotically simple and empty spacetime together with its properties. Asymptotically simple and empty spacetimes are also called asymptotically flat spacetimes and are characterized in that their metrics approach the Minkowski metric `at infinity'. Thus, using the correspondence between geometry and energy-matter distributions due to Einstein's field equations, asymptotically flat spacetimes describe universes with a finite matter distribution. We refer to \cite{Ashtekar,Friedrich} for excellent overview articles on these matters. In analogy with asymptotically simple spacetimes we then introduce the slightly stronger notion of a \emph{simple spacetime}. This then enables us to provide a global definition of the graviton field and a definition for the Fourier transformation on the sections of particle fields.\footnote{We remark that a global definition of the graviton field is possible if and only if the spacetime is parallelizable.} With that we define the graviton field as the deviation of the metric with respect to the Minkowski background metric. Finally, we comment on the diffeomorphism invariance of General Relativity. This includes the Lie group structure of the diffeomorphism group as a regular Lie group in the sense of Milnor, cf.\ \cite{Milnor,Schmeding_PhD}, and infinitesimal transformation properties. Additionally, we also comment on the viewpoint of General Relativity as a generalized gauge theory in the framework of Lie groupoids and Lie algebroids. Finally, we comment on the terminology of `Linearized General Relativity', which we use for the expansion of the Einstein--Hilbert Lagrange density with respect to the graviton coupling constant.

\enter

\begin{con}[Sign choices] \label{con:sign_choices}
	We use the sign-convention \((-++)\), as classified by \cite{Misner_Thorne_Wheeler}, i.e.:
	\begin{enumerate}
		\item Minkowski metric: \(\eta_{\mu \nu} = \begin{pmatrix} 1 & 0 & 0 & 0 \\ 0 & - 1 & 0 & 0 \\ 0 & 0 & - 1 & 0 \\ 0 & 0 & 0 & - 1 \end{pmatrix}_{\mu \nu}\)
		\item Riemann tensor: \(\tensor{R}{^\rho _\sigma _\mu _\nu} = \partial_\mu \tensor{\Gamma}{^\rho _\nu _\sigma} - \partial_\nu \tensor{\Gamma}{^\rho _\mu _\sigma} + \tensor{\Gamma}{^\rho _\mu _\lambda} \tensor{\Gamma}{^\lambda _\nu _\sigma} - \tensor{\Gamma}{^\rho _\nu _\lambda} \tensor{\Gamma}{^\lambda _\mu _\sigma}\)
		\item Einstein field equations: \(G_{\mu \nu} = \kappa T_{\mu \nu}\)
	\end{enumerate}
	Additionally we use the plus-signed Clifford relation, i.e.\ \(\left \{ \gamma_m , \gamma_n \right \} = 2 \eta_{m n} \id_{\Sigma M}\), cf.\ the discussion in \remref{rem:signature_metric_clifford_relation}.
\end{con}

\enter

\begin{defn}[Spacetime] \label{def:spacetime}
	Let \((M,\met)\) be a \(d\)-dimensional Lorentzian manifold. We call \((M,\met)\) a spacetime, if it is smooth, connected and time-orientable.
\end{defn}

\enter

\begin{defn}[Asymptotically simple (and empty) spacetime] \label{defn:asymptotically-simple-empty-spacetime}
	Let \((M,\met)\) be an oriented and causal spacetime. We call \((M,\met)\) an asymptotically simple spacetime, if it admits a conformal extension \(\big ( \conext,\conmet \big )\) in the sense of Penrose \cite{Penrose_1,Penrose_2,Penrose_3,Penrose_4}: That is, if there exists an embedding \(\iota \colon M \hookrightarrow \conext\) and a smooth function \(\varsigma \in C^\infty \big ( \conext \big )\), such that:
	\begin{enumerate}
		\item \(\conext\) is a manifold with interior \(\iota \left ( M \right )\) and boundary \(\scri\), i.e.\ \(\conext \cong \iota \left ( M \right ) \sqcup \scri\)
		\item \(\eval[2]{\varsigma}_{\iota \left ( M \right )} > 0\), \(\eval[2]{\varsigma}_{\scri} \equiv 0\) and \(\eval[2]{\dif \varsigma}_{\scri} \not \equiv 0\); additionally \(\iota_* \met \equiv \varsigma^2 \conmet\)
		\item Each null geodesic of \(\big ( \conext,\conmet \big )\) has two distinct endpoints on \(\scri\)
	\end{enumerate}
	We call \((M,\met)\) an asymptotically simple and empty spacetime, if additionally:\footnote{This condition can be modified to allow electromagnetic radiation near \(\scri\). We remark that asymptotically simple and empty spacetimes are also called asymptotically flat.}
	\begin{enumerate}
		\setcounter{enumi}{3}
		\item \(\eval[2]{\left ( R_{\mu \nu} \right )}_{\iota^{-1} ( \widetilde{O} )} \equiv 0\), where \(\widetilde{O} \subset \conext\) is an open neighborhood of \(\scri \subset \conext\)
	\end{enumerate}
\end{defn}

\enter

\begin{prop} \label{prop:ase-parallelizable}
	Let \((M,\met)\) be an asymptotically simple and empty spacetime. Then \((M,\met)\) is globally hyperbolic and thus parallelizable in four dimensions of spacetime.
\end{prop}

\begin{proof}
	The first part of the statement, i.e.\ that \((M,\met)\) is globally hyperbolic, is a classical result due to Ellis and Hawking \cite[Proposition 6.9.2]{Hawking_Ellis}. We conclude the second part, i.e.\ that \((M,\met)\) is parallelizable, by noting that we have additionally assumed spacetimes to be four-dimensional: Thus, being globally hyperbolic, there is a well-defined three-dimensional space-submanifold, which therefore is parallelizable as it is orientable by assumption.
\end{proof}

\enter

\begin{col}
	Any asymptotically simple and empty four-dimensional spacetime is spin.
\end{col}

\begin{proof}
	This follows immediately from \propref{prop:ase-parallelizable}, as parallelizable manifolds are trivially spin.
\end{proof}

\enter

\begin{prop}[\cite{Geroch_1,Geroch_2}] \label{prop:parallelizable-spin}
	A four-dimensional spacetime \((M,\met)\) is spin if and only if it is globally hyperbolic. Equivalently, \((M,\met)\) is spin if and only if it is parallelizable.
\end{prop}

\begin{proof}
	These are two classical results by Geroch.
\end{proof}

\enter

\begin{defn}[Simple spacetime] \label{defn:simple_spacetime}
	Let \((M,\met)\) be a spacetime. We call the triple \((M,\met,\trivmap)\) a simple spacetime, if \(M\) is diffeomorphic to the Minkowski spacetime \(\bbM\) and \(\trivmap \colon M \to \bbM\) is a fixed such diffeomorphism (not necessarily an isometry), called the trivialization map. Furthermore, we use \(\trivmap\) to pushforward the metric \(\met\) to the Minkowski spacetime \(\bbM\) via \(g := \trivmap_* \met \in \Gamma \big ( \bbM, \operatorname{Sym}^2 T^* \bbM \big )\) to obtain an equivalence between the physical spacetime \((M,\met)\) and its background Minkowski spacetime \((\bbM,g)\) and to define the graviton field thereon, cf.\ \defnref{defn:md_and_gf}.
\end{defn}

\enter

\begin{ass} \label{ass:simple_spacetime}
	From now on, we assume spacetimes to be simple.
\end{ass}

\enter

\begin{rem} \label{rem:Minkowski_background}
	The rather restrictive setup of \assref{ass:simple_spacetime} is motivated by \propref{prop:ase-parallelizable} and \propref{prop:parallelizable-spin}: It is physically reasonable to consider asymptotically simple and empty spacetimes, as well as to demand a spin structure for fermionic particles. Thus, the spacetime \((M,\met)\) has the same asymptotic structure as the Minkowski spacetime \((\bbM,\eta)\) and is furthermore parallelizable. This implies that it is diffeomorphic to the Minkowski spacetime of the same dimension, modulo possible singularities. However, as we need the eigenvalues of the metric \(g\) to be bounded by \assref{ass:bdns_gf} for the following constructions, we exclude singularities in our setup. This assumption allows us to view particle fields, in particular the graviton field, as sections over Minkowski spacetime \(\sctn{\bbM,E}\), where \(\pi_E \colon E \to \bbM\) is a suitable vector bundle for the particle fields under consideration, cf.\ \defnref{defn:correspondence_Minkowski_spacetime}. In turn, this enables us to use Wigner's classification of elementary particles via irreducible representations of the Poincar\'{e} group \cite{Wigner}, which will be studied in \cite{Prinz_8}. Thus we can proceed as usual by constructing the Fock space to describe the quantum states of our corresponding Quantum Field Theory. Finally, this setup provides a well-defined Fourier transformation for particle sections, cf.\ \defnref{defn:fourier_transform}.
\end{rem}

\enter

\begin{defn}[Metric decomposition and graviton field] \label{defn:md_and_gf}
	Let \((M,\met,\trivmap)\) be a simple spacetime. Then we use the following metric decomposition on the background Minkowski spacetime \((\bbM,\eta)\)
	\begin{equation} \label{eqn:metric_decomposition}
		h_{\mu \nu} := \frac{1}{\gcoupling} \left ( g_{\mu \nu} - \eta_{\mu \nu} \right ) \iff g_{\mu \nu} \equiv \eta_{\mu \nu} + \gcoupling h_{\mu \nu} \, ,
	\end{equation}
	where \(\gcoupling := \sqrt{\kappa}\) is the graviton coupling constant (with \(\kappa := 8 \pi G\) the Einstein gravitational constant). Thus, the graviton field \(h_{\mu \nu}\) is given as a rescaled, symmetric \((0,2)\)-tensor field on the background Minkowski spacetime, i.e.\ as the section \(\gcoupling h \in \Gamma \big ( \bbM, \operatorname{Sym}^2 T^* \bbM \big )\).
\end{defn}

\enter

\begin{rem}
	Given the situation of \defnref{defn:md_and_gf}, the graviton field \(h\) depends directly on the choice of the trivialization map \(\trivmap\). It can be shown, however, that this dependence can be absorbed, if the theory is diffeomorphism-invariant \cite{Prinz_8}. Thus, this construction is in particular well-defined for Linearized General Relativity.
\end{rem}

\enter

\begin{ass} \label{ass:bdns_gf}
	Given the metric decomposition from \defnref{defn:md_and_gf}, we assume the following boundedness condition for the gravitational constant \(\gcoupling\) and the graviton field \(h_{\mu \nu}\):
	\begin{equation}
		\left | \gcoupling \right | \left \| h \right \|_{\max} := \left | \gcoupling \right | \max_{\lambda \in \operatorname{EW} \left ( h \right )} \left | \lambda \right | < 1 \, ,
	\end{equation}
	where \(\operatorname{EW} \left ( h \right )\) denotes the set of eigenvalues of \(h\). This will be relevant for preceding assertions to assure the convergence of series involving the graviton coupling constant \(\gcoupling\) and the graviton field \(h_{\mu \nu}\).
\end{ass}

\enter

\begin{defn}[Correspondence to Minkowski spacetime] \label{defn:correspondence_Minkowski_spacetime}
	Let \((M,\met,\trivmap)\) be a simple spacetime and \(\pi_E \colon E \to M\) a vector bundle for particle fields. Then we extend the vector bundle for particle fields via \(\left ( \trivmap \circ \pi_E \right ) \colon E \to \bbM\) to a vector bundle over the background Minkowski spacetime \(\bbM\), i.e.\ we have:
	\begin{equation}
	\begin{tikzcd}[row sep=huge]
		E \arrow[swap]{d}{\pi_E} \arrow[dashed]{dr}{\trivmap \circ \pi_E} & \\
		M \arrow[swap]{r}{\tau} & \bbM
	\end{tikzcd}
	\end{equation}
\end{defn}

\enter

\begin{ter}
Given a spacetime \((M, \gamma)\) and the Lagrange density of General Relativity \(\mathcal{L}_\text{GR}\) together with a metric decomposition \(\gamma_{\mu \nu} \equiv b_{\mu \nu} + \varkappa \theta_{\mu \nu}\), where \(b_{\mu \nu}\) is a fixed background metric on \(M\), \(\varkappa\) the gravitational coupling constant and \(\theta_{\mu \nu}\) the corresponding graviton field on \(M\). Then we refer to the expansion of \(\mathcal{L}_\text{GR}\) in powers of \(\varkappa\) as \emph{Linearized General Relativity}.\footnote{This slight abuse of terminology is common in the physics literature.} In the realm of this dissertation we assume spacetimes to be simple, and thus push forward the metric decomposition and the particle fields to the background Minkowski spacetime \(\bbM\), in order to apply Wigner's elementary particle classification and define a sensible Fourier transformation. Furthermore, we choose \(b_{\mu \nu} := \tau^* \eta_{\mu \nu}\), so that we expand the graviton field around the Minkowski metric \(\eta_{\mu \nu}\) on the background Minkowski spacetime \((\bbM, \eta)\). Thus, in the following we discuss general constructions on the spacetime \((M,\gamma)\) and switch to the background Minkowski spacetime \((\bbM, \eta)\) whenever we apply the linearization process.
\end{ter}

\enter

\begin{rem}
	Given the situation of \assref{ass:simple_spacetime}, the gravitational path integral then corresponds to an integral over the space of symmetric \((0,2)\)-tensor fields over the background Minkowski spacetime \(\bbM\). As the construction of such integral measures over function spaces is rather troublesome, we simply consider the \(\hbar \ll 0\) limit, where the Feynman graph expansion can be interpreted as a `perturbative definition' of the path integral. We refer to \cite{Hamber} for a more physical treatment.
\end{rem}

\enter

\begin{ass} \label{ass:diffeo_homotopic_to_identity}
	We assume from now on that diffeomorphisms are homotopic to the identity, i.e.\ \(\phi \in \diff\).
\end{ass}

\enter

\begin{rem}
	\assref{ass:diffeo_homotopic_to_identity} is motivated by the fact that diffeomorphisms homotopic to the identity are generated via the flows of compactly supported vector fields, \(X \in \vectc{M}\), and differ from the identity only on compactly supported domains. Thus, diffeomorphisms homotopic to the identity preserve the asymptotic structure of spacetimes. We remark that, different from finite dimensional Lie groups, the Lie exponential map
	\begin{align}
		\operatorname{exp} \, & : \quad \vectc{M} \to \diff
		\intertext{is no longer locally surjective, which leads to the notion of an evolution map}
		\operatorname{Evol} \, & : \quad C^\infty \left ( [0,1], \vectc{M} \right ) \to C^\infty \left ( [0,1], \diff \right )
	\end{align}
	that maps smooth curves in the Lie algebra to smooth curves in the corresponding Lie group. We refer to \cite{Schmeding_PhD} for further details.
\end{rem}

\enter

\begin{defn}[Transformation under (infinitesimal) diffeomorphisms] \label{defn:transformation_diffeo}
	Given the situation of \defnref{defn:md_and_gf} and \assref{ass:diffeo_homotopic_to_identity}, we define the action of diffeomorphisms \(\phi \in \diff\) on the graviton field via
	\begin{align}
		\left ( \trivmap \circ \phi \right )_* \left ( \gcoupling h \right ) & := \left ( \trivmap \circ \phi \right )_* g \, ,
		\intertext{such that the background Minkowski metric can be conveniently defined to be invariant, i.e.}
		\left ( \trivmap \circ \phi \right )_* \eta & := 0 \, ,
	\end{align}
	and on the other particle fields \(\particlefield \in \sctn{\bbM,E}\) as usual, i.e.\ via
	\begin{equation}
		\left ( \trivmap \circ \phi \right )_* \particlefield \, .
	\end{equation}
	In particular, the action of infinitesimal diffeomorphisms is given via the Lie derivative with respect the generating vector field \(X \in \vectc{\bbM}\), i.e.\
	\begin{align}
		\delta_X h_{\mu \nu} & \equiv \frac{1}{\gcoupling} \left ( \nabla^{(g)}_\mu X_\nu + \nabla^{(g)}_\nu X_\mu \right ) \, , \\
		\delta_X \eta_{\mu \nu} & \equiv 0
		\intertext{and}
		\delta_X \particlefield & \equiv \Lie_X \particlefield \, ,
	\end{align}
	where \(\nabla^{(g)}\) denotes the covariant derivative with respect to the connection \(\Gamma\) induced via \(g\) on \(\bbM\), i.e.\ via
	\begin{equation}
		\tensor{\Gamma}{^\rho _\mu _\nu} := \frac{1}{2} g^{\rho \sigma} \left ( \partial_\mu g_{\sigma \nu} + \partial_\nu g_{\mu \sigma} - \partial_\sigma g_{\mu \nu} \right ) \, .
	\end{equation}
	The concrete action will be considered further in \defnref{defn:diffeomorphism-group-and-group-of-gauge-transformations}.
\end{defn}

\enter

\begin{rem} \label{rem:Lie_groupoid_and_algebroid}
	Using the setup from \asssnref{ass:simple_spacetime}{ass:bdns_gf}{ass:diffeo_homotopic_to_identity} and \defnref{defn:md_and_gf}, \ref{defn:correspondence_Minkowski_spacetime}~and~\ref{defn:transformation_diffeo}, we can view Linearized General Relativity coupled to matter from the Standard Model as a `generalized gauge theory' on the background Minkowski spacetime: The `gauge group' is then given via the pushforward of diffeomorphisms homotopic to the identity by the trivialization map, i.e.\ \(\widetilde{\mathcal{D}} := \trivmap_* \diff \cong \diffbbM\). Furthermore, their infinitesimal actions are given via Lie derivatives with respect to compactly supported vector fields \(\vectc{\bbM}\). In particular, the right setting to study the gauge theoretic properties of such a theory is given via the Lie groupoid \(\big ( \widetilde{\mathcal{D}} \times \mathcal{B} \big ) \rightrightarrows \mathcal{B}\) over the background Minkowski spacetime-matter bundle \(\mathcal{B} := \bbM \times V\). Additionally, the action of infinitesimal diffeomorphisms is embedded into this picture via the corresponding Lie algebroid \(((\vectc{\bbM} \times \mathcal{B}) \to \mathcal{B}, [\cdot, \cdot ], \rho)\): More precisely, \([\cdot, \cdot ]\) is the Lie bracket on \(\vectc{\bbM}\) and \(\rho \colon (\vectc{\bbM} \times \mathcal{B}) \to T \mathcal{B}\) the anchor map. Then, as in the case of `ordinary gauge theories' --- that is gauge theories coming from a principle bundle structure --- the invariance of the theory under diffeomorphisms provides an obstacle for the calculation of the propagator. We solve this issue by introducing a linearized de Donder gauge fixing together with the corresponding ghost and antighost fields, \(\gravitonghost \in \sctn{\bbM, T^* [1,0] \bbM}\) and \(\overline{\gravitonghost} \in \sctn{\bbM, T [-1,0] \bbM}\), respectively. This viewpoint will be elaborated in \cite{Prinz_8}.
\end{rem}

\enter

\begin{ter}[Linearized General Relativity]
	Given General Relativity modeled on a simple spacetime \((M, \gamma, \tau)\) together with a metric decomposition \(g_{\mu \nu} \equiv h_{\mu \nu} + \varkappa \beta_{\mu \nu}\) on \(\bbM\). Then we use the term \emph{Linearized General Relativity} by slight abuse of terminology not only for linear monomials, but for the complete expansion of the Lagrange density in powers of \(\varkappa\). This is common usage in the physics literature. Additionally, we will not distinguish in the following between the spacetime \(M\), the Minkowski background spacetime \(\bbM\) as well as the corresponding (metric) tensors \(\gamma, \theta, b \in \sctnbig{M, \operatorname{Sym}^2 \left ( T^* M \right ) \!}\) and their respective pushforwards \(g, h, \beta \in \sctnbig{\bbM, \operatorname{Sym}^2 \left ( T^* \bbM \right ) \!}\) via \(\tau\). We refer to \sectionref{sec:linearized_gravity_and_further_preparations} for a detailed derivation thereof.
\end{ter}

\section{The geometry of particle fields} \label{sec:the-geometry-of-particle-fields}

In this section, we discuss the geometry of particle fields. To this end, we first define graded supermanifolds together with a supercommutator that turns the module of super vector fields into a Lie superalgebra. Then we discuss homological and cohomological vector fields as odd supercommuting vector fields of degree \(1\) or \(-1\), respectively. With that, we define the spacetime-matter vector bundle as the globally trivial \(\mathbb{Z}^2\)-graded super vector bundle that describes the particle fields of (effective) Quantum General Relativity coupled to the Standard Model as its sheaf of sections. Then we discuss again the diffeomorphism group as well as the group of gauge transformations together with their actions. Furthermore, we discuss metrics, (inverse) vielbein fields, connections, the Clifford multiplication and relation, the (twisted) Dirac operator, curvature tensors, the Riemannian and Minkowskian volume forms and the Fourier transformation for particle fields.

\enter

\begin{defn}[\(\mathbb{Z}^2\)-graded supermanifold] \label{defn:ztwo-graded-supermanifold}
Let \(\mathcal{M}\) be a topological manifold. We call \(\mathcal{M}\) a \(\mathbb{Z}^2\)-graded supermanifold, if it is isomorphic to a vector bundle \(\pi \colon \mathcal{M} \to M\) that splits into a direct sum bundle such that the following diagram commutes
\begin{equation}
\begin{tikzcd}[row sep=huge]
	\mathcal{M} \arrow[swap]{dr}{\pi} \arrow{rr}{\sim} & & \displaystyle \bigoplus_{(i,j) \in \mathbb{Z}^2} \mathcal{M}_{(i,j)} \arrow{dl}{\tilde{\pi}} \\
	& M &
\end{tikzcd} \, ,
\end{equation}
where \((i,j) \in \mathbb{Z}^2\) denotes the degree of the subbundles and \(\mathcal{M}_{(0,0)} \cong \{0\}\), i.e.\ the degree \((0,0)\) is concentrated in the so-called body \(M\). We call the first integer \(i\) the graviton-ghost degree, the second integer \(j\) the gauge ghost degree and their sum \(k := i + j\) the total ghost degree. Additionally, we assume that the grading is compatible with the super structure of \(\mathcal{M}\), which means that the parity of a subbundle is given via
\begin{equation}
	p \equiv i + j \quad \text{Mod 2} \, , \label{eqn:compatibility_grading_super_structure}
\end{equation}
where \(0 \in \mathbb{Z}_2\) denotes even coordinates and \(1 \in \mathbb{Z}_2\) denotes odd coordinates. Concretely, on the level of graded super functions \(\mathcal{C} \left ( \mathcal{U} \right )\) for \(\mathcal{U} \subseteq \mathcal{M}\) this means that
\begin{equation}
	\mathcal{C} \big ( \mathcal{U}_{(i,j)} \big ) \cong \begin{cases} C^\infty \big ( \mathcal{U}_{(0,0)} \big ) & \text{if \((i,j) = (0,0)\)} \\ \mathcal{S} \big ( \mathcal{U}_{(i,j)} \big ) & \text{if \(p = 0\)} \\ \mathcal{A} \big ( \mathcal{U}_{(i,j)} \big ) & \text{if \(p = 1\)} \end{cases} \, ,
\end{equation}
where \(\mathcal{U}_{(i,j)} \subseteq \mathcal{M}_{(i,j)}\) is an open subset, \(C^\infty \big ( \mathcal{U}_{(0,0)} \big )\) denotes real smooth functions on \(U \subseteq M\), \(\mathcal{S} \big ( \mathcal{U}_{(i,j)} \big )\) denotes symmetric formal power series and \(\mathcal{A} \big ( \mathcal{U}_{(i,j)} \big )\) denotes antisymmetric formal power series. Finally, we define the grade shift via
\begin{equation}
	\mathcal{M}_{(i,j)} [m,n] := \mathcal{M}_{(i+m,j+n)} \, ,
\end{equation}
which additionally implies a potential shift in parity according to \eqnref{eqn:compatibility_grading_super_structure}. We refer to \cite{Voronov} for more details in this direction.
\end{defn}

\enter

\begin{defn}[Supercommutator] \label{defn:supercommutator}
	Let \(\mathcal{M}\) be a supermanifold and \(\boldsymbol{X}_1, \boldsymbol{X}_2 \in \mathfrak{X} \left ( \mathcal{M} \right )\) be two super vector fields of distinct parity \(p_1, p_2 \in \mathbb{Z}_2\). Then we introduce the supercommutator as follows:
	\begin{equation}
		\commutatorbig{\boldsymbol{X}_1}{\boldsymbol{X}_2} := \boldsymbol{X}_1 \left ( \boldsymbol{X}_2 \right ) - \left ( -1 \right )^{p_1 p_2} \boldsymbol{X}_2 \left ( \boldsymbol{X}_1 \right )
	\end{equation}
	This turns the module \((\mathfrak{X} \left ( \mathcal{M} \right ), \left [ \cdot , \cdot \right ])\) into a Lie superalgebra.
\end{defn}

\enter

\begin{ass}
	In the following, we assume that the grading is compatible with the super structure in the sense of \eqnref{eqn:compatibility_grading_super_structure}.
\end{ass}

\enter

\begin{defn}[Homological and cohomological vector fields]
	Let \(\mathcal{M}\) be a \(\mathbb{Z}\)-graded supermanifold with compatible super structure. Then we denote the subspace of pure super vector fields \(\boldsymbol{X}^\mu\) with degree \(z \in \mathbb{Z}\) by \(\mathfrak{X}_{(z)} \left ( \mathcal{M} \right )\). Then an odd vector field \(\Xi \in \mathfrak{X} \left ( \mathcal{M} \right )\) with the property
	\begin{equation}
	\commutatorbig{\Xi}{\Xi} \equiv 2 \, \Xi^2 \equiv 0
	\end{equation}
	is called homological if \(\Xi \in \mathfrak{X}_{(-1)} \left ( \mathcal{M} \right )\) and cohomological if \(\Xi \in \mathfrak{X}_{(1)} \left ( \mathcal{M} \right )\). This turns \((\mathcal{C}_\bullet \left ( \mathcal{M} \right ),\Xi)\) into a chain complex and the pair \((\mathcal{M}, \Xi)\) is called a differential-graded manifold.
\end{defn}

\enter

\begin{exmp}
	Let \(M\) be a manifold with \(\Omega^\bullet \left ( M \right )\) its sheaf of differential forms. Let furthermore \(\mathcal{M} := T[1]M\) denote its degree shifted tangent bundle. Then we can identify \(\mathcal{M} \cong \Omega^\bullet \left ( M \right )\) and obtain a cohomological vector field via the de Rham differential \(\mathrm{d} \in \mathfrak{X}_{(1)} \left ( \mathcal{M} \right )\) and a homological vector field via the de Rham codifferential \(\delta \in \mathfrak{X}_{(-1)} \left ( \mathcal{M} \right )\).
\end{exmp}

\enter

\begin{defn}[Spacetime-matter bundle] \label{defn:spacetime-matter_bundle}
Let \(( M , \gamma )\) be a simple spacetime. Then we define the spacetime-matter bundle of (effective) Quantum General Relativity coupled to the Standard Model as the globally trivial \(\mathbb{Z}^2\)-graded super bundle \(\beta_\Q \colon \mathcal{B}_\Q \to M\), where \(\mathcal{B}_\Q := M \times_M \mathcal{V}_\Q\) is the fiber product over \(M\) with
\begin{equation}
\begin{split}
	\mathcal{V}_\Q & := \operatorname{LorMet} \left ( M \right ) \times \left ( T^* M \otimes_\mathbb{R} E \right ) \times \left ( G \times_\rho \left ( H^{(i)} \oplus \Sigma M^{\oplus j} \right ) \right ) \\ & \phantom{:=} \times \operatorname{Conn} \left ( M, \mathfrak{g} \right ) \times \Big ( T^* [1,0] M \oplus T [-1,0] M \oplus T M \Big ) \times \Big ( \mathfrak{g} [0,1] \oplus \mathfrak{g}^* [0,-1] \oplus \mathfrak{g}^* \Big ) \, , \label{eqn:spacetime-matter_bundle}
\end{split}
\end{equation}
where \(\operatorname{LorMet} \left ( M \right ) \subset \operatorname{Sym}^2_\mathbb{R} \left ( T^* M \right )\) is the vector bundle of Lorentzian metrics with signature \((1, d-1)\) and \(T^* M \otimes_\mathbb{R} E\) is the vector bundle of vielbein fields with \(E\) a real \(d\)-dimensional vector bundle. Furthermore, \(G \times_\rho \big ( H^{(i)} \oplus \Sigma M^{\oplus j} \big )\) is the fiber product with respect to the action \(\rho\) of the gauge group \(G\) on the Higgs bundle \(H^{(i)} := \mathbb{C}^i\) and the vector of spinor bundles \(\Sigma M^{\oplus j} \cong \mathbb{C}^{4j}\). Moreover, \(\operatorname{Conn} \left ( M, \mathfrak{g} \right ) \subset \Omega^1 \left ( M, \mathfrak{g} \right )\) denotes the vector bundle of local connection forms. Finally, \(T^* [1,0] M \oplus T [-1,0] M\) and \(\mathfrak{g} [0,1] \oplus \mathfrak{g}^* [0,-1]\) denote the degree-shifted vector bundles for graviton-ghosts and gauge ghosts, respectively, and the additional bundles \(T M\) and \(\mathfrak{g}^*\) are for the Lautrup--Nakanishi auxiliary fields. We equip the spacetime-matter bundle with metrics in \defnref{defn:metrics_spacetime-matter_bundle} which, in turn, naturally include the corresponding dual bundles.\footnote{We denote the dual bundles via the asterisk, \(*\), except for the spinor bundle, twisted spinor bundle and the parity shifted Lie algebra and dual Lie algebra bundles of the gauge groups for which we use the overline, \(\overline{\phantom{*}}\).} Additionally, we also equip it with connections in \defnref{defn:connections_spacetime-matter_bundle} such that we have a notion of curvature, cf.\ \defnref{defn:curvatures_spacetime-matter_bundle}.
\end{defn}

\vspace{\baselineskip}

\begin{rem}
The global triviality of the spacetime-matter bundle in \defnref{defn:spacetime-matter_bundle} is motivated by the following facts: The tangent bundle \(TM\) and the spinor bundle \(\Sigma M\) are globally trivial since simple spacetimes are defined to be diffeomorphic to the Minkowski spacetime and are thus in particular parallelizable, cf.\ \defnref{defn:simple_spacetime}. Furthermore, the vector bundle \(E\) is chosen to be globally trivial for the definition of vielbeins and inverse vielbeins, cf.\ \defnref{defn:vielbeins_inverse_vielbeins}. Moreover, the \(G\)-principle bundle is chosen to be globally trivial to allow for global sections and avoid instanton solutions. Finally, the degree shifted Lie algebra and dual Lie algebra bundles of the gauge groups \(T [1,0] M \oplus T^* [-1,0] M\) and \(\mathfrak{g} [0,1] \oplus \mathfrak{g} [0,-1]\) are globally trivial because of the global triviality of the bundles \(TM\) and \(G\).
\end{rem}

\enter

\begin{defn}[Sheaf of particle fields] \label{defn:sheaf-of-particle-fields}
Let \(( M , \gamma )\) be a simple spacetime with topology \(\mathcal{T}_M\) and \(\beta_\Q \colon \mathcal{B}_\Q \to M\) the spacetime-matter bundle from \defnref{defn:spacetime-matter_bundle}. Then we define the sheaf of particle fields via
\begin{equation}
	\mathcal{F}_\Q \, : \quad \mathcal{T}_M \to \Gamma \left ( M, \mathcal{B}_\Q \right ) \, , \quad U \mapsto \Gamma \left ( U, B \right ) \, ,
\end{equation}
where \(B \subset \mathcal{B}_\Q\) is one of the subbundles from \eqnref{eqn:spacetime-matter_bundle}. More precisely, we consider the following particle fields:
\begin{itemize}
\item Lorentzian metrics \(\gamma \in \operatorname{LorMet} \left ( M \right )\)
\item Graviton fields \(h \in \operatorname{Grav} \left ( M \right )\)
\item Vielbein fields \(e \in \sctnbig{M, T^* M \otimes_\mathbb{R} E}\)
\item Vector of \(2i\) Higgs and Goldstone fields \(\Phi \in \sctnbig{M, H^{(i)}}\)
\item Vector of \(j\) fermion fields \(\Psi \in \sctnbig{M, \Sigma M^{\oplus j}}\)
\item Gauge boson fields \(\imaginary \mathrm{g} A \in \operatorname{Conn} \left ( M, \mathfrak{g} \right )\)
\item Graviton-ghost fields \(C \in \sctnbig{M, T^* [1,0] M}\)
\item Graviton-antighost fields \(\overline{C} \in \sctnbig{M, T [-1,0] M}\)
\item Gauge ghost fields \(c \in \sctnbig{M, \mathfrak{g} [0,1]}\)
\item Gauge antighost fields \(\overline{c} \in \sctnbig{M, \mathfrak{g}^* [0,-1]}\)
\item Graviton-Lautrup--Nakanishi auxiliary fields \(B \in \sctnbig{M, TM}\)
\item Gauge Lautrup--Nakanishi auxiliary fields \(b \in \sctnbig{M, \mathfrak{g}^*}\)
\end{itemize}
\end{defn}

\enter

\begin{defn}[Diffeomorphism group and group of gauge transformations] \label{defn:diffeomorphism-group-and-group-of-gauge-transformations}
	Given the situation of \defnref{defn:transformation_diffeo}, the physical theories that we are studying are invariant under the action of two groups: The diffeomorphism group homotopic to the identity \(\mathcal{D} := \diff\) and under the group of gauge transformations \(\mathcal{G} := \Gamma \left ( M, M \times G \right )\), where \(G \cong U(1) \times \widetilde{G}\) is the gauge group with \(\widetilde{G}\) compact and semisimple, cf.\ \cite{Saller,Baez,Baez_Huerta,Tong} for a discussion on the Standard Model gauge group. The diffeomorphism group homotopic to the identity acts via
	\begin{subequations}
	\begin{align}
		\varrho & \, : \quad \mathcal{D} \times \mathcal{B}_\Q \to \mathcal{B}_\Q \quad \left ( \phi, \varphi \right ) \mapsto \phi_* \varphi \, ,
		\intertext{where \(\varrho\) acts naturally on \(M\) and via pushforward on the corresponding particle bundles.\footnotemark \phantom{ } Furthermore, the group of gauge transformations acts via}
		\rho & \, : \quad \mathcal{G} \times \mathcal{B}_\Q \to \mathcal{B}_\Q \quad \left ( \gamma, \varphi \right ) \mapsto \gamma \cdot \varphi \, ,
	\end{align}
	\end{subequations}
	\footnotetext{The action on the spinor bundle is more involved, as its construction depends crucially on the metric \(g\). We refer to \cite{Mueller_Nowaczyk} for an explicit construction.}%
	\noindent where \(\rho\) acts via the matrix representation on the vectors of Higgs and spinor fields. Additionally, we also consider the action of infinitesimal diffeomorphisms via
	\begin{subequations}
	\begin{align}
		\boldsymbol{\varrho} & \, : \quad \mathfrak{D} \times \mathcal{B}_\Q \to \mathcal{B}_\Q \quad \left ( X, \varphi \right ) \mapsto \Lie_X \varphi \, ,
		\intertext{where \(\mathfrak{D} := \mathfrak{diff} \left ( M \right ) \cong \mathfrak{X}_\text{c} \left ( M \right )\) is the Lie algebra of compactly supported vector fields and \(\Lie_X\) denotes the Lie derivative of the geodesic exponential map. Moreover, we also consider the action of infinitesimal gauge transformations via}
		\boldsymbol{\rho} & \, : \quad \mathfrak{G} \times \mathcal{B}_\Q \to \mathcal{B}_\Q \quad \left ( Z, \varphi \right ) \mapsto \ell_Z \varphi \, ,
	\end{align}
	\end{subequations}
	where \(\mathfrak{G} := \Gamma \left ( M, M \times \mathfrak{g} \right )\) is the Lie algebra of \(\mathfrak{g}\)-valued vector fields, with \(\mathfrak{g}\) the Lie algebra of the gauge group \(G\), and \(\ell_Z\) denotes the Lie derivative of the Lie exponential map.
\end{defn}

\vspace{\baselineskip}

\begin{defn}[Metrics on the spacetime-matter bundle] \label{defn:metrics_spacetime-matter_bundle}
We consider the following metrics on the spacetime-matter bundle \(\mathcal{B}_\Q\): On the tangent bundle \(TM\) we consider the Lorentzian metric \(g\) with West coast (``mostly minus'') signature, mapping vector fields \(X_1, X_2 \in \Gamma \left ( M, TM \right )\) to the real number
\begin{equation}
	\left \langle X_1 , X_2 \right \rangle_{TM} := g_{\mu \nu} X_1^\mu X_2^\nu \in \mathbb{R} \, .
\end{equation}
Furthermore, on the real vector bundle \(E\) we consider the constant Minkowski metric \(\eta\) with West coast (``mostly minus'') signature, mapping vector fields \(Y_1, Y_2 \in \Gamma \left ( M, E \right )\) to the real number
\begin{equation}
	\left \langle Y_1 , Y_2 \right \rangle_E := \eta_{m n} Y_1^m Y_2^n \in \mathbb{R} \, .
\end{equation}
Moreover, on the Higgs bundle \(H^{(i)}\) we use the Hermitian metric, mapping vectors of scalar fields \(\Phi_1, \Phi_2 \in \Gamma \big ( M, H^{(i)} \big )\) to the complex number
\begin{equation}
	\left \langle \Phi_1 , \Phi_2 \right \rangle_{H^{(i)}} := \Phi_1^\dagger \Phi_2 \in \mathbb{C} \, ,
\end{equation}
where \(\dagger\) denotes Hermitian conjugation. We extend this metric \(G\)-invariantly to the twisted Higgs bundle \(G \times_\rho H^{(i)}\), i.e.\ such that
\begin{equation}
	\left \langle \rho \left ( g \right ) \Phi_1 , \rho \left ( g \right ) \Phi_2 \right \rangle_{G \times_\rho H^{(i)}} \equiv \left \langle \Phi_1 , \Phi_2 \right \rangle_{G \times_\rho H^{(i)}}
\end{equation}
holds for all \(g \in \mathcal{G}\).
Finally, on the vector of spinor bundles \(\Sigma M^{\oplus j}\) we use the Hermitian metric together with the Clifford multiplication by the unit timelike vector field, mapping vectors of spinor fields \(\Psi_1, \Psi_2 \in \Gamma \left ( M, \Sigma M^{\oplus j} \right )\) to the complex number
\begin{equation}
	\left \langle \Psi_1 , \Psi_2 \right \rangle_{\Sigma M^{\oplus j}} := \overline{\Psi_1} \Psi_2 \in \mathbb{C} \, ,
\end{equation}
where we have set \(\overline{\Psi} := e_0^m \psi^\dagger \boldsymbol{\gamma}_m\), with the curved unit timelike vector field components of a vielbein \(e_0^m\),\footnote{We remark that since we consider matter-compatible spacetimes to be diffeomorphic to the Minkowski spacetime and thus in particular diffeomorphic to a globally hyperbolic spacetime, it is also possible to consider charts in which \(e^m_0 \equiv \delta^m_0\) such that in particular \(e^m_0 \gamma_m \equiv \gamma_0\), and some references use this implicitly, e.g.\ \cite{Choi_Shim_Song}. However it should be noted that then the theory is only invariant under diffeomorphisms which preserve global hyperbolicity. \label{foot:choi_shim_song_gh-manifold}} being introduced in \defnref{defn:vielbeins_inverse_vielbeins}, and the Clifford multiplication \(\boldsymbol{\gamma}_m\), being introduced in \defnref{defn:clifford_multiplication}. We extend this metric \(G\)-invariantly to the twisted spinor bundle \(G \times_\rho \Sigma M^{\oplus j}\), i.e.\ such that
\begin{equation}
	\left \langle \rho \left ( g \right ) \Psi_1 , \rho \left ( g \right ) \Psi_2 \right \rangle_{G \times_\rho \Sigma M^{\oplus j}} \equiv \left \langle \Psi_1 , \Psi_2 \right \rangle_{G \times_\rho \Sigma M^{\oplus j}}
\end{equation}
holds for all \(g \in \mathcal{G}\).
\end{defn}

\vspace{\baselineskip}

\begin{defn}[Vielbeins and inverse vielbeins] \label{defn:vielbeins_inverse_vielbeins}
Let \(\mathcal{B}_\text{QGR-QED}\) be the spacetime-matter bundle. Then we can define global vector bundle isomorphisms \(e \in \Gamma \left ( M, T^*M \otimes_\mathbb{R} E \right )\), called vielbeins, such that
\begin{subequations}
\begin{equation}
	g_{\mu \nu} = \eta_{m n} e^m_\mu e^n_\nu \, . \label{eqn:vielbeins}
\end{equation}
Furthermore, we can define global inverse vector bundle isomorphisms \(e^* \in \Gamma \left ( M, TM \otimes_\mathbb{R} E^* \right )\), called inverse vielbeins, such that\footnote{We omit the asterisk, \(*\), for inverse vielbeins \(e^*\) when the abstract index notation is used because of \eqnref{eqn:transformation_vielbeins_inverse_vielbeins}.}
\begin{equation}
	\eta_{m n} = g_{\mu \nu} e^\mu_m e^\nu_n \, . \label{eqn:inverse_vielbeins}
\end{equation}
\end{subequations}
Greek indices, here \(\mu\) and \(\nu\), belong to the tangent bundle \(TM\) and are referred to as curved indices. Thus, they are raised and lowered using the usual metric \(g_{\mu \nu}\) and its inverse \(g^{\mu \nu}\). Latin indices, here \(m\) and \(n\), belong to the vector bundle \(E\) and are referred to as flat indices. Thus, they are raised and lowered using the Minkowski metric \(\eta_{m n}\) and its inverse \(\eta^{m n}\). Therefore, inverse vielbeins are related to vielbeins via
\begin{equation}
	e^\mu_m = g^{\mu \nu} \eta_{m n} e^n_\nu \, . \label{eqn:transformation_vielbeins_inverse_vielbeins}
\end{equation}
Moreover, notice that \eqnsaref{eqn:vielbeins}{eqn:inverse_vielbeins} are equivalent to
\begin{equation}
	g_{\mu \nu} e^\nu_n = \eta_{m n} e^m_\mu = e_{\mu n} \, , \label{eqn:inverse_vielbeins_eigenvalue_equation}
\end{equation}
which states that, when suppressing the Einstein summation convention on the flat indices and viewing them as numbers, inverse vielbeins \(\left \{ e^\mu_m \right \}_{m \in \left \{ 1,2,3,4 \right \}}\) are the set of the \(4\) eigenvector fields of the metric \(g_{\mu \nu}\) with eigenvalues \(\eta_{m m} \in \left \{ \pm 1 \right \}\), which are normalized to unit length
\begin{equation}
	\left \| e_m \right \|_g := \sqrt{\left \vert g_{\mu \nu} e^\mu_m e^\nu_m \right \vert} = \sqrt{\left \vert \eta_{m m} \right \vert} = 1 \, .
\end{equation}
Finally, we remark that the global definition of vielbeins and inverse vielbeins is only possible for parallelizable manifolds, such as the simple spacetimes of \defnref{defn:simple_spacetime}. This is because the tangent frame bundle \(FM\) allows for a global section if and only if the manifold is parallelizable, i.e.\ it is then possible to choose a global coordinate system on \(TM\) that consists of normalized eigenvector fields of the metric \(g_{\mu \nu}\).
\end{defn}

\vspace{\baselineskip}

\begin{rem} \label{rem:ambiguity_vielbeins_inverse_vielbeins}
Given the situation of \defnref{defn:vielbeins_inverse_vielbeins}, notice that the definitions of vielbeins and inverse vielbeins are not unique, since for any local Lorentz transformation acting on \(E^*\), i.e.\ a section of the corresponding orthogonal frame bundle \(\Lambda \in \Gamma \left ( U, F_O E^* \right )
\) acting via the standard representation on \(E^*\), we have
\begin{equation}
\begin{split}
	g_{\mu \nu} e^\mu_m e^\nu_n & = \eta_{m n}\\
	& = \eta_{r s} \tensor{\Lambda}{^r _m} \tensor{\Lambda}{^s _n}\\
	& = g_{\mu \nu} e^\mu_r e^\nu_s \tensor{\Lambda}{^r _m} \tensor{\Lambda}{^s _n}\\
	& = g_{\mu \nu} \tilde{e}^\mu_m \tilde{e}^\nu_n \, .
\end{split}
\end{equation}
Here, we denoted the transformed inverse vielbeins via \(\tilde{e}^\mu_m := e^\mu_r \tensor{\Lambda}{^r _m}\). Obviously the same calculation also holds for vielbeins instead of inverse vielbeins with local Lorentz transformations acting on \(E\). This ambiguity will lead to the first term in the spin connection, \(e_\nu^n \left ( \partial_\mu e^\nu_l \right )\), cf.\ \eqnref{eqn:spin_connection}. In fact, the first term in the spin connection can be viewed as the gauge field associated to local Lorentz transformations.
\end{rem}

\vspace{\baselineskip}

\begin{defn}[Connections on the spacetime-matter bundle] \label{defn:connections_spacetime-matter_bundle}
We use the following connections on the spacetime-matter bundle \(\mathcal{B}_\Q\): For the tangent bundle \(TM\) of the manifold \(M\) we use the Levi-Civita connection \(\nabla^{TM}_\mu\), acting on a vector field \(X \in \Gamma \left ( M, TM \right )\) via
\begin{subequations}
\begin{align}
	\nabla^{TM}_\mu X^\nu & := \partial_\mu X^\nu + \Gamma^\nu_{\mu \lambda} X^\lambda
\intertext{and on the corresponding covector field via}
	\nabla^{TM}_\mu X_\nu & := \partial_\mu X_\nu - \Gamma^\lambda_{\mu \nu} X_\lambda \, ,
\end{align}
\end{subequations}
with the Christoffel symbol \(\Gamma^\nu_{\mu \lambda}\), given in the case of the Levi-Civita connection via
\begin{equation}
	\Gamma^\nu_{\mu \lambda} := \frac{1}{2} g^{\nu \tau} \left ( \partial_\mu g_{\tau \lambda} + \partial_\lambda g_{\mu \tau} - \partial_\tau g_{\mu \lambda} \right ) \, .
\end{equation}
Furthermore, for the vector bundle \(E\) we use the covariant derivative \(\nabla^{E}_\mu\), induced via the connection on the tangent bundle using vielbeins and inverse vielbeins and acting on a vector field \(Y \in \Gamma \left ( M, E \right )\) via
\begin{subequations}
\begin{align}
	\nabla^{E}_\mu Y^n & := \partial_\mu Y^n + \omega_{\mu l}^n Y^l
\intertext{and on the corresponding covector field via}
	\nabla^{E}_\mu Y_n & := \partial_\mu Y_n - \omega_{\mu n}^l Y_l \, ,
\end{align}
\end{subequations}
with the spin connection
\begin{equation}
\begin{split}
	\omega_{\mu l}^n & := e_\nu^n \left ( \nabla^{TM}_\mu e^\nu_l \right )\\ & \phantom{:} \equiv e_\nu^n \left ( \partial_\mu e^\nu_l \right ) + e_\nu^n \Gamma^\nu_{\mu \lambda} e^\lambda_l \, . \label{eqn:spin_connection}
\end{split}
\end{equation}
We remark that the first term in the spin connection, \(e_\nu^n \left ( \partial_\mu e^\nu_l \right )\), is due to the ambiguity in the definition of vielbeins and inverse vielbeins, as was discussed in \remref{rem:ambiguity_vielbeins_inverse_vielbeins}. Moreover, on the twisted Higgs bundle \(G \times_\rho H^{(i)}\), we use the covariant derivative \(\nabla^{G \times_\rho H^{(i)}}_\mu\), acting on the twisted Higgs and Goldstone fields \(\Phi \in \Gamma \big ( M, G \times_\rho H^{(i)} \big )\) via
\begin{subequations}
\begin{align}
	\nabla^{G \times_\rho H^{(i)}}_\mu \Phi^k & := \partial_\mu \Phi^k + \imaginary \mathrm{g} A^b_\mu \tensor{\mathfrak{H}}{_b ^k _l} \Phi^l
\intertext{and on the corresponding dual section via}
	\nabla^{G \times_\rho H^{(i)}}_\mu \Phi^*_k & := \partial_\mu \Phi^*_k - \imaginary \mathrm{g} A^b_\mu \tensor{\mathfrak{H}}{_b _k ^l} \Phi^*_l \, ,
\end{align}
\end{subequations}
where \(\mathfrak{H} \in \sctnbig{M, \mathfrak{g}^* \otimes_\mathbb{R} \operatorname{End} \big ( H^{(i)} \big )}\) denotes the representation of \(\mathfrak{g}\) via \(\rho\) on the Higgs bundle \(H^{(i)}\). Finally, on the twisted vector of spinor bundles we use the covariant derivative \(\nabla^{G \times_\rho \Sigma M^{\oplus j}}_\mu\), acting on a twisted vector of spinor fields \(\Psi \in \Gamma \left ( M, G \times_\rho \Sigma M^{\oplus j} \right )\) via
\begin{subequations}
\begin{align}
	\nabla^{G \times_\rho \Sigma M^{\oplus j}}_\mu \Psi^k & := \partial_\mu \Psi^k + \varpi_\mu \Psi^k + \imaginary \mathrm{g} A^b_\mu \tensor{\mathfrak{S}}{_b ^k _l} \Psi^l
\intertext{and on the corresponding dual section via}
	\nabla^{G \times_\rho \Sigma M^{\oplus j}}_\mu \overline{\Psi}_k & := \partial_\mu \overline{\Psi}_k - \overline{\Psi}_k \varpi_\mu - \imaginary \mathrm{g} A^b_\mu  \tensor{\mathfrak{S}}{_b _k ^l} \overline{\Psi}_l \, ,
\end{align}
\end{subequations}
with the spinor bundle spin connection
\begin{equation}
	\varpi_\mu := - \frac{\imaginary}{4} \omega_\mu^{r s} \boldsymbol{\sigma}_{r s} \, ,
\end{equation}
where \(\boldsymbol{\sigma}_{r s} := \frac{\imaginary}{2} \left [ \boldsymbol{\gamma}_r , \boldsymbol{\gamma}_s \right ]\) denotes the Clifford representation of \(\mathfrak{spin} (1,3)\) on \(\Sigma M^{\oplus j}\) and \(\mathfrak{S} \in \sctnbig{M, \mathfrak{g}^* \otimes_\mathbb{R} \operatorname{End} \big ( \Sigma M^{\oplus j} \big )}\) denotes the representation of \(\mathfrak{g}\) via \(\rho\) on the vector of spinor bundles \(\Sigma M^{\oplus j}\). We remark that the Levi-Civita connection is metric and torsion-free and the other connections are metric.
\end{defn}

\vspace{\baselineskip}

\begin{rem}[Tetrad postulate] \label{rem:tetrad_postulate}
Given the situation of \defnref{defn:vielbeins_inverse_vielbeins}, the tetrad postulate states that vielbeins and inverse vielbeins are parallel sections in \(\Gamma \left ( U, T^*M \otimes_\mathbb{R} E \right )\) and \(\Gamma \left ( U, TM \otimes_\mathbb{R} E^* \right )\), respectively, with respect to the corresponding tensor product bundle connection, cf.\ \defnref{defn:connections_spacetime-matter_bundle}, i.e.\ we have
\begin{subequations}
\begin{align}
	\nabla^{TM \otimes_\mathbb{R} E}_\mu e_\nu^n & = \partial_\mu e_\nu^n - \Gamma_{\mu \nu}^\lambda e_\lambda^n + \sigma_{\mu l}^n e_\nu^l \equiv 0
\intertext{and}
	\nabla^{TM \otimes_\mathbb{R} E}_\mu e^\nu_n & = \partial_\mu e^\nu_n + \Gamma_{\mu \lambda}^\nu e^\lambda_n - \sigma_{\mu n}^l e^\nu_l \equiv 0 \, .
\end{align}
\end{subequations}
In particular, this implies that it is irrelevant whether vielbeins and inverse vielbeins are placed before or after the covariant derivative \(\nabla^{TM \otimes_\mathbb{R} E}_\mu\) on the product bundle \(TM \otimes_\mathbb{R} E\). Finally, we remark that despite its name the tetrad postulate is not a postulate but a consequence of the definition of the connection on \(E\) via the connection on \(TM\) using vielbeins and inverse vielbeins.
\end{rem}

\vspace{\baselineskip}

\begin{defn}[Clifford multiplication] \label{defn:clifford_multiplication}
Given the spacetime-matter bundle \(\mathcal{B}_\Q\), we define the Clifford multiplication \(\boldsymbol{\gamma} \in \Gamma \big ( M, E^* \otimes_\mathbb{R} \operatorname{End} \left ( \Sigma M^{\oplus j} \right ) \! \big )\) for vector fields \(Y \in \Gamma \left ( M, E \right )\) and vectors of spinor fields \(\Psi \in \Gamma \left ( M, \Sigma M^{\oplus j} \right )\) via
\begin{align}
	Y \cdot \Psi & := Y^m \boldsymbol{\gamma}_m \Psi \, .
\intertext{Furthermore, using vielbeins \(e \in \Gamma \left ( M, T^*M \otimes_\mathbb{R} E \right )\), we can define the extended Clifford multiplication \(\boldsymbol{\gamma} \circ e \in \Gamma \big ( M, T^*M \otimes_\mathbb{R} \operatorname{End} \left ( \Sigma M^{\oplus j} \right ) \! \big )\) for vector fields \(X \in \Gamma \left ( M, TM \right )\) via}
	X \cdot \Psi & := X^\mu e_\mu^m \boldsymbol{\gamma}_m \Psi \, .
\end{align}
Additionally, we extend the above definitions naturally to the twisted vector of spinor bundles \(G \times_\rho \Sigma M^{\oplus j}\).
\end{defn}

\vspace{\baselineskip}

\begin{rem}
We remark that the Clifford multiplication \(\boldsymbol{\gamma}\) or \(\boldsymbol{\gamma} \circ e\) turns the spaces of spinor fields \(\Gamma \left ( M, \Sigma M^{\oplus j} \right )\) and twisted spinor fields \(\Gamma \left ( M, G \times_\rho \Sigma M^{\oplus j} \right )\) into modules over the space of vector fields \(\Gamma \left ( M, E \right )\) or \(\Gamma \left ( M, TM \right )\), respectively.\footnote{This is similar to the fact that all spaces of (super) vector fields are modules over the space of real or complex functions.} Furthermore, it induces an automorphism if and only if the corresponding vector field \(Y \in \Gamma \left ( M, E \right )\) or \(X \in \Gamma \left ( M, TM \right )\) has nowhere vanishing seminorm with respect to the metric \(\eta_{m n}\) or \(g_{\mu \nu}\), respectively, i.e.\ \(\left \| Y \right \|_\eta := \sqrt{\left \vert \eta_{m n} Y^m Y^n \right \vert} \not \equiv 0\) or \(\left \| X \right \|_g := \sqrt{\vert g_{\mu \nu} X^\mu X^\nu \vert} \not \equiv 0\).\footnote{Again, this is similar to the fact that the multiplication of a (super) vector field with a real or complex function induces an automorphism if and only if the corresponding function 
has nowhere vanishing absolute value.}
\end{rem}

\vspace{\baselineskip}

\begin{defn}[Clifford relation] \label{defn:clifford_relation}
We set the Clifford relation for the vector bundle \(E\) as
\begin{align}
	\left \{ \boldsymbol{\gamma}_m , \boldsymbol{\gamma}_n \right \} & = 2 \eta_{m n} \id_{\Sigma M^{\oplus j}} \, ,
\intertext{or equivalently, using vielbeins \(e \in \Gamma \left ( M, T^*M \otimes_\mathbb{R} E \right )\), for the tangent bundle \(TM\) as}
	e_\mu^m e_\nu^n \left \{ \boldsymbol{\gamma}_m , \boldsymbol{\gamma}_n \right \} & = 2 g_{\mu \nu} \id_{\Sigma M^{\oplus j}} \, .
\end{align}
\end{defn}

\vspace{\baselineskip}

\begin{rem} \label{rem:signature_metric_clifford_relation}
We remark that the West coast (``mostly minus'') signature for the metrics \(g\) and \(\eta\) together with the ``plus signed'' Clifford relation induces a quaternionic representation for the Clifford algebra \(\operatorname{Cliff} \left ( 1, 3 \right )\) as the matrix algebra \(\operatorname{Mat} \left (2, \mathbb{H} \right )\). Choosing the Pauli matrices as a representation for the quaternions, we obtain the usual complex Dirac representation as the matrix algebra \(\operatorname{Mat} \left (4, \mathbb{C} \right )\), whose generators are Hermitian.
\end{rem}

\vspace{\baselineskip}

\begin{defn}[(Twisted) Dirac operator] \label{defn:dirac_operator}
Let \(\mathcal{B}_\Q\) be the spacetime-matter bundle from \defnref{defn:spacetime-matter_bundle}. Then we define the Dirac operator \(\slashed{\nabla}^{\Sigma M^{\oplus j}}\) on vectors of spinor fields \(\Psi \in \Gamma \left ( M, \Sigma M^{\oplus j} \right )\) such that the following diagram commutes:
\begin{subequations}
\begin{equation}
\begin{tikzcd}[row sep=huge]
	\Gamma \left ( M, \Sigma M^{\oplus j} \right ) \arrow{r}{\slashed{\nabla}^{\Sigma M^{\oplus j}}} \arrow[swap]{d}{\nabla^{\Sigma M^{\oplus j}}_\cdot} & \Gamma \left ( M, \Sigma M^{\oplus j} \right )\\
\Gamma \left ( M, T^* M \otimes_\mathbb{R} \Sigma M^{\oplus j} \right ) \arrow[swap]{r}{\raisebox{-4mm}{$g^{-1} \otimes_{\mathbb{R}} \id_{\Sigma M^{\oplus j}}$}} & \Gamma \left ( M, T M \otimes_\mathbb{R} \Sigma M^{\oplus j} \right ) \arrow[swap]{u}{\boldsymbol{\gamma} \circ e}
\end{tikzcd}
\end{equation}
Here, \(\nabla^{\Sigma M^{\oplus j}}_\mu = \partial_\mu + \varpi_\mu\) is the covariant derivative on the vector of spinor bundles \(\Sigma M^{\oplus j}\) and \(e_\nu^n \boldsymbol{\gamma}_n\) is the local representation of the Clifford-multiplication. Thus, the local description of the Dirac-operator on the spinor bundle \(\Sigma M^{\oplus j}\) is given via\footnote{Notice that by the tetrad postulate given in \remref{rem:tetrad_postulate} it does not matter whether we place the inverse vielbeins \(e^{\mu m}\) before or after the covariant derivative \(\nabla^{\Sigma M}_\mu\).}
\begin{equation}
\begin{split}
	\slashed{\nabla}^{\Sigma M^{\oplus j}} & := e^{\mu m} \boldsymbol{\gamma}_m \nabla^{\Sigma M^{\oplus j}}_{\mu} \\ & \phantom{:} \equiv e^{\mu m} \boldsymbol{\gamma}_m \left ( \partial_\mu + \varpi_\mu \right ) \, .
\end{split}
\end{equation}
Additionally, we extend this definition to the twisted vector of spinor bundles \(G \times_\rho \Sigma M^{\oplus j}\) by using the covariant derivative \(\nabla^{G \times_\rho \Sigma M^{\oplus j}}_\mu\) instead of \(\nabla^{\Sigma M^{\oplus j}}_\mu\): This yields the local description of the twisted Dirac-operator on the twisted vector of spinor bundles \(G \times_\rho \Sigma M^{\oplus j}\)
	\begin{equation}
	\begin{split}
	\slashed{\nabla}^{G \times_\rho \Sigma M^{\oplus j}} & := e^{\mu m} \boldsymbol{\gamma}_m \nabla^{G \times_\rho \Sigma M^{\oplus j}}_{\mu} \\ & \phantom{:} \equiv e^{\mu m} \boldsymbol{\gamma}_m \left ( \partial_\mu + \varpi_\mu + \imaginary \mathrm{g} A^a_\mu \mathfrak{S}_a \right ) \, .
	\end{split}
	\end{equation}
\end{subequations}
\end{defn}

\vspace{\baselineskip}

\begin{defn}[Curvatures of the spacetime-matter bundle] \label{defn:curvatures_spacetime-matter_bundle}
Using the connections from \defnref{defn:connections_spacetime-matter_bundle}, we can construct the following curvature tensors: We start with the Riemann tensor of the tangent bundle, which acts on a vector field \(X \in \Gamma \left ( M, TM \right )\) via
\begin{subequations}
\begin{align}
	\tensor{R}{^\rho _\sigma _\mu _\nu} X^\sigma & := \left [ \nabla^{TM}_\mu , \nabla^{TM}_\nu \right ] X^\rho
\intertext{and reads}
	\tensor{R}{^\rho _\sigma _\mu _\nu} & = \partial_\mu \Gamma^\rho_{\nu \sigma} - \partial_\nu \Gamma^\rho_{\mu \sigma} + \Gamma^\rho_{\mu \lambda} \Gamma^\lambda_{\nu \sigma} - \Gamma^\rho_{\nu \lambda} \Gamma^\lambda_{\mu \sigma} \, .
\end{align}
\end{subequations}
From this, we can define the Ricci tensor as the following contraction
\begin{align}
	R_{\mu \nu} & := \tensor{R}{^\rho _\mu _\rho _\nu} \, ,
\intertext{and then the Ricci scalar as the final contraction}
	R & := g^{\mu \nu} R_{\mu \nu} \, .
\end{align}
Furthermore, we obtain the curvature form on the \(G\)-principle bundle via
\begin{subequations}
\begin{align}
	\tensor{f}{^a _b _c} F_{\mu \nu}^b \mathfrak{s}^c & := \left [ \nabla^{G}_\mu , \nabla^{G}_\nu \right ] \mathfrak{s}^a \, ,
\intertext{where \(f \in \sctnbig{M, \mathfrak{g}^* \otimes_\mathbb{R} \operatorname{End} \left ( \mathfrak{g} \right ) \!}\) denotes the adjoint representation, and reads}
	F_{\mu \nu}^a & = \mathrm{g} \left ( \partial_\mu A_\nu - \partial_\nu A_\mu - \mathrm{g} \tensor{f}{^a _b _c} A^b_\mu A^c_\nu \right ) \, , \label{eqn:fmunu}
\end{align}
\end{subequations}
where the indices \(a\), \(b\) and \(c\) are with respect to the adjoint representation \(f\) of \(\mathfrak{g}\) on itself.
\end{defn}

\vspace{\baselineskip}

\begin{prop}[Ricci scalar for the Levi-Civita connection, \cite{Prinz_2}] \label{prop:ricci_scalar_for_the_levi_civita_connection}
Using the Levi-Civita connection, the Ricci scalar is given via partial derivatives of the metric and its inverse as follows:
\begin{equation} \label{eqn:ricci_scalar_metric}
\begin{split}
	R & = g^{\mu \rho} g^{\nu \sigma} \left ( \partial_\mu \partial_\nu g_{\rho \sigma} - \partial_\mu \partial_\rho g_{\nu \sigma} \right )\\
	& \phantom{ = } + g^{\mu \rho} g^{\nu \sigma} g^{\kappa \lambda} \left ( \left ( \partial_\mu g_{\kappa \lambda} \right ) \left ( \partial_\nu g_{\rho \sigma} - \frac{1}{4} \partial_\rho g_{\nu \sigma} \right ) + \left ( \partial_\nu g_{\rho \kappa} \right ) \left ( \frac{3}{4} \partial_\sigma g_{\mu \lambda} - \frac{1}{2} \partial_\mu g_{\sigma \lambda} \right ) \right .\\
	& \phantom{ \phantom{ = } + g^{\mu \rho} g^{\nu \sigma} g^{\kappa \lambda} \{ } \left . \vphantom{\left ( \frac{1}{2} \right )} - \left ( \partial_\mu g_{\rho \kappa} \right ) \left ( \partial_\nu g_{\sigma \lambda} \right ) \right )
\end{split}
\end{equation}
\end{prop}

\begin{proof}
See \propref{prop:ricci_scalar_for_the_levi_civita_connection} for the proof of a slightly more general statement.
\end{proof}

\vspace{\baselineskip}

\begin{rem}
Assuming the Levi-Civita connection, the Riemann tensor \(\tensor{R}{^\rho _\sigma _\mu _\nu}\) and the Ricci tensor \(R_{\mu \nu}\) are not sensitive to the choice of the signature of the metric \(g_{\mu \nu}\), whereas the lowered Riemann tensor \(R_{\rho \sigma \mu \nu} := g_{\rho \lambda} \tensor{R}{^\lambda _\sigma _\mu _\nu}\) and the Ricci scalar \(R\) are. Therefore, the Einstein-Hilbert Lagrange density is sensitive to this choice as well, which is the reason for the minus sign in front of the Einstein-Hilbert Lagrange density in \eqnref{eqn:EH-Lagrange_density} in our West coast (``mostly minus'') signature, cf.\ \conref{con:sign_choices}.
\end{rem}

\vspace{\baselineskip}

\begin{defn}[Riemannian and Minkowskian volume forms] \label{def:riemannian_and_minkowskian_volume_form}
Given a manifold \((M,g)\), we define the Riemannian volume form for the metric \(g\) via
\begin{equation}
	\dif V_g := \sqrt{- \dt{g}} \dif t \wedge \dif x \wedge \dif y \wedge \dif z \, .
\end{equation}
Additionally, given the Minkowski metric \(\eta\), we define the Minkowskian volume form via
\begin{equation}
	\dif V_\eta = \dif t \wedge \dif x \wedge \dif y \wedge \dif z \, .
\end{equation}
\end{defn}

\enter

\begin{defn}[Fourier transformation] \label{defn:fourier_transform}
	Let \((M,\met,\trivmap)\) be a simple spacetime with background Minkowski spacetime \((\bbM,\eta)\). Using the correspondence from \defnref{defn:correspondence_Minkowski_spacetime}, we define the Fourier transformation for particle fields, i.e.\ sections \(\particlefield \in \sctn{M,E}\), as follows:
\begin{subequations}
\begin{align}
	\mathscr{F} \, : \quad \Gamma \big ( M, E \big ) \to \widehat{\Gamma} \big ( M, E \big ) \, , \quad \particlefield \left ( x^\alpha \right ) \mapsto \widehat{\particlefield} \left ( p^\alpha \right )
	\intertext{with}
	\widehat{\particlefield} \left ( p^\alpha \right ) := \frac{1}{\left ( 2 \pi \right )^2} \tau^* \left ( \int_\sbbM \left ( \tau_* \particlefield \right ) \big ( y^\beta \big ) e^{-i \eta \left ( y, p \right )} \dif V_\eta \right )
\end{align}
\end{subequations}
\end{defn}

\section{Lagrange densities} \label{sec:lagrange-densities}

In this section, we discuss the Lagrange densities of (effective) Quantum General Relativity coupled to the Standard Model. Lagrange densities are local functionals on the sheaf of particle fields that induce their equations of motion via an Euler--Lagrange variation. More precisely, a Lagrange functional is a map
\begin{subequations}
\begin{equation}
	\LQ \colon \BQ \to \tf{M} \, ,
\end{equation}
where \(\BQ\) is the spacetime-matter bundle and \(\tf{M}\) the vector space of top forms on \(M\). More precisely, the locality of \(\LQ\) means that it factors uniquely as follows
\begin{equation}
\begin{tikzcd}[row sep=huge]
	\BQ \arrow{rr}{\LQ} \arrow[swap]{dr}{J^k} & & \tf{M} \\
	& J^k \BQ \arrow[swap]{ur}{j^k \LQ} &
\end{tikzcd} \, ,
\end{equation}
\end{subequations}
where \(J^k \BQ\) denotes the \(k\)-th jet bundle of \(\BQ\) for some \(k \in \mathbb{N}\). In particular, for (effective) Quantum General Relativity coupled to the Standard Model we can chose \(k = 2\) or even \(k = 1\) after suitable partial integrations if boundary terms can be neglected.

\subsection{Gravitons and graviton-ghosts} \label{ssec:gravitons_and_graviton-ghosts}

We start in this subsection with the Lagrange density of (effective) Quantum General Relativity, i.e.\ gravitons and their ghost.

\enter

\begin{con}[Lagrange density] \label{con:Lagrange_density}
	We choose the following signs and prefactors for the Lagrange density, where \(\dif V_g := \sqrt{- \dt{g}} \dif t \wedge \dif x \wedge \dif y \wedge \dif z\) and \(\dif V_\eta := \dif t \wedge \dif x \wedge \dif y \wedge \dif z\) denote the Riemannian and Minkowskian volume forms, respectively:
	\begin{enumerate}
		\item Einstein-Hilbert Lagrange density: \begin{equation} \mathcal{L}_\text{GR} := - \frac{1}{2 \gcoupling^2} R \dif V_g \, , \label{eqn:EH-Lagrange_density} \end{equation} with \(R := g^{\nu \sigma} \tensor{R}{^\mu _\sigma _\mu _\nu}\)
		\item Linearized de Donder Gauge fixing Lagrange density: \begin{equation} \mathcal{L}_\text{GF} := - \frac{1}{4 \gcoupling^2 \zeta}  \eta^{\mu \nu} \deDonder^{(1)}_\mu \deDonder^{(1)}_\nu \dif V_\eta \, , \end{equation} with \(\deDonder^{(1)}_\mu := \eta^{\rho \sigma} \Gamma_{\rho \sigma \mu} \equiv \gcoupling \eta^{\rho \sigma} \left ( \partial_\rho h_{\mu \sigma} - \frac{1}{2} \partial_\mu h_{\rho \sigma} \right )\)
		\item Ghost Lagrange density: \begin{equation} \begin{split} \mathcal{L}_\text{Ghost} & := - \frac{1}{2 \zeta} \eta^{\rho \sigma} \overline{C}^\mu \left ( \partial_\rho \partial_\sigma C_\mu \right ) \dif V_\eta \\ & \phantom{:=} - \frac{1}{2} \eta^{\rho \sigma} \overline{C}^\mu \left ( \partial_\mu \big ( \tensor{\Gamma}{^\nu _\rho _\sigma} C_\nu \big ) - 2 \partial_\rho \big ( \tensor{\Gamma}{^\nu _\mu _\sigma} C_\nu \big ) \right ) \dif V_\eta \, , \end{split} \end{equation} with \(\gravitonghost \in \Gamma \left ( M, T^*[1,0] M \right )\) and \(\overline{\gravitonghost} \in \Gamma \left ( M, T[-1,0] M \right )\)
	\end{enumerate}
	The Lagrange density of (effective) Quantum General Relativity is then the sum of the three, i.e.\
	\begin{equation} \label{eqn:QGR_Lagrange_density}
		\begin{split}
		\mathcal{L}_\text{QGR} & := \mathcal{L}_\text{GR} + \mathcal{L}_\text{GF} + \mathcal{L}_\text{Ghost} \\
		& \phantom{:} \equiv - \frac{1}{2 \gcoupling^2} \left ( \sqrt{- \dt{g}} R + \frac{1}{2 \zeta}  \eta^{\mu \nu} \deDonder^{(1)}_\mu \deDonder^{(1)}_\nu \right ) \dif V_\eta \\
		& \phantom{:=} - \frac{1}{2} \eta^{\rho \sigma} \left ( \frac{1}{\zeta} \overline{C}^\mu \left ( \partial_\rho \partial_\sigma C_\mu \right ) + \overline{C}^\mu \left ( \partial_\mu \big ( \tensor{\Gamma}{^\nu _\rho _\sigma} C_\nu \big ) - 2 \partial_\rho \big ( \tensor{\Gamma}{^\nu _\mu _\sigma} C_\nu \big ) \right ) \right ) \dif V_\eta \, ,
		\end{split}
	\end{equation}
cf.\ \cite[Section 2.2]{Prinz_2}. We remark that the ghost Lagrange density is constructed via Faddeev-Popov's method \cite{Faddeev_Popov}, cf.\ \cite[Subsection 2.2.3]{Prinz_2} and \cite{Prinz_4,Prinz_5}, which can be embedded into the more elaborate settings of BRST cohomology and BV formalism.
\end{con}

\enter

\begin{rem}
	The reason for the sign choices from \conref{con:Lagrange_density} are as follows: The minus sign for the Einstein-Hilbert Lagrange density is due to the sign choice for the Minkowski metric, cf.\ \conref{con:sign_choices}. Then, the minus sign for the gauge fixing Lagrange density is such that \(\zeta = 1\) corresponds to the de Donder gauge fixing. Finally, the sign for the ghost Lagrange density is, as usual, an arbitrary choice, and is chosen such that all Lagrange densities have the same sign.
\end{rem}

\enter

\begin{rem} \label{rem:diffeo_invariance_EH-Lagrange_density}
	Contrary to Yang--Mills Lagrange densities, which are strictly invariant under gauge transformations, the Einstein-Hilbert Lagrange density is not invariant under general diffeomorphisms as it is a tensor density of weight \(1\). More precisely, the action of an infinitesimal diffeomorphism adds a total derivative to the Einstein--Hilbert Lagrange density if the corresponding vector field is not Killing (else the Lagrange density would remain unchanged).
\end{rem}

\subsection{Gravitons and matter from the Standard Model}

We proceed in this subsection by discussing the couplings of gravitons to matter from the Standard Model. To this end, we first classify all appearing gravity-matter interactions into 10 Lagrange densities, henceforth referred to as matter-model Lagrange densities. Then we discuss the Lagrange densities for gravitons with scalar particles, spinor particles, gauge bosons and gauge ghosts in detail.

\enter

\begin{lem} \label{lem:matter-model-Lagrange-densities}
	Consider (effective) Quantum General Relativity coupled to the Standard Model (QGR-SM). Then the interaction Lagrange densities between gravitons and matter particles are of the following 10 types:\footnote{We remark that the tensors \(\tensor[_k]{\! T}{}\) are not related to Hilbert stress-energy tensors. More precisely, they are defined as the graviton-free matter contributions of the corresponding Lagrange densities.}
	{\allowdisplaybreaks
	\begin{align}
		\tensor[_1]{\mathcal{L}}{_{\textup{QGR-SM}}} & := \tensor[_1]{\! T}{} \dif V_g \, , \\
		\tensor[_2]{\mathcal{L}}{_{\textup{QGR-SM}}} & := \left ( g^{\mu \nu} \, \tensor[_2]{\! T}{_\mu _\nu} \right ) \dif V_g \, , \\
		\tensor[_3]{\mathcal{L}}{_{\textup{QGR-SM}}} & := \left ( g^{\mu \nu} g^{\rho \sigma} \, \tensor[_3]{\! T}{_\mu _\nu _\rho _\sigma} \right ) \dif V_g \, , \\
		\tensor[_4]{\mathcal{L}}{_{\textup{QGR-SM}}} & := \left ( g^{\mu \nu} \tensor{\Gamma}{^\tau _\mu _\nu} \, \tensor[_4]{\! T}{_\tau} \right ) \dif V_g \, , \\
		\tensor[_5]{\mathcal{L}}{_{\textup{QGR-SM}}} & := \left ( g^{\mu \nu} g^{\rho \sigma} \tensor{\Gamma}{^\tau _\mu _\nu} \, \tensor[_5]{\! T}{_\rho _\sigma _\tau} \right ) \dif V_g \, , \\
		\tensor[_6]{\mathcal{L}}{_{\textup{QGR-SM}}} & := \left ( g^{\mu \nu} g^{\rho \sigma} \tensor{\Gamma}{^\kappa _\mu _\nu} \tensor{\Gamma}{^\lambda _\rho _\sigma} \, \tensor[_6]{\! T}{_\kappa _\lambda} \right ) \dif V_g \, , \\
		\tensor[_7]{\mathcal{L}}{_{\textup{QGR-SM}}} & := \left ( e_0^o \, \tensor[_7]{\! T}{_o} \right ) \dif V_g \, , \\
		\tensor[_8]{\mathcal{L}}{_{\textup{QGR-SM}}} & := \left ( e_0^o e^{\rho r} \, \tensor[_8]{\! T}{_o _\rho _r} \right ) \dif V_g \, , \\
		\tensor[_9]{\mathcal{L}}{_{\textup{QGR-SM}}} & := \left ( e_0^o e^{\rho r} e^{\sigma s} \left ( \partial_\rho e_\sigma^t \right ) \, \tensor[_9]{\! T}{_o _r _s _t} \right ) \dif V_g \, , \\
		\intertext{and}
		\tensor[_{10}]{\mathcal{L}}{_{\textup{QGR-SM}}} & := \left ( e_0^o e^{\rho r} e^{\sigma s} e_\tau^t \tensor{\Gamma}{^\tau _\rho _\sigma} \, \tensor[_{10}]{\! T}{_o _r _s _t} \right ) \dif V_g \, .
	\end{align}
	}%
\end{lem}

\begin{proof}
	A direct computation shows that the scalar particles form the Standard Model are of type 1 and 2. Furthermore, the spinor particles from the Standard Model are of type 7, 8, 9 and 10. Moreover, the bosonic gauge boson particles from the Standard Model are of type 3, 5 and 6 and additionally 1, 2 and 4 for spontaneous symmetry breaking. Finally, the gauge ghosts are of type 2 and 4 and additionally 1 for spontaneous symmetry breaking. This is discussed in detail in the following four Subsubsections.
\end{proof}

\subsubsection{Gravitons and scalar particles} \label{sssec:gravitons_and_scalar_particles}

Scalar particles from the Standard Model are the Higgs and Goldstone bosons.\footnote{The gauge ghosts are discussed in \sssecref{sssec:gravitons_and_gauge_ghosts}.} In the following we describe the interaction of gravitons with a real scalar field and twisted Higgs and Goldstone fields (which leads to spontaneous symmetry breaking). Geometrically they are described via sections \(\phi \in \Gamma \left ( M, \mathbb{R} \right )\) and \(\Phi \in \Gamma \big ( M, H^{(i)} \big )\), respectively. Then, the corresponding Lagrange densities are given by
\begin{align}
	\mathcal{L}_\text{GR-Scalar} & = \left ( \frac{1}{2} g^{\mu \nu} \left ( \partial_\mu \phi \right ) \left ( \partial_\nu \phi \right ) + \sum_{i \in \boldsymbol{I}_\phi} \frac{\alpha_i}{i!} \phi^i \right ) \dif V_g
	\intertext{and}
	\mathcal{L}_\text{GR-Higgs} & = \left ( g^{\mu \nu} \big ( \nabla^{G \times_\rho \mathbb{C}^i}_\mu \Phi \big )^\dagger \big ( \nabla^{G \times_\rho \mathbb{C}^i}_\nu \Phi \big ) + \sum_{i \in \boldsymbol{I}_\Phi} \frac{\alpha_i}{i!} \big ( \Phi^\dagger \Phi \big )^i \right ) \dif V_g \, , \label{eqn:vector_complex_scalar_field}
\end{align}
where \(\boldsymbol{I}_\phi\) and \(\boldsymbol{I}_\Phi\) denote the interaction sets with particle mass \(- \alpha_2\) and coupling constants \(\alpha_i\) for \(i \neq 2\). The Higgs bundle from the Standard Model is of the form \eqnref{eqn:vector_complex_scalar_field} with further interactions coming from the gauge fixing of the corresponding Electroweak gauge bosons, cf.\ \sssecref{sssec:gravitons_and_gauge_bosons}. These interactions correspond to type 1 and 2 from \lemref{lem:matter-model-Lagrange-densities}. More precisely, we have
\begin{align}
	\tensor[_1]{\! T}{} & := \sum_{i \in \boldsymbol{I}_\phi} \frac{\alpha_i}{i!} \phi^i + \sum_{i \in \boldsymbol{I}_\Phi} \frac{\alpha_i}{i!} \big ( \Phi^\dagger \Phi \big )^i
	\intertext{and}
	\tensor[_2]{\! T}{_\mu _\nu} & := \frac{1}{2} \left ( \partial_\mu \phi \right ) \left ( \partial_\nu \phi \right ) + \big ( \nabla^{G \times_\rho \mathbb{C}^i}_\mu \Phi \big )^\dagger \big ( \nabla^{G \times_\rho \mathbb{C}^i}_\nu \Phi \big ) \, .
\end{align}

\subsubsection{Gravitons and spinor particles}

Spinor particles from the Standard Model are leptons and quarks. In the following we describe the interaction of gravitons with spinor fields and a vector of twisted spinor fields. Geometrically they are described via sections \(\psi \in \Gamma \left ( M, \Sigma M \right )\) and \(\Psi \in \Gamma \left ( M, \Sigma M^{\oplus j} \right )\), respectively. The corresponding dual spinor fields are defined via
\begin{align}
	\overline{\psi} & := e_0^o \left ( \gamma_o \psi \right )^\dagger
	\intertext{and}
	\overline{\Psi} & := e_0^o \left ( \boldsymbol{\gamma}_o \Psi \right )^\dagger \, ,
\end{align}
where \(e_0^o\) is a vielbein with its curved index fixed to \(\mu \equiv 0\) and flat index \(o\), i.e.\ a vielbein contracted with the normalized timelike vector field \(e \left ( \dif t \right )\), and \(\gamma_m\) and \(\boldsymbol{\gamma}_m\) are the Dirac matrices for the Minkowski background metric \(\eta\) on \(M\) and \(M^{\oplus j}\), respectively. Thus, dual spinor fields depend on the metric via the vielbein \(e_0^o\) with fixed timelike curved index.\footnote{We emphasize the placement of \(\gamma_o\) and \(\boldsymbol{\gamma}_o\) in the following equations, as only the timelike Dirac matrices \(\gamma_0\) and \(\boldsymbol{\gamma}_0\) are hermitian, whereas the other Dirac matrices are antihermitian.} We remark that if the spacetime \((M,\met)\) is globally hyperbolic, it is possible to choose charts in which \(e_0^o \equiv \delta_0^o\), as is done implicitly in e.g.\ \cite{Choi_Shim_Song,Schuster,Rodigast_Schuster_1,Rodigast_Schuster_2}. However it should be noted that in this setting the theory is no longer invariant under general diffeomorphisms, but only under the subgroup of diffeomorphisms preserving global hyperbolicity. As we do not want to restrict our analysis to such charts and diffeomorphisms, we set
\begin{align}
	\overline{\psi}_o & := \left ( \gamma_o \psi \right )^\dagger
	\intertext{and}
	\overline{\Psi}_o & := \left ( \boldsymbol{\gamma}_o \Psi \right )^\dagger
\end{align}
for later use. Then, the corresponding Lagrange densities are given by
\begin{align}
	\mathcal{L}_\text{GR-Spinor} & = \Big ( \overline{\psi} \big ( \imaginary \slashed{\nabla}^{\Sigma M} - m_\psi \big ) \psi \Big ) \dif V_g
	\intertext{and}
	\mathcal{L}_\text{GR-Spinor\(^j\)} & = \Big ( \overline{\Psi} \big ( \imaginary \slashed{\nabla}^{G \times_\rho \Sigma M^{\oplus j}} - \boldsymbol{m}_\Psi \big ) \Psi \Big ) \dif V_g \, ,
\end{align}
where \(\boldsymbol{m}_\Psi\) is a diagonal \(j \times j\)-matrix with entries given via the corresponding spinor particle masses, and with the Dirac operators given via
\begin{align}
	\slashed{\nabla}^{\Sigma M} & := e^{\mu m} \gamma_m \left ( \partial_\mu + \varpi_\mu \right )
	\intertext{and}
	\slashed{\nabla}^{G \times_\rho \Sigma M^{\oplus j}} & := e^{\mu m} \boldsymbol{\gamma}_m \left ( \partial_\mu + \varpi_\mu \right ) + e^{\mu m} \gamma_m \left ( \imaginary \mathrm{g} A^a_\mu \mathfrak{b}_a \right ) \, ,
\end{align}
where \(\varpi_\mu \in \Gamma \left ( M, T^* M \otimes \operatorname{End} \left ( \Sigma M \right ) \right )\) is the spin connection form and \(\imaginary \mathrm{g} A \in \Gamma \big ( M, T^* M \otimes \operatorname{End} \left ( \Sigma M^{\oplus j} \right ) \! \big )\) the corresponding gauge group connection form. These interactions correspond to type 7, 8, 9 and 10 from \lemref{lem:matter-model-Lagrange-densities}. More precisely, we have
\begin{align}
	\tensor[_7]{\! T}{_o} & := - m_\psi \overline{\psi}_o \psi - \overline{\Psi}_o \boldsymbol{m}_\Psi \Psi \, , \\
	\tensor[_8]{\! T}{_o _\rho _r} & := \overline{\psi}_o \gamma_r \left ( \partial_\rho \psi \right ) + \overline{\Psi}_o \boldsymbol{\gamma}_r \left ( \partial_\rho \Psi \right ) \, , \\
	\tensor[_9]{\! T}{_o _r _s _t} & := - \frac{\imaginary}{4} \overline{\psi}_o \left ( \gamma_r \sigma_{s t} \right ) \psi - \frac{\imaginary}{4} \overline{\Psi}_o \left ( \boldsymbol{\gamma}_r \boldsymbol{\sigma}_{s t} \right ) \Psi \, ,
	\intertext{with \(\sigma_{s t} := \frac{\imaginary}{2} \left [ \gamma_s, \gamma_t \right ]\) and \(\boldsymbol{\sigma}_{s t} := \frac{\imaginary}{2} \left [ \boldsymbol{\gamma}_s, \boldsymbol{\gamma}_t \right ]\), and}
	\begin{split}
		\tensor[_{10}]{\! T}{_o _r _s _t} & := - \frac{\imaginary}{4} \overline{\psi}_o \left ( \gamma_r \sigma_{s t} \right ) \psi - \frac{\imaginary}{4} \overline{\Psi}_o \left ( \boldsymbol{\gamma}_r \boldsymbol{\sigma}_{s t} \right ) \Psi \\
		& \hphantom{:} \equiv \tensor[_9]{\! T}{_o _r _s _t} \, .
	\end{split}
\end{align}
We remark that the interaction of leptons and quarks with the Higgs and Goldstone bosons are given by
\begin{equation}
	\mathcal{L}_\text{Yukawa} = - \left ( \sum_{\{ \phi, \overline{\psi}_o, \psi \} \in \boldsymbol{I}_Y} \alpha_{\{ \phi, \overline{\psi}_o, \psi \}} \phi \overline{\psi}_o \psi \right ) \dif V_g
\end{equation}
which represent the Yukawa interaction terms for the interaction set \(\boldsymbol{I}_Y\), with corresponding coupling constants \(\alpha_{\{ \phi, \overline{\psi}_o, \psi \}}\). These interactions are of type 7 from \lemref{lem:matter-model-Lagrange-densities}. More precisely, we have
\begin{equation}
	\tensor[_7]{\! T}{_o} := - \sum_{\{ \phi, \overline{\psi}_o, \psi \} \in \boldsymbol{I}_Y} \alpha_{\{ \phi, \overline{\psi}_o, \psi \}} \phi \overline{\psi}_o \psi \, .
\end{equation}

\subsubsection{Gravitons and gauge bosons} \label{sssec:gravitons_and_gauge_bosons}

Gauge bosons from the Standard Model are the photon, the \(Z\)- and \(W^\pm\)-bosons, and the gluons. In the following we describe the interaction of gravitons with gauge bosons from a Quantum Yang--Mills theory. We denote the Yang--Mills gauge group by \(G\) and its Lie algebra by \(\mathfrak{g}\). Geometrically, gauge bosons are described via connection forms \(\imaginary \mathrm{g} A \in \Gamma \left ( M, T^* M \otimes \mathfrak{g} \right )\) on the underlying principle bundle. More precisely, they are given as the components with respect to a basis choice \(\set{\mathfrak{b}_a}\) on \(\mathfrak{g}\). Then, the corresponding Lagrange densities are given by\footnote{We remark that this obviously also includes abelian gauge theories, such as (quantum) electrodynamics, by setting \(\mathfrak{g}\) to be abelian, i.e.\ \(f^{abc} \equiv 0\). \label{ftn:YM-ED}}
\begin{align}
	\mathcal{L}_\text{GR-YM} & = \left ( - \frac{1}{4 \mathrm{g}^2} \delta_{ab} g^{\mu \nu} g^{\rho \sigma} F^a_{\mu \rho} F^b_{\nu \sigma} \right ) \dif V_g
	\intertext{and the Lorenz gauge fixing by\footnotemark}
	\mathcal{L}_\text{GR-YM-GF} & = \left ( \frac{1}{2 \xi} \delta_{ab} g^{\mu \nu} g^{\rho \sigma} \big ( \nabla^{TM}_\mu A^a_\nu \big ) \big ( \nabla^{TM}_\rho A^b_\sigma \big ) \right ) \dif V_g \, . \label{eqn:GR-YM-GF}
\end{align}
\footnotetext{It is convenient to use the covariant Lorenz gauge fixing \(g^{\mu \nu} \nabla^{TM}_\mu A^a_\nu \overset{!}{=} 0\), as this choice avoids couplings from graviton-ghosts to gauge ghosts \cite{Prinz_5}.}%
\noindent These interactions correspond to type 3, 5 and 6 from \lemref{lem:matter-model-Lagrange-densities}. More precisely, we have\footnote{We remark the minus sign due to the covariant derivative on forms and the additional factor of 2 due to the binomial theorem in \eqnref{eqn:gauge_boson_vii}.}
\begin{align}
	\tensor[_3]{\! T}{_\mu _\nu _\rho _\sigma} & := - \frac{1}{4 \mathrm{g}^2} \delta_{ab} F^a_{\mu \rho} F^b_{\nu \sigma} + \frac{1}{2 \xi} \delta_{ab} \big ( \partial_\mu A^a_\nu \big ) \big ( \partial_\rho A^b_\sigma \big ) \, , \\
	\tensor[_5]{\! T}{_\mu _\nu _\tau} & := - \frac{1}{\xi} \delta_{ab} \big ( \partial_\mu A^a_\nu \big ) A^b_\tau \label{eqn:gauge_boson_vii}
	\intertext{and}
	\tensor[_6]{\! T}{_\kappa _\lambda} & := \frac{1}{2 \xi} \delta_{ab} A^a_\kappa A^b_\lambda \, .
\end{align}
We remark that the Lorenz gauge fixing Lagrange densities for the \(Z\)- and \(W^\pm\)-bosons need slight modifications due to the spontaneous symmetry breaking and are given by
\begin{align}
	\mathcal{L}_\text{\(Z\)-Boson-GF} & = \left ( \frac{1}{2 \xi_Z} g^{\mu \nu} g^{\rho \sigma} \big ( \nabla^{TM}_\mu Z_\nu ) \big ( \nabla^{TM}_\rho Z_\sigma \big ) + m_Z \phi_Z g^{\mu \nu} \big ( \nabla^{TM}_\mu Z_\nu \big ) + \frac{\xi_Z}{2} m_Z^2 \phi_Z^2 \right ) \dif V_g
	\intertext{and}
	\begin{split}
		\mathcal{L}_\text{\(W\)-Boson-GF} & = \left ( \frac{1}{\xi_W} g^{\mu \nu} g^{\rho \sigma} \big ( \nabla^{TM}_\mu W^-_\nu \big ) \big ( \nabla^{TM}_\rho W^+_\sigma \big ) + \xi_W m_W^2 \phi_{W^-} \phi_{W^+} \right . \\ & \phantom{= (} \left . + \imaginary m_W g^{\mu \nu} \left ( \phi_{W^+} \big ( \nabla^{TM}_\mu W^-_\nu \big ) - \phi_{W^-} \big ( \nabla^{TM}_\mu W^+_\nu \big ) \right ) \right ) \dif V_g \, ,
	\end{split}
\end{align}
where \(\xi_s\) is the corresponding gauge fixing parameter and \(m_s\) the corresponding mass for \(s \in \set{Z, W^+, W^-}\), and \(\phi_Z\), \(\phi_{W^+}\) and \(\phi_{W^-}\) are the Goldstone bosons. These interactions additionally require type 1, 2 and 4 from \lemref{lem:matter-model-Lagrange-densities}. More precisely, we have
\begin{align}
	\tensor[_1]{\! T}{} & := \frac{\xi_Z}{2} m_Z^2 \phi_Z^2 + \xi_W m_W^2 \phi_{W^-} \phi_{W^+} \, , \\
	\tensor[_2]{\! T}{_\mu _\nu} & := \xi_Z m_Z \phi_Z \big ( \partial_\mu Z_\nu \big ) + \imaginary \xi_W m_W \left ( \phi_{W^+} \big ( \partial_\mu W^-_\nu \big ) - \phi_{W^-} \big ( \partial_\mu W^+_\nu \big ) \right ) \, , \\
	\tensor[_3]{\! T}{_\mu _\nu _\rho _\sigma} & := \frac{1}{2 \xi_Z} \big ( \partial_\mu Z_\nu \big ) \big ( \partial_\rho Z_\sigma \big ) + \frac{1}{\xi_W} \big ( \partial_\mu W^-_\nu \big ) \big ( \partial_\rho W^+_\sigma \big )
	\intertext{and}
	\tensor[_4]{\! T}{_\tau} & := \left ( \xi_s m_s \right ) \phi^s A^{-s}_\tau \, .
\end{align}
We refer to \sssecref{sssec:gravitons_and_scalar_particles} for further interactions between \(Z\)- and \(W^\pm\)-bosons and Higgs and Goldstone bosons coming from the covariant derivative on the Higgs bundle.

\subsubsection{Gravitons and gauge ghosts} \label{sssec:gravitons_and_gauge_ghosts}

Gauge ghosts and gauge antighosts from the Standard Model, accompanying their corresponding gauge bosons \(\imaginary \mathrm{g} A \in \Gamma \left ( M, T^* M \otimes \mathfrak{g} \right )\), are fermionic \(\mathfrak{g}\)-valued scalar particles \(c \in \Gamma \left ( M, \mathfrak{g}[0,1] \right )\) and \(\overline{c} \in \Gamma \left ( M, \mathfrak{g}^*[0,-1] \right )\). Then, the corresponding Lagrange density is given by
\begin{equation}
	\mathcal{L}_\text{GR-YM-Ghost} = - \left ( g^{\mu \nu} \overline{c}_a \big ( \nabla^{TM}_\mu \left ( \partial_\nu c^a \right ) \! \big ) + \mathrm{g} g^{\mu \nu} \tensor{f}{^a _b _c} \overline{c}_a \big ( \nabla^{TM}_\mu A^b_\nu c^c \big ) \right ) \dif V_g \, .
\end{equation}
These interactions correspond to type 2 and 4 from \lemref{lem:matter-model-Lagrange-densities}. More precisely, we have\footnote{The ghost Lagrange densities are calculated with Faddeev--Popov's method \cite{Faddeev_Popov}, cf.\ \cite[Subsubsection 2.2.3]{Prinz_2} and \cite{Prinz_4,Prinz_5}. We mention that this construction can be embedded into a more general context, using BRST and anti-BRST operators \cite{Baulieu_Thierry-Mieg,Prinz_6}.}
\begin{align}
	\tensor[_2]{\! T}{_\mu _\nu} & := - \overline{c}_a \big ( \partial_\mu \partial_\nu c^a \big ) + \mathrm{g} \tensor{f}{^a _b _c} \overline{c}_a \big ( \partial_\mu A^b_\nu c^c \big )
	\intertext{and}
	\tensor[_4]{\! T}{_\tau} & := - \overline{c}_a \big ( \partial_\tau c^a \big ) + \mathrm{g} \tensor{f}{^a _b _c} \overline{c}_a A^b_\tau c^c \, .
\end{align}
We remark that the interaction of Electroweak gauge ghosts with the Higgs and Goldstone bosons are given by
\begin{equation}
	\mathcal{L}_\text{EW-Ghost} = - \left ( \sum_{\set{s_1, s_2, s_3} \in \boldsymbol{I}_\text{EW-Ghost}} \left ( \xi_{s_2} m_{s_2} \right ) \phi^{s_1} \overline{c}^{s_2} c^{s_3} \right ) \dif V_g \, ,
\end{equation}
where \(\xi_{s_i}\) is the corresponding gauge fixing parameter, \(m_{s_i}\) the corresponding mass for \(s_i \in \set{A, Z, W^+, W^-, H}\) and \(\boldsymbol{I}_\text{EW-Ghost}\) is the corresponding interaction set. These interactions are of type 1 from \lemref{lem:matter-model-Lagrange-densities}. More precisely, we have
\begin{equation}
	\tensor[_1]{\! T}{} := - \sum_{\set{s_1, s_2, s_3} \in \boldsymbol{I}_\text{EW-Ghost}} \left ( \xi_{s_2} m_{s_2} \right ) \phi^{s_1} \overline{c}^{s_2} c^{s_3} \, .
\end{equation}
We comment that with our chosen covariant Lorenz gauge fixing in \eqnref{eqn:GR-YM-GF} there are no interactions between graviton-ghosts and gauge ghosts present. This is due to the fact, that the gauge-fixing Lagrange density is a tensor density of weight 1, cf.\ \cite{Prinz_5}.

\section{The diffeomorphism-gauge BRST double complex} \label{sec:diffeomorphism-gauge-brst-double-complex}

BRST cohomology is a powerful tool to study quantum gauge theories together with their gauge fixings and corresponding ghosts via homological algebra \cite{Becchi_Rouet_Stora_1,Becchi_Rouet_Stora_2,Tyutin,Becchi_Rouet_Stora_3}. More precisely, a nilpotent operator \(D\) is introduced that performs an infinitesimal gauge transformation in direction of the ghost field. This so-called BRST operator \(D\) can be seen either as an odd super vector field on the super vector bundle of particle fields or as an odd superderivation on the superalgebra of particle fields. The nilpotency of \(D\) can then be used to compute its cohomology. This is useful, as physical states of the system can be identified with elements in the \(0\)-th cohomology class. Furthermore, this formalism can be used to unify the gauge fixing and ghost Lagrange densities as follows: First of all, we understand a quantum gauge theory Lagrange density \(\mathcal{L}_\text{QGT}\) as the sum of the classical gauge theory Lagrange density \(\mathcal{L}_\text{GT}\) together with a gauge fixing Lagrange density \(\mathcal{L}_\text{GF}\) and its corresponding ghost Lagrange density \(\mathcal{L}_\text{Ghost}\), i.e.\
\begin{equation}
	\mathcal{L}_\text{QGT} := \mathcal{L}_\text{GT} + \mathcal{L}_\text{GF} + \mathcal{L}_\text{Ghost} \, .
\end{equation}
By construction, the gauge fixing and ghost Lagrange densities are not independent: In the Faddeev--Popov setup the ghost Lagrange density is designed such that the ghost field satisfies residual gauge transformations of the chosen gauge fixing as equations of motion, with the antighost as Lagrange multiplier.\footnote{We remark that it is possible to generalize this setup by BRST and anti-BRST transformations, so that ghosts and antighosts can be treated on an equal footing \cite{Baulieu_Thierry-Mieg}.} In the BRST framework, both the gauge fixing and ghost Lagrange densities can be generated from a so-called gauge fixing fermion \(\Upsilon\) via the action of \(D\), i.e.\
\begin{equation}
	\mathcal{L}_\text{GF} + \mathcal{L}_\text{Ghost} \equiv D \Upsilon \, .
\end{equation}
Since the term \(D \Upsilon\) is \(D\)-exact, it is also \(D\)-closed and thus does not contribute to the 0th cohomology class. Thus, in particular, it does not affect physical observables. To incorporate the gauge fixing, we additionally add the corresponding Lautrup--Nakanishi auxiliary fields \cite{Nakanishi,Lautrup}. These are Lie algebra valued fields that act as Lagrange multipliers and whose equations of motion are precisely the gauge fixing conditions. We remark that it is also possible to define anti-BRST operators, which are homological differentials, by essentially replacing ghosts with antighosts in addition with a slightly modified action on the corresponding ghost, antighost and Lautrup--Nakanishi auxiliary fields, cf.\ \cite{Baulieu_Thierry-Mieg,Nakanishi_Ojima,Faizal} and the definitions below. In addition, we refer to the following introductory texts \cite{Barnich_Brandt_Henneaux,Mnev,Wernli}, the historical overview \cite{Becchi} and earlier investigations on the BRST setup of perturbative quantum gravity \cite{Nakanishi_Ojima,Faizal,Upadhyay}.

In this section, we extend and study this setup for (effective) Quantum General Relativity coupled to the Standard Model: This implies first of all the existence of two such operators, \(P\) and \(Q\): The first performs infinitesimal diffeomorphisms and the second performs infinitesimal gauge transformations, cf.\ \defnsaref{defn:diffeomorphism_brst_operator}{defn:gauge_brst_operator}. Then we provide the two gauge fixing fermions \(\stigma\) and \(\digamma\): The first implements the de Donder gauge fixing together with the respective graviton-ghosts and the second implements the Lorenz gauge fixing together with the respective gauge ghosts, cf.\ \propsaref{prop:de_donder_gauge_fixing_fermion}{prop:lorenz_gauge_fixing_fermion}. In particular, we have reworked the conventions such that the quadratic gauge fixing and ghost Lagrange densities are rescaled by the respective inverses of the gauge fixing parameters \(\zeta\) and \(\xi\): This induces that unphysical propagators are then rescaled via the respective gauge fixing parameters. Thus, the degree of a Feynman graph in the gauge fixing parameters is a measure for the unphysicalness of its virtual particles, which we will use later in \chpref{chp:hopf_algebraic_renormalization} and \sectionref{sec:outlook}. Furthermore, we show that all non-constant functionals on the superalgebra of particle fields that are essentially closed with respect to \(P\) are scalar tensor densities of weight \(w = 1\), cf.\ \lemref{lem:p_tensor_densities}. This allows us to show that the graviton-ghosts decouple from matter of the Standard Model if the gauge theory gauge fixing fermion is a tensor density of weight \(w = 1\), cf.\ \thmref{thm:no-couplings-grav-ghost-matter-sm}. In particular, every gauge theory gauge fixing fermion can be modified uniquely to satisfy said condition. Moreover, we prove that the two BRST operators anticommute, i.e.\ \(\commutatorbig{P}{Q} = 0\), in \thmref{thm:total_brst_operator}.\footnote{We emphasize that we use the symbol \(\left [ \cdot , \cdot \right ]\) for the supercommutator: In particular, it denotes the anticommutator if both arguments are odd, cf.\ \defnref{defn:supercommutator}.} This is a non-trivial observation, as infinitesimal diffeomorphisms concern all particle fields and thus in particular the operator \(Q\). As a result, their sum \(D := P + Q\) is also a differential, which we call \emph{total BRST operator}. This then allows us again to identify the physical states as elements of the respective \(0\)-th cohomology class. In addition, we show that the sum over the suitably modified gauge fixing fermions \(\Upsilon := \stigma^{(1)} + \digamma \! \! _{\{ 1 \}}\) is again a gauge fixing fermion, which we call \emph{total gauge fixing fermion}, cf.\ \thmref{thm:total_gauge_fixing_fermion}. In particular, we obtain the complete gauge fixing and ghost Lagrange densities of (effective) Quantum General Relativity coupled to the Standard Model via \(D \Upsilon\). Finally, we also introduce the corresponding anti-BRST differentials in \defnsaref{defn:diffeomorphism_anti-brst_operator}{defn:gauge_anti-brst_operator} and show that all BRST and anti-BRST operators mutually anticommute in \colssaref{col:anti-diffeomorphism_brst_operator}{col:anti-gauge_brst_operator}{col:total_anti-brst_operator} together with the already mentioned result \thmref{thm:total_brst_operator}.

\subsection{The diffeomorphism complex}

In this subsection, we study the diffeomorphism BRST operator \(P\) together with the de Donder gauge fixing fermion \(\stigma\) and its linearized variant \(\stigma^{(1)}\):

\enter

\begin{defn} \label{defn:diffeomorphism_brst_operator}
	We define the diffeomorphism BRST operator \(P \in \mathfrak{X}_{(1,0)} \left ( \mathcal{B}_\Q \right )\) as the following odd vector field on the spacetime-matter bundle with graviton-ghost degree 1:
	\begin{equation}
	\begin{split}
		P & := \left ( \frac{1}{\zeta} \partial_\mu \gravitonghost_\nu + \frac{1}{\zeta} \partial_\nu \gravitonghost_\mu - 2 C_\rho \tensor{\Gamma}{^\rho _\mu _\nu} \right ) \frac{\partial}{\partial h_{\mu \nu}} + \varkappa C^\rho \left ( \partial_\rho C_\sigma \right ) \frac{\partial}{\partial C_\sigma} + \frac{1}{\zeta} B^\sigma \frac{\partial}{\partial \overline{C}^\sigma} \\
		& \phantom{:=} + \varkappa \left ( C^\rho \big ( \partial_\rho A_\mu^a \big ) + \left ( \partial_\mu C^\rho \right ) A_\rho^a \right ) \frac{\partial}{\partial A_\mu^a} \\
		& \phantom{:=} + \varkappa C^\rho \left ( \partial_\rho c^a \right ) \frac{\partial}{\partial c^a} + \varkappa C^\rho \left ( \partial_\rho \overline{c}^a \right ) \frac{\partial}{\partial \overline{c}^a} + \varkappa C^\rho \left ( \partial_\rho b^a \right ) \frac{\partial}{\partial b^a} \\
		& \phantom{:=} + \varkappa C^\rho \left ( \partial_\rho \Phi \right ) \frac{\partial}{\partial \Phi} + \varkappa C^\rho \left ( \nabla^{\Sigma M}_\rho \Psi + \frac{\imaginary}{4} \left ( \partial_\mu X_\nu - \partial_\nu X_\mu \right ) e^{\mu m} e^{\nu n} \boldsymbol{\sigma}_{mn} \Psi \right ) \frac{\partial}{\partial \Psi}
	\end{split}
	\end{equation}
	Equivalently, its action on fundamental particle fields is given as follows:
	{\allowdisplaybreaks
	\begin{subequations}
	\begin{align}
		P h_{\mu \nu} & := \frac{1}{\zeta} \partial_\mu \gravitonghost_\nu + \frac{1}{\zeta} \partial_\nu \gravitonghost_\mu - 2 C_\rho \tensor{\Gamma}{^\rho _\mu _\nu} \\
		P \gravitonghost_\rho & := \varkappa \gravitonghost^\sigma \left ( \partial_\sigma \gravitonghost_\rho \right ) \\
		P \overline{\gravitonghost}^\rho & := \frac{1}{\zeta} B^\rho \\
		P B^\rho & := 0 \\
		P \eta_{\mu \nu} & := 0 \\
		P \partial_\mu & := 0 \\
		P \tensor{\Gamma}{^\rho _\mu _\nu} & := \left ( C^\sigma \big ( \partial_\sigma \tensor{\Gamma}{^\rho _\mu _\nu} \big ) + \left ( \partial_\mu C^\sigma \right ) \tensor{\Gamma}{^\rho _\sigma _\nu} + \left ( \partial_\nu C^\sigma \right ) \tensor{\Gamma}{^\rho _\mu _\sigma} - \left ( \partial_\sigma C^\rho \right ) \tensor{\Gamma}{^\sigma _\mu _\nu} + \partial_\mu \partial_\nu C^\rho \right ) \\
		P A_\mu^a & := \varkappa \left ( C^\rho \big ( \partial_\rho A_\mu^a \big ) + \left ( \partial_\mu C^\rho \right ) A_\rho^a \right ) \\
		P c^a & := \varkappa C^\rho \left ( \partial_\rho c^a \right ) \\
		P \overline{c}_a & := \varkappa C^\rho \left ( \partial_\rho \overline{c}^a \right ) \\
		P b_a & := \varkappa C^\rho \left ( \partial_\rho b^a \right ) \\
		P \delta_{ab} & := 0 \\
		P \Phi & := \varkappa C^\rho \left ( \partial_\rho \Phi \right ) \\
		P \Psi & := \varkappa C^\rho \left ( \nabla^{\Sigma M}_\rho \Psi + \frac{\imaginary}{4} \left ( \partial_\mu X_\nu - \partial_\nu X_\mu \right ) e^{\mu m} e^{\nu n} \left ( \boldsymbol{\sigma}_{mn} \cdot \Psi \right ) \right )
	\end{align}
	\end{subequations}
	}%
	We remark that the action of \(P\) on all fields \(\particlefield \notin \big \{ h, C, \overline{C}, B, \eta \big \}\) is given via the geodesic Lie derivative with respect to \(C\) and rescaled via \(\varkappa\), i.e.\ \(P \particlefield \equiv \varkappa \Lie_C \particlefield\).\footnote{We remark that the Lie derivative of spinor fields is a non-trivial notion: This is due to the fact that the construction of the spinor bundle depends directly on the metric, which is itself affected by the Lie derivative. We use the formula of Kosmann \cite{Kosmann}, which uses the connection on the spinor bundle. It can be shown, however, that the result is indeed independent of the chosen connection. We remark that this formula can be embedded into the construction of a universal spinor bundle cf.\ \cite{Mueller_Nowaczyk}.}
\end{defn}

\enter

\begin{prop} \label{prop:p-cohomological-vector-field}
	Given the situation of \defnref{defn:diffeomorphism_brst_operator}, we have
	\begin{equation}
		\commutatorbig{P}{P} \equiv 2 P^2 \equiv 0 \, ,
	\end{equation}
	i.e.\ \(P\) is a cohomological vector field with respect to the graviton-ghost degree.
\end{prop}

\begin{proof}
	This follows immediately after a short calculation using the Jacobi identity.
\end{proof}

\enter

\begin{lem} \label{lem:p_tensor_densities}
	Let \(\mathfrak{f} \in C^\infty_{\{w\}, (0,0)} \left ( \mathcal{B}_\Q , \mathbb{R} \right )\) be a non-constant local functional of tensor density weight \(w \in \mathbb{R}\).\footnote{I.e.\ \(\mathfrak{f} \equiv \left ( - \dt{g} \right )^{w/2} f\) for an ordinary functional \(f \in C^\infty_{\{0\}, (0,0)} \left ( \mathcal{B}_\Q , \mathbb{R} \right )\).} Then we have
	\begin{equation}
		P \mathfrak{f} \simeq_\textup{TD} 0
	\end{equation}
	if and only if \(w = 1\), where \(\simeq_\textup{TD}\) means equality modulo total derivatives.\footnote{The same statement also holds for the diffeomorphism anti-BRST operator \(\overline{P}\), cf.\ \defnref{defn:diffeomorphism_anti-brst_operator}.}
\end{lem}

\begin{proof}
	We calculate
	\begin{equation}
	\begin{split}
		P \mathfrak{f} & = \Lie_C \mathfrak{f} \\
		& = C^\rho \left ( \partial_\rho \mathfrak{f} \right ) + w \left ( \partial_\rho C^\rho \right ) \mathfrak{f} \\
		& = \partial_\rho \left ( C^\rho \mathfrak{f} \right ) + \left ( w - 1 \right ) \left ( \partial_\rho C^\rho \right ) \mathfrak{f} \, ,
		\end{split}
	\end{equation}
	which is a total derivative if and only if \(w = 1\), and thus proves the claimed statement.
\end{proof}

\enter

\begin{prop} \label{prop:de_donder_gauge_fixing_fermion}
	The Quantum General Relativity gauge fixing Lagrange density and its accompanying ghost Lagrange density
	\begin{equation}
	\begin{split}
		\mathcal{L}_\textup{QGR-GF} + \mathcal{L}_\textup{QGR-Ghost} & = - \frac{1}{4 \gcoupling^2 \zeta}  g^{\mu \nu} \deDonder_\mu \deDonder_\nu \dif V_g \\ & \phantom{=} - \frac{1}{2 \zeta} g_{\mu \nu} g^{\rho \sigma} \overline{C}^\mu \left ( \partial_\rho \partial_\sigma C^\nu \right ) \dif V_g \\ & \phantom{=} - \frac{1}{2} \overline{C}^\mu \left ( \left ( \partial_\nu \deDonder_\mu \right ) C^\nu + \deDonder_\nu \left ( \partial_\mu C^\nu \right ) \right ) \dif V_g
	\end{split}
	\end{equation}
	for the de Donder gauge fixing functional \(\deDonder_\mu := g^{\rho \sigma} \Gamma_{\rho \sigma \mu}\) can be obtained from the following gauge fixing fermion \(\stigma \in \mathcal{C}_{(-1,0)} \left ( \mathcal{B}_\Q \right )\)
	\begin{equation}
		\stigma := \frac{1}{2} \overline{C}^\rho \left ( \frac{1}{\varkappa} \deDonder_\rho + \frac{1}{2} B_\rho \right ) \dif V_g
	\end{equation}
	via \(P \stigma\).
\end{prop}

\begin{proof}
	The claimed statement follows directly from the calculations
	\begin{subequations}
	\begin{align}
		\begin{split}
			P \stigma & = \frac{1}{2 \zeta} B^\rho \left ( \frac{1}{\varkappa} \deDonder_\rho + \frac{1}{2} B_\rho \right ) \dif V_g  - \frac{1}{2 \varkappa} \overline{C}^\rho \left ( P \deDonder_\rho \right ) \dif V_g \\ & \phantom{=} - \frac{1}{2} \overline{C}^\rho \left ( \frac{1}{\varkappa} \deDonder_\rho + \frac{1}{2} B_\rho \right ) \left ( P \dif V_g \right ) \label{eqn:p_stigma}
		\end{split}
		\intertext{with}
		\begin{split}
			P \deDonder_\rho & = P \left ( g^{\mu \nu} \Gamma_{\rho \mu \nu} \right ) \\
			& = C^\sigma \left ( \partial_\sigma \deDonder_\rho \right ) + \left ( \partial_\rho C^\sigma \right ) \deDonder_\sigma + g^{\mu \nu} g_{\rho \sigma} \left ( \partial_\mu \partial_\nu C^\sigma \right )
		\end{split}
		\intertext{along with the total derivative}
		P \dif V_g & = \partial_\rho \left ( C^\rho \dif V_g \right )
		\intertext{and then finally eliminating the Lautrup--Nakanishi auxiliary field \(B^\rho\) by inserting its equation of motion}
		\operatorname{EoM} \left ( B_\rho \right ) & = - \frac{1}{\varkappa} \deDonder_\rho \, ,
	\end{align}
	\end{subequations}
	which are obtained as usual via an Euler--Lagrange variation of \eqnref{eqn:p_stigma}.
\end{proof}

\enter

\begin{col} \label{col:linearized_de_donder_gauge_fixing_fermion}
	Given the situation of \propref{prop:de_donder_gauge_fixing_fermion}. Then the linearized de Donder gauge fixing and ghost Lagrange densities read
	\begin{equation}
	\begin{split}
		\mathcal{L}_\textup{QGR-GF} + \mathcal{L}_\textup{QGR-Ghost} & = - \frac{1}{4 \gcoupling^2 \zeta}  \eta^{\mu \nu} \deDonder^{(1)}_\mu \deDonder^{(1)}_\nu \dif V_\eta \\ & \phantom{=} - \frac{1}{2 \zeta} \eta^{\rho \sigma} \overline{C}^\mu \left ( \partial_\rho \partial_\sigma C_\mu \right ) \dif V_\eta \\ & \phantom{=} - \frac{1}{2} \eta^{\rho \sigma} \overline{C}^\mu \left ( \partial_\mu \big ( \tensor{\Gamma}{^\nu _\rho _\sigma} C_\nu \big ) - 2 \partial_\rho \big ( \tensor{\Gamma}{^\nu _\mu _\sigma} C_\nu \big ) \right ) \dif V_\eta
	\end{split}
	\end{equation}
	with the linearized de Donder gauge fixing functional \(\deDonder^{(1)}_\mu := \eta^{\rho \sigma} \Gamma_{\rho \sigma \mu}\). They can be obtained from the following gauge fixing fermion \(\stigma^{(1)} \in \mathcal{C}_{(-1,0)} \left ( \mathcal{B}_\Q \right )\)
	\begin{equation}
		\stigma^{(1)} := \frac{1}{2} \overline{C}^\rho \left ( \frac{1}{\varkappa} \deDonder^{(1)}_\rho + \frac{1}{2} B_\rho \right ) \dif V_\eta \label{eqn:linearized_de_donder_gauge_fixing_fermion}
	\end{equation}
	via \(P \stigma^{(1)}\).
\end{col}

\begin{proof}
	This can be shown analogously to the proof of \propref{prop:de_donder_gauge_fixing_fermion}.
\end{proof}

\enter

\begin{rem}
	In the following, we will only use the linearized de Donder gauge fixing and ghost Lagrange densities from \colref{col:linearized_de_donder_gauge_fixing_fermion}. The reason is that the perturbative expansion becomes simpler if the gauge fixing functional does only contribute to the propagator. Nevertheless, the complete de Donder gauge fixing can also be useful, as it does not depend on the choice of a background metric.
\end{rem}

\enter

\begin{defn} \label{defn:diffeomorphism_anti-brst_operator}
	Given the situation of \defnref{defn:diffeomorphism_brst_operator}, we additionally define the diffeomorphism anti-BRST operator \(\overline{P} \in \mathfrak{X}_{(-1,0)} \left ( \mathcal{B}_\Q \right )\) as the following odd vector field on the spacetime-matter bundle with graviton-ghost degree -1:
	\begin{subequations}
	\begin{align}
		\overline{P} & := \eval{P}_{C \rightsquigarrow \overline{C}}
		\intertext{together with the following additional changes}
		\overline{P} \gravitonghost_\rho & := - \frac{1}{\zeta} B^\rho + \varkappa \overline{C}^\sigma \left ( \partial_\sigma C_\rho \right ) + \varkappa \big ( \partial_\rho \overline{C}^\sigma \big ) C_\sigma \\
		\overline{P} \overline{\gravitonghost}^\rho & := \varkappa \overline{\gravitonghost}^\sigma \big ( \partial_\sigma \overline{\gravitonghost}_\rho \big ) \\
		\overline{P} B^\rho & := \varkappa \overline{C}^\sigma \left ( \partial_\sigma B_\rho \right ) - \varkappa \big ( \partial_\sigma \overline{C}^\rho \big ) B^\sigma
	\end{align}
	\end{subequations}
\end{defn}

\enter

\begin{col} \label{col:anti-diffeomorphism_brst_operator}
	Given the situation of \defnref{defn:diffeomorphism_anti-brst_operator}, we have
	\begin{align}
		\commutatorbig{\overline{P}}{\overline{P}} & \equiv 2 \overline{P}^2 \equiv 0
		\intertext{and}
		\commutatorbig{P}{\overline{P}} & \equiv P \circ \overline{P} + \overline{P} \circ P \equiv 0 \, ,
	\end{align}
	i.e.\ \(\overline{P}\) is a homological vector field with respect to the graviton-ghost degree that anticommutes with \(P\).
\end{col}

\begin{proof}
	This statement can be shown analogously to \propref{prop:p-cohomological-vector-field}.
\end{proof}

\subsection{The gauge complex}

In this subsection, we study the gauge BRST operator \(Q\) together with the Lorenz gauge fixing fermion \(\digamma\) and its density variant \(\digamma \! \! _{\{ 1 \}}\):

\enter

\begin{defn} \label{defn:gauge_brst_operator}
	We define the gauge BRST operator \(Q \in \mathfrak{X}_{(0,1)} \left ( \mathcal{B}_\Q \right )\) as the following odd vector field on the spacetime-matter bundle with gauge ghost degree 1:
	\begin{equation}
	\begin{split}
		Q & := \left ( \frac{1}{\xi} \partial_\mu c^a + \mathrm{g} \tensor{f}{^a _b _c} c^b A_\mu^c \right ) \frac{\partial}{\partial A_\mu^a} + \frac{\mathrm{g}}{2} \tensor{f}{^a _b _c} c^b c^c \frac{\partial}{\partial c^a} + \frac{1}{\xi} b_a \frac{\partial}{\partial \overline{c}_a} \\ & \phantom{:=} + \mathrm{g} c^a \left ( \mathfrak{H}_a \cdot \Phi \right ) \frac{\partial}{\partial \Phi} + \mathrm{g} c^a \left ( \mathfrak{S}_a \cdot \Psi \right ) \frac{\partial}{\partial \Psi}
	\end{split}
	\end{equation}
	Equivalently, its action on fundamental particle fields is given as follows:
	{\allowdisplaybreaks
	\begin{subequations}
	\begin{align}
		Q A_\mu^a & := \frac{1}{\xi} \partial_\mu c^a + \mathrm{g} \tensor{f}{^a _b _c} c^b A_\mu^c \\
		Q c^a & := \frac{\mathrm{g}}{2} \tensor{f}{^a _b _c} c^b c^c \\
		Q \overline{c}_a & := \frac{1}{\xi} b_a \\
		Q b_a & := 0 \\
		Q \delta_{ab} & := 0 \\
		Q h_{\mu \nu} & := 0 \\
		Q \gravitonghost_\rho & := 0 \\
		Q \overline{\gravitonghost}^\rho & := 0 \\
		Q B^\rho & := 0 \\
		Q \eta_{\mu \nu} & := 0 \\
		Q \partial_\mu & := 0 \\
		Q \tensor{\Gamma}{^\rho _\mu _\nu} & := 0 \\
		Q \Phi & := \mathrm{g} c^a \left ( \mathfrak{H}_a \cdot \Phi \right ) \\
		Q \Psi & := \mathrm{g} c^a \left ( \mathfrak{S}_a \cdot \Psi \right )
	\end{align}
	\end{subequations}
	}%
	We remark that the action of \(Q\) on all fields \(\particlefield \notin \big \{ A, c, \overline{c}, b, \delta \big \}\) is given via the gauge Lie derivative with respect to \(c\) and rescaled via \(\mathrm{g}\), i.e.\ \(Q \particlefield \equiv \mathrm{g} \lie_c \particlefield\).
\end{defn}

\enter

\begin{prop} \label{prop:q-cohomological-vector-field}
	Given the situation of \defnref{defn:gauge_brst_operator}, we have
	\begin{equation}
		\commutatorbig{Q}{Q} \equiv 2 Q^2 \equiv 0 \, ,
	\end{equation}
	i.e.\ \(Q\) is a cohomological vector field with respect to the gauge ghost degree.
\end{prop}

\begin{proof}
	This follows immediately after a short calculation using the Jacobi identity.
\end{proof}

\begin{prop} \label{prop:lorenz_gauge_fixing_fermion}
	The Quantum Yang--Mills theory gauge fixing Lagrange density and its accompanying ghost Lagrange density
	\begin{equation}
	\begin{split}
		\mathcal{L}_{\textup{QYM-GF}} + \mathcal{L}_{\textup{QYM-Ghost}} & = - \frac{1}{2 \mathrm{g}^2 \xi}  \delta_{a b} \linLorenz^a \linLorenz^b \dif V_\eta \\ & \phantom{=} - \frac{1}{\xi} \eta^{\mu \nu} \overline{c}_a \left ( \partial_\mu \partial_\nu c^a \right ) \dif V_\eta \\ & \phantom{=} - \mathrm{g} \eta^{\mu \nu} \tensor{f}{^a _b _c} \overline{c}_a \left ( \partial_\mu \big ( c^b A^c_\nu \big ) \right ) \dif V_\eta
	\end{split}
	\end{equation}
	for the Minkowski metric Lorenz gauge fixing functional \(\linLorenz^a := \mathrm{g} \eta^{\mu \nu} \left ( \partial_\mu A_\nu^a \right )\) can be obtained from the following gauge fixing fermion \(\digamma \in \mathcal{C}_{\{ 0 \}, (0,-1)} \left ( \mathcal{B}_\Q \right )\)
	\begin{equation}
		\digamma := \overline{c}_a \left ( \frac{1}{\mathrm{g}} \linLorenz^a + \frac{1}{2} b^a \right ) \dif V_\eta
	\end{equation}
	via \(Q \digamma\).
\end{prop}

\begin{proof}
	The claimed statement follows directly from the calculations
	\begin{subequations}
	\begin{align}
		Q \digamma & = \frac{1}{\xi} b_a \left ( \frac{1}{\mathrm{g}} \linLorenz^a + \frac{1}{2} b^a \right ) \dif V_\eta - \frac{1}{\mathrm{g}} \overline{c}_a \left ( Q \linLorenz^a \right ) \dif V_\eta \label{eqn:q_digamma_eta}
		\intertext{with}
		\begin{split}
		Q \linLorenz^a & = \mathrm{g} \eta^{\mu \nu} \partial_\mu \left ( Q A_\nu^a \right ) \\
		& = \mathrm{g} \eta^{\mu \nu} \partial_\mu \left ( \frac{1}{\xi} \partial_\nu c^a + \mathrm{g} \tensor{f}{^a _b _c} c^b A_\nu^c \right )
		\end{split}
		\intertext{and then finally eliminating the Lautrup--Nakanishi auxiliary field \(b_a\) by inserting its equation of motion}
		\operatorname{EoM} \left ( b_a \right ) & = - \frac{1}{\mathrm{g}} \linLorenz^a \, ,
	\end{align}
	\end{subequations}
	which are obtained as usual via an Euler--Lagrange variation of \eqnref{eqn:q_digamma_eta}.
\end{proof}

\enter

\begin{col} \label{col:covariant_lorenz_gauge_fixing_fermion}
	Given the situation of \propref{prop:lorenz_gauge_fixing_fermion}. Then the covariant Lorenz gauge fixing and ghost Lagrange densities read
	\begin{equation}
	\begin{split}
		\mathcal{L}_{\textup{QYM-GF}} + \mathcal{L}_{\textup{QYM-Ghost}} & = - \frac{1}{2 \mathrm{g}^2 \xi}  \delta_{a b} \Lorenz^a L^b \dif V_g \\ & \phantom{=} - \frac{1}{\xi} g^{\mu \nu} \overline{c}_a \left ( \nabla^{TM}_\mu \left ( \partial_\nu c^a \right ) \right ) \dif V_g \\ & \phantom{:=} - \mathrm{g} g^{\mu \nu} \tensor{f}{^a _b _c} \overline{c}_a \left ( \nabla^{TM}_\mu \big ( c^b A^c_\nu \big ) \right ) \dif V_g \, ,
	\end{split}
	\end{equation}
	with the covariant Lorenz gauge fixing functional \(\Lorenz^a := \mathrm{g} g^{\mu \nu} \big ( \nabla^{TM}_\mu A_\nu^a \big )\). They can be obtained from the following gauge fixing fermion \(\digamma \! \! _{\{ 1 \}} \in \mathcal{C}_{\{ 1 \}, (0,-1)} \left ( \mathcal{B}_\Q \right )\)
	\begin{equation}
		\digamma \! \! _{\{ 1 \}} := \overline{c}_a \left ( \frac{1}{\mathrm{g}} \Lorenz^a + \frac{1}{2} b^a \right ) \dif V_g \label{eqn:covariant_lorenz_gauge_fixing_fermion}
	\end{equation}
	via \(Q \digamma \! \! _{\{ 1 \}}\).
\end{col}

\begin{proof}
	This can be shown analogously to the proof of \propref{prop:lorenz_gauge_fixing_fermion}.
\end{proof}

\enter

\begin{defn} \label{defn:gauge_anti-brst_operator}
	Given the situation of \defnref{defn:gauge_brst_operator}, we additionally define the gauge anti-BRST operator \(\overline{Q} \in \mathfrak{X}_{(0,-1)} \left ( \mathcal{B}_\Q \right )\) as the following odd vector field on the spacetime-matter bundle with gauge ghost degree -1:
	\begin{subequations}
	\begin{align}
		\overline{Q} & := \eval{Q}_{c \rightsquigarrow \overline{c}}
		\intertext{together with the following additional changes}
		\overline{Q} c^a & := - \frac{1}{\xi} b_a + \mathrm{g} \tensor{f}{^a _b _c} \overline{c}^b c^c \\
		\overline{Q} \overline{c}_a & := \frac{\mathrm{g}}{2} \tensor{f}{^a _b _c} \overline{c}^b \overline{c}^c \\
		\overline{Q} b_a & := \mathrm{g} \tensor{f}{^a _b _c} \overline{c}^b b^c
	\end{align}
	\end{subequations}
\end{defn}

\enter

\begin{col} \label{col:anti-gauge_brst_operator}
	Given the situation of \defnref{defn:gauge_anti-brst_operator}, we have
	\begin{align}
		\commutatorbig{\overline{Q}}{\overline{Q}} & \equiv 2 \overline{Q}^2 \equiv 0
		\intertext{and}
		\commutatorbig{Q}{\overline{Q}} & \equiv Q \circ \overline{Q} + \overline{Q} \circ Q \equiv 0 \, ,
	\end{align}
	i.e.\ \(\overline{Q}\) is a homological vector field with respect to the gauge ghost degree that anticommutes with \(Q\).
\end{col}

\begin{proof}
	This statement can be shown analogously to \propref{prop:q-cohomological-vector-field}.
\end{proof}

\subsection{The double complex}

In this subsection, we show that the two BRST operators \(P\) and \(Q\) anticommute and thus give rise to the \emph{total BRST operator} as the sum \(D := P + Q\). Additionally, we show that each gauge theory gauge fixing fermion can be modified uniquely to become a tensor density of weight \(w = 1\). This is a useful choice, as then the graviton-ghosts decouple from matter of the Standard Model. Finally, we introduce the \emph{total gauge fixing fermion} as the sum \(\Upsilon := \stigma^{(1)} + \digamma \! \! _{\{ 1 \}}\), where \(\stigma^{(1)}\) is the gauge fixing fermion corresponding to the linearized de Donder gauge fixing and \(\digamma \! \! _{\{ 1 \}}\) is the gauge fixing fermion corresponding to the covariant Lorenz gauge fixing. This setup allows us to create the complete gauge fixing and ghost Lagrange densities of (effective) Quantum General Relativity coupled to the Standard Model via \(D \Upsilon\):

\enter

\begin{thm} \label{thm:total_brst_operator}
	Given the two BRST operators \(P \in \mathfrak{X}_{(1,0)} \left ( \mathcal{B}_\Q \right )\) and \(Q \in \mathfrak{X}_{(0,1)} \left ( \mathcal{B}_\Q \right )\) from \defnref{defn:diffeomorphism_brst_operator} and \defnref{defn:gauge_brst_operator}, respectively. Then we have
	\begin{align}
		\commutatorbig{P}{Q} & \equiv P \circ Q + Q \circ P \equiv 0 \, ,
		\intertext{i.e.\ their sum}
		D & := P + Q
	\end{align}
	is also a cohomological vector field with respect to the total ghost degree. We call \(D \in \mathfrak{X}_{(1)} \left ( \mathcal{B}_\Q \right )\) the total BRST operator.
\end{thm}

\begin{proof}
	We show this statement by an explicit calculation:
	\begin{equation}
	\begin{split}
		P \circ Q & = \varkappa \bigg ( \frac{1}{\xi} \left ( \partial_\mu C^\rho \right ) \left ( \partial_\rho c^a \right ) + \frac{1}{\xi} C^\rho \left ( \partial_\mu \partial_\rho c^a \right ) + \mathrm{g} \tensor{f}{^a _b _c} C^\rho \big ( \partial_\rho c^b \big ) A_\mu^c \\
		& \phantom{= \varkappa \bigg (} + \mathrm{g} \tensor{f}{^a _b _c} C^\rho c^b \big ( \partial_\rho A_\mu^c \big ) + \mathrm{g} \tensor{f}{^a _b _c} \left ( \partial_\mu C^\rho \right ) c^b A_\rho^c \bigg ) \frac{\partial}{\partial A_\mu^a} \\
		& \phantom{=} + \frac{\varkappa \mathrm{g}}{2} \tensor{f}{^a _b _c} C^\rho \bigg ( \! \big ( \partial_\rho c^b \big ) c^c + c^b \big ( \partial_\rho c^c \big ) \! \bigg ) \frac{\partial}{\partial c^a} + \frac{\varkappa}{\xi} C^\rho \left ( \partial_\rho b^a \right ) \frac{\partial}{\partial \overline{c}_a} \\
		& \phantom{=} + \varkappa \mathrm{g} C^\rho c^a \mathfrak{H}_a \cdot \left ( \partial_\rho \Phi \right ) \frac{\partial}{\partial \Phi} \\
		& \phantom{=} + \varkappa \mathrm{g} C^\rho c^a \mathfrak{S}_a \cdot \left ( \nabla^{\Sigma M}_\rho \Psi + \frac{\imaginary}{4} \left ( \partial_\mu X_\nu - \partial_\nu X_\mu \right ) e^{\mu m} e^{\nu n} \left ( \boldsymbol{\sigma}_{mn} \cdot \Psi \right ) \right ) \frac{\partial}{\partial \Psi} \\
		& = - Q \circ P \, ,
	\end{split}
	\end{equation}
	where we have used \(C^\rho c^a \equiv - c^a C^\rho\) and \(\commutatorbig{\mathfrak{S}_a}{\boldsymbol{\sigma}_{mn}} \equiv 0\).
\end{proof}

\enter

\begin{col} \label{col:total_anti-brst_operator}
	Given the two anti-BRST operators \(\overline{P} \in \mathfrak{X}_{(-1,0)} \left ( \mathcal{B}_\Q \right )\) and \(\overline{Q} \in \mathfrak{X}_{(0,-1)} \left ( \mathcal{B}_\Q \right )\) from \defnref{defn:diffeomorphism_anti-brst_operator} and \defnref{defn:gauge_anti-brst_operator}, respectively. Then we have
	\begin{align}
		\commutatorbig{\overline{P}}{\overline{Q}} & \equiv \overline{P} \circ \overline{Q} + \overline{Q} \circ \overline{P} \equiv 0 \, ,
		\intertext{i.e.\ their sum}
		\overline{D} & := \overline{P} + \overline{Q}
	\end{align}
	is also a homological vector field with respect to the total ghost degree. We call \(D \in \mathfrak{X}_{(-1)} \left ( \mathcal{B}_\Q \right )\) the total anti-BRST operator. Furthermore, given the two BRST operators \(P \in \mathfrak{X}_{(1,0)} \left ( \mathcal{B}_\Q \right )\) and \(Q \in \mathfrak{X}_{(0,1)} \left ( \mathcal{B}_\Q \right )\) from \defnref{defn:diffeomorphism_brst_operator} and \defnref{defn:gauge_brst_operator}, respectively. Then we have additionally
	\begin{align}
		\commutatorbig{P}{\overline{Q}} & \equiv P \circ \overline{Q} + \overline{Q} \circ P \equiv 0
		\intertext{and}
		\commutatorbig{\overline{P}}{Q} & \equiv \overline{P} \circ Q + Q \circ \overline{P} \equiv 0 \, ,
	\end{align}
	i.e.\ all BRST and anti-BRST operators mutually anticommute.
\end{col}

\begin{proof}
	This statement can be shown analogously to \thmref{thm:total_brst_operator}.
\end{proof}

\enter

\begin{thm} \label{thm:no-couplings-grav-ghost-matter-sm}
	The graviton-ghosts decouple from matter of the Standard Model if and only if the gauge fixing fermion of Yang--Mills theory is a tensor density of weight \(w = 1\).\footnote{Equivalently: If the gauge fixing and gauge ghost Lagrange densities of Yang--Mills theory are tensor densities of weight \(w = 1\).} In particular, every such gauge fixing fermion can be modified uniquely to satisfy said condition; the case of the Lorenz gauge fixing is given via \(\digamma \! \! _{\{ 1 \}}\) in Equation~(\ref{eqn:covariant_lorenz_gauge_fixing_fermion}).
\end{thm}

\begin{proof}
	We start with the first assertion: Let \(\omega\) be a tensor density of weight \(w = 1\). Then, due to \lemref{lem:p_tensor_densities}, we have
	\begin{equation}
		P \omega \simeq_\text{TD} 0 \, ,
	\end{equation}
	where \(\simeq_\text{TD}\) means equality modulo total derivatives. Furthermore, due to \thmref{thm:total_brst_operator}, we have
	\begin{equation}
		\left ( P \circ Q \right ) = - \left ( Q \circ P \right ) \, ,
	\end{equation}
	which implies
	\begin{equation}
		\left ( P \circ Q \right ) \omega \simeq_\text{TD} 0 \, .
	\end{equation}
	Thus, on the level of Lagrange densities, we have
	\begin{align}
		P \mathcal{L}_\text{QYM} & \simeq_\text{TD} 0
		\intertext{if and only if \(\mathcal{L}_\text{QYM}\) is a tensor density of weight \(w = 1\). Using the above reasoning, this is equivalent to}
		\mathcal{L}_\text{QYM} & := \mathcal{L}_\text{YM} + Q \chi
	\end{align}
	with \(\chi\) being a gauge fixing fermion of tensor density weight \(w = 1\). It directly follows that there are no interactions between graviton-ghosts and the fields \(\particlefield \in \big \{ A, c, \overline{c}, b, \Phi, \Psi \big \}\), given that \(\mathcal{L}_\text{YM}\) and the matter Lagrange densities are covariant, and thus tensor densities of weight \(w = 1\). The second assertion follows directly from the fact that any functional on the spacetime-matter bundle \(f \in C^\infty_{\{0\}, (0,0)} \left ( \mathcal{B}_\Q , \mathbb{R} \right )\) can be modified uniquely to obtain tensor density weight \(w = 1\) by the following replacements
	\begin{equation}
		f \rightsquigarrow \sqrt{- \dt{g}} \eval{f}_{\subalign{\eta_{\mu \nu} & \rightsquigarrow g_{\mu \nu} \\ \partial_\mu & \rightsquigarrow \nabla^{TM}_\mu \\}} \, ,
	\end{equation}
	which concludes the proof.
\end{proof}

\enter

\begin{thm} \label{thm:total_gauge_fixing_fermion}
	Given the total BRST operator \(D\) from \thmref{thm:total_brst_operator} and let \(\Upsilon \in \mathcal{C}_{(-1)} \left ( \mathcal{B}_\Q \right )\)
	\begin{equation}
		\Upsilon := \stigma^{(1)} + \digamma \! \! _{\{ 1 \}}
	\end{equation}
	be the sum of the gauge fixing fermions from \colsaref{col:linearized_de_donder_gauge_fixing_fermion}{col:covariant_lorenz_gauge_fixing_fermion}, respectively. Then the complete gauge fixing and ghost Lagrange densities for (effective) Quantum General Relativity coupled to the Standard Model can be generated via \(D \Upsilon\).
\end{thm}

\begin{proof}
	This follows directly from the calculation
	\begin{equation}
	\begin{split}
		D \Upsilon & = \big ( P + Q \big ) \big ( \stigma^{(1)} + \digamma \! \! _{\{ 1 \}} \big ) \\
		& = P \stigma^{(1)} + P \digamma \! \! _{\{ 1 \}} + Q \stigma^{(1)} + Q \digamma \! \! _{\{ 1 \}} \\
		& \simeq_\text{TD} P \stigma^{(1)} + Q \digamma \! \! _{\{ 1 \}} \, ,
	\end{split}
	\end{equation}
	where we have used \lemref{lem:p_tensor_densities} and \(\simeq_\text{TD}\) means equality modulo total derivatives.
\end{proof}

\chapter{Hopf algebraic renormalization} \label{chp:hopf_algebraic_renormalization}

In this chapter, we discuss the Hopf algebraic renormalization of gauge theories and gravity. This includes first some preliminaries on Connes--Kreimer renormalization theory. Then we discuss Quantum Field Theories that lead to an ill-defined renormalization Hopf algebra and possible solutions together with their physical implications. In addition, we introduce the notion of a \emph{predictive Quantum Field Theory} as a (possibly non-renormalizable) Quantum Field Theory which allows for a well-defined perturbative expansion. Next, we study combinatorial properties of the superficial degree of divergente in order to generalize known coproduct identities and a theorem of van Suijlekom so that they are applicable to (effective) Quantum General Relativity coupled to the Standard Model. This theorem relates (generalized) gauge symmetries to Hopf ideals and we close this chapter with criteria that ensure that these Hopf ideals are compatible with renormalized Feynman rules.

\section{Preliminaries of Hopf algebraic renormalization} \label{sec:preliminaries_of_hopf_algebraic_renormalization}

We start this chapter by briefly recalling the relevant definitions and notations from Hopf algebraic renormalization: We consider \(\Q\) to be a local Quantum Field Theory (QFT), i.e.\ a QFT given by a Lagrange functional. Then, in a nutshell, the renormalization Hopf algebra\footnote{We use the symbol \(\HQ\) by abuse of notation simultaneously for the vector space \(\HQ\) as well as for the complete renormalization Hopf algebra \((\HQ, m, \one, \Delta, \coone, S)\).} \(\HQ\) of a QFT \(\Q\) consists of a vector space \(\HQ\) with algebra structure \((\HQ, m, \one)\), coalgebra structure \((\HQ, \Delta, \coone)\) and antipode \(S \colon \HQ \to \HQ\). More precisely, given the set \(\GQ\) of 1PI Feynman graphs of \(\Q\), the vector space \(\HQ\) is defined as the vector space over \(\mathbb{Q}\) generated by the elements of the set \(\GQ\) and disjoint unions thereof. Then, the product \(m\) is simply given via disjoint union, with the empty graph as unit. The interesting structures are the coproduct \(\Delta\) and the antipode \(S\): It was realized by Kreimer that the organization of subdivergences of Feynman graphs can be encoded into a coalgebra structure on \(\HQ\) \cite{Kreimer_Hopf_Algebra}. Then, building upon this, Connes and Kreimer formulated the renormalized Feynman rules \(\Phi_\mathscr{R}\) as an algebraic Birkhoff decomposition with respect to the renormalization scheme \(\mathscr{R}\) \cite{Connes_Kreimer_0}. See \defnref{defn:renormalization_hopf_algebra} and \defnref{defn:fr_reg_ren_counterterm} for the formal definitions and \cite[Section 3]{Prinz_2} for a more detailed introduction using the same notations and conventions.\footnote{We remark that some of the introductory material in this section is borrowed from \cite[Section 3]{Prinz_2}.} This mathematical formulation of the renormalization operation allows for a precise analysis of symmetries via Hopf ideals. We want to deepen this viewpoint in the context of quantum gauge theories by generalizing results from \cite{Kreimer_Anatomy,vSuijlekom_QED,vSuijlekom_QCD,vSuijlekom_BV,Kreimer_vSuijlekom}. Finally, we also mention some detailed introductory texts \cite{Sweedler,Manchon,Figueroa_Gracia-Bondia,Guo,Panzer,Yeats}.

\enter

\begin{defn}[Multiset over a set] \label{defn:multiset_over_a_set}
	Let \(M\) and \(S\) be sets. The set \(M\) is called a multiset over \(S\), if \(M\) contains elements of \(S\) in arbitrary multiplicity. Then, the multiset
	\begin{equation} \label{eqn:multiset_over_a_set}
		\pi \, : \quad M \to S \, , \quad m_s \mapsto s \, ,
	\end{equation}
	where \(\pi\) projects the elements of \(M\) to \(S\), can be canonically identified with the set \((s, n_s) \in \widetilde{M}^S \subset S \times \mathbb{N}_0\), where the natural number \(n_s\) indicates the multiplicity of each element \(s\) in \(M\) (which can possibly be zero). We call \(\widetilde{M}^S\) the multiset representation of \(M\) over \(S\). Given the multiset representation \(\widetilde{M}^S\) of \(M\) over \(S\), we define the two projections
	\begin{subequations}
	\begin{align}
		\varsigma \, & : \quad \widetilde{M}^S \to S \, , \quad (s, n_s) \mapsto s
		\intertext{and}
		\varrho \, & : \quad \widetilde{M}^S \to \mathbb{N}_0 \, , \quad (s, n_s) \mapsto n_s \, .
	\end{align}
	\end{subequations}
	Additionally, if the elements in the set \(S\) are ordered, we define the corresponding multiset-vector as the vector \(\mathbf{n} := (n_1, \dots, n_\mathfrak{s})^\intercal \in \mathbb{N}_0^\mathfrak{s}\), where \(\mathfrak{s}\) is the cardinality of the set \(S\) and \(n_i\) denotes the multiplicity of the element \(s_i\) in \(M\). Furthermore, two multisets \(M_1\) and \(M_2\) over the same set \(S\) are called isomorphic, if each element \(s \in S\) has the same multiplicity \(n_s\) in either \(M_1\) and \(M_2\). In the following, we will always assume that the underlying set \(S\) is ordered and thus use the equivalence between multisets and their multiset-vectors.
\end{defn}

\enter

\begin{rem}
	Given the situation of \defnref{defn:multiset_over_a_set}, a multiset \(M\) over \(S\) and its multiset representation \(\widetilde{M}^S\) are in general different sets, as they might have different cardinalities. As an extreme example, every set can be seen as a multiset over the singleton. Therefore, its multiset representation consist only of the element \((*, n)\), where \(n\) is the cardinality of the set. On the other hand, a set viewed as a multiset over itself has the same cardinality as its underlying set, but its elements are distinctly marked. Accordingly, its multiset representation consist of elements \((s, 1)\), for each element \(s\) in the underlying set. In this spirit, a multiset \(M\) over a set \(S\) can be seen as a \(S\)-colored set, by means of the map \(\pi\) from \eqnref{eqn:multiset_over_a_set}.
\end{rem}

\enter

\begin{defn}[Residue, amplitude and coupling constant set] \label{defn:residue_amplitude_and_coupling_constant_set}
	Let \(\Q\) be a QFT given via the Lagrange density \(\mathcal{L}_\Q\). Then each monomial in \(\mathcal{L}_\Q\) describes either a fundamental interaction or a propagation of the involved particles. We collect this information in two sets, called vertex residue set \(\RQO\) and edge residue set \(\RQI\), as follows: The first set consists of all fundamental interactions and the second set consists of all propagators, or, equivalently, particle types of \(\Q\). Finally, the residue set is then defined as the disjoint union
	\begin{equation}
		\RQ := \RQO \sqcup \RQI \, .
	\end{equation}
	We denote the cardinality of the vertex set via \(\mathfrak{v}_\Q := \# \RQO\). Furthermore, we define the set of amplitudes \(\AQ\) as the set containing all possible external leg structures of 1PI Feynman graphs. In particular, it is given as the disjoint union
	\begin{equation}
		\AQ := \RQ \sqcup \mathcal{Q}_\Q \, ,
	\end{equation}
	where \(\mathcal{Q}_\Q\) denotes the set of pure quantum corrections, that is, interactions which are only possible via trees or Feynman graphs, but not directly via residues in the set \(\RQ\). If \(\Q\) is a quantum gauge theory, we add additional labels to the edge-types: One for the physical degrees of freedom and at least one for the unphysical degrees of freedom, cf.\ \remref{rem:longitudinal_and_transversal_gauge_fields} and \cite{Prinz_9}. Moreover, we denote by \(\qQ\) the set of physical coupling constants and, if present, gauge fixing parameters appearing in the Lagrange density \(\mathcal{L}_\Q\). Finally, we define the function
	\begin{equation} \label{eqn:coupling-coloring_function}
		\theta \, : \quad \AQ \to \qQ \, , \quad r \mapsto \begin{cases} q_r := q_v \left ( \prod_{e \in E \left ( v \right )} \sqrt{q_e} \right ) & \text{if \(r = v \in \RQO\)} \\ 1 & \text{else, i.e.\ \(r \in \left ( \AQ \setminus \RQO \right )\)} \end{cases} \, ,
	\end{equation}
	where \(q_v\) denotes the coupling constant that scales the vertex-type \(v\), \(q_e\) denotes the gauge fixing parameter that is associated to the edge-type \(e\) if it is unphysical and finally \(E \left ( v \right )\) denotes the set of edges that are attached to the vertex \(v\). We denote the cardinality of the set of physical coupling constants via \(\mathfrak{q}_\Q := \# \qQ\).
\end{defn}

\enter

\begin{defn}[Transversal structure] \label{defn:transversal_structure}
	Let \(\Q\) be a quantum gauge theory. Then each independent gauge fixing term induces a longitudinal projection operator \(\boldsymbol{L}\) for the propagator of the corresponding gauge field. Together with the respective identity operator \(\boldsymbol{I}\) we define the associated transversal projection operator \(\boldsymbol{T}\) via
	\begin{equation}
		\boldsymbol{T} := \boldsymbol{I} - \boldsymbol{L} \, .
	\end{equation}
	We refer to the set \(\set{\boldsymbol{L}, \boldsymbol{I}, \boldsymbol{T}}\) as transversal structure. Additionally, let \(\mathfrak{f}_\Q\) denote the number of independent gauge fixing terms of \(\Q\).\footnote{This includes in particular the coupling of gravity to gauge theories, which requires independent gauge fixing terms for the diffeomorphism invariance and the gauge invariance, cf.\ e.g.\ \cite{Prinz_2,Prinz_4}. With that we also obtain two separate transversal structures: \(\set{L, I, T}\) for the Quantum Yang--Mills theory part and \(\set{\bbL, \bbI, \bbT}\) for the (effective) Quantum General Relativity part, cf.\ \eqnsref{eqn:projection_tensors_qym} and \eqnsref{eqn:projection_tensors_qgr}.} Then we consider the union
	\begin{equation}
		\mathcal{T}_\Q := \bigcup_{k = 1}^{\mathfrak{f}_\Q} \set{\boldsymbol{L}, \boldsymbol{I}, \boldsymbol{T}}_k
	\end{equation}
	and refer to it as the transversal structure of \(\Q\).
\end{defn}

\enter

\begin{rem} \label{rem:longitudinal_and_transversal_gauge_fields}
	The `physical' and `unphysical' labels together with the particle-type labels of \defnref{defn:residue_amplitude_and_coupling_constant_set} connect as follows to the physics of quantum gauge theories: Physical particle-types are transversal gauge field edges, canceled ghost field edges and matter field edges, respectively. Contrary, unphysical particle-types are longitudinal or canceled gauge field edges, ghost field edges and canceled matter field edges, respectively. Thus, our `physical' and `unphysical' labels are related to cancellation identities \cite{tHooft_Veltman,Citanovic,Sars_PhD,Kissler_Kreimer,Gracey_Kissler_Kreimer,Kissler} and the marking of edges in the construction of Feynman graph cohomology \cite{Kreimer_Sars_vSuijlekom,Berghoff_Knispel}. Additionally, if \(\Q\) has several longitudinal projection operators we need to keep track which longitudinal projection induced the cancellation of an edge. This is the reason why we extend our setup to allow for possibly several distinct `unphysical' labels, each of which is related to the corresponding gauge fixing parameter. This discussion will be studied in detail in \cite{Prinz_9}, cf.\ \cite{Prinz_5,Prinz_7}.
\end{rem}

\enter

\begin{defn}[(Feynman) graphs and Feynman graph set] \label{defn:feynman_graphs}
	A graph \(G := \left ( V, E, \beta \right )\) is given via a set of vertices \(V\), a set of edges \(E = E_0 \sqcup E_1\), where \(E_0\) is the subset of unoriented and \(E_1\) is the subset of oriented edges, and a morphism\footnote{We remark that the map \(\beta\) is necessary if graphs are allowed to have multi-edges or simultaneously oriented and unoriented edges, which is typically the case in physics.}
	\begin{equation}
		\beta \, : \quad E \to \left ( V \times V \times \mathbb{Z}_2 \right ) \, , \quad e \mapsto \begin{cases} \left ( v_1, v_2; 0 \right ) & \text{if \(e \in E_0\)} \\ \left ( v_i, v_t; 1 \right ) & \text{if \(e \in E_1\)} \end{cases}\, ,
	\end{equation}
	mapping edges to tuples of vertices together with their binary orientation information; if the edge is oriented, the order of the vertices is first initial then terminal. Given a graph \(G\), the corresponding sets are denoted via \(V \equiv V \left ( G \right ) \equiv G^{[0]}\) and \(E \equiv E \left ( G \right ) \equiv G^{[1]}\), where we omit the dependence on the graph \(G\) only if there is no ambiguity possible. Finally, given a QFT \(\Q\), we define a Feynman graph \(\Gamma := ( G, \{ *_p, *_f \}, E_\text{Ext}, \tau )\) as a graph \(G\) with the following extra structure: We add a set of external edges \(E_\text{Ext}\) and two external vertices \(\{ *_p, *_f \}\), where \(*_p\) is the endpoint for past external edges and \(*_f\) is the endpoint for future external edges. Then, we extend the map \(\beta\) to the set of external edges \(E_\text{Ext}\) via
	\begin{equation}
		\eval{\beta}_{E_\text{Ext}} \! \! \! \! \! \! \! : \quad E_\text{Ext} \to \left ( \big ( V \sqcup \set{*_p, *_f} \! \big ) \times \big ( V \sqcup \set{*_p, *_f} \! \big ) \times \mathbb{Z}_2 \right ) \, , \quad e \mapsto \begin{cases} \left ( v_1, v_2; 0 \right ) & \text{if \(e \in E_0\)} \\ \left ( v_i, v_t; 1 \right ) & \text{if \(e \in E_1\)} \end{cases}\, .
	\end{equation}
	Additionally, the vertex set \(V\) and the edge set \(E \sqcup E_\text{Ext}\) are considered as multisets over the vertex residue set \(\RQO\) and the edge residue set \(\RQI\), respectively:
	\begin{equation} \label{eqn:residue-coloring_function}
		\tau \, : \quad \left ( V \sqcup E \sqcup E_\text{Ext} \right ) \to \RQ \, , \quad r \mapsto \begin{cases} r_v \in \RQO & \text{if \(r \in V\)} \\ r_e \in \RQI & \text{if \(r \in E \sqcup E_\text{Ext}\)} \end{cases} \, ,
	\end{equation}
	where the map \(\tau\) corresponds to the map \(\pi\) from \eqnref{eqn:multiset_over_a_set}. Thus, using the coloring function \(\tau\), we view Feynman graphs as \(\RQ\)-colored graphs. Two Feynman graphs from the same QFT \(\Q\) are considered to be isomorphic, if they are isomorphic as \(\RQ\)-colored graphs and if they have the same external leg structure, cf.\ \defnref{defn:automorphisms_of_feynman_graphs}. Furthermore, using the composition \(\theta \circ \tau\) with the coupling constant function \(\theta\) from \eqnref{eqn:coupling-coloring_function}, we can also view the vertex set \(V\) as a multiset over the coupling constant set \(\qQ\), where then the composition \(\theta \circ \tau\) corresponds to the map \(\pi\) from \eqnref{eqn:multiset_over_a_set}. We remark, however, that edges are unlabeled in this coloring (labeled by \(1\)). Finally, a connected graph is called `one-particle irreducible (1PI)', if it is still connected after the removal of any of its internal edges.\footnote{In the mathematical literature these graphs are called bridge-free.} We denote the set of all 1PI Feynman graphs of \(\Q\) by \(\GQ\).
\end{defn}

\enter

\begin{defn}[Residue of a Feynman graph]
	Let \(\Q\) be a local QFT with residue set \(\RQ\) and 1PI Feynman graph set \(\GQ\). Then the external leg structure of a Feynman graph \(\Gamma \in \GQ\) is called its residue and denoted via \(\res{\Gamma} \in \AQ\). It is considered as the graph obtained from \(\Gamma\) by shrinking all its internal edges to a single vertex.
\end{defn}

\enter

\begin{defn}[Sets of half-edges, corollas and external vertex residues] \label{defn:half-edges_corollas_external-vertex-residue-set}
	Given a graph \(G\), we define the set of half-edges \(H \left ( G \right ) \equiv G^{[1/2]}\) via
	\begin{equation}
		H \left ( G \right ) := \set{h_v \cong (v,e) \left \vert \, v \in V , \, e \in E \text{ and } v \in \beta \left ( e \right ) \right .} \, ,
	\end{equation}
	where \(v \in \beta \left ( e \right )\) means that the vertex \(v\) is attached to the edge \(e\). The set of half-edges is then accompanied by the involution \(\iota\), which interchanges an internal half-edge with the internal half-edge that it forms the internal edge with and furthermore fixates external half-edges:
	\begin{equation}
		\iota \, : \quad H \left ( G \right ) \to H \left ( G \right ) \, , \quad (v,e) \mapsto \begin{cases} (w,e) & \text{if \(e \in E \left ( G \right )\) and \(v, w \in \beta \left ( e \right )\)} \\ (v,e) & \text{if \(e \in E_\text{Ext} \left ( G \right )\)} \end{cases}
	\end{equation}
	Thus, \(\iota\) can be used to reproduce the set \(E \left ( G \right ) \sqcup E_\text{Ext} \left ( G \right )\) from the set \(H \left ( G \right )\). Additionally, we define the set of corollas \(C \left ( G \right )\)
	\begin{equation}
		C \left ( G \right ) := \set{\left . c_v \cong \left ( v, \set{h_v \in H \left ( G, v \right )} \right ) \right \vert v \in V} \, ,
	\end{equation}
	where \(\{ h_v \in H \left ( G, v \right ) \}\) is the set of half-edges attached to the vertex \(v\). We also apply these constructions to Feynman graphs \(\Gamma\) by means of its underlying graph \(G\). Finally, we define the set of external vertex residues \(W \left ( \Gamma \right )\) of a (not necessary connected) Feynman graph \(\Gamma\) via
	\begin{equation}
		W \left ( \Gamma \right ) := \set{r_\gamma := \res{\gamma} \left \vert \, \text{\(\gamma\) connected component of \(\Gamma\) and } r_\gamma \in \RQO \right . } \, ,
	\end{equation}
	i.e.\ \(W \left ( \Gamma \right )\) is a multiset over \(\RQO\), by means of \(\tau\) from \eqnref{eqn:residue-coloring_function}, and furthermore a multiset over \(\qQ\), by means of \(\theta \circ \tau\) from \eqnsaref{eqn:coupling-coloring_function}{eqn:residue-coloring_function}.\footnote{We remark that if \(\Gamma\) is connected, then \(W \left ( \Gamma \right )\) contains at most one element.}
\end{defn}

\enter

\begin{defn}[Automorphisms of (Feynman) graphs] \label{defn:automorphisms_of_feynman_graphs}
	Let \(G\) be a graph. A map \(a \colon G \to G\), or, more precisely, the collection of its two underlying maps
	\begin{subequations}
	\begin{align}
		a_V \, & : \quad V \to V \, , \quad v_1 \mapsto v_2
		\intertext{and}
		a_E \, & : \quad E \to E \, , \quad e_1 \mapsto e_2
	\end{align}
	\end{subequations}
	is called an automorphism of \(G\), if \(a_V\) and \(a_E\) are bijections which are additionally compatible with \(\beta\) in the sense that \(\beta \circ a_E \equiv ( a_V \times a_V \times \id_{\mathbb{Z}_2} ) \circ \beta\). Furthermore, given a Feynman graph \(\Gamma\) and a map \(\alpha \colon \Gamma \to \Gamma\). Then \(\alpha\) is called automorphism of \(\Gamma\), if \(\alpha\) is an automorphism of the underlying graph, additionally compatible with the coloring function \(\tau\), i.e.\ \(\tau \circ \alpha \equiv \tau\), and the identity on external edges. The group of automorphisms of a Feynman graph \(\Gamma\) will be denoted via \(\operatorname{Aut} \left ( \Gamma \right )\) and its rank via \(\operatorname{Sym} \left ( \Gamma \right )\), to which we refer as its symmetry factor.
\end{defn}

\enter

\begin{defn}[Feynman graph invariants] \label{defn:betti-numbers_and_multi-indices}
	Let \(\Q\) be a QFT, \(\GQ\) its 1PI Feynman graph set, \(\RQO\) its vertex residue set and \(\qQ\) its coupling constant set. We equip the elements in the sets \(\RQO\) and \(\qQ\) with an arbitrary ordering, in order to have well-defined multiset vectors in the sense of \defnref{defn:multiset_over_a_set}. Given a Feynman graph \(\Gamma \in \GQ\), we associate the following two numbers and four multi-indices to it:
	\begin{itemize}
		\item \(\bettio{\Gamma} \in \mathbb{N}_0\) denotes the number of its connected components
		\item \(\bettii{\Gamma} \in \mathbb{N}_0\) denotes the number of its independent loops,\footnote{In the mathematical literature this is usually called a cycle.} where we only consider loops by internal edges, i.e.\ remove the two external vertices \(\{ *_p, *_f \}\)
		\item \(\intvtx{\Gamma} \in \ZvQ\) denotes the multiset vector of \(V \left ( \Gamma \right )\) over \(\RQO\), with respect to \(\tau\) from \eqnref{eqn:residue-coloring_function}
		\item \(\extvtx{\Gamma} \in \ZvQ\) denotes the multiset vector of \(W \left ( \Gamma \right )\) over \(\RQO\), with respect to \(\tau\) from \eqnref{eqn:residue-coloring_function}
		\item \(\intcpl{\Gamma} \in \ZqQ\) denotes the multiset vector of \(V \left ( \Gamma \right )\) over \(\qQ\), with respect to \(\theta \circ \tau\) from \eqnsaref{eqn:coupling-coloring_function}{eqn:residue-coloring_function}
		\item \(\extcpl{\Gamma} \in \ZqQ\) denotes the multiset vector of \(W \left ( \Gamma \right )\) over \(\qQ\), with respect to \(\theta \circ \tau\) from \eqnsaref{eqn:coupling-coloring_function}{eqn:residue-coloring_function}
	\end{itemize}
	Then we extend these invariants to the unit \(\one \in \HQ\) by \(0 \in \mathbb{N}_0\), \(\mathbf{0} \in \ZvQ\) and \(\mathbf{0} \in \ZqQ\), respectively, and to disjoint unions of 1PI Feynman graphs via addition, i.e.
	\begin{equation}
		\operatorname{Inv} \left ( \Gamma_1 \sqcup \Gamma_2 \right ) := \operatorname{Inv} \left ( \Gamma_1 \right ) + \operatorname{Inv} \left ( \Gamma_2 \right ) \, ,
	\end{equation}
	where \(\operatorname{Inv} \left ( \Gamma \right )\) is any of the invariants above and \(\Gamma_1, \Gamma_2 \in \GQ\).
\end{defn}

\enter

\begin{defn}[Weight of residues]
	Let \(\mathcal Q\) be a QFT with residue set \(\RQ\). We introduce a weight function
	\begin{equation}
		\omega \, : \quad \RQ \to \mathbb{Z} \, , \, , \quad r \mapsto \degp{\FR{r}} \, ,
	\end{equation}
	which maps a residue \(r \in \RQ\) to the degree of its corresponding Feynman rule \(\FR{r}\), viewed as a polynomial in momenta (or, in position space, derivatives).
\end{defn}

\enter

\begin{defn}[Superficial degree of divergence] \label{defn:sdd}
	Let \(\mathcal Q\) be a QFT with weighted residue set \((\RQ, \omega)\) and Feynman graph set \(\GQ\). We turn \(\GQ\) into a weighted set as well by extending \(\omega\) to the function
	\begin{equation}
		\omega \, : \quad \GQ \to \mathbb{Z} \, , \quad \Gamma \mapsto d \lambda \left ( \Gamma \right ) + \sum_{v \in V \left ( \Gamma \right )} \omega \left ( v \right ) + \sum_{e \in E \left ( \Gamma \right )} \omega \left ( e \right ) \, , \label{eqn:sdd}
	\end{equation}
	where \(d\) is the dimension of spacetime. Then, the weight \(\omega \left ( \Gamma \right )\) of a Feynman graph \(\Gamma\) is called its `superficial degree of divergence (SDD)'. A Feynman graph \(\Gamma\) is called superficially divergent if \(\sdd{\Gamma} \geq 0\) and superficially convergent if \(\sdd{\Gamma} < 0\). Finally, we set \(\sdd{\one} := 0\) for convenience.
\end{defn}

\enter

\begin{defn}[Set of superficially divergent subgraphs of a Feynman graph] \label{defn:set_of_divergent_subgraphs}
	Let \(\mathcal Q\) be a QFT with weighted Feynman graph set \((\GQ, \omega)\) and  \(\Gamma \in \GQ\) a Feynman graph. Then we denote by \(\DQ{\Gamma}\) the set of superficially divergent subgraphs of \(\Gamma\), i.e.
	\begin{subequations}
	\begin{align}
		\DQ{\Gamma} & := \set{\one \subseteq \gamma \subseteq \Gamma \, \left \vert \; \gamma = \bigsqcup_i \gamma_i \text{ with } \gamma_i \in \GQ \text{ and } \sdd{\gamma_i} \geq 0 \right .} \, ,
		\intertext{and by \(\mathcal{D}^\prime \left ( \Gamma \right )\) the set of proper divergent subgraphs}
		\DQprime{\Gamma} & := \set{ \gamma \in \DQ{\Gamma} \, \left \vert \; \one \subsetneq \gamma \subsetneq \Gamma \right .} \, .
	\end{align}
	\end{subequations}
	We remark that the condition \(\res{\gamma_i} \in \RQ\) for all divergent \(\gamma_i\) ensures the well-definedness of the renormalization Hopf algebra, cf.\ \cite[Subsection 3.3]{Prinz_2}.
\end{defn}

\enter

\begin{defn}[The (associated) renormalization Hopf algebra] \label{defn:renormalization_hopf_algebra}
	Given a QFT \(\Q\) with weighted Feynman graph set \((\GQ, \omega)\). Then the renormalization Hopf algebra is modeled on the \(\mathbb{Q}\)-vector space generated by 1PI Feynman graphs from the set \(\GQ\) and disjoint unions thereof. More precisely, the product \(m \colon \HQ \otimes_\mathbb{Q} \HQ \to \HQ\) is given via disjoint union and the coproduct \(\Delta \colon \HQ \to \HQ \otimes_\mathbb{Q} \HQ\) is given via the decomposition of (products of) 1PI Feynman graphs into the sum of all pairs of divergent subgraphs with the remaining cographs:
	\begin{equation}
		\Delta \, : \quad \HQ \to \HQ \otimes_\mathbb{Q} \HQ \, , \quad \Gamma \mapsto \sum_{\gamma \in \DQ{\Gamma}} \gamma \otimes_\mathbb{Q} \Gamma / \gamma \, ,
	\end{equation}
	where the cograph \(\Gamma / \gamma\) is defined by shrinking the internal edges of \(\gamma\) in \(\Gamma\) to a new vertex for each connected component of \(\gamma\). Furthermore, the unit \(\one \colon \mathbb{Q} \hookrightarrow \HQ\) is given via a multiple of the empty graph and the counit \(\coone \colon \HQ \surject \mathbb{Q}\) is given via the map sending all non-empty graphs to zero and the empty graph to its prefactor.\footnote{We remark that technically the unit and counit of the algebra and their respective functions are separate objects, which can be conveniently identified.} Moreover, the antipode is recursively defined as the negative of the convolution product with itself and the projector onto the augmentation ideal, cf.\ \defnref{defn:convolution_product} and \defnref{defn:augmentation_ideal}, i.e.\ via the normalization \(S \left ( \one \right ) := \one\) and else as follows:
	\begin{equation}
		S \, : \quad \HQ \to \HQ \, , \quad \Gamma \mapsto - \Gamma - \sum_{\DQprime{\Gamma}} S \left ( \gamma \right ) \Gamma / \gamma \, .
	\end{equation}
	In the following, we will omit the ground field from the tensor product, i.e.\ set \(\otimes := \otimes_\mathbb{Q}\). Finally, we remark that especially in the context of quantum gauge theories the above construction can be ill-defined, cf.\ the comment at the end of \defnref{defn:set_of_divergent_subgraphs}. This then leads to the notion of an \emph{associated renormalization Hopf algebra}, introduced in \cite[Subsection 3.3]{Prinz_2} and reproduced in \sectionref{sec:the_associated_renormalization_hopf_algebra}.
\end{defn}

\enter

\begin{defn}[Convolution product] \label{defn:convolution_product}
	Let \(\ring\) be a ring, \(A\) a \(\ring\)-algebra and \(C\) a \(\ring\)-coalgebra. Then, using the product \(\mult_A\) on \(A\) and the coproduct \(\Delta_C\) on \(C\), we can turn the \(\ring\)-module \(\text{Hom}_{\ring-\mathsf{Mod}} \left ( C , A \right )\) of \(\ring\)-linear maps from \(C\) to \(A\) into a \(\ring\)-algebra as well, by defining the convolution product \(\star\) as follows: Given \(f, g \in \text{Hom}_{\ring-\mathsf{Mod}} \left ( C , A \right )\), then we set
	\begin{equation}
		f \star g := \mult_A \circ \left ( f \otimes g \right ) \circ \Delta_C \, .
	\end{equation}
	Obviously, this definition extends trivially if \(A\) or \(C\) possesses additionally a bi- or Hopf algebra structure. It is commutative, if \(C\) is cocommutative and \(A\) is commutative. Finally, given a \(\ring\)-Hopf algebra \(H\), we remark that the algebra of endomorphisms \((\text{Hom}_{\ring-\mathsf{Mod}} \left ( H , H \right ), \star)\) is a group, with the antipode \(S\) being the \(\star\)-inverse to the identity morphism \(\operatorname{Id}_H\).
\end{defn}

\enter

\begin{defn}[Augmentation ideal] \label{defn:augmentation_ideal}
	Given a bi- or a Hopf algebra \(B\), then the kernel of the coidentity
	\begin{equation}
		\operatorname{Aug} \left ( \HQ \right ) := \operatorname{Ker} \big ( \mspace{1mu} \coone \mspace{1mu} \big )
	\end{equation}
	is an ideal, called the augmentation ideal. Additionally, we denote the projection map to it via \(\mathscr{A}\), i.e.\
	\begin{equation}
		\mathscr{A} \, : \quad \HQ \surject \operatorname{Aug} \left ( \HQ \right ) \subset \HQ \, , \quad \mathfrak{G} \mapsto \sum_{\substack{\set{\alpha_\text{s}, \mathfrak{G}_\text{s}} \in \SQ{\mathfrak{G}}\\\coone \left ( \mathfrak{G}_\text{s} \right ) = 0}} \alpha_\text{s} \mathfrak{G}_\text{s} \, ,
	\end{equation}
	where \(\SQ{\mathfrak{G}}\) denotes the set of summands of \(\mathfrak{G} \in \HQ\), cf.\ \defnref{defn:sets_summands_connected_components}.
\end{defn}

\enter

\begin{defn}[Basis decomposition of Hopf algebra elements] \label{defn:sets_summands_connected_components}
	Let \(\Q\) be a QFT, \(\GQ\) the set of its 1PI Feynman graphs and \(\HQ\) its (associated) renormalization Hopf algebra. Given an element \(\mathfrak{G} \in \HQ\), we are interested in its decomposition with respect to elements in the set \(\GQ\). Therefore, we denote by \(\SQ{\mathfrak{G}}\) the set of its summands, grouped into tuples of prefactors \(\alpha_\text{s} \in \mathbb{Q}\) (where we exclude the trivial case \(\alpha_\text{s} = 0\) if \(\mathfrak{G}_\text{s} \in \operatorname{Aug} \left ( \HQ \right )\)) and graphs \(\mathfrak{G}_\text{s} \in \HQ\) that can be disjoint unions, i.e.\ \(\mathfrak{G}_\text{s} = \bigsqcup_i \Gamma_i\) for \(\Gamma_i \in \GQ\),\footnote{This is actually the decomposition from \eqnref{eqn:decomposition_connected_components}.} such that
	\begin{equation}
		\mathfrak{G} \equiv \sum_{\set{\alpha_\text{s}, \mathfrak{G}_\text{s}} \in \SQ{\mathfrak{G}}} \alpha_\text{s} \mathfrak{G}_\text{s} \, .
	\end{equation}
	Additionally, we also write \(\mathfrak{G}_\text{s} \in \SQ{\mathfrak{G}}\) instead of \(\mathfrak{G}_\text{s} \in \set{\alpha_\text{s}, \mathfrak{G}_\text{s}} \in \SQ{\mathfrak{G}}\), if we are only interested in properties of the graph \(\mathfrak{G}_\text{s}\). Furthermore, given such a \(\mathfrak{G}_\text{s} \in \SQ{\mathfrak{G}}\), we denote by \(\CQ{\mathfrak{G}_\text{s}}\) the set of its connected components (where we exclude the identity \(\one \in \HQ\) if \(\mathfrak{G}_\text{s} \in \operatorname{Aug} \left ( \HQ \right )\)), such that
	\begin{equation} \label{eqn:decomposition_connected_components}
		\mathfrak{G}_\text{s} \equiv \prod_{\mathfrak{G}_\text{c} \in \CQ{\mathfrak{G}_\text{s}}} \mathfrak{G}_\text{c} \, .
	\end{equation}
	In particular, we have
	\begin{equation}
		\mathfrak{G} \equiv \sum_{\set{\alpha_\text{s}, \mathfrak{G}_\text{s}} \in \SQ{\mathfrak{G}}} \alpha_\text{s} \left ( \prod_{\mathfrak{G}_\text{c} \in \CQ{\mathfrak{G}_\text{s}}} \mathfrak{G}_\text{c} \right ) \, .
	\end{equation}
\end{defn}

\enter

\begin{defn}[Connectedness and gradings of the renormalization Hopf algebra] \label{defn:connectedness_gradings_renormalization_hopf_algebra}
	Given the situation of \defnref{defn:betti-numbers_and_multi-indices}, we construct the following three gradings on the renormalization Hopf algebra \(\HQ\): Let \(\mathfrak{G} \in \HQ\) be an element with \(\mathfrak{G}_\text{s} \in \SQ{\mathfrak{G}}\), we associate the following number and two multi-indices to \(\mathfrak{G}_\text{s}\):
	\begin{itemize}
		\item Loop-grading, denoted via \(L, l \in \mathbb{N}_0\), and given by
		\begin{equation}
		\operatorname{LoopGrd} \left ( \mathfrak{G}_\text{s} \right ) := \sum_{\mathfrak{G}_\text{c} \in \CQ{\mathfrak{G}_\text{s}}} \lambda \left ( \mathfrak{G}_\text{c} \right )
		\end{equation}
		\item Vertex-grading, denoted via \(\mathbf{V}, \mathbf{v} \in \ZvQ\), and given by
		\begin{equation}
		\vtxgrd{\mathfrak{G}_\text{s}} := \intvtx{\mathfrak{G}_\text{s}} - \extvtx{\mathfrak{G}_\text{s}} \label{eqn:resgrd}
		\end{equation}
		\item Coupling-grading, denoted via \(\mathbf{C}, \mathbf{c} \in \ZqQ\), and given by
		\begin{equation}
		\cplgrd{\mathfrak{G}_\text{s}} := \intcpl{\mathfrak{G}_\text{s}} - \extcpl{\mathfrak{G}_\text{s}}
		\end{equation}
	\end{itemize}
	In statements that are valid in any of these three gradings, we denote the grading function by \(\operatorname{Grd}\) and the gradings via \(\mathbf{G}\) and \(\mathbf{g}\). Furthermore, we denote the unit multi-index with respect to a vertex residue \(v \in \RQO\) or a coupling constant \(q \in \qQ\) via \(\mathbf{e}_v\) and \(\mathbf{e}_q\), respectively. Moreover, we denote the restriction of an object or an element to any of these three gradings via
	\begin{equation}
		\left ( \HQ \right )_\mathbf{G} := \eval{\HQ}_\mathbf{G} \, , \label{eqn:notation_restriction_grading}
	\end{equation}
	and omit the brackets, if no lower index is present. Clearly,
	\begin{equation}
		\left ( \mathcal{H}_{\Q} \right )_{L = 0} \cong \left ( \mathcal{H}_{\Q} \right )_{\mathbf{V} = \mathbf{0}} \cong \left ( \mathcal{H}_{\Q} \right )_{\mathbf{C} = \mathbf{0}} \cong \mathbb{Q} \, ,
	\end{equation}
	and thus \(\mathcal{H}_{\Q}\) is connected in all three gradings.
\end{defn}

\enter

\begin{rem}
	The three gradings from \defnref{defn:connectedness_gradings_renormalization_hopf_algebra} are further refinements of each other. In particular, the vertex-grading is equivalent to the coupling-grading if each vertex is associated with a unique coupling constant and it is furthermore equivalent to the loop-grading if the theory has only one vertex type. Additionally, the coupling-grading is equivalent to the loop-grading if the corresponding theory has only one coupling constant. We remark that both statements are due to the Euler identity, given in \eqnref{eqn:euler_characteristic}. Moreover, the numbers and multi-indices from \defnref{defn:betti-numbers_and_multi-indices} and the gradings from \defnref{defn:connectedness_gradings_renormalization_hopf_algebra} are compatible with the product of \(\HQ\) via addition, but not with the addition of \(\HQ\), as summands can live in different gradings.
\end{rem}

\enter

\begin{defn}[Projection to divergent graphs] \label{defn:projection_divergent_graphs}
	Let \(\Q\) be a QFT, \(\GQ\) its 1PI Feynman graph set and \(\HQ\) its (associated) renormalization Hopf algebra. We define the projection map to divergent Feynman graphs via
	\begin{subequations}
	\begin{align}
		\Omega \, & : \quad \GQ \to \GQ \, , \quad \Gamma \mapsto \begin{cases} \Gamma & \text{if \(\sdd{\Gamma} \geq 0\)} \\ 0 & \text{else, i.e.\ \(\sdd{\Gamma} < 0\)} \end{cases}
		\intertext{and then extend it additively and multiplicatively to \(\HQ\), i.e.}
		\Omega \, & : \quad \HQ \to \HQ \, , \quad \mathfrak{G} \mapsto \sum_{\set{\alpha_\text{s}, \mathfrak{G}_\text{s}} \in \SQ{\mathfrak{G}}} \alpha_\text{s} \left ( \prod_{\mathfrak{G}_\text{c} \in \CQ{\mathfrak{G}_\text{s}}} \Omega \left ( \mathfrak{G}_\text{c} \right ) \right ) \, ,
	\end{align}
	\end{subequations}
	that is, we keep the summands of \(\mathfrak{G}\) only, if all of its connected components are divergent. Additionally, we also use the following shorthand-notation:
	\begin{subequations}
		\begin{align}
		\overline{\HQ} & := \operatorname{Im} \left ( \Omega \right )
		\intertext{and}
		\overline{\mathfrak{G}} & := \Omega \left ( \mathfrak{G} \right ) \, .
		\end{align}
	\end{subequations}
	We remark that this definition will be useful for combinatorial Green's functions \(\combgreen^r\), combinatorial charges \(\combcharge^v\) and products thereof in the context of Hopf subalgebras for multiplicative renormalization, cf.\ \sectionref{sec:coproduct_and_antipode_identities}.
\end{defn}

\enter

\begin{defn}[Superficially compatible grading] \label{defn:superficially_compatible_grading}
	Given the situation of \defnref{defn:sdd} and \defnref{defn:connectedness_gradings_renormalization_hopf_algebra}, a grading is called superficially compatible, if all Feynman graphs with a given residue and a given grading have the same superficial degree of divergence. Equivalently, the degree of divergence of a Feynman graph depends only on its residue and the given grading. This will be studied in \propref{prop:superficial_grade_compatibility}.
\end{defn}

\enter

\begin{defn}[(Restricted) combinatorial Green's functions] \label{def:combinatorial_greens_functions}
	Let \(\Q\) be a QFT, \(\AQ\) the set of its amplitudes and \(\mathcal{G}_{\Q}\) the set of its Feynman graphs. Given an amplitude \(r \in \AQ\), we set
	\begin{equation}
		\precombgreen^r := \sum_{\substack{\Gamma \in \mathcal{G}_{\Q}\\\res{\Gamma} = r}} \frac{1}{\sym{\Gamma}} \Gamma
	\end{equation}
	and then define the combinatorial Green's function with amplitude \(r\) as the following sum:
	\begin{equation} \label{eqn:combgreen}
		\combgreen^r := \begin{cases} \one + \precombgreen^r & \text{if \(r \in \RQ^{[0]}\)} \\ \one - \precombgreen^r & \text{if \(r \in \RQ^{[1]}\)} \\ \precombgreen^r & \text{else, i.e.\ \(r \in \mathcal{Q}_\Q\)}
	\end{cases}
	\end{equation}
	Furthermore, we denote the restriction of \(\combgreen^r\) to one of the gradings \(\mathbf{g}\) from \defnref{defn:connectedness_gradings_renormalization_hopf_algebra} via
	\begin{equation}
		\rescombgreen^r_\mathbf{g} := \eval{\combgreen^r}_{\mathbf{g}} \, .
	\end{equation}
\end{defn}

\enter

\begin{rem} \label{rem:different_conventions_restricted_greens_functions}
	We remark that restricted combinatorial Green's functions are in the literature often denoted via \(c^r_\mathbf{g}\) and differ by a minus sign from our definition. Our convention is such that they are given as the restriction of the complete combinatorial Green's function to the corresponding grading, which provides additional minus signs for propagator graphs.
\end{rem}

\enter

\begin{defn}[(Restricted) combinatorial charges]
	Let \(v \in \RQ^{[0]}\) be a vertex residue, then we define its combinatorial charge \(\combcharge^v\) via
	\begin{equation}
		\combcharge^v := \frac{\combgreen^v}{\prod_{e \in E \left ( v \right )} \sqrt{\combgreen^e}} \, ,
	\end{equation}
	where \(E \left ( v \right )\) denotes the set of all edges attached to the vertex \(v\). Furthermore, we denote the restriction of \(\combcharge^v\) to one of the gradings \(\mathbf{g}\) from \defnref{defn:connectedness_gradings_renormalization_hopf_algebra} via
	\begin{equation}
		\combcharge^v_\mathbf{g} := \eval{\combcharge^v}_{\mathbf{g}} \, .
	\end{equation}
\end{defn}

\enter

\begin{defn}[(Restricted) products of combinatorial charges] \label{defn:combinatorial_charges}
	Let \(\mathbf{v} \in \ZvQ\) be a multi-index of vertex residues. Then we define the product of combinatorial charges associated to \(\mathbf{v}\) via
	\begin{equation} \label{eqn:combinatorial_charges}
		\combcharge^\mathbf{v} := \prod_{k = 1}^{\mathfrak{v}_\Q} \left ( \combcharge^{v_k} \right )^{(\mathbf{v})_k} \, ,
	\end{equation}
	where \((\mathbf{v})_k\) denotes the \(k\)-th entry of \(\mathbf{v}\). In particular, given a vertex residue \(v \in \RQO\) and a natural number \(n \in \mathbb{N}_+\), we define the exponentiation of the combinatorial charge \(\combcharge^v\) by \(n\) via
	\begin{equation}
		\combcharge^{nv} := \left ( \combcharge^{v} \right )^n \, .
	\end{equation}
	Furthermore, we denote the restriction of \(\combcharge^\mathbf{v}\) to one of the gradings \(\mathbf{g}\) from \defnref{defn:connectedness_gradings_renormalization_hopf_algebra} via
	\begin{equation} \label{eqn:restricted_products_combinatorial_charges}
		\combcharge^\mathbf{v}_\mathbf{g} := \eval{\left ( \prod_{k = 1}^{\mathfrak{v}_\Q} \left ( \combcharge^{v_k} \right )^{(\mathbf{v})_k} \right )}_\mathbf{g} \, .
	\end{equation}
\end{defn}

\enter

\begin{defn}[Sets of combinatorial and physical charges, projection map] \label{defn:sets_of_coupling_constants}
	Let \(\Q\) be a QFT. Then we denote via \(\QQ\) and \(\qQ\) the sets of combinatorial and physical charges, respectively. We associate to each vertex residue \(v \in \RQO\) a combinatorial charge and to each interaction monomial in the Lagrange density \(\mathcal{L}_\Q\) a (not necessarily distinct) physical coupling constant. Additionally, we define the set-theoretic projection map\footnote{This map is the set-theoretic restriction of the renormalized Feynman rules, which map combinatorial charges to Feynman integrals corresponding to renormalized physical charges.}
	\begin{equation}
		\operatorname{Cpl} \, : \quad \QQ \surject \qQ \, , \quad \combcharge^v \mapsto \theta \left ( v \right ) \, ,
	\end{equation}
	where \(\theta \colon \RQO \to \qQ\) is the map from \eqnref{eqn:coupling-coloring_function}.
\end{defn}

\enter

\begin{lem} \label{lem:v-h-e-sets_and_r-v}
	Given a Feynman graph \(\Gamma \in \GQ\), the sets \(V \left ( \Gamma \right )\) and \(E \left ( \Gamma \right )\) from \defnref{defn:feynman_graphs} and \(H \left ( \Gamma \right )\) and \(C \left ( \Gamma \right )\) from \defnref{defn:half-edges_corollas_external-vertex-residue-set}, viewed as multisets over \(\RQ\), depend only on its residue \(\res{\Gamma}\) and its vertex-grading multi-index \(\vtxgrd{\Gamma}\). In particular, we obtain well-defined sets \(V \left ( r, \mathbf{v} \right )\), \(E \left ( r, \mathbf{v} \right )\), \(H \left ( r, \mathbf{v} \right )\) and \(C \left ( r, \mathbf{v} \right )\), such that we have \(V \left ( r, \mathbf{v} \right ) \cong V \left ( \Gamma \right )\), \(E \left ( r, \mathbf{v} \right ) \cong E \left ( \Gamma \right )\), \(H \left ( r, \mathbf{v} \right ) \cong H \left ( \Gamma \right )\) and \(C \left ( r, \mathbf{v} \right ) \cong C \left ( \Gamma \right )\) as multisets over \(\RQ\), for all \(\Gamma \in \GQ\) with \(\res{\Gamma} = r\) and \(\vtxgrd{\Gamma} = \mathbf{v}\).
\end{lem}

\begin{proof}
	Given \(\Gamma \in \GQ\), then by definition its vertex set \(V \left ( \Gamma \right )\) is a multiset over \(\RQO\), using \(\tau\) from \eqnref{eqn:residue-coloring_function}. Thus it is uniquely characterized via its internal residue multi-index \(\intvtx{\Gamma}\), as it displays the multiplicity of each vertex residue \(r_v \in \RQO\) in the vertex set \(V \left ( \Gamma \right )\). Furthermore, we can reconstruct \(\intvtx{\Gamma}\) from \(\res{\Gamma}\) and \(\vtxgrd{\Gamma}\) using the definition, \eqnref{eqn:resgrd}, i.e.\
	\begin{equation}
		V \left ( r, \mathbf{v} \right ) \cong \vtxgrd{\Gamma} + \extvtx{\Gamma} \, ,
	\end{equation}
	while noting that \(\extvtx{\Gamma}\) is given for connected Feynman graphs \(\Gamma \in \GQ\) with \(\res{\Gamma} \in \RQO\) as the multi-indices having a one for the corresponding vertex residue and zeros else, i.e.\
	\begin{equation}
		\left ( \extvtx{\Gamma} \right )_j = \begin{cases} 1 & \text{if \(\res{\Gamma} = v_j \in \RQO\)} \\ 0 & \text{else} \end{cases} \, ,
	\end{equation}
	and for Feynman graphs \(\Gamma \in \GQ\) with \(\res{\Gamma} \in \big ( \AQ \setminus \RQO \big )\) as the zero-multi-index, i.e.\
	\begin{equation}
		\extvtx{\Gamma} = \mathbf{0} \, .
	\end{equation}
	Thus we have have shown that the set \(V \left ( r, \mathbf{v} \right )\) is well-defined as a multiset over \(\RQO\). Moreover, we can obtain the half-edge set \(H \left ( r, \mathbf{v} \right )\) from \(\res{\Gamma}\), \(V \left ( r, \mathbf{v} \right )\) and the coloring function \(\tau\) via
	\begin{equation}
		H \left ( r, \mathbf{v} \right ) := \set{h_v \in \bigsqcup_{v \in V \left ( r, \mathbf{v} \right )} H \left ( v, \tau \left ( v \right ) \right )} \setminus H \left ( \res{\Gamma}, \tau \left ( \res{\Gamma} \right ) \right ) \, ,
	\end{equation}
	where \(\tau\) indicates the vertex-type of \(v \in V \left ( r, \mathbf{v} \right )\), i.e.\ which edge-types are attachable to it. Then we denote via \(H \left ( v, \tau \left ( v \right ) \right )\) the set of all such pairings \(h_v \cong (v, e)\), take its disjoint union and then remove the set of external half-edges of \(\Gamma\). Additionally, we obtain the edge set \(E \left ( r, \mathbf{v} \right )\) as a multiset over \(\RQI\) from the half-edge set \(H \left ( r, \mathbf{v} \right )\) as follows: We use an equivalence relation \(\sim\) which identifies two half-edges to a single edge, if they are of the same particle type, i.e.\ \(h_1 \sim h_2\) if \(\tau \left ( e_1 \right ) = \tau \left ( e_2 \right )\), and then consider the quotient
	\begin{equation}
		E \left ( r, \mathbf{v} \right ) := H \left ( r, \mathbf{v} \right ) / \sim \, .
	\end{equation}
	We remark that there are in general many possibilities to define \(\sim\), but the resulting multisets are isomorphic, hence it suffices to take an arbitrary choice. In particular, one such choice is the involution \(\iota\) from \defnref{defn:half-edges_corollas_external-vertex-residue-set}. Finally, we obtain the corolla set \(C \left ( r, \mathbf{v} \right )\) from the vertex set \(V \left ( r, \mathbf{v} \right )\) and the half-edge set \(H \left ( r, \mathbf{v} \right )\) by simply associating to each vertex the set of half-edges attached to it, i.e.\
	\begin{equation}
		C \left ( r, \mathbf{v} \right ) := \set{\left . c_v \cong \left ( v, \set{h_v \in H \left ( r, \mathbf{v} \right )} \right ) \right \vert v \in V \left ( r, \mathbf{v} \right )} \, ,
	\end{equation}
	which completes the proof.
\end{proof}

\enter

\begin{defn}[Set of superficially divergent insertable graphs for a Feynman graph] \label{defn:set_of_divergent_insertable_graphs_graphs}
	Let \(\Q\) be a QFT and \(\Gamma \in \mathcal{G}_{\Q}\) a Feynman graph of \(\Q\). Then we denote by \(\IQ{\Gamma}\) the set of superficially divergent graphs that are insertable into \(\Gamma\), i.e.\footnote{We remark that we have \(\one \in \IQ{\Gamma}\) for all \(\Gamma \in \GQ\).}
	\begin{equation}
	\begin{split}
		\IQ{\Gamma} := \left \{ \gamma \in \HQ \, \vphantom{\gamma \in \HQ \, \vert \; \extvtx{\gamma} \leq \intvtx{\Gamma} \text{ and } \sdd{\gamma_\text{c}} \geq 0 \text{ for all } \gamma_\text{c} \in \CQ{\gamma} \phantom{\vert} \text{ and } \res{\gamma_\text{p}} \in E \left ( \Gamma \right ) \text{ for all } \gamma_\text{p} \in \mathcal{P} \left ( \gamma \right )} \right . & \left \vert \; \extvtx{\gamma} \leq \intvtx{\Gamma} \text{ and } \sdd{\gamma_\text{c}} \geq 0 \text{ for all } \gamma_\text{c} \in \CQ{\gamma} \vphantom{\gamma \in \HQ \, \vert \; \extvtx{\gamma} \leq \intvtx{\Gamma} \text{ and } \sdd{\gamma_\text{c}} \geq 0 \text{ for all } \gamma_\text{c} \in \CQ{\gamma} \phantom{\vert} \text{ and } \res{\gamma_\text{p}} \in E \left ( \Gamma \right ) \text{ for all } \gamma_\text{p} \in \mathcal{P} \left ( \gamma \right )} \right . \\ 
		& \phantom{\vert} \left . \; \text{and } \res{\gamma_\text{p}} \in E \left ( \Gamma \right ) \text{ for all } \gamma_\text{p} \in \mathcal{P} \left ( \gamma \right ) \vphantom{\gamma \in \HQ \, \vert \; \extvtx{\gamma} \leq \intvtx{\Gamma} \text{ and } \sdd{\gamma_\text{c}} \geq 0 \text{ for all } \gamma_\text{c} \in \CQ{\gamma} \phantom{\vert} \text{ and } \res{\gamma_\text{p}} \in E \left ( \Gamma \right ) \text{ for all } \gamma_\text{p} \in \mathcal{P} \left ( \gamma \right )} \right \} \, ,
	\end{split}
	\end{equation}
	where \(\mathcal{P} \left ( \gamma \right ) \subseteq \CQ{\gamma}\) denotes the set of connected components of \(\gamma\) which are propagator graphs.
\end{defn}

\enter

\begin{defn}[Insertion factors] \label{defn:ins-aut_ins_insrv}
	Let \(\Q\) be a QFT, \(\GQ\) its Feynman graph set and \(\HQ\) its (associated) renormalization Hopf algebra. Given two Feynman graphs \(\Gamma, \Gamma^\prime \in \mathcal{G}_{\Q}\) and an element in the Hopf algebra \(\gamma \in \HQ\), we want to characterize possible insertions. To this end, we define the following four combinatorial factors:
	\begin{itemize}
		\item \(\ins{\gamma}{\Gamma}\) denotes the number of ways to insert \(\gamma\) into \(\Gamma\)
		\item \(\insaut{\gamma}{\Gamma}{\Gamma^\prime}\) denotes the number of ways to insert \(\gamma\) into \(\Gamma\), such that the insertion is automorphic to \(\Gamma^\prime\)
		\item \(\insrr{\gamma}\) denotes the number of ways to insert \(\gamma\) into a Feynman graph with residue \(r\) and vertex-grading multi-index \(\mathbf{v}\), which is well-defined due to \lemref{lem:v-h-e-sets_and_r-v}
		\item \(\isoemb{\gamma}{\Gamma}\) denotes the number of non-trivial isomorphic embeddings of \(\gamma\) as a subgraph of \(\Gamma\)
	\end{itemize}
	We remark that these numbers are zero, if either \(\gamma\) is not insertable into \(\Gamma\), i.e.\ \(\gamma \notin \IQ{\Gamma}\), if there is no insertion possible which is automorphic to \(\Gamma^\prime\) or if there is no isomorphic embedding possible. Finally, we set for all \(\Gamma \in \GQ\)
	\begin{equation}
		\insaut{\one}{\Gamma}{\Gamma} = \ins{\one}{\Gamma} = \insrr{\one} = \isoemb{\one}{\Gamma} := 1 \, .
	\end{equation}
\end{defn}

\enter

\begin{prop} \label{prop:isomorphism_i}
	Given the situation of \defnref{defn:set_of_divergent_insertable_graphs_graphs}, we have for all Feynman graphs \(\Gamma \in \GQ\)
	\begin{equation}
		\sum_{\gamma \in \IQ{\Gamma}} \frac{\ins{\gamma}{\Gamma}}{\sym{\gamma}} \gamma = \frac{\prod_{v \in V \left ( \Gamma \right )} \overline{\combgreen}^v}{\prod_{e \in E \left ( \Gamma \right )} \overline{\combgreen}^e} \, .
	\end{equation}
\end{prop}

\begin{proof}
	We can insert in each vertex \(v \in V \left ( \Gamma \right )\) at most one superficially divergent vertex correction \(\gamma^v\) with \(\res{\gamma^v} = v\), i.e. a summand of \(\overline{\combgreen}^v\). Furthermore, we can insert in each edge \(e \in E \left ( \Gamma \right )\) arbitrary many superficially divergent edge corrections \(\gamma^e = \prod_i \gamma^e_i\) with \(\operatorname{Res} \big ( \gamma^e_i \big ) = e\) for all \(i\), i.e.\ a summand of \(\textfrac{1}{\overline{\combgreen}^e}\), where the fraction is understood as the formal geometric series \(\textfrac{1}{\left ( 1-x \right )} \equiv \sum_{k = 0}^\infty x^k\).\footnote{We remark that this viewpoint is the reason for the minus sign in the definition of combinatorial Green's function for propagators, i.e.\ \eqnref{eqn:combgreen} of \defnref{def:combinatorial_greens_functions}.} Finally, the prefactor \(\ins{\gamma}{\Gamma}\) corresponds to the multiplicity of similar \(\RQ\)-colored vertices and edges of \(\Gamma\), using \(\tau\) from \eqnref{eqn:residue-coloring_function}.
\end{proof}

\enter

\begin{defn}[Set of superficially divergent insertable graphs for residue and vertex-grading] \label{defn:set_of_divergent_insertable_graphs_grading-residue}
	Let \(\Q\) be a QFT, \(\AQ\) its amplitude set and \(\HQ\) its (associated) renormalization Hopf algebra. Given an amplitude \(r \in \AQ\), a vertex-grading multi-index \(\mathbf{v} \in \ZvQ\) with \(\mathbf{v} \neq \mathbf{0}\) and a Feynman graph \(\Gamma \in \GQ\) with \(\res{\Gamma} = r\) and \(\vtxgrd{\Gamma} = \mathbf{v}\). Then we define the set \(\IQrv\) of superficially divergent graphs insertable into Feynman graphs with residue \(r\) and vertex-grading \(\mathbf{v}\) via\footnote{We remark that we have \(\one \in \IQrv\) for all \(r \in \AQ\) and \(\mathbf{v} \in \ZvQ\) with \(\mathbf{v} \neq \mathbf{0}\).}
	\begin{equation}
		\IQrv := \IQ{\Gamma} \, ,
	\end{equation}
	which is well-defined due to \colref{col:set_of_divergent_insertable_graphs_grading-residue}.
\end{defn}

\enter

\begin{col} \label{col:set_of_divergent_insertable_graphs_grading-residue}
	Given the situation of \defnref{defn:set_of_divergent_insertable_graphs_grading-residue}, the set \(\IQrv\) satisfies
	\begin{equation} \label{eqn:set_of_divergent_insertable_graphs_grading-residue}
		\sum_{\gamma \in \IQrv} \frac{\insrr{\gamma}}{\sym{\gamma}} \gamma = \frac{\prod_{v \in V \left ( r, \mathbf{v} \right )} \overline{\combgreen}^v}{\prod_{e \in E \left ( r, \mathbf{v} \right )} \overline{\combgreen}^e} \, ,
	\end{equation}
	and is thus in particular well-defined.
\end{col}

\begin{proof}
	This follows directly from \lemref{lem:v-h-e-sets_and_r-v} and \propref{prop:isomorphism_i}.
\end{proof}

\enter

\begin{prop} \label{prop:isomorphic_insertable_graph_sets}
	Given the situation of \defnref{defn:set_of_divergent_insertable_graphs_grading-residue}, we have for all amplitudes \(r \in \AQ\) and vertex-grading multi-indices \(\mathbf{v} \in \ZvQ\)
	\begin{equation}
		\sum_{\gamma \in \IQrv} \frac{\insrr{\gamma}}{\sym{\gamma}} \gamma = \begin{cases} \overline{\combgreen}^r \overline{\combcharge}^\mathbf{v} & \text{if \(r \in \RQ\)} \\ \prod_{e \in E \left ( r \right )} \sqrt{\overline{\combgreen}^e} \overline{\combcharge}^\mathbf{v} & \text{else, i.e.\ \(r \in \mathcal{Q}_\Q\)} \end{cases} \, ,
	\end{equation}
	where \(E \left ( r \right )\) denotes the set of edges attached to \(r\) and the square-root is defined via the formal series \(\sqrt{x} \equiv \sum_{k = 0}^\infty \binom{\textfrac{1}{2}}{k} \left ( x - 1 \right )^k\).
\end{prop}

\begin{proof}
	The numerator of the right hand side of \eqnref{eqn:set_of_divergent_insertable_graphs_grading-residue} of \colref{col:set_of_divergent_insertable_graphs_grading-residue} can be expressed as follows:\footnote{The two cases emerge due to the vertex-grading, which treats Feynman graphs with vertex residues differently in order to obtain a valid grading of the renormalization Hopf algebra, cf.\ \eqnref{eqn:resgrd} of \defnref{defn:connectedness_gradings_renormalization_hopf_algebra}.}
	\begin{subequations}
	\begin{equation} \label{eqn:combgreen_vertexset}
		\prod_{v \in V \left ( r, \mathbf{v} \right )} \overline{\combgreen}^v = \begin{cases} \overline{\combgreen}^r \overline{\combgreen}^\mathbf{v} & \text{if \(r \in \RQO\)} \\ \overline{\combgreen}^\mathbf{v} & \text{else, i.e.\ \(r \in \left ( \AQ \setminus \RQO \right )\)} \end{cases} \, ,
	\end{equation}
	where the notation \(\overline{\combgreen}^\mathbf{v} := \prod_{k = 1}^{\mathfrak{v}_\Q} \big ( \overline{\combgreen}^{v_k} \big )^{\mathbf{v}_k}\) is analogous to \eqnref{eqn:combinatorial_charges} of \defnref{defn:combinatorial_charges}. Furthermore, the denominator of the right hand side of \eqnref{eqn:set_of_divergent_insertable_graphs_grading-residue} of \colref{col:set_of_divergent_insertable_graphs_grading-residue} can be expressed as follows:
	\begin{equation} \label{eqn:combgreen_edgeset}
	\begin{split}
		\frac{\one}{\prod_{e \in E \left ( r, \mathbf{v} \right )} \overline{\combgreen}^e} & = \begin{cases} \dfrac{\one}{\prod_{v \in V \left ( r, \mathbf{v} \right )} \left ( \prod_{e \in E \left ( v \right )} \sqrt{\overline{\combgreen}^e} \right )} & \text{if \(r \in \RQO\)} \\[5pt] \dfrac{\prod_{{e_1} \in E \left ( r \right )} \sqrt{\overline{\combgreen}^{e_1}}}{\prod_{v \in V \left ( r, \mathbf{v} \right )} \left ( \prod_{{e_2} \in E \left ( v \right )} \sqrt{\overline{\combgreen}^{e_2}} \right )} & \text{else, i.e.\ \(r \in \left ( \AQ \setminus \RQO \right )\)} \end{cases} \\
		& = \begin{cases} \dfrac{\overline{\combcharge}^\mathbf{v}}{\overline{\combgreen}^\mathbf{v}} & \text{if \(r \in \RQO\)} \\[10pt] \dfrac{\overline{\combgreen}^r \overline{\combcharge}^\mathbf{v}}{\overline{\combgreen}^\mathbf{v}} & \text{if \(r \in \RQI\)} \\[5pt] \dfrac{\prod_{e \in E \left ( r \right )} \sqrt{\overline{\combgreen}^e} \overline{\combcharge}^\mathbf{v}}{\overline{\combgreen}^\mathbf{v}} & \text{else, i.e.\ \(r \in \mathcal{Q}_\Q\)} \end{cases}
	\end{split}
	\end{equation}
	\end{subequations}
	Multiplying \eqnref{eqn:combgreen_vertexset} with \eqnref{eqn:combgreen_edgeset}, we obtain
	\begin{equation}
		\frac{\prod_{v \in V \left ( r, \mathbf{v} \right )} \overline{\combgreen}^v}{\prod_{e \in E \left ( r, \mathbf{v} \right )} \overline{\combgreen}^e} = \begin{cases} \overline{\combgreen}^r \overline{\combcharge}^\mathbf{v} & \text{if \(r \in \RQ\)} \\ \prod_{e \in E \left ( r \right )} \sqrt{\overline{\combgreen}^e} \overline{\combcharge}^\mathbf{v} & \text{else, i.e.\ \(r \in \mathcal{Q}_\Q\)} \end{cases} \, .
	\end{equation}
	Finally, the prefactor \(\insrr{\gamma}\) corresponds to the multiplicity of similar \(\RQ\)-colored vertices and edges of Feynman graphs with residue \(r\) and vertex-grading multi-index \(\mathbf{v}\), using \(\tau\) from \eqnref{eqn:residue-coloring_function}.
\end{proof}

\enter

\begin{lem}[{\cite[Lemma 12]{vSuijlekom_QCD}}] \label{lem:sym-factors_and_ins-factors}
	Given the situation of \defnref{defn:automorphisms_of_feynman_graphs} and \defnref{defn:ins-aut_ins_insrv}, we have for all Feynman graphs \(\Gamma \in \GQ\) and their corresponding subgraphs \(\one \subseteq \gamma \subseteq \Gamma\)
	\begin{equation}
		\frac{\sym{\gamma} \sym{\Gamma / \gamma}}{\sym{\Gamma}} = \frac{\insaut{\gamma}{\Gamma / \gamma}{\Gamma}}{\isoemb{\gamma}{\Gamma}} \, .
	\end{equation}
\end{lem}

\begin{proof}
	Let \(\Gamma \in \GQ\) be a Feynman graph. Then, by definition, we have
	\begin{equation}
		\sym{\Gamma} = \# \operatorname{Aut} \left ( \Gamma \right ) \, ,
	\end{equation}
	where the automorphisms are fixing the external leg structure by definition, cf.\ \defnref{defn:automorphisms_of_feynman_graphs}. Thus, for a given subgraph \(\gamma \subseteq \Gamma\), we have
	\begin{equation}
		\sym{\gamma} \sym{\Gamma / \gamma} = \# \operatorname{Aut} \left ( \gamma \right ) \# \operatorname{Aut} \left ( \Gamma / \gamma \right ) \, ,
	\end{equation}
	which counts all automorphisms of \(\Gamma / \gamma\) times those of \(\gamma\), fixing both their external legs. Thus, comparing to \(\sym{\Gamma}\), the following two things can appear: The automorphism group \(\operatorname{Aut} \left ( \Gamma \right )\) might contain automorphisms which exchange non-trivial isomorphic embeddings \(\gamma, \gamma^\prime \subseteq \Gamma\) and can thus contain automorphisms exceeding the set \(\operatorname{Aut} \left ( \gamma \right ) \cup \operatorname{Aut} \left ( \Gamma / \gamma \right )\). Contrary, the quotient graph \(\Gamma / \gamma\) might possess symmetries which get spoiled after the insertion of \(\gamma\) into \(\Gamma / \gamma\). These two possibilities are reflected by the quotient \(\insaut{\gamma}{\Gamma / \gamma}{\Gamma} / \isoemb{\gamma}{\Gamma}\), as it counts the number of equivalent insertions of \(\gamma\) into \(\Gamma / \gamma\) automorphic to \(\Gamma\) modulo the additional symmetries that might appear, cf.\ \defnref{defn:ins-aut_ins_insrv}. Thus we obtain
	\begin{equation}
		\frac{\sym{\gamma} \sym{\Gamma / \gamma}}{\sym{\Gamma}} = \frac{\insaut{\gamma}{\Gamma / \gamma}{\Gamma}}{\isoemb{\gamma}{\Gamma}} \, ,
	\end{equation}
	as claimed.
\end{proof}

\enter

\begin{defn}[Algebra of formal (Feynman) integral expressions] \label{defn:formal_feynman_integral_expressions}
	Let \(\field\) be a field and \(\mathcal{E}\) the \(\field\)-vector space generated by the set of formal integral expressions, that is pairs \((D,I)\), where \(D\) is a domain and \(I\) a differential form on it. Addition is then declared via
	\begin{equation}
		(D_1,I_1) + (D_2,I_2) := (D_1 \oplus D_2, I_1 \oplus 0_2 + 0_1 \oplus I_2) \, ,
	\end{equation}
	where \(0_i\) is the zero differential form on the domain \(D_i\), and scalar multiplication is declared via
	\begin{equation}
		k (D,I) := (D,kI)
	\end{equation}
	for \(k \in \field\). Furthermore, we turn \(\mathcal{E}\) into an algebra by declaring the multiplication via
	\begin{equation}
		(D_1,I_1) \times (D_2,I_2) := (D_1 \otimes_\field D_2, I_1 \otimes_\field I_2) \, , \label{eqn:multiplication_map_ffie}
	\end{equation}
	which we call \(\mu\). Moreover, we address the name `formal integral expression' by defining the evaluation character (i.e.\ algebra morphism)
	\begin{equation}
		\operatorname{Int} \, : \quad \mathcal{E}_\text{Fin} \to \mathbb{C} \, , \quad (D,I) \mapsto \int_D I \, ,
	\end{equation}
	where \(\mathcal{E}_\text{Fin} \subset \mathcal{E}\) is the subalgebra where the evaluation map is finite and thus well-defined. In particular, we fix the normal subgroups \(\mathbf{1}_{\mathcal{E}_\text{Fin}} := \operatorname{Int}^{-1} \left ( 1 \right ) \subset \mathcal{E}_\text{Fin}\) and \(\mathbf{1}_{\mathcal{E}} := \iota \left ( \mathbf{1}_{\mathcal{E}_\text{Fin}} \right )\), where \(\iota \colon \mathcal{E}_\text{Fin} \hookrightarrow \mathcal{E}\) is the natural inclusion map. Both of these groups consist of formal integral expressions \((D,I)\) with \(\operatorname{Int} \left ( D,I \right ) = 1\). Therefore, we will treat \(\mathbf{1}_{\mathcal{E}}\) as the equivalence class of `units' on \(\mathcal{E}\). Finally, given a QFT \(\Q\), we define the algebra of its formal Feynman integral expressions as follows: We set \(\field := \mathbb{Q}\) and restrict the allowed domains \(D\) and differential forms \(I\) according to the chosen Feynman integral representation (position space, momentum space, parametric space, etc.).
\end{defn}

\enter

\begin{rem}
	 The setup of \defnref{defn:formal_feynman_integral_expressions} allows us in particular to address ill-defined integral expressions by externalizing the integration process.
\end{rem}

\enter

\begin{defn}[Feynman rules, regularization and renormalization schemes] \label{defn:fr_reg_ren_counterterm}
	Let \(\Q\) be a QFT, \(\HQ\) its (associated) renormalization Hopf algebra and \(\EQ\) its algebra of formal Feynman integral expressions. Then we define its Feynman rules as the following character (i.e.\ algebra morphism)
	\begin{equation}
		\Phi \, : \quad \HQ \to \EQ \, , \quad \Gamma \mapsto (D_\Gamma, I_\Gamma) \, ,
	\end{equation}
	where \((D_\Gamma, I_\Gamma)\) is the formal Feynman integral expression for the Feynman graph \(\Gamma\). Furthermore, we introduce a regularization scheme \(\mathscr{E}\) as a map\footnote{There exist renormalization schemes, such as kinematic renormalization schemes, that are well-defined without a previous regularization step. These can be seen as embedded into our framework by simply setting \(\mathscr{E} := \operatorname{Id}_{\EQ}\) and considering it as the natural inclusion of \(\EQ\) into \(\EQ^\varepsilon\).}
	\begin{equation} \label{eqn:regularization_scheme}
		\mathscr{E} \, : \quad \EQ \hookrightarrow \EQ^\varepsilon \, , \quad (D,I) \mapsto \left ( D,I_\mathscr{E} \left ( \varepsilon \right ) \right ) := \left ( D,\sum_{i = 0}^\infty I_i \, \varepsilon^i \right ) \, ,
	\end{equation}
	where \(\EQ^\varepsilon := \EQ [ [ \varepsilon ] ] \supset \EQ\) and the coefficients of the Taylor series are differential forms \(I_i\) on \(D\). Additionally, the regulated formal integral expressions \(\left ( D,I_\mathscr{E} \left ( \varepsilon \right ) \right )\) are subject to the boundary condition \(I \left ( 0 \right ) \equiv I\), which is equivalent to \(I_0 := I\), and the integrability condition \(\operatorname{Int} \left ( D,I_\mathscr{E} \left ( \varepsilon \right ) \right ) < \infty\), for all \(\varepsilon \in J\) with \(J \subseteq [0,\infty)\) a non-empty interval. We then set the regularized Feynman rules as the map
	\begin{equation}
		\regFR \, : \quad \HQ \to {\EQ}^\varepsilon \, , \quad \Gamma \mapsto \left ( \mathscr{E} \circ \Phi \right ) \left ( \Gamma, \varepsilon \right ) \, .
	\end{equation}
	Moreover, we introduce a renormalization scheme as a linear map\footnote{Sometimes, if convenient, we view \(\mathscr{R}\) also as endomorphism on \(\EQ^\varepsilon\) with image \(\EQ^\varepsilon_-\) and cokernel \(\EQ^\varepsilon_+\).}
	\begin{equation} \label{eqn:renormalization_scheme}
		\mathscr{R} \, : \quad \EQ^\varepsilon \surject \EQ^\varepsilon_- \, , \quad \left ( D,I_\mathscr{E} \left ( \varepsilon \right ) \right ) \mapsto \begin{cases} (D,0_D) & \text{if \(\left ( D,I_\mathscr{E} \left ( \varepsilon \right ) \right ) \in \operatorname{Ker} \left ( \mathscr{R} \right )\)} \\ \left ( D,I_{\mathscr{E}, \mathscr{R}} \left ( \varepsilon \right ) \right ) & \text{else} \end{cases} \, ,
	\end{equation}
	where \({\EQ}^\varepsilon_- := \operatorname{Im} \left ( \mathscr{R} \right ) \subset \EQ^\varepsilon\) and \(0_D\) is the zero differential form on \(D\), for all \(\varepsilon \in \mathbb{R}\). Additionally, to ensure locality of the counterterm, \(\mathscr{R}\) needs to be a Rota-Baxter operator of weight \(\lambda = -1\), i.e.\ fulfill
	\begin{equation}
		\mu \circ \left ( \mathscr{R} \otimes \mathscr{R} \right ) + \mathscr{R} \circ \mu = \mathscr{R} \circ \mu \circ \left ( \mathscr{R} \otimes \id + \id \otimes \mathscr{R} \right ) \, ,
	\end{equation}
	where \(\mu\) denotes the multiplication on \(\EQ^\varepsilon\) (and by abuse of notation also on \(\EQ^\varepsilon_-\) via restriction) from \eqnref{eqn:multiplication_map_ffie}. In particular, \((\EQ^\varepsilon, \mathscr{R})\) is a Rota-Baxter algebra of weight \(\lambda = -1\) and \(\mathscr{R}\) induces the splitting
	\begin{equation}
		\EQ^\varepsilon \cong \EQ^\varepsilon_+ \oplus \EQ^\varepsilon_-
	\end{equation}
	with \(\EQ^\varepsilon_+ := \operatorname{CoKer} \left ( \mathscr{R} \right )\) and \(\EQ^\varepsilon_- := \operatorname{Im} \left ( \mathscr{R} \right )\). Then we can introduce the counterterm map \(\countertermsymbol\), sometimes also called `twisted antipode', recursively via the normalization \(\counterterm{\one} \in \mathbf{1}_{\EQ^\varepsilon}\) and
	\begin{equation}
		\countertermsymbol \, : \quad \operatorname{Aug} \left ( \HQ \right ) \to \EQ^\varepsilon_- \, , \quad \Gamma \mapsto - \renscheme{\countertermsymbol \star \left ( \regFR \circ \mathscr{A} \right )} \left ( \Gamma \right )
	\end{equation}
	else, where \(\mathscr{A} \colon \HQ \surject \operatorname{Aug} \left ( \HQ \right )\) is the projector onto the augmentation ideal from \defnref{defn:augmentation_ideal}. Next we define renormalized Feynman rules via
	\begin{equation}
		\Phi_\mathscr{R} \, : \quad \HQ \to \EQ^\varepsilon_+ \, , \quad \Gamma \mapsto \underset{\varepsilon \mapsto 0}{\operatorname{Lim}} \left ( \countertermsymbol \star \Phi \right ) \left ( \Gamma \right ) \, ,
	\end{equation}
	where the corresponding formal Feynman integral expression is well-defined in the limit \(\varepsilon \mapsto 0\), if the cokernel \(\operatorname{CoKer} \left ( \mathscr{R} \right )\) consists only of convergent formal Feynman integral expressions, cf.\ \lemref{lem:finite_renormalization_schemes}. We remark that the renormalized Feynman rules \(\Phi_\mathscr{R}\) and the counterterm map \(\countertermsymbol\) correspond to the algebraic Birkhoff decomposition of the Feynman rules \(\Phi\) with respect to the renormalization scheme \(\mathscr{R}\), as was first observed in \cite{Connes_Kreimer_0} and e.g.\ reviewed in \cite{Guo,Panzer}. Finally, we remark that the above discussion can be also lifted to the algebra of meromorphic functions \(\mathcal{M}^\varepsilon := \mathbb{C} \big [ \varepsilon^{-1}, \varepsilon \big ] \big ] \), if a suitable regularization scheme \(\mathscr{E}\) is chosen,\footnote{In the sense that the integrated expressions do not contain essential singularities in the regulator.} by setting
	\begin{equation}
		\widetilde{\mathscr{E}} \, : \quad \EQ \to \mathcal{M}^\varepsilon \, , \quad \left ( D,I_\mathscr{E} \left ( \varepsilon \right ) \right ) \mapsto f_\mathscr{E} \left ( \varepsilon \right ) := \int_D \eval{\left ( I_\mathscr{E} \left ( \varepsilon \right ) \right )}_{\varepsilon \in J} \, ,
	\end{equation}
	for fixed external momentum configurations away from Landau singularities. Then we can proceed as before by setting a renormalization scheme as a linear map
	\begin{equation}
		\widetilde{\mathscr{R}} \, : \quad \mathcal{M}^\varepsilon \surject \mathcal{M}^\varepsilon_- \, , \quad f_\mathscr{E} \left ( \varepsilon \right ) \mapsto \begin{cases} 0 & \text{if \(f_\mathscr{E} \left ( \varepsilon \right ) \in \operatorname{Ker} \big ( \widetilde{\mathscr{R}} \big )\)} \\ f_{\mathscr{E}, \mathscr{R}} \left ( \varepsilon \right ) & \text{else} \end{cases} \, ,
	\end{equation}
	where \(\mathcal{M}^\varepsilon_- := \operatorname{Im} \big ( \widetilde{\mathscr{R}} \big ) \subset \mathcal{M}^\varepsilon\), and the rest analogously.
\end{defn}

\enter

\begin{defn}[Hopf subalgebras for multiplicative renormalization] \label{defn:hopf_subalgebras_renormalization_hopf_algebra}
	Let \(\Q\) be a QFT, \(\RQ\) its weighted residue set, \(\HQ\) its (associated) renormalization Hopf algebra and \(\rescombgreen^r_\mathbf{G} \in \HQ\) a restricted combinatorial Green's function, where \(\mathbf{G}\) and \(\mathbf{g}\) denote one of the gradings from \defnref{defn:connectedness_gradings_renormalization_hopf_algebra}. We are interested in Hopf subalgebras which correspond to multiplicative renormalization, i.e.\ Hopf subalgebras of \(\HQ\) such that the coproduct factors over restricted combinatorial Green's functions as follows:
	\begin{equation}
		\Delta \left ( \rescombgreen^r_{\mathbf{G}} \right ) = \sum_{\mathbf{g}} \mathfrak{P}_{\mathbf{g}} \left ( \rescombgreen^r_{\mathbf{G}} \right ) \otimes \rescombgreen^r_{\mathbf{G} - \mathbf{g}} \, , \label{eqn:hopf_subalgebras_multi-index}
	\end{equation}
	where \(\mathfrak{P}_{\mathbf{g}} \left ( \rescombgreen^r_{\mathbf{G}} \right ) \in \HQ\) is a polynomial in graphs such that each summand has multi-index \(\mathbf{g}\).\footnote{There exist closed expressions for the polynomials \(\mathfrak{P}_{\mathbf{g}} \left ( \rescombgreen^r_{\mathbf{G}} \right )\) as we will see in \sectionref{sec:coproduct_and_antipode_identities}, in particular \propref{prop:coproduct_greensfunctions},, which were first introduced in \cite{Yeats_PhD}.}
\end{defn}

\enter

\begin{rem} \label{rem:hopf_subalgebras_renormalization_hopf_algebra}
	Given the situation of \defnref{defn:fr_reg_ren_counterterm} and assume that the (associated) renormalization Hopf algebra \(\HQ\) possesses Hopf subalgebras in the sense of \defnref{defn:hopf_subalgebras_renormalization_hopf_algebra}. Then we can calculate the \(Z\)-factor for a given residue \(r \in \RQ\) via
	\begin{equation}
		Z^r_{\mathscr{E}, \mathscr{R}} \left ( \varepsilon \right ) := \counterterm{\combgreen^r} \, .
	\end{equation}
	More details in this direction can be found in \cite{Panzer,vSuijlekom_Multiplicative} (with a different notation). Additionally, we remark that the existence of the Hopf subalgebras from \defnref{defn:hopf_subalgebras_renormalization_hopf_algebra} depends crucially on the grading \(\mathbf{g}\). In particular, for the loop-grading these Hopf subalgebras exist if and only if the QFT has only one fundamental interaction, i.e.\ \(\RQO\) is a singleton. Furthermore, they exist for the coupling-grading if and only if the QFT has for each fundamental interaction a different coupling constant, i.e.\ \(\# \QQ = \# \qQ\). Finally, they exist always for the vertex-grading, as we will see in \propref{prop:coproduct_greensfunctions}, cf.\ \sectionsaref{sec:coproduct_and_antipode_identities}{sec:quantum_gauge_symmetries_and_subdivergences}.
\end{rem}

\enter

\begin{lem} \label{lem:finite_renormalization_schemes}
	The image of the renormalized Feynman rules \(\operatorname{Im} \left ( \Phi_\mathscr{R} \right )\) consists of convergent integral expressions, if the cokernel \(\operatorname{CoKer} \left ( \mathscr{R} \right )\) of the corresponding renormalization scheme \(\mathscr{R} \in \operatorname{End} \left ( \EQ^\varepsilon \right )\) does.
\end{lem}

\begin{proof}
	The theorem about the algebraic Birkhoff decomposition, first observed in \cite{Connes_Kreimer_0}, states in this context that
	\begin{align}
		\Phi_\mathscr{R} & \, : \quad \HQ \to \EQ^\varepsilon_+
		\intertext{and}
		\countertermsymbol & \, : \quad \HQ \to \EQ^\varepsilon_- \, ,
	\end{align}
	where \(\EQ^\varepsilon_+ := \operatorname{CoKer} \left ( \mathscr{R} \right )\) and \(\EQ^\varepsilon_- := \operatorname{Im} \left ( \mathscr{R} \right )\), and thus \(\operatorname{Im} \left ( \Phi_\mathscr{R} \right )\) consists of finite integral expressions, if \(\operatorname{CoKer} \left ( \mathscr{R} \right )\) does.
\end{proof}

\enter

\begin{defn}[Proper renormalization schemes] \label{defn:proper_renormalization_scheme}
	A renormalization scheme \(\mathscr{R} \in \operatorname{End} \left ( \EQ^\varepsilon \right )\) is called proper, if both its kernel \(\operatorname{Ker} \left ( \mathscr{R} \right )\) and its cokernel \(\operatorname{CoKer} \left ( \mathscr{R} \right )\) consist only of convergent integral expressions.\footnote{We allow, as an exception, superficially divergent graphs to be in the kernel of \(\mathscr{R}\), if they would lead to an ill-defined coalgebra structure on the renormalization Hopf algebra. See \cite[Subsection 3.3]{Prinz_2} for a detailed discussion on this matter.} In particular, we demand that
	\begin{equation}
		\operatorname{Im} \left ( \Phi \circ \Omega \right ) \subseteq \operatorname{CoIm} \left ( \mathscr{R} \right ) \, ,
	\end{equation}
	i.e.\ the image of superficially divergent graphs under the Feynman rules is a subset of the coimage of a proper renormalization scheme.
\end{defn}

\enter

\begin{rem}
	\defnref{defn:proper_renormalization_scheme} is motivated by the fact that in physics we want renormalization schemes to produce finite, i.e.\ integrable, renormalized Feynman rules and furthermore preserve the locality of the theory, i.e.\ remove divergences of Feynman integrals via contributions from themselves.
\end{rem}

\section{The associated renormalization Hopf algebra} \label{sec:the_associated_renormalization_hopf_algebra}

In this section, we describe a problem which may occur in the construction of the renormalization Hopf algebra \(\HQ\) to a given Quantum Field Theory \(\Q\) using \defnref{defn:renormalization_hopf_algebra}. Then we present four different solutions to still obtain a renormalization Hopf algebra (which are not isomorphic if the problem occurs) and discuss their physical interpretation. This section is taken from \cite[Subsection 3.3]{Prinz_2}.

\vspace{\baselineskip}

\begin{pro} \label{pro:problem}
Given a general Quantum Field Theory \(\Q\), \defnref{defn:renormalization_hopf_algebra} may not yield a well-defined Hopf algebra due to the following reason: Let \(\Q\) be such that there exist divergent Feynman graphs \(\gamma \in \mathcal{G}_\Q\) whose residue is not in the residue set, i.e.\ we have \(\sdd{\gamma} \geq 0\) and \(\res{\gamma} \notin \RQ\). Then given any Feynman graph \(\Gamma \in \mathcal{G}_\Q\) with \(\gamma \in \mathcal{D} \left ( \Gamma \right )\), the quotients of the form \(\Gamma / \gamma\) for \(\gamma \subsetneq \Gamma\) are ill-defined, as they generate a new vertex \(\res{\gamma} \notin \RQ^{[0]}\). As a consequence, the definitions of the coproduct and the antipode are ill-defined as well. Physically speaking, there is no monomial in the Lagrange density \(\LQ\) to absorb the divergences coming from the graph \(\gamma\).
\end{pro}

\vspace{\baselineskip}

\begin{rem}
In order to remedy \proref{pro:problem} we need to change some of the definitions. This is explained in the following \solsaref{sol:solution_1}{sol:solution_2}{sol:solution_3}{sol:solution_4}. In order to distinguish the different objects, we add tildes to the objects with modified definitions. Then, after this section, we drop the tildes again and simply use the original symbols.
\end{rem}

\vspace{\baselineskip}

\begin{defn}[Modified Feynman graph set] \label{defn:modified_feynman_graph_set}
Let \(\Q\) be a Quantum Field Theory with residue set \(\mathcal{R}_{\Q}\). Recall from \defnref{defn:feynman_graphs} that we denote by \(\mathcal{G}_{\Q}\) the set of all one-particle irreducible (1PI) Feynman graphs that can be generated by the residue set \(\mathcal{R}_{\Q}\) of \(\Q\). Moreover, we define the set \(\widetilde{\mathcal{G}}_{\Q}\) of all 1PI Feynman graphs of \(\Q\) which does not contain superficially divergent subgraphs whose residue is not in the residue set \(\mathcal{R}_{\Q}\), i.e.
\begin{equation}
	\widetilde{\mathcal{G}}_{\Q} := \left \{ \Gamma \in \mathcal{G}_{\Q} \, \left \vert \; \res{\gamma} \in \RQ \, \text{ for all } \, \gamma \in \DQ{\Gamma} \right . \right \} \, .
\end{equation}
This set will be used in \solref{sol:solution_1}.
\end{defn}

\vspace{\baselineskip}

\begin{defn}[Modified sets of superficially divergent subgraphs] \label{defn:modified_sets_of_superficially_divergent_subgraphs}
Let \(\Q\) be a Quantum Field Theory and \(\Gamma \in \mathcal{G}_{\Q}\) a Feynman graph of \(\Q\). Recall from \defnref{defn:set_of_divergent_subgraphs} that we denote by \(\mathcal{D} \left ( \Gamma \right )\) the set of superficially divergent subgraphs of \(\Gamma\) and by \(\mathcal{D}^\prime \left ( \Gamma \right )\) the set of superficially divergent proper subgraphs of \(\Gamma\). Moreover, we define the two additional sets \(\widetilde{\mathcal{D}} \left ( \Gamma \right )\) and \(\widetilde{\mathcal{D}}^\prime \left ( \Gamma \right )\), corresponding to \(\mathcal{D} \left ( \Gamma \right )\) and \(\mathcal{D}^\prime \left ( \Gamma \right )\), respectively, which do not contain Feynman graphs with superficially divergent subgraphs whose residue is not in the residue set \(\mathcal{R}_{\Q}\), i.e.\
\begin{subequations}
\begin{equation}
	\widetilde{\mathcal{D}} \left ( \Gamma \right ) := \left \{ \gamma \in \mathcal{D} \left ( \Gamma \right ) \, \left \vert \; \res{\gamma_\text{c}} \in \RQ \, \text{ for all } \, \gamma_\text{c} \in \CQ{\gamma} \right . \right \} \, ,
\end{equation}
where \(\CQ{\gamma}\) denotes the set of connected components of \(\gamma\), and
\begin{equation}
	\widetilde{\mathcal{D}}^\prime \left ( \Gamma \right ) := \left \{ \gamma \in \widetilde{\mathcal{D}} \left ( \Gamma \right ) \, \left \vert \; \gamma \subsetneq \Gamma \right . \right \} \, ,
\end{equation}
\end{subequations}
These sets will be used in \solref{sol:solution_2}.
\end{defn}

\vspace{\baselineskip}

\begin{sol} \label{sol:solution_1}
The first solution to \proref{pro:problem} to replace the Feynman graph set \(\mathcal{G}_\Q\) from \defnref{defn:feynman_graphs} by \(\widetilde{\mathcal{G}}_\Q\) from \defnref{defn:modified_feynman_graph_set}. Then, we can construct the renormalization Hopf algebra as in \defnref{defn:renormalization_hopf_algebra}, which we now denote by \(\widetilde{\mathcal{H}}^{(1)}_\Q\).
\end{sol}

\vspace{\baselineskip}

\begin{sol} \label{sol:solution_2}
The second solution to \proref{pro:problem} is to simply remove all divergent Feynman graphs whose residue is not in the residue set from the sets of divergent subgraphs, i.e.\ replace the sets \(\mathcal{D} \left ( \cdot \right )\) and \(\mathcal{D}^\prime \left ( \cdot \right )\) from \defnref{defn:set_of_divergent_subgraphs} by \(\widetilde{\mathcal{D}} \left ( \cdot \right )\) and \(\widetilde{\mathcal{D}}^\prime \left ( \cdot \right )\) from \defnref{defn:modified_sets_of_superficially_divergent_subgraphs}, respectively. Then, we can construct the renormalization Hopf algebra as in \defnref{defn:renormalization_hopf_algebra}, which we now again denote by \(\widetilde{\mathcal{H}}^{(2)}_\Q\).
\end{sol}

\vspace{\baselineskip}

\begin{sol} \label{sol:solution_3}
The third solution to \proref{pro:problem} is to add all missing residues to the residue set and set its weights to the value of a divergent Feynman graph with this particular residue.\footnote{If there exist two or more such graphs with different superficial degree of divergence we take the highest for uniqueness.} This enlarges also the set of Feynman graphs. Then, we can define the Hopf algebra using this enlarged set of Feynman graphs as in \defnref{defn:renormalization_hopf_algebra}, which we now again denote by \(\widetilde{\mathcal{H}}^{(3)}_\Q\).
\end{sol}

\vspace{\baselineskip}

\begin{sol} \label{sol:solution_4}
The fourth solution to \proref{pro:problem} works only in special cases: Given that there exist tree diagrams with the residue of a divergent Feynman graph whose residue is not in the residue set. Then, we can construct the renormalization Hopf algebra as in \defnref{defn:renormalization_hopf_algebra} by defining the shrinking process as the replacement of the aforementioned graphs with the sum over all possible such trees, which we now again denote by \(\widetilde{\mathcal{H}}^{(4)}_\Q\).\footnote{This procedure is a priori non-local, but could effectively become local via appropriate cancellation identities that are similar to the Slavnov--Taylor identities in Quantum Yang--Mills theory, except that there is no higher-valent vertex residue. We refer to \sectionref{sec:longitudinal_and_transversal_projections} for a discussion on Quantum Yang--Mills theory and (effective) Quantum General Relativity, where we conjecture that this is the right solution for gravity-matter couplings.}
\end{sol}

\vspace{\baselineskip}

\begin{rem}
Equivalently, the renormalization Hopf algebra from \solref{sol:solution_1} could be constructed in two different ways: The first possibility is to define \(\mathcal{H}_\Q\) as the \(\mathbb{Q}\)-algebra generated by the set \(\mathcal{G}_\Q\) of Feynman graphs with the product and unit as in \defnref{defn:renormalization_hopf_algebra}. Then we define the ideal \(\mathfrak{j}_\Q\) generated via the Feynman graphs in the set \(\mathcal{G}_\Q \setminus \widetilde{\mathcal{G}}_\Q\) and consider the quotient
\begin{equation}
	\widetilde{\mathcal{H}}^{(1)}_\Q := \mathcal{H}_\Q / \mathfrak{j}_\Q \, ,
\end{equation}
on which we can define the additional Hopf algebra structures as in \defnref{defn:renormalization_hopf_algebra}. The second possibility is to use the Hopf algebra \(\widetilde{\mathcal{H}}^{(3)}_\Q\) from \solref{sol:solution_3} and define the ideal \(\mathfrak{k}_\Q\) generated via Feynman graphs which contain vertices that are not in the residue set \(\RQO\). Then we obtain
\begin{equation}
	\widetilde{\mathcal{H}}^{(1)}_\Q := \widetilde{\mathcal{H}}^{(3)}_\Q / \left ( \mathfrak{j}_\Q + \mathfrak{k}_\Q \right ) \, .
\end{equation}
\end{rem}

\vspace{\baselineskip}

\begin{rem}[Physical interpretation] \label{rem:physical_interpretation}
Physically, \proref{pro:problem} states that divergent Feynman graphs whose residue is not in the residue set contribute in principle to a divergent Green's function which cannot be renormalized if the corresponding vertex is missing in the Quantum Field Theory. However, there could be still three possibilities that the unrenormalized Feynman rules remedy the problem themselves: The first one is that the problematic Feynman graphs itself or the corresponding restricted combinatorial Green's functions turn out to be in the kernel of the unrenormalized Feynman rules which corresponds to \solref{sol:solution_1}. The second one is that the problematic Feynman graphs itself or the corresponding restricted combinatorial Green's functions turn out to be already finite when applying the unrenormalized Feynman rules which corresponds to \solref{sol:solution_2}. However, if this is not the case, we need to add the corresponding vertices with suitable Feynman rules in order to absorb the divergences of the corresponding restricted combinatorial Green's functions via multiplicative renormalization corresponding to \solref{sol:solution_3}. Luckily, for all established physical Quantum Field Theories this situation did not appear so far. Finally, the last scenario could be still circumvented, if corresponding tree graphs exist together with corresponding cancellation identities that render this a priori non-local process local, which corresponds to \solref{sol:solution_4}.
\end{rem}

\vspace{\baselineskip}

\begin{exmp}[QED]
An illuminating example to \proref{pro:problem} is QED since we need to apply both, \solref{sol:solution_1} and \solref{sol:solution_3}: Consider QED with its combinatorics as a renormalizable Quantum Field Theory (i.e.\ the superficial degree of divergence of a Feynman graph depends only on its external leg structure). Then, the Feynman graphs contributing to the three- and four-point function are divergent --- however in contrast to non-abelian quantum gauge theories, there is no three- and four-photon vertex present to absorb the corresponding divergences. Luckily, when summing all Feynman graphs of a given loop order we have the following cancellations after applying the unrenormalized Feynman rules: The Feynman graphs contributing to the three-point function cancel pairwise due to Furry's Theorem, cf.\ \cite[Theorem 3.55]{Prinz_2} for a generalization thereof and the divergences of the Feynman graphs contributing to the four-point function cancel pairwise due to gauge invariance \cite{Aldins_Brodsky_Dufner_Kinoshita}. Thus, QED is a renormalizable Quantum Field Theory after all without the need to add a three- and four-photon vertex to the theory.
\end{exmp}

\vspace{\baselineskip}

\begin{rem}[The situation of QGR-SM]
The situation of QGR-SM is worse than the one for QED, since QGR is non-renormalizable as a Quantum Field Theory (in particular, the superficial degree of divergence of a pure gravity Feynman graph depends only on its loop number). In particular, there exist Feynman graphs whose external leg structure consists of any combination of even numbers of matter particles, with additional gravitons as virtual particles. These graphs are superficially divergent, however there do not exist corresponding residues in the residue set. For this scenario we suggest \solref{sol:solution_4} to avoid introducing further vertices, as the corresponding trees are already part of the theory. In particular, we claim that corresponding generalized versions of Slavnov--Taylor identities will render this a priori non-local renormalization operation local. We refer to the more detailed discussion in the beginning of \sectionref{sec:longitudinal_and_transversal_projections}.
\end{rem}

\vspace{\baselineskip}

\begin{defn}[Renormalization Hopf algebra associated to a Quantum Field Theory] \label{defn:renormalization_hopf_algebra_associated_to_a_local_qft}
Let \(\Q\) be a Quantum Field Theory. Then we denote by \(\HQ\) the renormalization Hopf algebra resulting from the application of any of the \solsoref{sol:solution_1}{sol:solution_2}{sol:solution_3}{sol:solution_4} to \defnref{defn:renormalization_hopf_algebra}. Then we call \(\HQ\) ``the renormalization Hopf algebra associated to \(\Q\)''.
\end{defn}

\vspace{\baselineskip}

\begin{rem}
The motivation for \defnref{defn:renormalization_hopf_algebra_associated_to_a_local_qft} is to simplify notation, as for the realm of this work it is not necessary to distinguish between \solsaref{sol:solution_1}{sol:solution_2}{sol:solution_3}{sol:solution_4}.
\end{rem}

\section{Predictive Quantum Field Theories} \label{sec:predictive-quantum-field-theories}

To address the situation of (effective) Quantum General Relativity more precisely, we introduce the notion of a `predictive Quantum Field Theory' accompanying the well-established classification into super-renormalizable, renormalizable and non-renormalizable Quantum Field Theories. This turns out to be a useful addendum, since renormalization theory allows for non-renormalizable Quantum Field Theories that have a well-defined perturbative expansion after a proper incorporation of their internal symmetries \cite{Prinz_3}. More precisely, it is always possible to render any Feynman integral finite by applying suitable subtractions. Thus, the relevant question becomes whether these subtractions can be organized in a systematic way into \(Z\)-factors to allow for multiplicative renormalization. In addition to the well-established classification of renormalizability, this situation can be essentially improved by internal symmetries of the theory. If this is possible, we call a Quantum Field Theory predictive.

\enter

\begin{defn}[Classification of quantum field theories] \label{defn:classification_of_qfts}
Let \(\Q\) be a quantum field theory and let \(\qQ\) denote the set of coupling constants of \(\Q\), that is the set of constants scaling monomials in \(\LQ\) involving at least three fields. Then, an interaction with coupling constant \(\alpha_i \in \qQ\) is called:
\begin{enumerate}
\item super-renormalizable, iff \(\left [ \alpha_i \right ] < 0\)
\item renormalizable, iff \(\left [ \alpha_i \right ] = 0\)
\item non-renormalizable, iff \(\left [ \alpha_i \right ] > 0\)
\end{enumerate}
where \(\left [ \alpha_i \right ]\) denotes the mass-dimension of \(\alpha_i\).
\end{defn}

\enter

These criteria for the classification of \(\Q\) can be equivalently rephrased in terms of its perturbative expansion via properties of its superficial degree of divergence:

\enter

\begin{lem} \label{lem:classification-qfts}
Let \(\Q\) be a quantum field theory with residue set \(\RQ\) and Feynman graph set \(\GQ\), then the interaction with respect to the residue \(r_i \in \RQ\) can be classified according to \defnref{defn:classification_of_qfts} equivalently as follows:
\begin{enumerate}
\item super-renormalizable, iff \(\frac{\partial}{\partial \mathbf{g}_i} \sdd{r_i, \mathbf{g}_i} < 0\)
\item renormalizable, iff \(\frac{\partial}{\partial \mathbf{g}_i} \sdd{r_i, \mathbf{g}_i} = 0\)
\item non-renormalizable, iff \(\frac{\partial}{\partial \mathbf{g}_i} \sdd{r_i, \mathbf{g}_i} > 0\)
\end{enumerate}
where \(\omega\) denotes the superficial degree of divergence, cf.\ \defnref{defn:sdd}, and \(\mathbf{g}_i\) denotes an appropriate grading, cf.\ \defnref{defn:connectedness_gradings_renormalization_hopf_algebra}, which we consider here as a continuous variable.
\end{lem}

\begin{proof}
We start by recalling that the action has mass dimension \(\left [ S \right ] = 0\), an infinitesimal line element \(\left [ \dif x \right ] = -1\) and the partial derivative \(\left [ \partial \right ] = 1\). Thus, the \(\ds\)-dimensional Minkowskian volume form \(\dif V_\eta\ := \dif t \wedge \dif x_1 \wedge \dots \wedge \dif x_{\ds - 1}\) has mass dimension \(\left [ \dif V_\eta \right ] = - \ds\), and the Lagrange function \(L \equiv L_\text{Kin} + L_\text{Int}\) of the Lagrange density \(\mathcal{L} := L \dif V_\eta\) has mass dimension \(\left [ L \right ] = \ds\). Applying this knowledge to the kinetic terms \(L_\text{Kin} \equiv \sum_{i,j} \left ( \partial \phi_i \right ) \left ( \partial \phi_i \right ) + \big ( \overline{\psi}_j \slashed{\partial} \psi_j \big )\), we obtain that bosonic and ghost fields \(\phi_i\) have mass dimension \(\left [ \phi_i \right ] = \textfrac{(d - 2)}{2}\), whereas fermionic fields \(\psi_j\) have mass dimension \(\left [ \psi_j \right ] = \textfrac{(d - 1)}{2}\). Therefore, the mass dimension of the coupling constant is directly linked to the number of partial derivatives in each interaction monomial, which, in turn, is responsible for the UV-behavior of the corresponding Feynman rule. More precisely, let \(L_\text{Int} \equiv \sum_k \alpha_k I_k\) be the interaction Lagrange function and \(\alpha_k I _k\) an interaction monomial therein. Let \(\# \partial_{I_k}\) denote the number of partial derivatives, \(\# \phi_k\) the number of bosonic and ghost fields and \(\# \psi_k\) the number of fermionic fields in \(I_k\). Using \(\left [ L \right ] = \ds\) again, we obtain
\begin{equation}
	\# \partial_{I_k} = d - \left [ \alpha_k \right ] - \left ( \frac{d - 2}{2} \right ) \# \phi_k - \left ( \frac{d - 1}{2} \right ) \# \psi_k \, .
\end{equation}
Additionally, we denote by Let \(\# \partial_{C_k}\) the number of partial derivatives for the corresponding corolla of the interaction, that is the vertex together with its half-edges. As each half-edge corresponds to half of a propagator, each bosonic or ghost edge contains \(- 1\) derivatives and each fermionic edge contains \(- \textfrac{1}{2}\) derivatives. Thus, we obtain
\begin{equation}
	\# \partial_{C_k} = d \left ( 1 - \frac{1}{2} \# \phi_k - \frac{1}{2} \# \psi_k \right ) - \left [ \alpha_k \right ] \, .
\end{equation}
Comparing this result to \colref{col:weights_of_corollas_and_renormalizability}, which states that an interaction is renormalizable if the weight of its corolla is
\begin{equation}
	\csdd{v} \equiv d \left ( 1 - \frac{1}{2} \val{v} \right ) \, ,
\end {equation}
super-renormalizable if the weight is smaller and non-renormalizable if the weight is bigger, we obtain the claimed result by identifying \(\val{v} \equiv \# \phi_k + \# \psi_k\).
\end{proof}

\enter

After having discussed the notion of renormalizability, we extend this traditional classification by the notion of a `predictive quantum field theory':

\enter

\begin{defn}[Predictive Quantum Field Theory] \label{defn:predictive_quantum_field_theory}
Let \(\Q\) be a Quantum Field Theory with coupling constant set \(\qQ\) and counterterm set \(\ZQ\). We call \(\Q\) predictive, if the set \(\qQ\) is finite and either the set \(\ZQ\) itself is finite or \(\Q\) possesses `quantum gauge symmetries', cf.\ \defnref{defn:quantum_gauge_symmetries}, inducing a finite set of counterterms \(\ZQGSQ\) built from the set \(\ZQ\).
\end{defn}

\enter

This then immediately leads to the following:

\enter

\begin{obs} \label{obs:s-ren-and-ren-qfts-are-predictive}
Every super-renormalizable or renormalizable Quantum Field Theory \(\Q\) with a finite coupling constant set \(\qQ\) is predictive.
\end{obs}

\enter

Finally, given the results of \cite{Prinz_3,Kreimer_QG1}, we claim the following:

\enter

\begin{conj} \label{conj:qgr-predictive}
	(Effective) Quantum General Relativity, possibly coupled to matter from the Standard Model, is predictive.
\end{conj}

\section{A superficial argument} \label{sec:a_superficial_argument}

In this section, we study combinatorial properties of the `superficial degree of divergence (SDD)' from \defnref{defn:sdd}. This integer,  combinatorially associated to each Feynman graph \(\Gamma \in \GQ\), provides a measure of the ultraviolet divergence of the corresponding Feynman integral. In fact, a result of Weinberg states that the ultraviolet divergence of the (formal) Feynman integral \(\Phi \left ( \Gamma \right )\) is bounded by a polynomial of degree \(\operatorname{Max}_{\one \subsetneq \gamma \subseteq \Gamma} \sdd{\gamma}\), where the maximum is considered over all non-empty 1PI subgraphs and \(\omega \colon \GQ \to \mathbb{Z}\) denotes the SDD of \(\gamma\) \cite{Weinberg}. More precisely, the corresponding Feynman integral converges if \(\sdd{\Gamma} < 0\) and \(\sdd{\gamma} < 0\) for all non-empty 1PI subgraphs \(\one \subsetneq \gamma \subset \Gamma\). We start this section by providing an alternative definition of the SDD in terms of weights of corollas in \defnref{defn:asdd} and \lemref{lem:asdd}. With this criterion on hand, we show in \thmref{thm:asdd} that the SDD of Feynman graphs with a fixed residue depends affine-linearly on their vertex-grading. Then we give in \colref{col:weights_of_corollas_and_renormalizability} an alternative characterization to the classification of QFTs into (super-/non-)renormalizability via weights of corollas. Building upon this result, we introduce the notion of a `cograph-divergent QFT' in \defnref{defn:cograph-divergent}, which is shown in \propref{prop:proj_div_graphs_coprod} to be the obstacle for the compatibility of coproduct identities with the projection to divergent Feynman graphs from \defnref{defn:projection_divergent_graphs}. Next we study the superficial compatibility of the coupling-grading and loop-grading in \propref{prop:superficial_grade_compatibility}. We complete this section by showing that (effective) Quantum General Relativity coupled to the Standard Model (QGR-SM) satisfies both, i.e.\ is cograph-divergent and has superficially compatible coupling-grading in \colref{col:qgr-sm_is_cograph-divergent} and \colref{col:qgr-sm_is_sqgsc}. The results in this section are in particular useful for the following sections, where we want to state our results not only for the renormalizable case, but include also the more involved super- and non-renormalizable cases, or even mixes thereof. Finally, this allows us to apply our results to QGR-SM.

\enter

\begin{defn}[Weights of corollas] \label{defn:asdd}
	Let \(\Q\) be a QFT and \((\RQ, \omega)\) its weighted residue set. Given a vertex residue \(v \in \RQO\), we define the weight of its corolla via
	\begin{equation} \label{eqn:pre-asdd}
		\csdd{v} \equiv \sdd{c_v} := \sdd{v} + \frac{1}{2} \sum_{e \in E \left ( v \right )} \sdd{e} \, .
	\end{equation}
\end{defn}

\enter

\begin{lem} \label{lem:asdd}
	Given the situation of \defnref{defn:asdd}, the superficial degree of divergence of a Feynman graph \(\Gamma \in \GQ\) can be equivalently calculated via\footnote{We remark that the equivalence holds only for non-trivial graphs, i.e.\ graphs with at least one vertex.}
	\begin{equation}
		\omega \, : \quad \GQ \to \mathbb{Z} \, , \quad \Gamma \mapsto d \bettii{\Gamma} + \sum_{v \in V \left ( \Gamma \right )} \csdd{v} - \frac{1}{2} \sum_{e \in E \left ( \res{\Gamma} \right )} \omega \left ( e \right ) \, . \label{eqn:prealternative_superficial_degree_of_divergence}
	\end{equation}
\end{lem}

\begin{proof}
	This follows directly from the combination of \eqnref{eqn:sdd} from \defnref{defn:sdd} and \eqnref{eqn:pre-asdd} from \defnref{defn:asdd}.
\end{proof}

\enter

\begin{thm}[Superficial degree of divergence via residue and vertex-grading] \label{thm:asdd}
	Given the situation of \lemref{lem:asdd}, the superficial degree of divergence of a Feynman graph \(\Gamma \in \GQ\) can be decomposed as follows:\footnote{The function \(\sigma\) can be expressed equivalently as a linear functional in the coarser gradings, if they are superficially compatible, cf.\ \defnref{defn:superficially_compatible_grading}.}
	\begin{subequations}
	\begin{align}
		\sdd{\Gamma} & \equiv \rsdd{\Gamma} + \ssdd{\Gamma} \, , \label{eqn:asdd}
		\intertext{where \(\rsdd{\Gamma}\) depends only on \(\res{\Gamma}\) and \(\ssdd{\Gamma}\) depends only on \(\vtxgrd{\Gamma}\). In particular, we have:}
		\rsdd{\Gamma} & := \begin{cases} \omega \left ( \res{\Gamma} \right ) & \text{if \(\res{\Gamma} \in \RQO\)} \\ d - \dfrac{1}{2} \sum_{e \in E \left ( \res{\Gamma} \right )} \omega \left ( e \right ) & \text{else, i.e.\ \(\res{\Gamma} \in \left ( \AQ \setminus \RQO \right )\)} \end{cases} \label{eqn:rsdd}
		\intertext{and}
		\ssdd{\Gamma} & := \sum_{i = 1}^{\mathfrak{v}_\Q} \left ( d \left ( \frac{1}{2} \val{v_i} - 1 \right ) + \csdd{v_i} \right ) \left ( \vtxgrd{\Gamma} \right )_i \label{eqn:ssdd}
	\end{align}
	\end{subequations}
	Notably, \(\ssdd{\Gamma}\) is linear in \(\vtxgrd{\Gamma}\) and thus \(\sdd{\Gamma}\) is affine linear in \(\vtxgrd{\Gamma}\).
\end{thm}

\begin{proof}
	We start with \eqnref{eqn:prealternative_superficial_degree_of_divergence} from \defnref{defn:asdd}: First we rewrite the loop number using the Euler characteristic\footnote{For the Euler characteristic, \eqnref{eqn:euler_characteristic}, we need to either ignore the external (half-)edges or assume that they are attached to external vertices, as in \defnref{defn:feynman_graphs}, and adjust the loop-number accordingly.}
	\begin{equation}
		\bettii{\Gamma} = \bettio{\Gamma} - \# V \left ( \Gamma \right ) + \# E \left ( \Gamma \right ) \, , \label{eqn:euler_characteristic}
	\end{equation}
	where \(\bettio{\Gamma} = 1\), as \(\Gamma \in \GQ\) is connected. Then we express \(\# V \left ( \Gamma \right )\) in terms of \(\vtxgrd{\Gamma}\) and \(\res{\Gamma}\) via
	\begin{equation} \label{eqn:vertex-set_and_grading}
	\begin{split}
		\# V \left ( \Gamma \right ) & = \sum_{i = 1}^{\mathfrak{v}_\Q} \left ( \intvtx{\Gamma} \right )_i\\
		& = \sum_{i = 1}^{\mathfrak{v}_\Q} \left ( \left ( \vtxgrd{\Gamma} \right )_i + \left ( \extvtx{\Gamma} \right )_i \right )\\
		& = \sum_{i = 1}^{\mathfrak{v}_\Q} \left ( \left ( \vtxgrd{\Gamma} \right )_i \right ) + \delta_{\res{\Gamma} \in \RQO} \, ,
	\end{split}
	\end{equation}
	where in the last equality we have used that \(\Gamma \in \GQ\) is connected by setting
	\begin{equation}
		\delta_{\res{\Gamma} \in \RQO} = \begin{cases} 1 & \text{if \(\res{\Gamma} \in \RQO\)} \\ 0 & \text{else, i.e.\ \(\res{\Gamma} \in \left ( \AQ \setminus \RQO \right )\)} \end{cases} \, .
	\end{equation}
	Furthermore, we express \(\# E \left ( \Gamma \right )\) in terms of \(\vtxgrd{\Gamma}\) and the valences of the corresponding vertices via
	\begin{equation} \label{eqn:edge-set_and_grading}
	\begin{split}
		\# E \left ( \Gamma \right ) & = \frac{1}{2} \sum_{i = 1}^{\mathfrak{v}_\Q} \val{v_i} \left ( \intvtx{\Gamma} \right )_i - \frac{1}{2} \sum_{i = 1}^{\mathfrak{v}_\Q} \val{v_i} \left ( \extvtx{\Gamma} \right )_i\\
		& = \frac{1}{2} \sum_{i = 1}^{\mathfrak{v}_\Q} \val{v_i} \left ( \vtxgrd{\Gamma} \right )_i \, .
	\end{split}
	\end{equation}
	Thus, combining \eqnssaref{eqn:euler_characteristic}{eqn:vertex-set_and_grading}{eqn:edge-set_and_grading}, we obtain
	\begin{equation}
		\bettii{\Gamma} = 1 - \delta_{\res{\Gamma} \in \RQO} + \frac{1}{2} \sum_{i = 1}^{\mathfrak{v}_\Q} \left ( \val{v_i} - 2 \right ) \left ( \vtxgrd{\Gamma} \right )_i \, . \label{eqn:bettii_residue_vertex-grading}
	\end{equation}
	We proceed by rewriting
	\begin{equation}
		\sum_{v \in V \left ( \Gamma \right )} \csdd{v} = \sum_{i = 1}^{\mathfrak{v}_\Q} \csdd{v_i} \left ( \vtxgrd{\Gamma} \right )_i \, .
	\end{equation}
	Finally, plugging the above results into \eqnref{eqn:prealternative_superficial_degree_of_divergence} from \defnref{defn:asdd} yields the claimed result.
\end{proof}

\enter

\begin{col}[Weights of corollas and renormalizability] \label{col:weights_of_corollas_and_renormalizability}
	Given the situation of \thmref{thm:asdd} and a vertex residue \(v \in \RQO\). Then its corolla \(c_v\) is renormalizable if its weight is
	\begin{equation}
		\csdd{v} \equiv d \left ( 1 - \frac{1}{2} \val{v} \right ) \, , \label{eqn:weight_corolla}
	\end{equation}
	non-renormalizable if its weight is bigger and super-renormalizable if its weight is smaller. In particular, the QFT \(\Q\) is renormalizable, non-renormalizable or super-renormalizable if all of its corollas are.
\end{col}

\begin{proof}
	Before presenting the actual argument, we recall that a QFT \(\Q\) is called renormalizable, if, for a fixed residue \(r\), the superficial degree of divergence is independent of the grading, non-renormalizable if the superficial degree of divergence increases with increasing grading and super-renormalizable if the superficial degree of divergence decreases with increasing grading. Using \eqnref{eqn:ssdd} from \thmref{thm:asdd}, we directly obtain the claimed bound
	\begin{equation}
	\csdd{v} = d \left ( 1 - \frac{1}{2} \val{v} \right ) \, ,
	\end{equation}
	as for this value the superficial degree of divergence of a Feynman graph is independent of the corolla \(c_v\), it is positively affected if its weight is bigger and it is negatively affected if its weight is smaller.
\end{proof}

\enter

\begin{col} \label{col:ssdd_decomposition}
	Given the situation of \thmref{thm:asdd}, the dependence of \(\ssdd{\Gamma}\) on \(\vtxgrd{\Gamma}\) can be furthermore refined by decomposing
	\begin{equation}
		\ssdd{\Gamma} \equiv \ssddn{\Gamma} + \ssddr{\Gamma} + \ssdds{\Gamma} \label{eqn:ssdd_decomposition}
	\end{equation}
	where
	\begin{subequations}
	\begin{align}
		\ssddn{\Gamma} & \equiv \left \vert \ssddn{\Gamma} \right \vert \, , \\
		\ssddr{\Gamma} & \equiv 0
		\intertext{and}
		\ssdds{\Gamma} & \equiv - \left \vert \ssdds{\Gamma} \right \vert \, .
	\end{align}
	\end{subequations}
	In particular, \(\ssddn{\Gamma}\) depends only on the non-renormalizable corollas, \(\ssddr{\Gamma}\) depends only on the renormalizable corollas and \(\ssdds{\Gamma}\) depends only on the super-renormalizable corollas.
\end{col}

\begin{proof}
	This follows directly from \thmref{thm:asdd} and \colref{col:weights_of_corollas_and_renormalizability}.
\end{proof}

\enter

\begin{defn}[Cograph-divergent QFTs] \label{defn:cograph-divergent}
	Let \(\Q\) be a QFT with residue set \(\RQ\) and weighted Feynman graph set \((\GQ, \omega)\). We call \(\Q\) cograph-divergent, if for each superficially divergent Feynman graph \(\Gamma \in \overline{\GQ}\) and its superficially divergent subgraphs \(\gamma \in \DQ{\Gamma}\), the corresponding cographs \(\Gamma / \gamma\) are also all superficially divergent, i.e.\ fulfill \(\sdd{\Gamma / \gamma} \geq 0\).
\end{defn}

\enter

\begin{rem}
	\defnref{defn:cograph-divergent} is trivially satisfied for super-renormalizable and renormalizable QFTs. However, it is the obstacle for the compatibility of coproduct identities with the projection to divergent graphs for non-renormalizable QFTs, as will be shown in \propref{prop:proj_div_graphs_coprod}. Furthermore, we show in \colref{col:qgr-sm_is_cograph-divergent} that (effective) Quantum General Relativity coupled to the Standard Model is cograph-divergent.
\end{rem}

\enter

\begin{prop} \label{prop:proj_div_graphs_coprod}
	Let \(\Q\) be a QFT with (associated) renormalization Hopf algebra \(\HQ\). Then coproduct identities are compatible with the projection to divergent graphs from \defnref{defn:projection_divergent_graphs} if and only if \(\Q\) is cograph-divergent: More precisely, given the identity (we assume \(\overline{\mathfrak{G}}_\mathbf{V} \neq 0\))\footnote{We remark that \(\mathfrak{h}_{\mathbf{V} - \mathbf{v}} \equiv \overline{\mathfrak{h}_{\mathbf{V} - \mathbf{v}}}\) by the definition of the coproduct in (associated) renormalization Hopf algebras, cf.\ \defnref{defn:renormalization_hopf_algebra}.}
	\begin{align}
		\D{\mathfrak{G}_\mathbf{V}} & = \sum_\mathbf{v} \mathfrak{h}_{\mathbf{V} - \mathbf{v}} \otimes \mathfrak{H}_\mathbf{v}
		\intertext{then this implies}
		\Delta \big ( \overline{\mathfrak{G}}_\mathbf{V} \big ) & = \sum_\mathbf{v} \mathfrak{h}_{\mathbf{V} - \mathbf{v}} \otimes \overline{\mathfrak{H}}_\mathbf{v}\label{eqn:div-coprod}
	\end{align}
	if and only if \(\Q\) is cograph-divergent.
\end{prop}

\begin{proof}
	Since the coproduct is additive, we can without loss of generality assume that \(\mathfrak{G}_\mathbf{V}\) consists of only one summand. Thus,
	\begin{equation}
		\mathfrak{G}_\mathbf{V} \equiv \alpha \prod_i \Gamma_i
	\end{equation}
	with \(\alpha \in \mathbb{Q}\) and \(\Gamma_i \in \GQ\) are 1PI Feynman graphs whose vertex-grading adds up to \(\mathbf{V}\). Furthermore, we have either \(\overline{\mathfrak{G}}_\mathbf{V} = \mathfrak{G}_\mathbf{V}\) or \(\overline{\mathfrak{G}}_\mathbf{V} = 0\) by definition, cf.\ \defnref{defn:projection_divergent_graphs}, where the latter case is excluded by assumption. Using the linearity and multiplicativity of the coproduct,
	\begin{equation}
		\D{\alpha \prod_i \Gamma_i} = \alpha \prod_i \D{\Gamma_i} \, ,
	\end{equation}
	we can reduce our calculation to the 1PI Feynman graphs \(\Gamma_i\) via
	\begin{equation}
		\D{\Gamma_i} = \sum_{\gamma_i \in \DQ{\Gamma_i}} \gamma_i \otimes \Gamma_i / \gamma_i \, .
	\end{equation}
	If \(\Q\) is cograph-divergent, we obtain \(\overline{\Gamma_i / \gamma_i} \equiv \Gamma_i / \gamma_i\), which concludes the proof.
\end{proof}

\enter

\begin{lem} \label{lem:cograph-divergence_criterion}
	Let \(\Q\) be a QFT consisting only of non-renormalizable and renormalizable corollas. Then \(\Q\) is cograph-divergent if \(\sdd{v} \geq 0\) for all vertex-residues \(v \in \RQO\).
\end{lem}

\begin{proof}
	This is a direct consequence of \eqnref{eqn:asdd} from \thmref{thm:asdd} together with \eqnref{eqn:ssdd_decomposition} from \colref{col:ssdd_decomposition}: We obtain \(\sdd{\Gamma} \geq 0\) for all Feynman graphs \(\Gamma \in \GQ\) with \(\res{\Gamma} \in \RQ\), if \(\rsdd{r} \geq 0\) for all residues \(r \in \RQ\). This follows from the fact that the contribution from \(\sigma\) is non-negative by the assumption that \(\Q\) consists only of non-renormalizable and renormalizable corollas. Furthermore, as the weight of an edge is negative in our conventions, \(\rsdd{e} \geq 0\) is trivially satisfied for all edge-residues \(e \in \RQI\). Thus, the property \(\rsdd{r} \geq 0\) reduces to \(\rsdd{v} \equiv \sdd{v} \geq 0\) for all vertex residues \(v \in \RQO\). As the residue of a Feynman graph \(\Gamma\) agrees with any of its cographs, i.e.\ \(\res{\Gamma} = \res{\Gamma / \gamma}\) for any subgraph \(\gamma \subset \Gamma\), we obtain the claimed statement.
\end{proof}

\enter

\begin{col} \label{col:qgr-sm_is_cograph-divergent}
	(Effective) Quantum General Relativity coupled to the Standard Model in 4 dimensions of spacetime is cograph-divergent.
\end{col}

\begin{proof}
	(Effective) Quantum General Relativity (QGR) coupled to the Standard Model (SM) consists of a non-renormalizable (QGR) and a renormalizable (SM) sub-QFT. Furthermore, its vertices \(v \in \mathcal{R}^{[0]}_\text{QGR-SM}\) are either independent, linear dependent or quadratic dependent on momenta, and thus satisfy \(\sdd{v} \geq 0\). Hence \lemref{lem:cograph-divergence_criterion} applies, which concludes the proof.
\end{proof}

\enter

\begin{rem} \label{rem:proj_div_graphs_non-cogr-div-qft}
	A direct consequence of \thmref{thm:asdd} is that, for a fixed residue \(r \in \RQ\), the zero locus of the superficial degree of divergence \(\omega\) is a hyperplane in \(\ZvQ\). More precisely, we define the function
	\begin{equation}
		\omega^r \, : \quad \ZvQ \to \mathbb{Z} \, , \quad \mathbf{v} \mapsto \rho \left ( r \right ) + \sigma \left ( \mathbf{v} \right )
	\end{equation}
	and set \(\mathscr{H}^r_0 := \left ( \omega^r \right )^{-1} \left ( 0 \right ) \subset \ZvQ\). Accordingly, we decompose
	\begin{equation}
		\ZvQ \cong \mathscr{H}^r_+ \sqcup \mathscr{H}^r_0 \sqcup \mathscr{H}^r_-
	\end{equation}
	as sets, where \(\mathscr{H}^r_\pm\) are defined such that \(\eval[1]{\omega^r}_{\mathscr{H}^r_+} > 0\) and \(\eval[1]{\omega^r}_{\mathscr{H}^r_-} < 0\). We apply this knowledge to \propref{prop:proj_div_graphs_coprod}: Suppose \(\Q\) is a non-cograph-divergent QFT and consider for simplicity that \(\mathfrak{G}_\mathbf{V}\) is given as the sum over 1PI Feynman graphs with a fixed residue \(r\). Then, the coproduct identity
	\begin{align}
		\D{\mathfrak{G}_\mathbf{V}} & = \sum_{\mathbf{v} \in \ZvQ} \mathfrak{h}_{\mathbf{V} - \mathbf{v}} \otimes \mathfrak{H}_\mathbf{v}
		\intertext{implies only the following splitted identity}
		\Delta \big ( \overline{\mathfrak{G}}_\mathbf{V} \big ) & = \sum_{\mathbf{v} \in \left ( \mathscr{H}^r_+ \sqcup \mathscr{H}^r_0 \right )} \mathfrak{h}_{\mathbf{V} - \mathbf{v}} \otimes \overline{\mathfrak{H}}_\mathbf{v} + \sum_{\mathbf{v}^\prime \in \mathscr{H}^r_-} \mathfrak{h}_{\mathbf{V} - \mathbf{v}^\prime} \otimes \mathfrak{H}_{\mathbf{v}^\prime} \, ,
	\end{align}
	where the elements \(\mathfrak{H}_{\mathbf{v}^\prime}\) consist precisely of the convergent cographs.
\end{rem}

\enter

\begin{rem}
	Given the situation of \propref{prop:proj_div_graphs_coprod} and \remref{rem:proj_div_graphs_non-cogr-div-qft}, analogous statements also hold in the case of the antipode due to \lemref{lem:coproduct_and_antipode_identities}.
\end{rem}

\enter

\begin{prop}[Superficial grade compatibility] \label{prop:superficial_grade_compatibility}
	Given the situation of \defnref{defn:superficially_compatible_grading}, the coupling-grading is superficially compatible if and only if
	\begin{equation} \label{eqn:coupling-grading_superficial-compatibility}
		m \csdd{v} = n \csdd{w}
	\end{equation}
	holds for all \(v, w \in \RQO\) and \(m, n \in \mathbb{N}_+\) with \(\theta \left ( v \right )^m = \theta \left ( w \right )^n\). Furthermore, the loop-grading is superficially compatible if and only if
	\begin{equation} \label{eqn:loop-grading_superficial-compatibility}
		\frac{1}{\left ( \val{v} - 2 \right )} \csdd{v} = \frac{1}{\left ( \val{w} - 2 \right )} \csdd{w}
	\end{equation}
	holds for all \(v, w \in \RQO\), where \(\val{v}\) denotes the valence of \(v\).
\end{prop}

\begin{proof}
	We start with the superficial compatibility of the coupling-grading: Two vertex residues \(v, w \in \RQO\) contribute to the same coupling-grading, if and only if there exist natural numbers \(m, n \in \mathbb{N}_+\) such that \(\theta \left ( v \right )^m = \theta \left ( w \right )^n\). Thus, \eqnref{eqn:coupling-grading_superficial-compatibility} is the condition for the superficial compatibility of the coupling-grading and we proceed to the superficial-compatibility of the loop-grading: Loop-grading is only sensible if just one coupling constant is present. In this case, different vertex residues might be scaled via different powers of the same coupling constant. Applying the previous result, loop-grading is superficially compatible if these powers depend only on the valences of the vertices, i.e.\ two vertex residues \(v, w \in \RQO\) are considered equivalent, if and only if \(\theta \left ( v \right )^{\val{w} - 2} = \theta \left ( w \right )^{\val{v} - 2}\). Thus, \eqnref{eqn:loop-grading_superficial-compatibility} is the condition for the superficial compatibility of the loop-grading, which concludes the proof.
\end{proof}

\enter

\begin{rem}
	Having a superficially compatible coupling-grading is a necessary condition for the validity of quantum gauge symmetries on the level of Feynman rules --- however it is not sufficient. This will be discussed in \sectionref{sec:quantum_gauge_symmetries_and_renormalized_feynman_rules}.
\end{rem}

\enter

\begin{col} \label{col:qgr-sm_is_sqgsc}
	(Effective) Quantum General Relativity coupled to the Standard Model has a superficially compatible coupling-grading.
\end{col}

\begin{proof}
	We start by considering the pure gravitational part, i.e.\ gravitons and graviton-ghosts: The Feynman rules of (effective) Quantum General Relativity (QGR) are such that each vertex \(v \in \mathcal{R}_\text{QGR}^{[0]}\) has weight \(\sdd{v} = 2\) and each edge \(e \in \mathcal{R}_\text{QGR}^{[1]}\) has weight \(\sdd{e} = -2\), as the corresponding Feynman rules are quadratic and inverse quadratic in momenta, respectively. Thus, the corolla-weight of a vertex \(v \in \mathcal{R}_\text{QGR}^{[0]}\) is
	\begin{equation}
		\csdd{v} = 2 - \val{v} \, .
	\end{equation}
	Applying \eqnref{eqn:loop-grading_superficial-compatibility} from \propref{prop:superficial_grade_compatibility} yields
	\begin{equation}
	\begin{split}
		\frac{1}{\left ( \val{v} - 2 \right )} \csdd{v} & = \frac{2 - \val{v}}{\val{v} - 2}\\
		& \equiv -1 \, ,
	\end{split}
	\end{equation}
	which shows that QGR has even a superficially compatible loop-grading. Furthermore, the pure Standard Model (SM) part is renormalizable. Thus, applying \colref{col:weights_of_corollas_and_renormalizability}, the corolla-weight of a vertex \(v \in \mathcal{R}_\text{SM}^{[0]}\) is
	\begin{equation}
	\begin{split}
		\frac{1}{\left ( \val{v} - 2 \right )} \csdd{v} & = \frac{d \left ( 1 - \textfrac{1}{2} \val{v} \right )}{\left ( \val{v} - 2 \right )}\\
		& \equiv - \frac{d}{2} \, ,
	\end{split}
	\end{equation}
	which shows that also the SM has even a superficially compatible loop-grading. Finally, we consider the mixed part, i.e.\ SM residues with a positive number of gravitons attached to it.\footnote{We remark that this includes also edges from the Standard Model, as they become vertices when attaching graviton half-edges to them.} It follows from the corresponding Feynman rules that the weights of these corollas depend only on the SM residue and are thus independent of the number of gravitons attached to them. Conversely, increasing the number of gravitons of such a vertex by gluing a three-valent graviton tree (consisting of a vertex and a propagator) also leaves the weight of the corolla unchanged (as the net difference is \(2 - 2 = 0\)), which finishes the proof.
\end{proof}

\section{Coproduct and antipode identities} \label{sec:coproduct_and_antipode_identities}

In this section, we state and generalize coproduct identities, known in the literature for the case of renormalizable QFTs with a single coupling constant \cite{vSuijlekom_QCD,vSuijlekom_BV,Yeats_PhD,Borinsky_Feyngen}. We reprove these identities and generalize them to cover the more involved cases of super- and non-renormalizable QFTs and QFTs with several vertex residues. Since coproduct identities imply recursive antipode identities by \lemref{lem:coproduct_and_antipode_identities}, we will in the following only discuss the coproduct cases.

\enter

\begin{lem}[Coproduct and antipode identities] \label{lem:coproduct_and_antipode_identities}
	Let \(\Q\) be a QFT with (associated) renormalization Hopf algebra \(\HQ\) and grading (multi-)indices \(\mathbf{G}\) and \(\mathbf{g}\). Then coproduct identities are equivalent to recursive antipode identities as follows: Given the identity
	\begin{align}
		\D{\mathfrak{G}} & = \sum_\mathbf{g} \mathfrak{h}_\mathbf{g} \otimes \mathfrak{H}_\mathbf{g} \, ,
		\intertext{then this is equivalent to the recursive identity}
		\antipode{\mathfrak{G}} & = - \sum_{\mathbf{g}} \antipode{\mathfrak{h}_\mathbf{g}} \mathscr{A} \left ( \mathfrak{H}_\mathbf{g} \right ) \, ,
	\end{align}
	where \(\mathscr{A}\) is the projector onto the augmentation ideal, cf.\ \defnref{defn:augmentation_ideal}.
\end{lem}

\begin{proof}
	This follows immediately from the recursive definition of the antipode on generators \(\Gamma \in \HQ\) via
	\begin{equation}
	\begin{split}
		\antipode{\Gamma} & := - \sum_{\gamma \in \DQ{\Gamma}} \antipode{\gamma} \mathscr{A} \left ( \Gamma / \gamma \right ) \\
		& \phantom{:} \equiv - \left ( S \star \mathscr{A} \right ) \left ( \Gamma \right ) \, ,
	\end{split}
	\end{equation}
	which is then linearly and multiplicatively extended to all of \(\HQ\), such that
	\begin{equation}
		S \star \operatorname{Id}_{\HQ} \equiv \one \circ \coone \equiv \operatorname{Id}_{\HQ} \star S
	\end{equation}
	holds, where \(\operatorname{Id}_{\HQ} \in \operatorname{End} \left ( \HQ \right )\) is the identity-endomorphism of \(\HQ\).\footnote{In particular, \(S\) is the \(\star\)-inverse to the identity \(\operatorname{Id}_{\HQ}\) and \(\one \circ \coone\) the \(\star\)-identity.}
\end{proof}

\enter

\begin{prop}[Coproduct identities for (divergent/restricted) combinatorial Green's functions\footnote{This is a direct generalization of \cite[Lemma 4.6]{Yeats_PhD}, \cite[Proposition 16]{vSuijlekom_QCD} and \cite[Theorem 1]{Borinsky_Feyngen} to super- and non-renormalizable theories, theories with several vertex residues and such with longitudinal and transversal degrees of freedom.}] \label{prop:coproduct_greensfunctions}
	Let \(\Q\) be a QFT, \(\HQ\) its (associated) renormalization Hopf algebra, \(r \in \RQ\) a residue and \(\mathbf{V} \in \ZvQ\) a vertex-grading multi-index. Then we have the following coproduct identities for combinatorial Green's functions:
	{
	\allowdisplaybreaks
	\begin{subequations}
	\begin{align}
		\Delta \left ( \combgreen^r \right ) & = \sum_{\mathbf{v} \in \ZvQ} \overline{\combgreen}^r \overline{\combcharge}^\mathbf{v} \otimes \rescombgreen^r_{\mathbf{v}} \, , \label{eqn:coproduct_greensfunctions} \\
		\Delta \left ( \rescombgreen^r_{\mathbf{V}} \right ) & = \sum_{\mathbf{v} \in \ZvQ} \eval{\left ( \overline{\combgreen}^r \overline{\combcharge}^\mathbf{v} \right )}_{\mathbf{V} - \mathbf{v}} \otimes \rescombgreen^r_{\mathbf{v}} \, , \label{eqn:coproduct_greensfunctions_restricted}
		\intertext{and, provided that \(\Q\) is cograph-divergent and \(\overline{\rescombgreen}^r_{\mathbf{V}} \neq 0\),}
		\Delta \big ( \overline{\rescombgreen}^r_{\mathbf{V}} \big ) & = \sum_{\mathbf{v} \in \ZvQ} \eval{\left ( \overline{\combgreen}^r \overline{\combcharge}^\mathbf{v} \right )}_{\mathbf{V} - \mathbf{v}} \otimes \overline{\rescombgreen}^r_{\mathbf{v}} \, . \label{eqn:coproduct_greensfunctions_restricted_divergent}
	\end{align}
	\end{subequations}
	}%
\end{prop}

\begin{proof}
	We start with \eqnref{eqn:coproduct_greensfunctions} by using the linearity of the coproduct, \propref{prop:isomorphic_insertable_graph_sets} and \lemref{lem:sym-factors_and_ins-factors}:
	\begin{subequations}
	\begin{equation}
	\begin{split}
		\D{\rescombgreen^r} & = \sum_{\substack{\Gamma \in \GQ\\\res{\Gamma} = r}} \frac{1}{\sym{\Gamma}} \D{\Gamma}\\
		& = \sum_{\substack{\Gamma \in \GQ\\\res{\Gamma} = r}} \sum_{\gamma \in \DQ{\Gamma}} \frac{1}{\sym{\Gamma}} \gamma \otimes \Gamma / \gamma\\
		& = \sum_{\substack{\Gamma \in \GQ\\\res{\Gamma} = r}} \sum_{\gamma \in \DQ{\Gamma}} \frac{\insaut{\gamma}{\Gamma / \gamma}{\Gamma}}{\isoemb{\gamma}{\Gamma}} \left ( \frac{1}{\sym{\gamma}} \gamma \right ) \otimes \left ( \frac{1}{\sym{\Gamma / \gamma}} \Gamma / \gamma \right )\\
		& = \sum_{\substack{\Gamma^\prime \in \GQ\\\resnolim{\Gamma^\prime} = r}} \left ( \sum_{\substack{\gamma^\prime \in \IQnolim{\Gamma^\prime}}} \frac{\ins{\gamma^\prime}{\Gamma^\prime}}{\sym{\gamma^\prime}} \gamma^\prime \right ) \otimes \left ( \frac{1}{\sym{\Gamma^\prime}} \Gamma^\prime \right ) + \overline{\combgreen}^r \otimes \one\\
		& = \sum_{\mathbf{v} \in \ZvQ} \left ( \sum_{\substack{\gamma^\prime \in \IQrv}} \frac{\insrr{\gamma^\prime}}{\sym{\gamma^\prime}} \gamma^\prime \right ) \otimes \left ( \sum_{\substack{\Gamma^\prime \in \GQ\\\resnolim{\Gamma^\prime} = r\\\vtxgrdnolim{\Gamma^\prime} = \mathbf{v}}} \frac{1}{\sym{\Gamma^\prime}} \Gamma^\prime \right ) + \overline{\combgreen}^r \otimes \one\\
		& = \sum_{\mathbf{v} \in \ZvQ} \overline{\combgreen}^r \overline{\combcharge}^{\mathbf{v}} \otimes \rescombgreen^r_{\mathbf{v}}
	\end{split}
	\end{equation}
	From this we proceed to \eqnref{eqn:coproduct_greensfunctions_restricted} via restriction:
	\begin{equation}
	\begin{split}
		\D{\combgreen^r_\mathbf{V}} & = \eval{\left ( \D{\combgreen^r} \right )}_\mathbf{V}\\
		& = \eval{\left ( \sum_{\mathbf{v} \in \ZvQ} \overline{\combgreen}^r \overline{\combcharge}^{\mathbf{v}} \otimes \rescombgreen^r_{\mathbf{v}} \right )}_\mathbf{V}\\
		& = \sum_{\mathbf{v} \in \ZvQ} \eval{\left ( \overline{\combgreen}^r \overline{\combcharge}^\mathbf{v} \right )}_{\mathbf{V} - \mathbf{v}} \otimes \rescombgreen^r_{\mathbf{v}} \, .
	\end{split}
	\end{equation}
	\end{subequations}
	Finally, \eqnref{eqn:coproduct_greensfunctions_restricted_divergent} follows from \eqnref{eqn:coproduct_greensfunctions_restricted} together with the assumption of \(\Q\) being cograph-divergent and the application of \propref{prop:proj_div_graphs_coprod}.
\end{proof}

\enter

\begin{prop}[Coproduct identities for (divergent/restricted) combinatorial charges] \label{prop:coproduct_charges}
	Let \(\Q\) be a QFT, \(\HQ\) its (associated) renormalization Hopf algebra, \(v \in \RQO\) a vertex residue and \(\mathbf{V} \in \ZvQ\) a vertex-grading multi-index. Then we have the following coproduct identities for combinatorial charges:
	{
	\allowdisplaybreaks
	\begin{subequations}
	\begin{align}
		\Delta \left ( \combcharge^v \right ) & = \sum_{\mathbf{v} \in \ZvQ} \overline{\combcharge}^v \overline{\combcharge}^\mathbf{v} \otimes \combcharge^v_{\mathbf{v}} \, , \label{eqn:coproduct_charges} \\
		\Delta \left ( \combcharge^v_{\mathbf{V}} \right ) & = \sum_{\mathbf{v} \in \ZvQ} \eval{\left ( \overline{\combcharge}^v \overline{\combcharge}^\mathbf{v} \right )}_{\mathbf{V} - \mathbf{v}} \otimes \combcharge^v_{\mathbf{v}} \, , \label{eqn:coproduct_charges_restricted}
		\intertext{and, provided that \(\Q\) is cograph-divergent and \(\overline{\combcharge}^v_{\mathbf{V}} \neq 0\),}
		\Delta \big ( \overline{\combcharge}^v_{\mathbf{V}} \big ) & = \sum_{\mathbf{v} \in \ZvQ} \eval{\left ( \overline{\combcharge}^v \overline{\combcharge}^\mathbf{v} \right )}_{\mathbf{V} - \mathbf{v}} \otimes \overline{\combcharge}^v_{\mathbf{v}} \, . \label{eqn:coproduct_charges_restricted_divergent}
	\end{align}
	\end{subequations}
	}%
\end{prop}

\begin{proof}
	We start with \eqnref{eqn:coproduct_charges} by using the linearity and multiplicativity of the coproduct and \propref{prop:coproduct_greensfunctions}:
	\begin{subequations}
	\begin{equation}
	\begin{split}
		\D{\combcharge^v} & = \frac{\D{\combgreen^v}}{\prod_{e \in E \left ( v \right )} \sqrt{\D{\combgreen^e}}}\\
		& = \frac{\sum_{{\mathbf{v}_v} \in \ZvQ} \overline{\combgreen}^v \overline{\combcharge}^{\mathbf{v}_v} \otimes \rescombgreen^v_{{\mathbf{v}_v}}}{\prod_{e \in E \left ( v \right )} \sqrt{\sum_{{\mathbf{v}_e} \in \ZvQ} \overline{\combgreen}^e \overline{\combcharge}^{\mathbf{v}_e} \otimes \rescombgreen^e_{{\mathbf{v}_e}}}}\\
		& = \frac{\left ( \overline{\combgreen}^v \otimes \one \right ) \left ( \sum_{{\mathbf{v}_v} \in \ZvQ} \overline{\combcharge}^{\mathbf{v}_v} \otimes \rescombgreen^v_{{\mathbf{v}_v}} \right )}{\prod_{e \in E \left ( v \right )} \sqrt{\left ( \overline{\combgreen}^e \otimes \one \right ) \left ( \sum_{{\mathbf{v}_e} \in \ZvQ} \overline{\combcharge}^{\mathbf{v}_e} \otimes \rescombgreen^e_{{\mathbf{v}_e}} \right )}}\\
		& = \left ( \overline{\combcharge}^v \otimes \one \right ) \left ( \frac{\left ( \sum_{{\mathbf{v}_v} \in \ZvQ} \overline{\combcharge}^{\mathbf{v}_v} \otimes \rescombgreen^v_{{\mathbf{v}_v}} \right )}{\prod_{e \in E \left ( v \right )} \sqrt{\left ( \sum_{{\mathbf{v}_e} \in \ZvQ} \overline{\combcharge}^{\mathbf{v}_e} \otimes \rescombgreen^e_{{\mathbf{v}_e}} \right )}} \right )\\
		& = \left ( \overline{\combcharge}^v \otimes \one \right ) \left ( \sum_{\mathbf{v} \in \ZvQ} \overline{\combcharge}^\mathbf{v} \otimes \combcharge^v_{\mathbf{v}} \right )\\
		& = \sum_{\mathbf{v} \in \ZvQ} \overline{\combcharge}^v \overline{\combcharge}^\mathbf{v} \otimes \combcharge^v_{\mathbf{v}}
	\end{split}
	\end{equation}
	From this we proceed to \eqnref{eqn:coproduct_charges_restricted} via restriction:
	\begin{equation}
	\begin{split}
		\D{\combcharge^v_\mathbf{V}} & = \eval{\left ( \D{\combcharge^v} \right )}_\mathbf{V}\\
		& = \eval{\left ( \sum_{\mathbf{v} \in \ZvQ} \overline{\combcharge}^v \overline{\combcharge}^{\mathbf{v}} \otimes \combcharge^v_{\mathbf{v}} \right )}_\mathbf{V}\\
		& = \sum_{\mathbf{v} \in \ZvQ} \eval{\left ( \overline{\combcharge}^v \overline{\combcharge}^\mathbf{v} \right )}_{\mathbf{V} - \mathbf{v}} \otimes \combcharge^v_{\mathbf{v}} \, ,
	\end{split}
	\end{equation}
	\end{subequations}
	Finally, \eqnref{eqn:coproduct_charges_restricted_divergent} follows from \eqnref{eqn:coproduct_charges_restricted} together with the assumption of \(\Q\) being cograph-divergent and the application of \propref{prop:proj_div_graphs_coprod}.
\end{proof}

\enter

\begin{prop}[Coproduct identities for exponentiated (divergent/restricted) combinatorial charges] \label{prop:coproduct_exponentiated_combinatorial_charges}
	Let \(\Q\) be a QFT, \(\HQ\) its (associated) renormalization Hopf algebra, \(v \in \RQO\) a vertex residue and \(\mathbf{V} \in \ZvQ\) a vertex-grading multi-index. Then we have the following coproduct identities for powers of combinatorial charges by \(m \in \mathbb{Q}\):\footnote{The power of an element in the renormalization Hopf algebra \(\mathfrak{G} \in \HQ\) via a non-natural number \(m \in \left ( \mathbb{Q} \setminus \mathbb{N}_0 \right )\) is understood via the formal binomial series, i.e.\
	\begin{equation}
		\mathfrak{G}^m := \sum_{n = 0}^\infty \binom{m}{n} \left ( \mathfrak{G} - \one \right )^n \, .
	\end{equation}
	More generally, if the renormalization Hopf algebra is considered over the field \(\field\), then the following statements are true for \(m \in \field\).}
	{
	\allowdisplaybreaks
	\begin{subequations}
	\begin{align}
		\Delta \left ( \combcharge^{mv} \right ) & = \sum_{\mathbf{v} \in \ZvQ} \overline{\combcharge}^{mv} \overline{\combcharge}^\mathbf{v} \otimes \combcharge^{mv}_{\mathbf{v}} \, , \label{eqn:coproduct_powers_charges} \\
		\Delta \left ( \combcharge^{mv}_{\mathbf{V}} \right ) & = \sum_{\mathbf{v} \in \ZvQ} \eval{\left ( \overline{\combcharge}^{mv} \overline{\combcharge}^\mathbf{v} \right )}_{\mathbf{V} - \mathbf{v}} \otimes \combcharge^{mv}_{\mathbf{v}} \, , \label{eqn:coproduct_powers_charges_restricted}
		\intertext{and, provided that \(\Q\) is cograph-divergent and \(\overline{\combcharge}^{mv}_{\mathbf{V}} \neq 0\),}
		\Delta \big ( \overline{\combcharge}^{mv}_{\mathbf{V}} \big ) & = \sum_{\mathbf{v} \in \ZvQ} \eval{\left ( \overline{\combcharge}^{mv} \overline{\combcharge}^\mathbf{v} \right )}_{\mathbf{V} - \mathbf{v}} \otimes \overline{\combcharge}^{mv}_{\mathbf{v}} \, . \label{eqn:coproduct_powers_charges_restricted_divergent}
	\end{align}
	\end{subequations}
	}%
\end{prop}

\begin{proof}
	We start with \eqnref{eqn:coproduct_powers_charges} by using the linearity and multiplicativity of the coproduct and \propref{prop:coproduct_charges}:
	\begin{subequations}
	\begin{equation}
	\begin{split}
		\Delta \left ( \combcharge^{mv} \right ) & = \left ( \D{\combcharge^v} \right )^m\\
		& = \left ( \sum_{\mathbf{v} \in \ZvQ} \overline{\combcharge}^v \overline{\combcharge}^\mathbf{v} \otimes \combcharge^v_{\mathbf{v}} \right )^m\\
		& = \left ( \left ( \overline{\combcharge}^v \otimes \one \right ) \left ( \sum_{\mathbf{v} \in \ZvQ} \overline{\combcharge}^\mathbf{v} \otimes \combcharge^v_{\mathbf{v}} \right ) \right )^m\\
		& = \left ( \overline{\combcharge}^{mv} \otimes \one \right ) \left ( \sum_{\mathbf{v} \in \ZvQ} \overline{\combcharge}^\mathbf{v} \otimes \combcharge^v_{\mathbf{v}} \right )^m\\
		& = \sum_{\mathbf{v} \in \ZvQ} \overline{\combcharge}^{mv} \overline{\combcharge}^\mathbf{v} \otimes \combcharge^{mv}_{\mathbf{v}}
	\end{split}
	\end{equation}
	From this we proceed to \eqnref{eqn:coproduct_powers_charges_restricted} via restriction:
	\begin{equation}
	\begin{split}
		\Delta \left ( \combcharge^{mv}_{\mathbf{V}} \right ) & = \eval{\left ( \D{\combcharge^{mv}} \right )}_\mathbf{V}\\
		& = \eval{\left ( \sum_{\mathbf{v} \in \ZvQ} \overline{\combcharge}^{mv} \overline{\combcharge}^\mathbf{v} \otimes \combcharge^{mv}_{\mathbf{v}} \right )}_\mathbf{V}\\
		& \phantom{ = } \sum_{\mathbf{v} \in \ZvQ} \eval{\left ( \overline{\combcharge}^{mv} \overline{\combcharge}^\mathbf{v} \right )}_{\mathbf{V} - \mathbf{v}} \otimes \combcharge^{mv}_{\mathbf{v}}
	\end{split}
	\end{equation}
	\end{subequations}
	Finally, \eqnref{eqn:coproduct_powers_charges_restricted_divergent} follows from \eqnref{eqn:coproduct_powers_charges_restricted} together with the assumption of \(\Q\) being cograph-divergent and the application of \propref{prop:proj_div_graphs_coprod}.
\end{proof}

\section{Quantum gauge symmetries and subdivergences} \label{sec:quantum_gauge_symmetries_and_subdivergences}

In this section, we give a precise definition of `quantum gauge symmetries (QGS)' in \defnref{defn:quantum_gauge_symmetries} and prove in \thmref{thm:quantum_gauge_symmetries_induce_hopf_ideals} that they induce Hopf ideals in the (associated) renormalization Hopf algebra, even for super- and non-renormalizable QFTs, QFTs with several coupling constants and QFTs with a transversal structure. This means that \(Z\)-factor identities coming from (generalized) gauge symmetries, such as \eqnsaref{eqns:z-factor_identities_qym}{eqns:z-factor_identities_qgr}, are compatible with the renormalization of subdivergences. Furthermore, we illustrate our framework with Quantum Yang--Mills theory in \exref{exmp:qym} and (effective) Quantum General Relativity in \exref{exmp:qgr}. Additionally, we mention the works \cite{Prinz_2,Borinsky_PhD}, which contain general results on Hopf ideals in the context of renormalization Hopf algebras.

\enter

\begin{defn}[Quantum gauge symmetries] \label{defn:quantum_gauge_symmetries}
	Let \(\Q\) be a QFT whose coupling-grading is superficially compatible, \(\QQ\) the set of its combinatorial charges with cardinality \(\mathfrak{v}_\Q\) and \(\qQ\) the set of its physical charges with cardinality \(\mathfrak{q}_\Q\), cf.\ \defnref{defn:sets_of_coupling_constants}. Suppose that \(\mathfrak{v}_\Q > \mathfrak{q}_\Q\), then we define the following set of equivalence relations, to which we refer as `quantum gauge symmetries (QGS)', via
	\begin{equation} \label{eqn:equivalence_relation_combinatorial_charges}
		\left ( \overline{\combcharge}^v_\mathbf{C} \right )^m \sim \left ( \overline{\combcharge}^w_\mathbf{C} \right )^n \quad : \! \! \iff \quad \cpl{\combcharge^v}^m \equiv \cpl{\combcharge^w}^n
	\end{equation}
	for all \(v, w \in \RQ^{[0]}\), \(m, n \in \mathbb{N}_+\) and \(\mathbf{C} \in \mathbb{Z}^{\mathfrak{q}_\Q}\). Explicitly, the product of combinatorial charges is defined as the sum over all possibilities to connect their external edges to trees with the respective combinatorial charges as vertices. If \(\Q\) possesses additionally a transversal structure, cf.\ \defnref{defn:transversal_structure}, then the external edges are additionally required to respect the `physical' and `unphysical' labels. We remark that this is automatic for internal edges due to the restriction to particular coupling-gradings, as `unphysical' edges are indexed by the respective gauge fixing parameters. In particular, this requires connecting gauge field edges to be:
	\begin{itemize}
		\item Transversal, if they are related to higher valent gauge field vertices
		\item Longitudinal, if they are related to ghost edges
	\end{itemize}
	Finally, the set of all quantum gauge symmetries of \(\Q\) is denoted via \(\operatorname{QGS}_\Q\) and elements therein are given as quadruples of the form \(\set{v, m; w, n} \in \operatorname{QGS}_\Q\).
\end{defn}

\enter

\begin{exmp}[Quantum Yang--Mills theory] \label{exmp:qym}
	We continue the example started around \eqnref{eqn:qym-lagrange-density-introduction}: Consider Quantum Yang--Mills theory with a Lorenz gauge fixing. Then, the Lagrange density is given via
	\begin{equation}
	\begin{split}
		\mathcal{L}_\text{QYM} & := \mathcal{L}_\text{YM} + \mathcal{L}_\text{GF} + \mathcal{L}_\text{Ghost} \\ & \phantom{:} = \eta^{\mu \nu} \eta^{\rho \sigma} \delta_{a b} \left ( - \frac{1}{4 \mathrm{g}^2} F^a_{\mu \rho} F^b_{\nu \sigma} - \frac{1}{2 \xi} \big ( \partial_\mu A^a_\nu \big ) \big ( \partial_\rho A^b_\sigma \big ) \right ) \dif V_\eta \\
		& \phantom{:=} + \eta^{\mu \nu} \left ( \frac{1}{\xi} \overline{c}_a \left ( \partial_\mu \partial_\nu c^a \right ) + \mathrm{g} \tensor{f}{^a _b _c} \overline{c}_a \left ( \partial_\mu \big ( c^b A^c_\nu \big ) \right ) \right ) \dif V_\eta \, ,
	\end{split}
	\end{equation}
	where \(F^a_{\mu \nu} := \mathrm{g} \big ( \partial_\mu A^a_\nu - \partial_\nu A^a_\mu \big ) - \mathrm{g}^2 \tensor{f}{^a _b _c} A^b_\mu A^c_\nu\) is the local curvature form of the gauge boson \(A^a_\mu\). Furthermore, \(\dif V_\eta := \dif t \wedge \dif x \wedge \dif y \wedge \dif z\) denotes the Minkowskian volume form. Additionally, \(\eta^{\mu \nu} \partial_\mu A^a_\nu \equiv 0\) is the Lorenz gauge fixing functional and \(\xi\) denotes the gauge fixing parameter. Finally, \(c^a\) and \(\overline{c}_a\) are the gauge ghost and gauge antighost, respectively. Then, we have the following identities, where \(l\) of the unspecified external gauge boson legs are considered to be longitudinally projected:\footnote{The residues were drawn with JaxoDraw \cite{Binosi_Theussl,Binosi_Collins_Kaufhold_Theussl}.}
	\begin{subequations}
	\begin{equation}
		\cpl{\left ( \combcharge^{\scriptstyle{T} \tcgreen{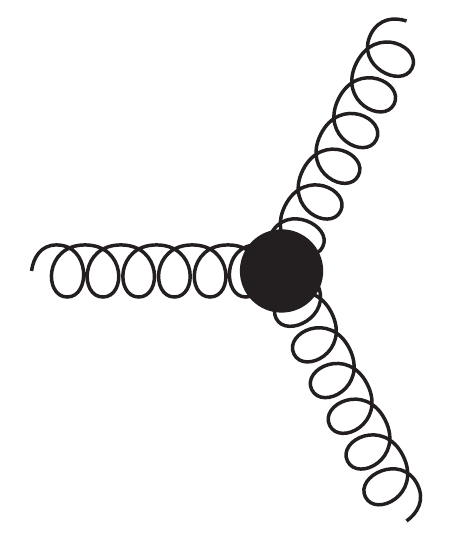}} \right )^{2_{\scriptscriptstyle{T}}}} = \xi^{\textfrac{l}{2}} \mathrm{g}^2 = \cpl{\combcharge^{\cgreen{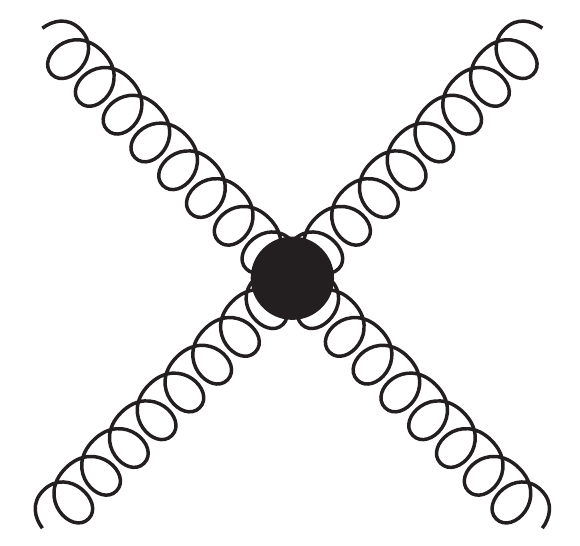}}} \, ,
	\end{equation}
	where \(T\) denotes the projection operator onto the transversal degree of freedom and \(2_{\scriptstyle{T}}\) the squaring by joining the two transversal half-edges, and
	\begin{equation}
		\cpl{\combcharge^{\tcgreen{c-gluontriple}^{\scriptstyle{L}}_{\scriptstyle{L}}}} = \xi^{1 + \textfrac{l}{2}} \mathrm{g} = \cpl{\combcharge^{\tcgreen{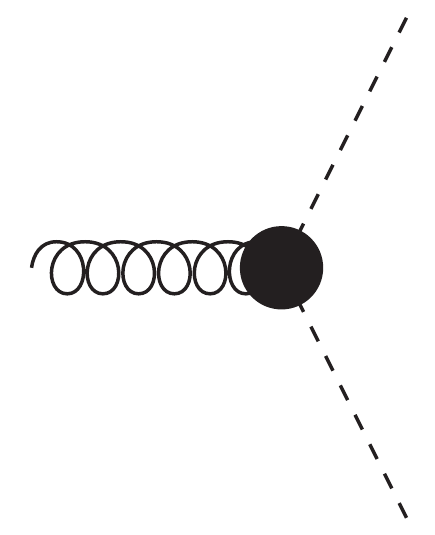}}} \, , \label{eqn:cpl_gluon_gluon-ghost}
	\end{equation}
	\end{subequations}
	where \(L\) denotes the projection operator onto the longitudinal degree of freedom, cf.\ \defnref{defn:transversal_structure}. More precisely, in the case of the Lorenz gauge fixing, we have
	\begin{subequations} \label{eqn:projection_tensors_qym}
	\begin{align}
		L^\nu_\mu & := \frac{1}{p^2} p^\nu p_\mu \, , \\
		I^\nu_\mu & := \delta^\nu_\mu
		\intertext{and}
		T^\nu_\mu & := I^\nu_\mu - L^\nu_\mu \, ,
	\end{align}
	\end{subequations}
	where a short calculation verifies \(L^2 = L\), \(I^2 = I\) and \(T^2 = T\). This implies:\footnote{We only display the generating set.}
	\begin{equation}
	\begin{split}
		\operatorname{QGS}_\text{QYM} & = \set{\set{{\scriptstyle{T} \tcgreen{c-gluontriple}}, 2_{\scriptstyle{T}}; {\cgreen{c-gluonquartic}}, 1} , \set{{\tcgreen{c-gluontriple}^{\scriptstyle{L}}_{\scriptstyle{L}}}, 1; {\tcgreen{c-gluonghosttriple}}, 1}}
	\end{split}
	\end{equation}
\end{exmp}

\enter

\begin{exmp}[(Effective) Quantum General Relativity] \label{exmp:qgr}
	We continue the example started around \eqnref{eqn:qgr-lagrange-density-introduction}: Consider (effective) Quantum General Relativity with the metric decomposition \(g_{\mu \nu} \equiv \eta_{\mu \nu} + \varkappa h_{\mu \nu}\), where \(h_{\mu \nu}\) is the graviton field and \(\varkappa := \sqrt{\kappa}\) the graviton coupling constant (with \(\kappa := 8 \pi G\) the Einstein gravitational constant), and a linearized de Donder gauge fixing. Then, the Lagrange density is given via
	\begin{equation}
	\begin{split}
		\mathcal{L}_\text{QGR} & := \mathcal{L}_\text{GR} + \mathcal{L}_\text{GF} + \mathcal{L}_\text{Ghost} \\ & \phantom{:} = - \frac{1}{2 \varkappa^2} \left ( \sqrt{- \dt{g}} R + \frac{1}{2 \zeta} \eta^{\mu \nu} \deDonder^{(1)}_\mu \deDonder^{(1)}_\nu \right ) \dif V_\eta \\ & \phantom{:=} - \frac{1}{2} \eta^{\rho \sigma} \left ( \frac{1}{\zeta} \overline{C}^\mu \left ( \partial_\rho \partial_\sigma C_\mu \right ) + \overline{C}^\mu \left ( \partial_\mu \big ( \tensor{\Gamma}{^\nu _\rho _\sigma} C_\nu \big ) - 2 \partial_\rho \big ( \tensor{\Gamma}{^\nu _\mu _\sigma} C_\nu \big ) \right ) \right ) \dif V_\eta \, ,
	\end{split}
	\end{equation}
	where \(R := g^{\nu \sigma} \tensor{R}{^\mu _\nu _\mu _\sigma}\) is the Ricci scalar (with \(\tensor{R}{^\rho _\sigma _\mu _\nu} := \partial_\mu \tensor{\Gamma}{^\rho _\nu _\sigma} - \partial_\nu \tensor{\Gamma}{^\rho _\mu _\sigma} + \tensor{\Gamma}{^\rho _\mu _\lambda} \tensor{\Gamma}{^\lambda _\nu _\sigma} - \tensor{\Gamma}{^\rho _\nu _\lambda} \tensor{\Gamma}{^\lambda _\mu _\sigma}\) the Riemann tensor). Again, \(\dif V_\eta := \dif t \wedge \dif x \wedge \dif y \wedge \dif z\) denotes the Minkowskian volume form, which is related to the Riemannian volume form \(\dif V_g\) via \(\dif V_g \equiv \sqrt{- \dt{g}} \dif V_\eta\). Additionally, \(\deDonder^{(1)}_\mu := \eta^{\rho \sigma} \Gamma_{\mu \rho \sigma} \equiv 0\) is the linearized de Donder gauge fixing functional and \(\zeta\) the gauge fixing parameter. Finally, \(C_\mu\) and \(\overline{C}^\mu\) are the graviton-ghost and graviton-antighost, respectively. Again, we refer to \cite{Prinz_2,Prinz_4} for more detailed introductions and further comments on the chosen conventions. Then, we have the following identities, where \(k \in \mathbb{N}_0\) denotes additional graviton legs and \(l\) of the unspecified external graviton legs are considered to be longitudinally projected:\footnote{Again, the residues were drawn with JaxoDraw \cite{Binosi_Theussl,Binosi_Collins_Kaufhold_Theussl}.}
	\begin{subequations}
	\begin{equation}
		\cpl{\left ( \combcharge^{\bbsT \tcgreen{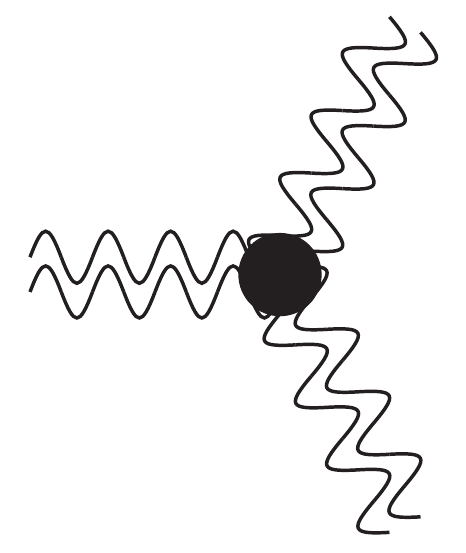}} \right ) \bullet_{\bbsT} \left ( \combcharge^{\bbsT \tcgreen{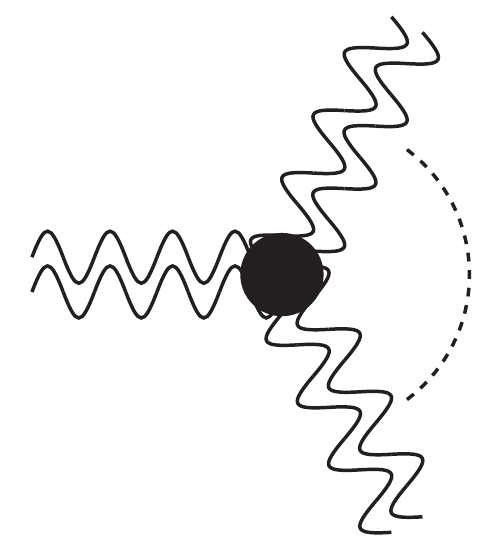} \legnumberexponent{k}} \right )} = \zeta^{\textfrac{l}{2}} \varkappa^{k+2} = \cpl{\combcharge^{\tcgreen{c-gravitonmultiple} \legnumberexponentlong{k+1}}}
	\end{equation}
	where \(\bbT\) denotes the projection operator onto the transversal degree of freedom and \(\bullet_{\bbsT}\) the product by joining the two transversal half-edges together, and
	\begin{equation}
		\cpl{\combcharge^{\tcgreen{c-gravitonmultiple}^{\bbsL}_{\bbsL} \legnumberexponentlongitudinal{k}}} = \zeta^{1 + \textfrac{l}{2}} \varkappa^{k+1} = \cpl{\combcharge^{\tcgreen{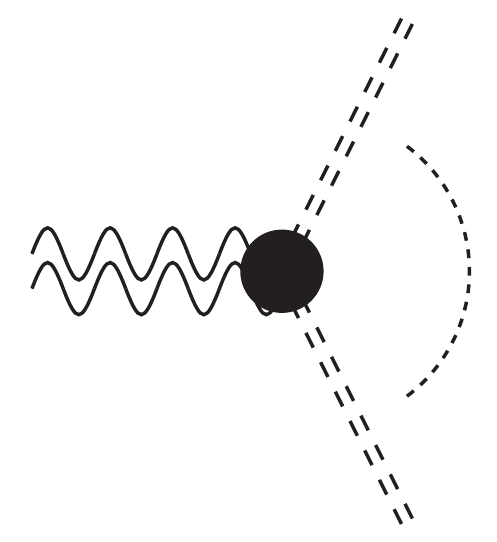} \legnumberexponentghost{k}}} \, , \label{eqn:cpl_graviton_graviton-ghost}
	\end{equation}
	\end{subequations}
	where \(\bbL\) denotes the projection operator onto the longitudinal degree of freedom, cf.\ \defnref{defn:transversal_structure}. More precisely, in the case of the linearized de Donder gauge fixing, we have\footnote{We remark that, unlike in the case of Quantum Yang--Mills theory with a Lorenz gauge fixing in Equations \eqref{eqn:projection_tensors_qym}, the projection operators \(\bbL\) and \(\bbT\) are not symmetric with respect to the index-pairs \(\mu \nu\) and \(\rho \sigma\). The indices \(\mu \nu\) belong to vertex Feynman rules and the indices \(\rho \sigma\) belong to the propagator Feynman rule. This reflects their respective weights as tensor densities, which will be discussed and studied further in \cite{Prinz_9,Prinz_5,Prinz_7}.}
	\begin{subequations} \label{eqn:projection_tensors_qgr}
	\begin{align}
		\bbL^{\rho \sigma}_{\mu \nu} & := \frac{1}{2 p^2} \left ( \delta^\rho_\mu p^\sigma p_\nu + \delta^\sigma_\mu p^\rho p_\nu + \delta^\rho_\nu p^\sigma p_\mu + \delta^\sigma_\nu p^\rho p_\mu - 2 \eta^{\rho \sigma} p_\mu p_\nu \right ) \, , \\
		\bbI^{\rho \sigma}_{\mu \nu} & := \frac{1}{2} \left ( \delta^\rho_\mu \delta^\sigma_\nu + \delta^\sigma_\mu \delta^\rho_\nu \right )
		\intertext{and}
		\bbT^{\rho \sigma}_{\mu \nu} & := \bbI^{\rho \sigma}_{\mu \nu} - \bbL^{\rho \sigma}_{\mu \nu} \, ,
	\end{align}
	\end{subequations}
	where a short calculation verifies \(\bbL^2 = \bbL\), \(\bbI^2 = \bbI\) and \(\bbT^2 = \bbT\). This implies:\footnote{Again, we only display the generating set.}
	\begin{subequations}
	\begin{equation}
	\begin{split}
		\operatorname{QGS}_\text{QGR} & = \set{ \left . \set{{{\raisebox{0.1ex}{$\bbsT$}} \tcgreen{c-gravitonmultiple} \legnumber{i}},  \tilde{\jmath} \left ( i, j \right )_{\bbsT}; {{\raisebox{0.1ex}{$\bbsT$}} \tcgreen{c-gravitonmultiple} \legnumber{j}}, \tilde{\imath} \left ( i, j \right )_{\bbsT}} \, \right \vert \, i, j \in \mathbb{N}_0} \\ & \phantom{= \{} \bigcup \set{\left . \set{{\tcgreen{c-gravitonmultiple}^{\raisebox{0.1ex}{$\bbsL$}}_{\raisebox{0.1ex}{$\bbsL$}} \legnumberlongitudinal{k}}, 1; {\tcgreen{c-gravitonghostmultiple} \legnumberghost{k}}, 1} \, \right \vert \, k \in \mathbb{N}_0} \, ,
	\end{split}
	\end{equation}
	where \(i, j, k \in \mathbb{N}_0\) denote additional graviton legs, and with
	\begin{align}
		\tilde{\imath} \left ( i, j \right ) & := \frac{i+3}{\GCD{i+3}{j+3}}
		\intertext{and}
		\tilde{\jmath} \left ( i, j \right ) & := \frac{j+3}{\GCD{i+3}{j+3}} \, ,
	\end{align}
	\end{subequations}
	where \(\GCD{m}{n}\) denotes the greatest common divisor of the two natural numbers \(m, n \in \mathbb{N}_+\),\footnote{More precisely, let \(m = \prod_{k = 0}^\infty \left ( p_k \right )^{i_k}\) and \(n = \prod_{k = 0}^\infty  \left ( p_k \right )^{j_k}\) be the respective prime factorizations. Then, we have \begin{equation} \GCD{m}{n} := \prod_{k = 0}^\infty \left ( p_k \right )^{\operatorname{Min} \left ( i_k, j_k \right )} \, . \end{equation}} and finally \(\tilde{\imath} \left ( i, j \right )_{\bbsT}\) and \(\tilde{\jmath} \left ( i, j \right )_{\bbsT}\) denote, as exponents of residues, all possibilities to glue them to trees with transversal intermediate graviton edges. This will be studied further in \cite{Prinz_5,Prinz_7}, using the general theory and Feynman rules developed in \cite{Prinz_2,Prinz_4}.
\end{exmp}

\enter

\begin{defn}[Quantum gauge symmetry ideal] \label{defn:qgs_ideal}
	Given the situation of \defnref{defn:quantum_gauge_symmetries}, a cograph-divergent QFT \(\Q\) and its (associated) renormalization Hopf algebra. Then we define for each quantum gauge symmetry \(\set{v,m;w,n} \in \operatorname{QGS}_\Q\) the ideal
	\begin{subequations} \label{eqns:qgs_ideal}
	\begin{align}
		\iQ^{\set{v,m;w,n}} & := \sum_{\mathbf{C}^\prime \in \mathbb{Z}^{\mathfrak{q}_\Q}} \left \langle \overline{\combcharge}^{mv}_{\mathbf{C}^\prime} - \overline{\combcharge}^{nw}_{\mathbf{C}^\prime} \right \rangle_{\HQ}
		\intertext{and then sum over all such quantum gauge symmetries \(\set{v,m;w,n} \in \operatorname{QGS}_\Q\) to obtain the ideal}
		\iQ & := \sum_{\set{v,m;w,n} \in \operatorname{QGS}_\Q} \iQ^{\set{v,m;w,n}} \, .
	\end{align}
	\end{subequations}
\end{defn}

\enter

\begin{lem} \label{lem:elements_in_iq}
	Given the situation of \defnref{defn:qgs_ideal} and let
	\begin{equation} \label{eqn:coupling-coloring_function_multi-indices}
		\boldsymbol{\theta} \, : \quad \ZvQ \to \ZqQ \, , \quad \mathbf{V} \mapsto \mathbf{C}
	\end{equation}
	be the function mapping a vertex-grading multi-index to its corresponding coupling-grading multi-index with respect to the function \(\theta\) from Equation~(\ref{eqn:coupling-coloring_function}). Then we have \(\big ( \overline{\combcharge}^\mathbf{v}_\mathbf{C} - \overline{\combcharge}^\mathbf{w}_\mathbf{C} \big ) \in \iQ\) if and only if \(\boldsymbol{\theta} \left ( \mathbf{v} \right ) = \boldsymbol{\theta} \left ( \mathbf{w} \right )\) for two vertex-grading multi-indices \(\mathbf{v}, \mathbf{w} \in \ZvQ\).
\end{lem}

\begin{proof}
	Before we start with the actual proof, we emphasize that \(\theta \colon \AQ \to \qQ\) is the function mapping a vertex residue to its associated (product of) coupling constant(s) and thus \(\boldsymbol{\theta} \colon \ZvQ \to \ZqQ\) is the function mapping a multi-index of vertex residues to its associated multi-index of coupling constants. From the very definition of the quantum gauge symmetry ideal \(\iQ\) we know that \(\big ( \overline{\combcharge}^{mv}_\mathbf{C} - \overline{\combcharge}^{nw}_\mathbf{C} \big ) \in \iQ\) if and only if \(\set{v,m;w,n} \in \operatorname{QGS}_\Q\) is a quantum gauge symmetry. Thus we will now show that we can rewrite the element \(\big ( \overline{\combcharge}^\mathbf{v}_\mathbf{C} - \overline{\combcharge}^\mathbf{w}_\mathbf{C} \big )\) as a sum of the form
	\begin{equation} \label{eqn:to_show_qgs-equivalence}
		\left ( \overline{\combcharge}^\mathbf{v}_\mathbf{C} - \overline{\combcharge}^\mathbf{w}_\mathbf{C} \right ) = \sum_{\mathbf{C}^\prime \in \mathbb{Z}^{\mathfrak{q}_\Q}} \left ( \prod_{i = 1}^K \overline{\combcharge}^{m_i v_i}_{\mathbf{C}^\prime} - \prod_{j = 1}^K \overline{\combcharge}^{n_j w_j}_{\mathbf{C}^\prime} \right ) \mathfrak{H}_{\mathbf{C} - \mathbf{C}^\prime} \, ,
	\end{equation}
	where \(\set{v_k ,m_k ; w_k ,n_k} \in \operatorname{QGS}_\Q\) are quantum gauge symmetries for all \(k \in \set{1, \dots, K}\) and \(\mathfrak{H}_{\mathbf{C} - \mathbf{C}^\prime} \in \HQ\) are elements in the associated renormalization Hopf algebra. For the first direction, we assume the equality
	\begin{equation}
		\mathbf{c} := \boldsymbol{\theta} \left ( \mathbf{v} \right ) = \boldsymbol{\theta} \left ( \mathbf{w} \right ) \, .
	\end{equation}
	Then, we observe that due to the very definition of quantum gauge symmetries, cf.\ \defnref{defn:quantum_gauge_symmetries}, there exist quantum gauge symmetries \(\set{\set{v_k ,m_k ; w_k ,n_k}}_{k = 1}^K\) such that the two vertex-grading multi-indices are related via
	\begin{equation}
		\mathbf{r} := \mathbf{v} - \sum_{k = 1}^K m_k \mathbf{e}_{v_k} = \mathbf{w} - \sum_{k = 1}^K n_k \mathbf{e}_{w_k} \, ,
	\end{equation}
	where \(\mathbf{e}_v\) is the unit multi-index with respect to the vertex \(v\). Furthermore, since
	\begin{equation}
		\left \vert \mathbf{c} \right \vert := \sum_{i = 1}^{\ZqQ} \mathbf{c}_i < \infty \, ,
	\end{equation}
	we observe that this is always possible for a finite number of quantum gauge symmetries, i.e.\ \(K \in \mathbb{N}\). Additionally, given that \(\mathbf{v} = \mathbf{w}\) leads to the trivial case \(\big ( \overline{\combcharge}^\mathbf{v}_\mathbf{C} - \overline{\combcharge}^\mathbf{w}_\mathbf{C} \big ) = 0 \in \iQ\), we assume from now on \(\mathbf{v} \neq \mathbf{w}\), which implies \(K \in \mathbb{N}_+\). We now proceed with \eqnref{eqn:to_show_qgs-equivalence} by writing out the term in the brackets by setting \(\mathfrak{H}_{\mathbf{C} - \mathbf{C}^\prime}  := \overline{\combcharge}^\mathbf{r}_{\mathbf{C} - \mathbf{C}^\prime}\) and using the definition given in \eqnref{eqn:restricted_products_combinatorial_charges}:
	\begin{equation}
	\begin{split}
		\left ( \prod_{i = 1}^K \overline{\combcharge}^{m_i v_i}_{\mathbf{C}^\prime} \right . & - \left . \prod_{j = 1}^K \overline{\combcharge}^{n_j w_j}_{\mathbf{C}^\prime} \right ) = \eval{\left ( \prod_{i = 1}^K \left ( \combcharge^{v_i} \right )^{m_i} - \prod_{j = 1}^K \left ( \combcharge^{w_j} \right )^{n_j} \right )}_{\mathbf{C}^\prime} \\
		 = & \eval{\left ( \! \left ( \prod_{i = 1}^{K-1} \left ( \combcharge^{v_i} \right )^{m_i} \right ) \! \Big ( \! \left ( \combcharge^{v_K} \right )^{m_K} - \left ( \combcharge^{w_K} \right )^{n_K} \! \Big ) \! \right )}_{\mathbf{C}^\prime} \\
		 & + \eval{\left ( \! \left ( \prod_{i = 1}^{K-2} \left ( \combcharge^{v_i} \right )^{m_i} \right ) \! \Big ( \! \left ( \combcharge^{v_{K-1}} \right )^{m_{K-1}} - \left ( \combcharge^{w_{K-1}} \right )^{n_{K-1}} \! \Big ) \! \left ( \combcharge^{w_K} \right )^{n_K} \right )}_{\mathbf{C}^\prime} \\
		 & + \eval{\left ( \! \left ( \prod_{i = 1}^{K-3} \left ( \combcharge^{v_i} \right )^{m_i} \right ) \! \Big ( \! \left ( \combcharge^{v_{K-2}} \right )^{m_{K-2}} - \left ( \combcharge^{w_{K-2}} \right )^{n_{K-2}} \! \Big ) \! \left ( \prod_{j = K-1}^{K} \left ( \combcharge^{w_j} \right )^{n_j} \right ) \! \right )}_{\mathbf{C}^\prime} \\
		 & + \dots \\
		 & + \eval{\left ( \! \left ( \prod_{i = 1}^2 \left ( \combcharge^{v_i} \right )^{m_i} \right ) \! \Big ( \! \left ( \combcharge^{v_3} \right )^{m_3} - \left ( \combcharge^{w_3} \right )^{n_3} \! \Big ) \! \left ( \prod_{j = 4}^{K} \left ( \combcharge^{w_j} \right )^{n_j} \right ) \! \right )}_{\mathbf{C}^\prime} \\
		 & + \eval{\left ( \! \left ( \combcharge^{v_1} \right )^{m_1} \! \Big ( \! \left ( \combcharge^{v_2} \right )^{m_2} - \left ( \combcharge^{w_2} \right )^{n_2} \! \Big ) \! \left ( \prod_{j = 3}^{K} \left ( \combcharge^{w_j} \right )^{n_j} \right ) \! \right )}_{\mathbf{C}^\prime} \\
		 & + \eval{\left ( \! \Big ( \! \left ( \combcharge^{v_1} \right )^{m_1} - \left ( \combcharge^{w_1} \right )^{n_1} \! \Big ) \! \left ( \prod_{j = 2}^{K} \left ( \combcharge^{w_j} \right )^{n_j} \right ) \! \right )}_{\mathbf{C}^\prime}
	\end{split}
	\end{equation}
	From this it is straightforward to see that each summand lies in the quantum gauge symmetry ideal \(\iQ\) and furthermore that the intermediate terms cancel pairwise. This directly implies the implication \(\big ( \overline{\combcharge}^\mathbf{v}_\mathbf{C} - \overline{\combcharge}^\mathbf{w}_\mathbf{C} \big ) \in \iQ\) if \(\boldsymbol{\theta} \left ( \mathbf{v} \right ) = \boldsymbol{\theta} \left ( \mathbf{w} \right )\). Conversely, given that \(\big ( \overline{\combcharge}^\mathbf{v}_\mathbf{C} - \overline{\combcharge}^\mathbf{w}_\mathbf{C} \big ) \in \iQ\), we obtain a decomposition as in \eqnref{eqn:to_show_qgs-equivalence} by the very definition of the quantum gauge symmetry ideal, cf.\ \defnref{defn:qgs_ideal}, which directly implies the equality \(\boldsymbol{\theta} \left ( \mathbf{v} \right ) = \boldsymbol{\theta} \left ( \mathbf{w} \right )\). This shows the claimed equivalence and thus concludes the proof.
\end{proof}

\enter

\begin{thm}[Quantum gauge symmetries induce Hopf ideals\footnote{This is a direct generalization of \cite[Theorem 15]{vSuijlekom_BV} to super- and non-renormalizable theories, theories with several coupling constants and such with longitudinal and transversal degrees of freedom.}] \label{thm:quantum_gauge_symmetries_induce_hopf_ideals}
	Given the situation of \defnref{defn:qgs_ideal}, the ideal \(\iQ\) is a Hopf ideal, i.e.\ satisfies:
	\begin{enumerate}
		\item \(\Delta \big ( \iQ \big ) \subseteq \HQ \otimes \iQ + \iQ \otimes \HQ\)
		\item \(\coone \big ( \iQ \big ) = 0\)
		\item \(S \big ( \iQ \big ) \subseteq \iQ\)
	\end{enumerate}
\end{thm}

\begin{proof}
	We start with the following calculation, using \propref{prop:coproduct_exponentiated_combinatorial_charges}:
	\begin{equation}
	\begin{split}
		\D{\overline{\combcharge}^{mv}_{\mathbf{V} - m \mathbf{e}_v} - \overline{\combcharge}^{nw}_{\mathbf{V} - n \mathbf{e}_w}} & = \sum_{\mathbf{v}^\prime \in \ZvQ} \left ( \eval{\left ( \overline{\combcharge}^{mv} \overline{\combcharge}^{\mathbf{v}^\prime} \right )}_{\mathbf{V} - \mathbf{v}^\prime - m \mathbf{e}_v} \otimes \overline{\combcharge}^{mv}_{\mathbf{v}^\prime} \right . \\ & \phantom{= \sum_{\mathbf{v}^\prime \in \ZvQ} (} \left . - \eval{\left ( \overline{\combcharge}^{nw} \overline{\combcharge}^{\mathbf{v}^\prime} \right )}_{\mathbf{V} - \mathbf{v}^\prime - n \mathbf{e}_w} \otimes \overline{\combcharge}^{nw}_{\mathbf{v}^\prime} \right ) \\
		& = \sum_{\mathbf{v}^\prime \in \ZvQ} \Big ( \overline{\combcharge}^{\mathbf{v}^\prime + m \mathbf{e}_v}_{\mathbf{V} - \mathbf{v}^\prime - m \mathbf{e}_v} \otimes \overline{\combcharge}^{mv}_{\mathbf{v}^\prime} - \overline{\combcharge}^{\mathbf{v}^\prime + n \mathbf{e}_w}_{\mathbf{V} - \mathbf{v}^\prime - n \mathbf{e}_w} \otimes \overline{\combcharge}^{nw}_{\mathbf{v}^\prime} \Big ) \\
		& = \sum_{\mathbf{v} \in \ZvQ} \left ( \overline{\combcharge}^\mathbf{v}_{\mathbf{V} - \mathbf{v}} \otimes \overline{\combcharge}^{mv}_{\mathbf{v} - m \mathbf{e}_v} - \overline{\combcharge}^\mathbf{v}_{\mathbf{V} - \mathbf{v}} \otimes \overline{\combcharge}^{nw}_{\mathbf{v} - n \mathbf{e}_w} \right ) \\
		& = \sum_{\mathbf{v} \in \ZvQ} \overline{\combcharge}^\mathbf{v}_{\mathbf{V} - \mathbf{v}} \otimes \left ( \overline{\combcharge}^{mv}_{\mathbf{v} - m \mathbf{e}_v} - \overline{\combcharge}^{nw}_{\mathbf{v} - n \mathbf{e}_w} \right )
	\end{split}
	\end{equation}
	Thus, when summing over all vertex-gradings \(\mathbf{v} \in \ZvQ\) that contribute to a particular coupling-grading \(\mathbf{c} \in \ZqQ\), briefly denoted via \(\mathbf{v} \in \boldsymbol{\theta}^{-1} \left ( \mathbf{c} \right )\) with \(\boldsymbol{\theta} \colon \ZvQ \to \ZqQ\) the function defined in \eqnref{eqn:coupling-coloring_function_multi-indices}, we obtain:
	\begin{equation} \label{eqn:coproduct_qgs}
	\begin{split}
		\D{\overline{\combcharge}^{mv}_{\mathbf{C} - m \mathbf{e}_{\theta \left ( v \right )}} - \overline{\combcharge}^{nw}_{\mathbf{C} - n \mathbf{e}_{\theta \left ( w \right )}}} & = \sum_{\mathbf{v} \in \ZvQ} \overline{\combcharge}^\mathbf{v}_{\mathbf{C} - \boldsymbol{\theta} \left ( \mathbf{v} \right )} \otimes \left ( \overline{\combcharge}^{mv}_{\mathbf{v} - m \mathbf{e}_v} - \overline{\combcharge}^{nw}_{\mathbf{v} - n \mathbf{e}_w} \right ) \\
		& = \sum_{\mathbf{c} \in \ZqQ} \sum_{\mathbf{v} \in \boldsymbol{\theta}^{-1} \left ( \mathbf{c} \right )} \overline{\combcharge}^\mathbf{v}_{\mathbf{C} - \mathbf{c}} \otimes \left ( \overline{\combcharge}^{mv}_{\mathbf{v} - m \mathbf{e}_v} - \overline{\combcharge}^{nw}_{\mathbf{v} - n \mathbf{e}_w} \right )
	\end{split}
	\end{equation}
	We proceed by introducing the following equivalence relation: Given two elements \(\mathfrak{G}, \mathfrak{H} \in \HQ\), we set
	\begin{subequations} \label{eqns:qgs_equivalence_relations}
	\begin{align}
		\mathfrak{G} \sim \mathfrak{H} \quad & : \iff \quad \exists \, \mathfrak{I} \subset \iQ \; \text{ such that } \; \mathfrak{G} = \mathfrak{H} + \mathfrak{I} \, , \label{eqn:first_equivalence}
	\intertext{which, on the level of restricted products of combinatorial charges, is due to \lemref{lem:elements_in_iq} equivalent to the following equivalence relation}
		\overline{\combcharge}^\mathbf{v}_{\mathbf{C}^\prime} \sim \overline{\combcharge}^\mathbf{w}_{\mathbf{C}^\prime} \quad & : \iff \quad \boldsymbol{\theta} \left ( \mathbf{v} \right ) = \boldsymbol{\theta} \left ( \mathbf{w} \right ) \, . \label{eqn:second_equivalence}
	\end{align}
	\end{subequations}
	Coming back to \eqnref{eqn:coproduct_qgs}, we now want to implement the equivalence relation of \eqnsref{eqns:qgs_equivalence_relations} on the left-hand side of the tensor product by adding terms in \(\iQ \otimes \HQ\): This implies that we can define equivalence classes of restricted combinatorial charges where the exponent is now a coupling-grading multi-index \(\mathbf{c} \in \ZqQ\), which we denote via \(\big [ \overline{\combcharge}^\mathbf{c}_{\mathbf{C} - \mathbf{c}} \big ]\). More precisely, let \(\mathbf{v}_1, \dots,  \mathbf{v}_L \in \ZvQ\) be all vertex-grading multi-indices with \(\boldsymbol{\theta} \left ( \mathbf{v}_l \right ) = \mathbf{c}\) for \(l \in \set{1, \dots, L}\) and \(L \in \mathbb{N}\). Then, the second sum in the second line of \eqnref{eqn:coproduct_qgs} reads
	\begin{equation}
		\sum_{\mathbf{v} \in \boldsymbol{\theta}^{-1} \left ( \mathbf{c} \right )} \overline{\combcharge}^\mathbf{v}_{\mathbf{C} - \mathbf{c}} \otimes \left ( \overline{\combcharge}^{mv}_{\mathbf{v} - m \mathbf{e}_v} - \overline{\combcharge}^{nw}_{\mathbf{v} - n \mathbf{e}_w} \right ) = \sum_{l = 1}^L \overline{\combcharge}^{\mathbf{v}_l}_{\mathbf{C} - \mathbf{c}} \otimes \left ( \overline{\combcharge}^{mv}_{\mathbf{v}_l - m \mathbf{e}_v} - \overline{\combcharge}^{nw}_{\mathbf{v}_l - n \mathbf{e}_w} \right )
	\end{equation}
	and the combinatorial charges \(\overline{\combcharge}^{\mathbf{v}_l}\) can be identified, modulo the addition of terms in \(\iQ \otimes \HQ\), to their equivalence class \(\big [ \overline{\combcharge}^\mathbf{c}_{\mathbf{C} - \mathbf{c}} \big ]\). Combining these results, we finally obtain:
	\begin{equation}
	\begin{split}
		\D{\overline{\combcharge}^{mv}_{\mathbf{C} - m \mathbf{e}_{\theta \left ( v \right )}} - \overline{\combcharge}^{nw}_{\mathbf{C} - n \mathbf{e}_{\theta \left ( w \right )}}} & \simeq_{\iQ \otimes \HQ} \sum_{\mathbf{c} \in \ZqQ} \left [ \overline{\combcharge}^\mathbf{c}_{\mathbf{C} - \mathbf{c}} \right ] \otimes \left ( \overline{\combcharge}^{mv}_{\mathbf{c} - m \mathbf{e}_{\theta \left ( v \right )}} - \overline{\combcharge}^{nw}_{\mathbf{c} - n \mathbf{e}_{\theta \left ( w \right )}} \right ) \\
		& \subseteq \HQ \otimes \iQ + \iQ \otimes \HQ \, ,
	\end{split}
	\end{equation}
	where \(\simeq_{\iQ \otimes \HQ}\) denotes equality modulo the addition of elements in \(\iQ \otimes \HQ\). Additionally, we remark the equality \(\theta \left ( v \right )^m \equiv \theta \left ( w \right )^n\) for quantum gauge symmetries \(\set{v,m;w,n} \in \operatorname{QGS}_\Q\), which relates the calculations in this proof to the definition of the quantum gauge symmetry ideal, cf.\ \defnref{defn:qgs_ideal}, by setting
	\begin{equation}
		\mathbf{C}^\prime := \mathbf{C} - m \mathbf{e}_{\theta \left ( v \right )} \equiv \mathbf{C} - n \mathbf{e}_{\theta \left ( w \right )} \, .
	\end{equation}
	This shows condition 1. Condition 2 follows immediately, as \(\one \notin \iQ\), i.e.\ \(\iQ \subset \operatorname{Aug} \left ( \HQ \right )\). Finally, condition 3 follows from \lemref{lem:coproduct_and_antipode_identities} together with condition 1, which finishes the proof.
\end{proof}

\enter

\begin{rem}
	\thmref{thm:quantum_gauge_symmetries_induce_hopf_ideals} describes the most general situation, as it includes also super- and non-renormalizable QFTs, QFTs with several coupling constants and QFTs with a transversal structure. Therefore it can be applied to (effective) Quantum General Relativity in the sense of \cite{Kreimer_QG1}, possibly coupled to matter from the Standard Model \cite{Romao_Silva}, cf.\ e.g.\ \cite{Prinz_2,Prinz_4}. Slightly less general results in this direction can be found in \cite{Kreimer_Anatomy,vSuijlekom_QED,vSuijlekom_QCD,vSuijlekom_BV,Kreimer_vSuijlekom,Kreimer_QG1,Kreimer_Core}, some of them using the language of Hochschild cohomology.\footnote{We also mention the relevant conference proceedings \cite{Kreimer_QG2,vSuijlekom_GF,vSuijlekom_BRST,vSuijlekom_Combinatorics,vSuijlekom_Ha-pQGT}.}
\end{rem}

\enter

\begin{col} \label{col:equivalence_vtx-grd_cpl-grd}
	Given the situation of \thmref{thm:quantum_gauge_symmetries_induce_hopf_ideals}, the vertex-grading and coupling-grading are equivalent if either \(\mathfrak{v}_\Q = \mathfrak{q}_\Q\) with \(\boldsymbol{\theta}\) bijective, or in the quotient Hopf algebra \(\HQ / \iQ\). In the latter case, the ideal \(\iQ\) is the smallest Hopf ideal with this property.
\end{col}

\begin{proof}
	This follows directly from the definition of the ideal \(\iQ\) in \defnref{defn:qgs_ideal}, which is designed such that in the quotient \(\HQ / \iQ\) all Feynman graphs are identified that contribute to restricted (powers of) combinatorial charges which are associated with (powers) of the same physical coupling constant.
\end{proof}

\enter

\begin{rem} \label{rem:quotient_vertex-grading_coupling-grading}
	The Hopf ideal \(\iQ\) from \thmref{thm:quantum_gauge_symmetries_induce_hopf_ideals} is defined such that in the quotient Hopf algebra \(\HQ / \iQ\) the coproduct and antipode identities from \sectionref{sec:coproduct_and_antipode_identities}, which are valid for vertex-grading, also hold for coupling-grading, cf.\ \colref{col:hopf_subalgebras_coupling-grading}. Thus it is possible to combine the \(Z\)-factors for the set \(\QQ\) to \(Z\)-factors for the set \(\qQ\), if the criteria from \thmref{thm:criterion_ren-hopf-mod} or \colref{col:qgs_and_rfr} are satisfied.
\end{rem}

\section{Quantum gauge symmetries and renormalized Feynman rules} \label{sec:quantum_gauge_symmetries_and_renormalized_feynman_rules}

Having established that `quantum gauge symmetries (QGS)' are compatible with the treatment of subdivergences in \thmref{thm:quantum_gauge_symmetries_induce_hopf_ideals}, we now turn our attention to their relation with renormalized Feynman rules. We start this section with the definition of the gauge theory renormalization Hopf module: Here we implement the quantum gauge symmetries only on the left-hand side of the tensor product of the coproduct, i.e.\ only on the superficially divergent subgraphs. As such, it is the weakest requirement for renormalized Feynman rules to possess quantum gauge symmetries. More precisely, in this setting the relations are only implemented on the \(Z\)-factors, i.e.\ \(\mathscr{R}\)-divergent contributions of the Feynman rules. In \thmref{thm:criterion_ren-hopf-mod} we provide criteria for this compatibility of quantum gauge symmetries with the unrenormalized Feynman rules and the chosen renormalization scheme. Then we show in \colref{col:qgs_and_rfr} that under mild assumptions on the unrenormalized Feynman rules this statement is independent of the chosen renormalization scheme. In particular, this result states that in this case we can implement the quantum gauge symmetries directly on the renormalization Hopf algebra by taking the quotient with respect to the quantum gauge symmetry Hopf ideal. Finally, we remark that for theories with a transversal structure these mild assumptions correspond precisely to the respective cancellation identities. Thus, combining these results, this shows the well-definedness of the Corolla polynomial without reference to a particular renormalization scheme, cf.\ \remref{rem:corolla_polynomial}.

\enter

\begin{defn}[Gauge theory renormalization Hopf module, \cite{Kissler_PhD}] \label{defn:renormalization_hopf_module}
	Let \(\Q\) be a cograph-divergent QFT with quantum gauge symmetries, i.e.\ \(\operatorname{QGS}_\Q \neq \emptyset\), \(\HQ\) its (associated) renormalization Hopf algebra and \(\iQ\) the  corresponding quantum gauge symmetry Hopf ideal. Let
	\begin{equation} \label{eqn:qgs_projection_map}
		\pi_\Q : \, \HQ \surject \HQ / \iQ
	\end{equation}
	denote the projection map. We consider \(\HQ\) as a left Hopf module over \(\HQ / \iQ\) with the usual Hopf structures as in \defnref{defn:renormalization_hopf_algebra}. The interesting map is the comodule map, defined via
	\begin{equation}
		\delta \, : \quad \HQ \to \left ( \HQ / \iQ \right ) \otimes \HQ \, , \quad \Gamma \mapsto \big ( \pi_\Q \otimes \id_{\HQ} \! \big ) \circ \D{\Gamma} \, .
	\end{equation}
	Then we define the renormalized Feynman rules \(\Phi_\mathscr{R}\) using the comodule map \(\delta\) instead of the coproduct \(\Delta\), i.e.\ defining the counterterm map \(\countertermsymbol\) on the quotient \(\HQ / \iQ\).
\end{defn}

\enter

\begin{col} \label{col:hopf_subalgebras_coupling-grading}
	Given the situation of \defnref{defn:renormalization_hopf_module}, we have
	\begin{subequations}
	\begin{align}
		\delta \left ( \combgreen^r \right ) & = \sum_{\mathbf{c} \in \mathbb{Z}^{\mathfrak{q}_\Q}} \left [ \overline{\combgreen}^r \overline{\combcharge}^\mathbf{c} \right ] \otimes \rescombgreen^r_{\mathbf{c}} \, , \\
		\delta \left ( \rescombgreen^r_{\mathbf{C}} \right ) & = \sum_{\mathbf{c} \in \mathbb{Z}^{\mathfrak{q}_\Q}} \eval{\left [ \overline{\combgreen}^r \overline{\combcharge}^\mathbf{c} \right ]}_{\mathbf{C} - \mathbf{c}} \otimes \rescombgreen^r_{\mathbf{c}} \, ,
		\intertext{and, provided that \(\overline{\rescombgreen}^r_{\mathbf{C}} \neq 0\),}
		\delta \left ( \overline{\rescombgreen}^r_{\mathbf{C}} \right ) & = \sum_{\mathbf{c} \in \mathbb{Z}^{\mathfrak{q}_\Q}} \eval{\left [ \overline{\combgreen}^r \overline{\combcharge}^\mathbf{c} \right ]}_{\mathbf{C} - \mathbf{c}} \otimes \overline{\rescombgreen}^r_{\mathbf{c}} \, ,
	\end{align}
	\end{subequations}
	where the equivalence classes on the left-hand side of the tensor product are with respect to the equivalence relation of Equations~(\ref{eqns:qgs_equivalence_relations}), i.e.\ modulo the addition of elements in \(\iQ\). Analogous results also hold in the cases of \propsaref{prop:coproduct_charges}{prop:coproduct_exponentiated_combinatorial_charges}.
\end{col}

\begin{proof}
	This follows directly from \propref{prop:coproduct_greensfunctions} together with \colref{col:equivalence_vtx-grd_cpl-grd}.
\end{proof}

\enter

\begin{rem}
	\colref{col:hopf_subalgebras_coupling-grading} states that the gauge theory renormalization Hopf module from \defnref{defn:renormalization_hopf_module} possesses Hopf subalgebras in the sense of \defnref{defn:hopf_subalgebras_renormalization_hopf_algebra} for coupling-grading, cf.\ \remref{rem:hopf_subalgebras_renormalization_hopf_algebra}. This implies that the subdivergence structure of QFTs is compatible with quantum gauge symmetries in the sense of \defnref{defn:quantum_gauge_symmetries}. Furthermore, it is obvious by construction that this is the weakest requirement for the compatibility of quantum gauge symmetries with multiplicative renormalization, cf.\ \colref{col:equivalence_vtx-grd_cpl-grd}. Their validity on the level of renormalized Feynman rules, i.e.\ the existence and well-definedness of the maps
	\begin{subequations}
	\begin{align}
	\widetilde{\countertermsymbol} & := \countertermsymbol \circ \left ( \pi_\Q \right )^{-1} \, : \quad \HQ / \iQ \to \EQ
	\intertext{and}
	\widetilde{\renFR} & := \renFR \circ \left ( \pi_\Q \right )^{-1} \, : \quad \HQ / \iQ \to \EQ \, ,
	\end{align}
	\end{subequations}
	where \(\left ( \pi_\Q \right )^{-1}\) is any right inverse to the projection map \(\pi_\Q\) from \eqnref{eqn:qgs_projection_map}, with respect to the following commuting diagrams
	\begin{equation}
	\begin{tikzcd}[row sep=huge]
		\HQ \arrow[swap]{d}{\pi_\Q} \arrow{r}{\countertermsymbol} & \EQ \\
		\HQ / \iQ \arrow[dashed, swap]{ur}{\widetilde{\countertermsymbol}} &
	\end{tikzcd}
	\qquad \text{and} \qquad
	\begin{tikzcd}[row sep=huge]
		\HQ \arrow[swap]{d}{\pi_\Q} \arrow{r}{\renFR} & \EQ \\
		\HQ / \iQ \arrow[dashed, swap]{ur}{\widetilde{\renFR}} &
	\end{tikzcd} \, ,
	\end{equation}
	are then studied in the following \lemref{lem:criterion_ren-hopf-mod}, \thmref{thm:criterion_ren-hopf-mod} and \colref{col:qgs_and_rfr}. Moreover, given a quantum gauge theory with a transversal structure, cf.\ \defnref{defn:transversal_structure}, we stress the following additional compatibility issue: Recall the setup of \defnref{defn:residue_amplitude_and_coupling_constant_set} where we represented each particle type by at least two edges to disentangle their physical and unphysical degrees of freedom, cf.\ \remref{rem:longitudinal_and_transversal_gauge_fields}. Then the Feynman rules are required to be compatible with this decomposition as follows: The divergent Feynman graphs and their residues need to behave similar with respect to these physical and unphysical projections. This ensures that the contraction of subdivergences is a well-defined operation and thus is a necessary condition to construct the renormalization Hopf algebra, cf.\ \cite[Subsection 3.3]{Prinz_2}. We will study this further in \cite{Prinz_9}, cf.\ \cite{Prinz_5,Prinz_7}, using cancellation identities and Feynman graph cohomology.
\end{rem}

\enter

\begin{lem} \label{lem:criterion_ren-hopf-mod}
	The gauge theory renormalization Hopf module from \defnref{defn:renormalization_hopf_module} is compatible with renormalized Feynman rules if \(\, \iQ \in \operatorname{Ker} \big ( \countertermsymbol \big )\). More precisely, if for all \(\set{v, m; w, n} \in \operatorname{QGS}_\Q\) and all \(\mathbf{C}^\prime \in \ZqQ\) we have
	\begin{equation}
		\counterterm{\overline{\combcharge}^{mv}_{\mathbf{C}^\prime}} = \counterterm{\overline{\combcharge}^{nw}_{\mathbf{C}^\prime}} \, . \label{eqn:well-definedness_counterterm-map}
	\end{equation}
\end{lem}

\begin{proof}
	This statement is equivalent to the well-definedness of the counterterm-map on the equivalence classes of the QGS-equivalence relation: Indeed, we have
	\begin{equation}
		\overline{\combcharge}^{mv}_{\mathbf{C}^\prime} \simeq_{\operatorname{QGS}_\Q} \overline{\combcharge}^{nw}_{\mathbf{C}^\prime} \, ,
	\end{equation}
	and thus \eqnref{eqn:well-definedness_counterterm-map} ensures that the counterterm-map can be unambiguously defined on the corresponding equivalence classes.
\end{proof}

\enter

\begin{thm}[Quantum gauge symmetries and renormalized Feynman rules] \label{thm:criterion_ren-hopf-mod}
	Given the situation of \lemref{lem:criterion_ren-hopf-mod} and a proper renormalization scheme \(\mathscr{R}\), then \eqnref{eqn:well-definedness_counterterm-map} is equivalent to one of the following identities for each \(\mathbf{c} \in \ZqQ\):
	\begin{enumerate}
		\item \(\big [ \overline{\combcharge}^\mathbf{c}_{\mathbf{C} - \mathbf{c}} \big ] = \left [ 0 \right ]\)
		\item \(\overline{\combcharge}^{mv}_{\mathbf{c} - m \mathbf{e}_{\theta \left ( v \right )}} = \overline{\combcharge}^{nw}_{\mathbf{c} - n \mathbf{e}_{\theta \left ( w \right )}} = 0\)
		\item \(\renscheme{\left ( \Phi \circ \mathscr{A} \right ) \big ( \overline{\combcharge}^{mv}_{\mathbf{c} - m \mathbf{e}_{\theta \left ( v \right )}} - \overline{\combcharge}^{nw}_{\mathbf{c} - n \mathbf{e}_{\theta \left ( w \right )}} \big )} = 0\)
	\end{enumerate}
\end{thm}

\begin{proof}
	Using \eqnref{eqn:coproduct_qgs} from the proof of \thmref{thm:quantum_gauge_symmetries_induce_hopf_ideals}, we obtain
	\begin{equation}
	\begin{split}
		\countertermsymbol \, \Big ( \overline{\combcharge}^{mv}_{\mathbf{C} - m \mathbf{e}_{\theta \left ( v \right )}} & - \overline{\combcharge}^{nw}_{\mathbf{C} - n \mathbf{e}_{\theta \left ( w \right )}} \Big ) = \\ & \phantom{=} - \sum_{\mathbf{c} \in \ZqQ} \renscheme{\counterterm{\left [ \overline{\combcharge}^\mathbf{c}_{\mathbf{C} - \mathbf{c}} \right ]} \FRP{\overline{\combcharge}^{mv}_{\mathbf{c} - m \mathbf{e}_{\theta \left ( v \right )}} - \overline{\combcharge}^{nw}_{\mathbf{c} - n \mathbf{e}_{\theta \left ( w \right )}}}} \, ,
	\end{split}
	\end{equation}
	which vanishes, if for all \(\mathbf{c} \in \ZqQ\)
	\begin{equation}
		\renscheme{\counterterm{\left [ \overline{\combcharge}^\mathbf{c}_{\mathbf{C} - \mathbf{c}} \right ]} \FRP{\overline{\combcharge}^{mv}_{\mathbf{c} - m \mathbf{e}_{\theta \left ( v \right )}} - \overline{\combcharge}^{nw}_{\mathbf{c} - n \mathbf{e}_{\theta \left ( w \right )}}}} = 0 \, . \label{eqn:qgs_criterion_fr_rs}
	\end{equation}
	Since the coupling-grading is required to be superficially compatible, cf.\ \defnref{defn:superficially_compatible_grading} and \defnref{defn:quantum_gauge_symmetries}, we have either
	\begin{equation}
		\left [ \overline{\combcharge}^\mathbf{c}_{\mathbf{C} - \mathbf{c}} \right ] = \left [ 0 \right ] \qquad \text{or} \qquad \left [ \overline{\combcharge}^\mathbf{c}_{\mathbf{C} - \mathbf{c}} \right ] = \Big [ \combcharge^\mathbf{c}_{\mathbf{C} - \mathbf{c}} \Big ]
	\end{equation}
	and either
	\begin{equation}
		\overline{\combcharge}^{mv}_{\mathbf{c}^\prime} = \overline{\combcharge}^{nw}_{\mathbf{c}^\prime} = 0 \qquad \text{or} \qquad \overline{\combcharge}^{mv}_{\mathbf{c}^\prime} = \combcharge^{mv}_{\mathbf{c}^\prime} \! \quad \text{and} \quad \overline{\combcharge}^{nw}_{\mathbf{c}^\prime} = \combcharge^{nw}_{\mathbf{c}^\prime} \, ,
	\end{equation}
	where \(\mathbf{c}^\prime := \mathbf{c} - m \mathbf{e}_{\theta \left ( v \right )} \equiv \mathbf{c} - n \mathbf{e}_{\theta \left ( w \right )}\). We proceed by noting that \(\mathscr{R}\) is a linear map, whose kernel and cokernel consists only of convergent formal Feynman integral expressions, as it is required to be proper, cf.\ \defnref{defn:proper_renormalization_scheme}. This directly implies, by the recursive structure of the counterterm, that \(\operatorname{Ker} \big ( \countertermsymbol \big )\) consists only of convergent formal Feynman integral expressions. Thus, \eqnref{eqn:qgs_criterion_fr_rs} is equivalent to one of the following identities for each \(\mathbf{c} \in \ZqQ\):
	\begin{enumerate}
		\item \(\big [ \overline{\combcharge}^\mathbf{c}_{\mathbf{C} - \mathbf{c}} \big ] = \left [ 0 \right ]\)
		\item \(\overline{\combcharge}^{mv}_{\mathbf{c} - m \mathbf{e}_{\theta \left ( v \right )}} = \overline{\combcharge}^{nw}_{\mathbf{c} - n \mathbf{e}_{\theta \left ( w \right )}} = 0\)
		\item \(\renscheme{\left ( \Phi \circ \mathscr{A} \right ) \big ( \overline{\combcharge}^{mv}_{\mathbf{c} - m \mathbf{e}_{\theta \left ( v \right )}} - \overline{\combcharge}^{nw}_{\mathbf{c} - n \mathbf{e}_{\theta \left ( w \right )}} \big )} = 0\)
	\end{enumerate}
	This is the claimed statement and thus finishes the proof.
\end{proof}

\enter

\begin{rem}
	Given the situation of \lemref{lem:criterion_ren-hopf-mod} and \thmref{thm:criterion_ren-hopf-mod}. We note that \eqnref{eqn:well-definedness_counterterm-map} and condition 3 are criteria for both, the unrenormalized Feynman rules \(\Phi\) and the renormalization scheme \(\mathscr{R}\). More precisely, \eqnref{eqn:well-definedness_counterterm-map} states that a common counterterm can be chosen for residues that are related via quantum gauge symmetries. Furthermore, condition 3 states that the corresponding \(\mathscr{R}\)-divergent contributions from restricted combinatorial charges coincide. In contrast, conditions 1 and 2 are solely criteria for the unrenormalized Feynman rules \(\Phi\).
\end{rem}

\enter

\begin{col}[Quantum gauge symmetries and renormalized Feynman rules] \label{col:qgs_and_rfr}
	The quotient Hopf algebra \(\HQ / \iQ\) is compatible with renormalized Feynman rules if \(\, \iQ \in \operatorname{Ker} \left ( \Phi_\mathscr{R} \right )\). More precisely, if for all \(\set{v, m; w, n} \in \operatorname{QGS}_\Q\) and all \(\mathbf{C}^\prime \in \ZqQ\) we have
	\begin{equation}
		\RFR{\overline{\combcharge}^{mv}_{\mathbf{C}^\prime}} = \RFR{\overline{\combcharge}^{nw}_{\mathbf{C}^\prime}} \, . \label{eqn:well-definedness_renormalized-fr}
	\end{equation}
	If the renormalization scheme \(\mathscr{R}\) is proper then this is equivalent to one of the following identities for each \(\mathbf{c} \in \ZqQ\):
	\begin{enumerate}
		\item \(\big [ \overline{\combcharge}^\mathbf{c}_{\mathbf{C} - \mathbf{c}} \big ] = \left [ 0 \right ]\)
		\item \(\overline{\combcharge}^{mv}_{\mathbf{c} - m \mathbf{e}_{\theta \left ( v \right )}} = \overline{\combcharge}^{nw}_{\mathbf{c} - n \mathbf{e}_{\theta \left ( w \right )}} = 0\)
		\item \(\Phi \big ( \overline{\combcharge}^{mv}_{\mathbf{c} - m \mathbf{e}_{\theta \left ( v \right )}} - \overline{\combcharge}^{nw}_{\mathbf{c} - n \mathbf{e}_{\theta \left ( w \right )}} \big ) = 0\)
	\end{enumerate}
\end{col}

\begin{proof}
	Again, using \eqnref{eqn:coproduct_qgs} from the proof of \thmref{thm:quantum_gauge_symmetries_induce_hopf_ideals} and the same reasoning as in the proof of \thmref{thm:criterion_ren-hopf-mod}, we obtain
	\begin{equation}
		\RFR{\overline{\combcharge}^{mv}_{\mathbf{C} - m \mathbf{e}_{\theta \left ( v \right )}} - \overline{\combcharge}^{nw}_{\mathbf{C} - n \mathbf{e}_{\theta \left ( w \right )}}} = \sum_{\mathbf{c} \in \ZqQ} \counterterm{\left [ \overline{\combcharge}^\mathbf{c}_{\mathbf{C} - \mathbf{c}} \right ]} \FR{\overline{\combcharge}^{mv}_{\mathbf{c} - m \mathbf{e}_{\theta \left ( v \right )}} - \overline{\combcharge}^{nw}_{\mathbf{c} - n \mathbf{e}_{\theta \left ( w \right )}}} \, ,
	\end{equation}
	which vanishes, if for all \(\mathbf{c} \in \ZqQ\)
	\begin{equation}
		\counterterm{\left [ \overline{\combcharge}^\mathbf{c}_{\mathbf{C} - \mathbf{c}} \right ]} \FR{\overline{\combcharge}^{mv}_{\mathbf{c} - m \mathbf{e}_{\theta \left ( v \right )}} - \overline{\combcharge}^{nw}_{\mathbf{c} - n \mathbf{e}_{\theta \left ( w \right )}}} = 0 \, . \label{eqn:qgs_criterion_rfr_fr}
	\end{equation}
	Once more, using the same reasoning as in the proof of \thmref{thm:criterion_ren-hopf-mod}, we obtain that \eqnref{eqn:qgs_criterion_rfr_fr} is equivalent to one of the following identities for each \(\mathbf{c} \in \ZqQ\):
	\begin{enumerate}
		\item \(\big [ \overline{\combcharge}^\mathbf{c}_{\mathbf{C} - \mathbf{c}} \big ] = \left [ 0 \right ]\)
		\item \(\overline{\combcharge}^{mv}_{\mathbf{c} - m \mathbf{e}_{\theta \left ( v \right )}} = \overline{\combcharge}^{nw}_{\mathbf{c} - n \mathbf{e}_{\theta \left ( w \right )}} = 0\)
		\item \(\Phi \big ( \overline{\combcharge}^{mv}_{\mathbf{c} - m \mathbf{e}_{\theta \left ( v \right )}} - \overline{\combcharge}^{nw}_{\mathbf{c} - n \mathbf{e}_{\theta \left ( w \right )}} \big ) = 0\)
	\end{enumerate}
	This is the claimed statement and thus finishes the proof.
\end{proof}

\enter

\begin{rem} \label{rem:corolla_polynomial}
	Given the situation of \colref{col:qgs_and_rfr} and a proper renormalization scheme \(\mathscr{R}\), cf.\ \defnref{defn:proper_renormalization_scheme}. We note that while \eqnref{eqn:well-definedness_renormalized-fr} is a criterion for both, the unrenormalized Feynman rules \(\Phi\) and the renormalization scheme \(\mathscr{R}\), conditions 1 to 3 are solely criteria for the unrenormalized Feynman rules \(\Phi\). More precisely, \eqnref{eqn:well-definedness_renormalized-fr} states that the values of renormalized Feynman rules are equivalent for residues that are related via quantum gauge symmetries. In contrast, conditions 1 to 3 state that the corresponding unrenormalized Feynman rules \(\Phi\) coincide on the restricted combinatorial charges. We remark, however, that if we consider a quantum gauge theory with a transversal structure, cf.\ \defnref{defn:transversal_structure}, then the vertex-residues of the combinatorial charges as well as their coupling-gradings includes the `physical' and `unphysical' labels. Thus, the mentioned criteria need only to hold for all such restrictions individually. This is directly related to cancellation identities \cite{tHooft_Veltman,Citanovic,Sars_PhD,Kissler_Kreimer,Gracey_Kissler_Kreimer,Kissler}, which are graphical versions of (generalized) Ward--Takahashi and Slavnov--Taylor identities \cite{Ward,Takahashi,Taylor,Slavnov}. More precisely, they indicate the behavior of unrenormalized (tree) Feynman diagrams with respect to longitudinal and transversal projections. Thus, given that they hold for Quantum Yang--Mills theory, \colref{col:qgs_and_rfr} implies the well-definedness of the Corolla polynomial without reference to a particular renormalization scheme \(\mathscr{R}\). The Corolla polynomial is a graph polynomial that relates amplitudes in Quantum Yang--Mills theory to amplitudes in \(\phi^3_4\)-Theory \cite{Kreimer_Yeats,Kreimer_Sars_vSuijlekom,Kreimer_Corolla}.  More precisely, this graph polynomial is based on half-edges and is used for the construction of the so-called Corolla differential that relates the corresponding parametric Feynman integral expressions \cite{Kreimer_Sars_vSuijlekom,Sars_PhD,Golz_PhD}. Thereby, the corresponding cancellation-identities are implicitly encoded into a double complex of Feynman graphs, called Feynman graph cohomology \cite{Kreimer_Sars_vSuijlekom,Berghoff_Knispel}. This double-complex can be interpreted as a perturbative version of BRST cohomology \cite{Becchi_Rouet_Stora_1,Becchi_Rouet_Stora_2,Tyutin,Becchi_Rouet_Stora_3}, where the precise relation will be studied in future work \cite{Prinz_9,Prinz_5,Prinz_7}. We remark that this construction has been successfully generalized to Quantum Yang--Mills theories with spontaneous symmetry breaking \cite{Prinz_1} and Quantum Electrodynamics with spinors \cite{Golz_1,Golz_2,Golz_3}. The possibility to reformulate (effective) Quantum General Relativity in this framework will be also studied in future work.
\end{rem}

\enter

\begin{rem}
	It is possible to endow the character group of the renormalization Hopf algebra \(\HQ\) with a manifold structure such that it becomes a regular Lie group in the sense of Milnor, cf.\ \cite{Milnor}. Then, in this setting, the character group on the quotient Hopf algebra \(\HQ / \iQ\) is a closed Lie subgroup thereof \cite{Bogfjellmo_Dahmen_Schmeding_1,Dahmen_Schmeding,Bogfjellmo_Schmeding,Bogfjellmo_Dahmen_Schmeding_2}.
\end{rem}

\chapter{Gravity-matter Feynman rules} \label{chp:linearized_gravity_and_feynman_rules}

In this chapter, we discuss the Feynman rules and transversal structure of gauge theories and gravity. This includes the linearization of the gravity-matter Lagrange densities as well as the calculation of the corresponding Feynman rules. We provide the Feynman rules of (Effective) Quantum General Relativity coupled to the Standard Model first for any vertex valence and with general gauge parameter. Then we display the concrete expressions for all gravity-matter propagators and three-valent vertices. Finally, we study the longitudinal structure of Quantum Yang--Mills theory with a Lorenz gauge fixing and (effective) Quantum General Relativity with a de Donder gauge fixing. This includes properties of the corresponding longitudinal, identical and transversal projection operators, the respective decompositions of the gauge boson and graviton propagators and the cancellation identities of all three-valent vertices in (effective) Quantum General Relativity coupled to the Standard Model.

\section{Linearized gravity and further preparations} \label{sec:linearized_gravity_and_further_preparations}

In this section, we start with the expansion of the Lagrange densities of (effective) Quantum General Relativity coupled to the Standard Model. In particular, we consider the expansion with respect to the graviton coupling constant \(\varkappa\) in the metric decomposition \(g_{\mu \nu} = \eta_{\mu \nu} + \varkappa h_{\mu \nu}\).

\subsection{Expansion of the Lagrange densities} \label{sec:expansion_lagrange_density}

Given the Quantum General Relativity Lagrange density
\begin{equation}
	\begin{split}
	\mathcal{L}_\text{QGR} & = - \frac{1}{2 \gcoupling^2} \left ( \sqrt{- \dt{g}} R + \frac{1}{2 \zeta}  \eta^{\mu \nu} \deDonder^{(1)}_\mu \deDonder^{(1)}_\nu \right ) \dif V_\eta \\
	& \phantom{:=} - \frac{1}{2} \eta^{\rho \sigma} \left ( \frac{1}{\zeta} \overline{C}^\mu \left ( \partial_\rho \partial_\sigma C_\mu \right ) + \overline{C}^\mu \left ( \partial_\mu \big ( \tensor{\Gamma}{^\nu _\rho _\sigma} C_\nu \big ) - 2 \partial_\rho \big ( \tensor{\Gamma}{^\nu _\mu _\sigma} C_\nu \big ) \right ) \right ) \dif V_\eta
	\end{split}
\end{equation}
from \conref{con:Lagrange_density}. In order to calculate the corresponding Feynman rules, we decompose \(\mathcal{L}_\text{QGR}\) with respect to its powers in the gravitational coupling constant \(\varkappa\) and the ghost field \(C\) as follows\footnote{We omit the term \(\mathcal{L}_\text{QGR}^{-1,0}\) as it is given by a total derivative.}
\begin{equation}
	\mathcal{L}_\text{QGR} \equiv \sum_{m = 0}^\infty \sum_{n = 0}^1 \mathcal{L}_\text{QGR}^{m,n} \, ,
\end{equation}
where we have set \(\mathcal{L}_\text{QGR}^{m,n} := \eval[1]{\left ( \mathcal{L}_\text{QGR} \right )}_{O (\varkappa^m C^n)}\). Given \(m \in \mathbb{N}_+\), the restricted Lagrange densities \(\mathcal{L}_\text{QGR}^{m,0}\) correspond to the potential terms for the interaction of \(\left ( m + 2 \right )\) gravitons and the restricted Lagrange densities \(\mathcal{L}_\text{QGR}^{m,1}\) correspond to the potential terms for the interaction of \(m\) gravitons with a graviton-ghost and graviton-antighost, while the terms \(m = 0\) and \(n \in \set{0,1}\) provide the kinetic terms for the graviton and graviton-ghost, respectively. The situation for the matter-model Lagrange densities from \lemref{lem:matter-model-Lagrange-densities} is then analogous.\footnote{The shift in \(m\) comes from the prefactor \(\textfrac{1}{\gcoupling^2}\) in \(\mathcal{L}_{\text{QGR}}\) and is convenient, because then propagators are of order \(\gcoupling^0\) and three-valent vertices of order \(\gcoupling^1\), etc.}

\enter

\begin{lem}[Inverse metric as Neumann series in the graviton field] \label{lem:inverse_metric_series}
	Given the metric decomposition from \defnref{defn:md_and_gf} and the boundedness condition from \assref{ass:bdns_gf}, the inverse metric is given via the Neumann series
	\begin{equation}
	g^{\mu \nu} = \sum_{k = 0}^\infty \left ( - \gcoupling \right )^k \left ( h^k \right )^{\mu \nu} \, ,
	\end{equation}
	where
	\begin{subequations}
		\begin{align}
		h^{\mu \nu} & := \eta^{\mu \rho} \eta^{\nu \sigma} h_{\rho \sigma} \, ,\\
		\left ( h^0 \vphantom{h^k} \right )^{\mu \nu} & := \eta^{\mu \nu}\\
		\intertext{and}
		\left ( h^k \right )^{\mu \nu} & := \underbrace{h^\mu_{\kappa_1} h^{\kappa_1}_{\kappa_2} \cdots h^{\kappa_{k-1} \nu}}_{\text{\(k\)-times}} \, , \; k \in \mathbb{N} \, .
		\end{align}
	\end{subequations}
	In particular, the first terms read as follows:
	\begin{equation}
		g^{\mu \nu} = \eta^{\mu \nu} - \gcoupling \eta^{\mu \rho} \eta^{\nu \sigma} h_{\rho \sigma} + \gcoupling^2 \eta^{\mu \rho} \eta^{\sigma \kappa} \eta^{\tau \nu} h_{\rho \sigma} h_{\kappa \tau} + \order{\gcoupling^3} \, .
	\end{equation}
\end{lem}

\begin{proof}
	We calculate
	\begin{equation}
	\begin{split}
		g_{\mu \nu} g^{\nu \rho} & = \left ( \eta_{\mu \nu} + \gcoupling h_{\mu \nu} \right ) \left ( \sum_{k = 0}^\infty \left ( - \gcoupling \right )^k \left ( h^k \right )^{\nu \rho} \right ) \\
		& = \eta_{\mu \nu} \eta^{\nu \rho} + \eta_{\mu \nu} \left ( \sum_{i = 1}^\infty \left ( - \gcoupling \right )^i \left ( h^i \right )^{\nu \rho} \right ) + \gcoupling h_{\mu \nu} \left ( \sum_{j = 0}^\infty \left ( - \gcoupling \right )^j \left ( h^j \right )^{\nu \rho} \right ) \\
		& = \delta_\mu^\rho - \gcoupling h_{\mu \nu} \left ( \sum_{i = 0}^\infty \left ( - \gcoupling \right )^i \left ( h^i \right )^{\nu \rho} \right ) + \gcoupling h_{\mu \nu} \left ( \sum_{j = 0}^\infty \left ( - \gcoupling \right )^j \left ( h^j \right )^{\nu \rho} \right ) \\
		& = \delta_\mu^\rho \, ,
	\end{split}
	\end{equation}
	as requested. Finally, we remark that the Neumann series
	\begin{equation}
		g^{\mu \nu} = \sum_{k = 0}^\infty \left ( - \gcoupling \right )^k \left ( h^k \right )^{\mu \nu}
	\end{equation}
	converges precisely for
	\begin{equation}
		\left | \gcoupling \right | \left \| h \right \|_{\max} := \left | \gcoupling \right | \max_{\lambda \in \operatorname{EW} \left ( h \right )} \left | \lambda \right | < 1 \, ,
	\end{equation}
	where \(\operatorname{EW} \left ( h \right )\) denotes the set of eigenvalues of \(h\), as stated.
\end{proof}

\enter

\begin{lem}[Vielbein and inverse vielbein as series in the graviton field] \label{lem:vielbeins_series}
	Given the metric decomposition from \defnref{defn:md_and_gf} and the boundedness condition from \assref{ass:bdns_gf}, the vielbein and inverse vielbein are given via the series
	\begin{subequations}
		\begin{align}
		e_\mu^m = \sum_{k = 0}^\infty \gcoupling^k \binom{\frac{1}{2}}{k} \left ( h^k \right )_\mu^m \, ,
		\intertext{with \(h_\mu^m := \eta^{m \nu} h_{\mu \nu}\), and}
		e_m^\mu = \sum_{k = 0}^\infty \gcoupling^k \binom{- \frac{1}{2}}{k} \left ( h^k \right )_m^\mu \, ,
		\end{align}
	\end{subequations}
	with \(h^\mu_m := \eta^{\mu \nu} \delta^\rho_m h_{\nu \rho}\).
\end{lem}

\begin{proof}
	We recall the defining equations for vielbeins and inverse vielbeins,
	\begin{align}
		g_{\mu \nu} & = \eta_{m n} e^m_\mu e^n_\nu
		\intertext{and}
		\eta_{m n} & = g_{\mu \nu} e^\mu_m e^\nu_n \, ,
	\end{align}
	cf.\ \cite[Definition 2.8]{Prinz_2}. Thus, we calculate
	\begin{equation}
	\begin{split}
		g_{\mu \nu} & = \eta_{m n} e^m_\mu e^n_\nu \\
		& = \eta_{m n} \left ( \sum_{i = 0}^\infty \gcoupling^i \binom{\frac{1}{2}}{i} \left ( h^i \right )_\mu^m \right ) \left ( \sum_{j = 0}^\infty \gcoupling^j \binom{\frac{1}{2}}{j} \left ( h^j \right )_\nu^n \right ) \\
		& = \sum_{i = 0}^\infty \sum_{j = 0}^\infty \gcoupling^{i + j} \binom{\frac{1}{2}}{i} \binom{\frac{1}{2}}{j} \left ( h^{i + j} \right )_{\mu \nu} \\
		& = \sum_{k = 0}^\infty \gcoupling^k \binom{1}{k} \left ( h^{k} \right )_{\mu \nu} \\
		& = \eta_{\mu \nu} + \gcoupling h_{\mu \nu} \, ,
	\end{split}
	\end{equation}
	where we have used Vandermonde's identity, and
	\begin{equation}
	\begin{split}
		g^{\mu \nu} & = \eta^{m n} e^\mu_m e^\nu_n \\
		& = \eta^{m n} \left ( \sum_{i = 0}^\infty \gcoupling^i \binom{- \frac{1}{2}}{i} \left ( h^i \right )_m^\mu \right ) \left ( \sum_{j = 0}^\infty \gcoupling^j \binom{- \frac{1}{2}}{j} \left ( h^j \right )_n^\nu \right ) \\
		& = \sum_{i = 0}^\infty \sum_{j = 0}^\infty \gcoupling^{i + j} \binom{- \frac{1}{2}}{i} \binom{- \frac{1}{2}}{j} \left ( h^{i + j} \right )^{\mu \nu} \\
		& = \sum_{k = 0}^\infty \gcoupling^k \binom{- 1}{k} \left ( h^k \right )^{\mu \nu} \\
		& = \sum_{k = 0}^\infty \left ( - \gcoupling \right )^k \left ( h^k \right )^{\mu \nu} \\
		& = g^{\mu \nu} \, ,
	\end{split}
	\end{equation}
	where we have again used Vandermonde's identity, the identity \(\binom{-1}{k} = \left ( -1 \right )^k\) and \lemref{lem:inverse_metric_series}. Finally, the series for the vielbein and inverse vielbein field converge precisely for
	\begin{equation}
		\left | \gcoupling \right | \left \| h \right \|_{\max} := \left | \gcoupling \right | \max_{\lambda \in \operatorname{EW} \left ( h \right )} \left | \lambda \right | < 1 \, ,
	\end{equation}
	where \(\operatorname{EW} \left ( h \right )\) denotes the set of eigenvalues of \(h\), as stated.
\end{proof}

\enter

\begin{prop}[Ricci scalar for the Levi-Civita connection, cf.\ \cite{Prinz_4}] \label{prop:ricci_scalar_for_the_levi_civita_connection}
	Using the Levi-Civita connection, the Ricci scalar is given via partial derivatives of the metric and its inverse as follows:
	\begin{equation} \label{eqn:ricci_scalar_metric}
	\begin{split}
		R & = g^{\mu \rho} g^{\nu \sigma} \left ( \partial_\mu \partial_\nu g_{\rho \sigma} - \partial_\mu \partial_\rho g_{\nu \sigma} \right ) \\
		& \hphantom{ = } + g^{\mu \rho} g^{\nu \sigma} g^{\kappa \lambda} \left ( \left ( \partial_\mu g_{\kappa \lambda} \right ) \left ( \partial_\nu g_{\rho \sigma} - \frac{1}{4} \partial_\rho g_{\nu \sigma} \right ) + \left ( \partial_\nu g_{\rho \kappa} \right ) \left ( \frac{3}{4} \partial_\sigma g_{\mu \lambda} - \frac{1}{2} \partial_\mu g_{\sigma \lambda} \right ) \right . \\
		& \hphantom{ = + g^{\mu \rho} g^{\nu \sigma} g^{\kappa \lambda} ( } \left . \vphantom{\left ( \frac{1}{2} \right )} - \left ( \partial_\mu g_{\rho \kappa} \right ) \left ( \partial_\nu g_{\sigma \lambda} \right ) \right )
	\end{split}
	\end{equation}
	Furthermore, we also consider the decomposition
	\begin{subequations}
	\begin{align}
		\begin{split}
		R & \equiv g^{\nu \sigma} \left ( \partial_\mu \Gamma^\mu_{\nu \sigma} - \partial_\nu \Gamma^\mu_{\mu \sigma} + \Gamma^\mu_{\mu \kappa} \Gamma^\kappa_{\nu \sigma} - \Gamma^\mu_{\nu \kappa} \Gamma^\kappa_{\mu \sigma} \right ) \\
		& =: R^{\partial \Gamma} + R^{\Gamma^2}
		\end{split}
		\intertext{with}
		R^{\partial \Gamma} & := g^{\nu \sigma} \left ( \partial_\mu \Gamma^\mu_{\nu \sigma} - \partial_\nu \Gamma^\mu_{\mu \sigma} \right ) \\
		\intertext{and}
		R^{\Gamma^2} & := g^{\nu \sigma} \left ( \Gamma^\mu_{\mu \kappa} \Gamma^\kappa_{\nu \sigma} - \Gamma^\mu_{\nu \kappa} \Gamma^\kappa_{\mu \sigma} \right ) \, .
	\end{align}
	\end{subequations}
	Then we obtain:
	\begin{equation}
	\begin{split}
		R^{\partial \Gamma} & = g^{\mu \rho} g^{\nu \sigma} \left ( \partial_\mu \partial_\nu g_{\rho \sigma} - \partial_\mu \partial_\rho g_{\nu \sigma} \right ) \\
		& \hphantom{ = } + g^{\mu \rho} g^{\nu \sigma} g^{\kappa \lambda} \left ( \left ( \partial_\mu g_{\rho \kappa} \right ) \left ( \frac{1}{2} \partial_\lambda g_{\nu \sigma} - \partial_\nu g_{\lambda \sigma} \right ) + \frac{1}{2} \left ( \partial_\nu g_{\mu \kappa} \right ) \left ( \partial_\sigma g_{\rho \lambda} \right ) \right )
	\end{split}
	\end{equation}
	and
	\begin{equation}
	\begin{split}
		R^{\Gamma^2} & = g^{\mu \rho} g^{\nu \sigma} g^{\kappa \lambda} \left ( \left ( \partial_\kappa g_{\mu \rho} \right ) \left ( \frac{1}{2} \partial_\nu g_{\sigma \lambda} - \frac{1}{4} \partial_\lambda g_{\nu \sigma} \right ) - \left ( \partial_\nu g_{\mu \kappa} \right ) \left ( \frac{1}{2} \partial_\rho g_{\sigma \lambda} - \frac{1}{4} \partial_\sigma g_{\rho \lambda} \right ) \right )
	\end{split}
	\end{equation}
\end{prop}

\begin{proof}
	The claim is verified by the calculations
	\begin{align}
		R & = R^{\partial \Gamma} + R^{\Gamma^2}
	\intertext{with}
	\begin{split} \label{eqn:R-partial-Gamma}
		R^{\partial \Gamma} & = g^{\nu \sigma} \left ( \partial_\mu \Gamma^\mu_{\nu \sigma} - \partial_\nu \Gamma^\mu_{\mu \sigma} \right ) \\
		& = \left ( \partial_\mu g^{\mu \rho} \right ) \left ( \partial_\nu g_{\rho \sigma} - \frac{1}{2} \partial_\rho g_{\nu \sigma} \right ) - \frac{1}{2} \left ( \partial_\nu g^{\mu \rho} \right ) \left ( \partial_\sigma g_{\mu \rho} \right ) \\
		& \hphantom{ = } + g^{\mu \rho} \left ( \partial_\mu \partial_\nu g_{\rho \sigma} - \frac{1}{2} \partial_\mu \partial_\rho g_{\nu \sigma} \right ) - \frac{1}{2} g^{\mu \rho} \left ( \partial_\nu \partial_\sigma g_{\mu \rho} \right ) \\
		& = g^{\mu \rho} g^{\nu \sigma} \left ( \partial_\mu \partial_\nu g_{\rho \sigma} - \partial_\mu \partial_\rho g_{\nu \sigma} \right ) \\
		& \hphantom{ = } + g^{\mu \rho} g^{\nu \sigma} g^{\kappa \lambda} \left ( \left ( \partial_\mu g_{\rho \kappa} \right ) \left ( \frac{1}{2} \partial_\lambda g_{\nu \sigma} - \partial_\nu g_{\lambda \sigma} \right ) + \frac{1}{2} \left ( \partial_\nu g_{\mu \kappa} \right ) \left ( \partial_\sigma g_{\rho \lambda} \right ) \right )
	\end{split}
	\intertext{and}
	\begin{split} \label{eqn:R-Gamma-2}
		R^{\Gamma^2} & = g^{\nu \sigma} \left ( \Gamma^\mu_{\mu \kappa} \Gamma^\kappa_{\nu \sigma} - \Gamma^\mu_{\nu \kappa} \Gamma^\kappa_{\mu \sigma} \right ) \\
		& = g^{\mu \rho} g^{\nu \sigma} g^{\kappa \lambda} \left ( \left ( \partial_\kappa g_{\mu \rho} \right ) \left ( \frac{1}{2} \partial_\nu g_{\sigma \lambda} - \frac{1}{4} \partial_\lambda g_{\nu \sigma} \right ) - \left ( \partial_\nu g_{\mu \kappa} \right ) \left ( \frac{1}{2} \partial_\rho g_{\sigma \lambda} - \frac{1}{4} \partial_\sigma g_{\rho \lambda} \right ) \right ) \, ,
	\end{split}
	\end{align}
	where we have used \(\left ( \partial_\rho g^{\nu \sigma} \right ) g_{\mu \sigma} = - g^{\nu \sigma} \left ( \partial_\rho g_{\mu \sigma} \right )\) in \eqnref{eqn:R-partial-Gamma} twice, which results from
	\begin{equation}
	\begin{split}
		0 & = \nabla^{TM}_\rho \delta_\mu^\nu \\
		& = \partial_\rho \delta_\mu^\nu + \tensor{\Gamma}{^\nu _\rho _\sigma} \delta_\mu^\sigma - \tensor{\Gamma}{^\sigma _\rho _\mu} \delta_\sigma^\nu \\
		& = \partial_\rho \delta_\mu^\nu + \tensor{\Gamma}{^\nu _\rho _\mu} - \tensor{\Gamma}{^\nu _\rho _\mu} \\
		& = \partial_\rho \delta_\mu^\nu \\
		& = \partial_\rho \left ( g_{\mu \sigma} g^{\nu \sigma} \right ) \\
		& = \left ( \partial_\rho g_{\mu \sigma} \right ) g^{\nu \sigma} + g_{\mu \sigma} \left ( \partial_\rho g^{\nu \sigma} \right ) \, .
	\end{split}
	\end{equation}
\end{proof}

\enter

\begin{col} \label{col:ricci_scalar_for_the_levi_civita_connection_restriction}
	Given the situation of \propref{prop:ricci_scalar_for_the_levi_civita_connection}, the grade-\(m\) part in the gravitational coupling constant \(\gcoupling\) of the Ricci scalar \(R\) is given via
	\begin{subequations}
	\begin{align}
		\eval{R^{\partial \Gamma}}_{\order{\gcoupling^0}} & = \eval{R^{\Gamma^2}}_{\order{\gcoupling^0}} = \eval{R^{\Gamma^2}}_{\order{\gcoupling^1}} = 0 \, , \\
		\eval{R^{\partial \Gamma}}_{\order{\gcoupling^1}} & = \gcoupling \eta^{\mu \rho} \eta^{\nu \sigma} \left ( \partial_\mu \partial_\nu h_{\rho \sigma} - \partial_\mu \partial_\rho h_{\nu \sigma} \right ) \\
		\intertext{and for \(m > 1\)}
		\begin{split}
			\eval{R^{\partial \Gamma}}_{\order{\gcoupling^m}} & = - \left ( - \gcoupling \right )^m \sum_{i+j = m-1} \left ( h^i \right )^{\mu \rho} \left ( h^j \right )^{\nu \sigma} \left ( \partial_\mu \partial_\nu h_{\rho \sigma} - \partial_\mu \partial_\rho h_{\nu \sigma} \right )\\
			& \hphantom{ = } + \left ( - \gcoupling \right )^m \sum_{i+j+k = m-2} \left ( h^i \right )^{\mu \rho} \left ( h^j \right )^{\nu \sigma} \left ( h^k \right )^{\kappa \lambda} \left ( \left ( \partial_\mu h_{\rho \kappa} \right ) \left ( \frac{1}{2} \partial_\lambda h_{\nu \sigma} - \partial_\nu h_{\lambda \sigma} \right ) \right . \\ & \hphantom{ = + \left ( - \gcoupling \right )^m \sum_{i+j+k = m-2} \left ( h^i \right )^{\mu \rho} \left ( h^j \right )^{\nu \sigma} \left ( h^k \right )^{\kappa \lambda} ( } \left . + \frac{1}{2} \left ( \partial_\nu h_{\mu \kappa} \right ) \left ( \partial_\sigma h_{\rho \lambda} \right ) \right )
		\end{split}
		\intertext{and}
		\begin{split}
			\eval{R^{\Gamma^2}}_{\order{\gcoupling^m}} & = \left ( - \gcoupling \right )^m \sum_{i+j+k = m-2} \left ( h^i \right )^{\mu \rho} \left ( h^j \right )^{\nu \sigma} \left ( h^k \right )^{\kappa \lambda} \left ( \left ( \partial_\kappa h_{\mu \rho} \right ) \left ( \frac{1}{2} \partial_\nu h_{\sigma \lambda} - \frac{1}{4} \partial_\lambda h_{\nu \sigma} \right ) \right . \\ & \hphantom{ = \left ( - \gcoupling \right )^m \sum_{i+j+k = m-2} \left ( h^i \right )^{\mu \rho} \left ( h^j \right )^{\nu \sigma} \left ( h^k \right )^{\kappa \lambda} ( } \! \! \! \! \! \left . - \left ( \partial_\nu h_{\mu \kappa} \right ) \left ( \frac{1}{2} \partial_\rho h_{\sigma \lambda} - \frac{1}{4} \partial_\sigma h_{\rho \lambda} \right ) \right ) \, .
		\end{split}
	\end{align}
	\end{subequations}
\end{col}

\begin{proof}
	This follows directly from \propref{prop:ricci_scalar_for_the_levi_civita_connection} together with \lemref{lem:inverse_metric_series}.
\end{proof}

\enter

\begin{prop}[Metric expression for the de Donder gauge fixing] \label{prop:metric_expression_for_de_donder_gauge_fixing}
	Given the square of the de Donder gauge fixing,
	\begin{equation}
		\deDonder^2 := g^{\mu \nu} \deDonder_\mu \deDonder_\nu
	\end{equation}
	with \(\deDonder_\mu := g^{\rho \sigma} \Gamma_{\rho \sigma \mu}\), this can be rewritten as
	\begin{equation}
		\deDonder^2 = g^{\mu \rho} g^{\nu \sigma} g^{\kappa \lambda} \left ( \left ( \partial_\nu g_{\sigma \mu} \right ) \left ( \partial_\kappa g_{\lambda \rho} \right ) - \left ( \partial_\nu g_{\sigma \mu} \right ) \left ( \partial_\rho g_{\kappa \lambda} \right ) + \frac{1}{4} \left ( \partial_\mu g_{\nu \sigma} \right ) \left ( \partial_\rho g_{\kappa \lambda} \right ) \right ) \, .
	\end{equation}
	Furthermore, its quadratic part is given by
	\begin{equation}
	\begin{split}
		\deDonder_{(2)}^2 & := \eval{\deDonder^2}_{\order{\gcoupling^2}} \\
		& \phantom{:} \equiv \eta^{\mu \nu} \deDonder^{(1)}_\mu \deDonder^{(1)}_\nu
	\end{split}
	\end{equation}
	with \(\deDonder^{(1)}_\mu := \eta^{\rho \sigma} \Gamma_{\rho \sigma \mu}\), and can be rewritten as
	\begin{equation}
		\deDonder_{(2)}^2 = \eta^{\mu \rho} \eta^{\nu \sigma} \eta^{\kappa \lambda} \left ( \left ( \partial_\nu g_{\sigma \mu} \right ) \left ( \partial_\kappa g_{\lambda \rho} \right ) - \left ( \partial_\nu g_{\sigma \mu} \right ) \left ( \partial_\rho g_{\kappa \lambda} \right ) + \frac{1}{4} \left ( \partial_\mu g_{\nu \sigma} \right ) \left ( \partial_\rho g_{\kappa \lambda} \right ) \right ) \, .
	\end{equation}
\end{prop}

\begin{proof}
	The claim is verified by the calculation
	\begin{equation}
	\begin{split}
		\deDonder^2 & = g^{\mu \nu} \deDonder_\mu \deDonder_\nu \\
		& = \frac{1}{4} g^{\mu \rho} g^{\nu \sigma} g^{\kappa \lambda} \left ( \partial_\nu g_{\sigma \mu} + \partial_\sigma g_{\mu \nu} - \partial_\mu g_{\nu \sigma} \right ) \left ( \partial_\kappa g_{\lambda \rho} + \partial_\lambda g_{\rho \kappa} - \partial_\rho g_{\kappa \lambda} \right ) \\
		& = \frac{1}{4} g^{\mu \rho} g^{\nu \sigma} g^{\kappa \lambda} \left ( \left ( \partial_\nu g_{\sigma \mu} \right ) \left ( \partial_\kappa g_{\lambda \rho} \right ) + \left ( \partial_\nu g_{\sigma \mu} \right ) \left ( \partial_\lambda g_{\rho \kappa} \right ) - \left ( \partial_\nu g_{\sigma \mu} \right ) \left ( \partial_\rho g_{\kappa \lambda} \right ) \right . \\
		& \hphantom{= \frac{1}{4} g^{\mu \rho} g^{\nu \sigma} g^{\kappa \lambda} (} \left . + \left ( \partial_\sigma g_{\mu \nu} \right ) \left ( \partial_\kappa g_{\lambda \rho} \right ) + \left ( \partial_\sigma g_{\mu \nu} \right ) \left ( \partial_\lambda g_{\rho \kappa} \right ) - \left ( \partial_\sigma g_{\mu \nu} \right ) \left ( \partial_\rho g_{\kappa \lambda} \right ) \right . \\
		& \hphantom{= \frac{1}{4} g^{\mu \rho} g^{\nu \sigma} g^{\kappa \lambda} (} \left . - \left ( \partial_\mu g_{\nu \sigma} \right ) \left ( \partial_\kappa g_{\lambda \rho} \right ) - \left ( \partial_\mu g_{\nu \sigma} \right ) \left ( \partial_\lambda g_{\rho \kappa} \right ) + \left ( \partial_\mu g_{\nu \sigma} \right ) \left ( \partial_\rho g_{\kappa \lambda} \right ) \right ) \\
		& = g^{\mu \rho} g^{\nu \sigma} g^{\kappa \lambda} \left ( \left ( \partial_\nu g_{\sigma \mu} \right ) \left ( \partial_\kappa g_{\lambda \rho} \right ) - \left ( \partial_\nu g_{\sigma \mu} \right ) \left ( \partial_\rho g_{\kappa \lambda} \right ) + \frac{1}{4} \left ( \partial_\mu g_{\nu \sigma} \right ) \left ( \partial_\rho g_{\kappa \lambda} \right ) \right ) \, ,
	\end{split}
	\end{equation}
	together with the obvious restriction to \(\order{\gcoupling^2}\).
\end{proof}

\enter

\begin{col} \label{col:metric_expression_for_de_donder_gauge_fixing_restriction}
	Given the situation of \propref{prop:metric_expression_for_de_donder_gauge_fixing}, the grade-\(m\) part in the gravitational coupling constant \(\gcoupling\) of the square of the de Donder gauge fixing \(\deDonder^2\) is given for \(m < 2\) via
	\begin{subequations}
	\begin{align}
		\eval{\deDonder^2}_{\order{\gcoupling^m}} & = 0
		\intertext{and for \(m > 1\) via}
		\begin{split}
		\eval{\deDonder^2}_{\order{\gcoupling^m}} & = \left ( - \gcoupling \right )^m \sum_{i+j+k = m-2} \left ( h^i \right )^{\mu \rho} \left ( h^j \right )^{\nu \sigma} \left ( h^k \right )^{\kappa \lambda} \\ & \hphantom{\sum_{i+j+k = m-2}} \times \left ( \left ( \partial_\nu h_{\sigma \mu} \right ) \left ( \partial_\kappa h_{\lambda \rho} \right ) - \left ( \partial_\nu h_{\sigma \mu} \right ) \left ( \partial_\rho h_{\kappa \lambda} \right ) + \frac{1}{4} \left ( \partial_\mu h_{\nu \sigma} \right ) \left ( \partial_\rho h_{\kappa \lambda} \right ) \right ) \, .
		\end{split}
	\end{align}
	\end{subequations}
	In particular, the quadratic term \(\deDonder_{(2)}^2\) is given by
	\begin{equation}
	\begin{split}
		\deDonder_{(2)}^2 & := \eval{\deDonder^2}_{\order{\gcoupling^2}} \\
	& \phantom{:} = \gcoupling^2 \eta^{\mu \rho} \eta^{\nu \sigma} \eta^{\kappa \lambda} \left ( \left ( \partial_\nu h_{\sigma \mu} \right ) \left ( \partial_\kappa h_{\lambda \rho} \right ) - \left ( \partial_\nu h_{\sigma \mu} \right ) \left ( \partial_\rho h_{\kappa \lambda} \right ) + \frac{1}{4} \left ( \partial_\mu h_{\nu \sigma} \right ) \left ( \partial_\rho h_{\kappa \lambda} \right ) \right ) \, .
	\end{split}
	\end{equation}
\end{col}

\begin{proof}
This follows directly from \propref{prop:metric_expression_for_de_donder_gauge_fixing} together with \lemref{lem:inverse_metric_series}.
\end{proof}

\enter

\begin{prop}[Determinant of the metric as a series in the graviton field] \label{prop:determinant_metric}
	Given the metric decomposition from \defnref{defn:md_and_gf}, the negative of the determinant of the metric, \(- \dt{g}\), is given via
	\begin{equation}
		- \dt{g} = 1 + \first + \second + \third + \fourth \label{eqn:determinant_metric}
	\end{equation}
	with
	{\allowdisplaybreaks
	\begin{subequations} \label{eqns:first-fourth}
	\begin{align}
		\begin{split} \label{eqn:first}
			\first & := \gcoupling \tr{\eta h}\\
			& \hphantom{ : } \equiv \gcoupling \eta^{\mu \nu} h_{\mu \nu} \, ,
		\end{split}
		\\
		\begin{split}
			\second & := \gcoupling^2 \left ( \frac{1}{2} \tr{\eta h}^2 - \frac{1}{2} \tr{\left ( \eta h \right )^2} \right )\\
			& \hphantom{ : } \equiv \gcoupling^2 \left ( \frac{1}{2} \eta^{\mu \nu} \eta^{\rho \sigma} - \frac{1}{2} \eta^{\mu \sigma} \eta^{\rho \nu} \right ) h_{\mu \nu} h_{\rho \sigma} \, ,
		\end{split}
		\\
		\begin{split}
			\third & := \gcoupling^3 \left ( \frac{1}{6} \left ( \tr{\eta h} \right )^3 - \frac{1}{2} \tr{\eta h} \tr{\left ( \eta h \right )^2} + \frac{1}{3} \tr{\left ( \eta h \right )^3} \right )\\
			& \hphantom{ : } \equiv \gcoupling^3 \left ( \frac{1}{6} \eta^{\mu \nu} \eta^{\rho \sigma} \eta^{\lambda \tau} - \frac{1}{2} \eta^{\mu \nu} \eta^{\rho \tau} \eta^{\lambda \sigma} + \frac{1}{3} \eta^{\mu \tau} \eta^{\rho \nu} \eta^{\lambda \sigma} \right ) h_{\mu \nu} h_{\rho \sigma} h_{\lambda \tau}
		\end{split}
		\intertext{and}
		\begin{split} \label{eqn:fourth}
			\fourth & := \gcoupling^4 \left ( \frac{1}{24} \left ( \tr{\eta h} \right )^4 - \frac{1}{4} \left ( \tr{\eta h} \right )^2 \tr{\left ( \eta h \right )^2} + \frac{1}{3} \tr{\eta h} \tr{\left ( \eta h \right )^3} \right .\\ & \hphantom{ := } \left . + \frac{1}{8} \left ( \tr{\left ( \eta h \right )^2} \right )^2 - \frac{1}{4} \tr{\left ( \eta h \right )^4} \right )\\
			& \hphantom{ : } \equiv \gcoupling^4 \left ( \frac{1}{24} \eta^{\mu \nu} \eta^{\rho \sigma} \eta^{\lambda \tau} \eta^{\vartheta \varphi} - \frac{1}{4} \eta^{\mu \nu} \eta^{\rho \sigma} \eta^{\lambda \varphi} \eta^{\vartheta \tau} + \frac{1}{3} \eta^{\mu \nu} \eta^{\rho \varphi} \eta^{\lambda \sigma} \eta^{\vartheta \tau} \right . \\ & \hphantom{ := ( } \left . + \frac{1}{8} \eta^{\mu \sigma} \eta^{\rho \nu} \eta^{\lambda \varphi} \eta^{\vartheta \tau} - \frac{1}{4} \eta^{\mu \varphi} \eta^{\rho \nu} \eta^{\lambda \sigma} \eta^{\vartheta \tau} \right ) h_{\mu \nu} h_{\rho \sigma} h_{\lambda \tau} h_{\vartheta \varphi} \, .
		\end{split}
	\end{align}
	\end{subequations}
	}%
	\end{prop}

\begin{proof}
	Given a \(4 \times 4\)-matrix \(\mtx \in \operatorname{Mat}_\mathbb{C} \left ( 4 \times 4 \right )\), from Newton's identities we get the relation
	\begin{equation} \label{eqn:Newtons_identities}
	\begin{split}
		\dt{\mtx} & = \frac{1}{4!} \operatorname{Det} \begin{pmatrix} \tr{\mtx} & 1 & 0 & 0 \\ \tr{\mtx^2} & \tr{\mtx} & 2 & 0 \\ \tr{\mtx^3} & \tr{\mtx^2} & \tr{\mtx} & 3 \\ \tr{\mtx^4} & \tr{\mtx^3} & \tr{\mtx^2} & \tr{\mtx} \end{pmatrix} \\
		& = \frac{1}{4!} \left ( \tr{\mtx}^4 - 6 \tr{\mtx}^2 \tr{\mtx^2} + 8 \tr{\mtx} \tr{\mtx^3} \right . \\
		& \hphantom{= \frac{1}{4!} (} \left . + 3 \tr{\mtx^2}^2 - 6 \tr{\mtx^4} \right ) \, .
	\end{split}
	\end{equation}
	Next, using the metric decomposition \(g = \eta + \gcoupling h\), we obtain\footnote{In accordance with index-notation, we set \(\delta\) to be the unit matrix.}
	\begin{equation}
	\begin{split}
		- \dt{g} & = - \dt{\eta + \gcoupling h}\\
		& = - \dt{\eta} \dt{\delta + \gcoupling \eta^{-1} h}\\
		& = \dt{\delta + \gcoupling \eta h} \, ,		
	\end{split}
	\end{equation}
	where we have used \(\dt{\eta} = - 1\) and \(\eta^{-1} = \eta\). Setting \(\mtx := \delta + \gcoupling \eta h\), using the linearity and cyclicity of the trace and the fact that \(\tr{\delta} = 4\), we get
	\begin{align}
		\tr{\delta + \gcoupling \eta h} & = 4 + \gcoupling \tr{\eta h}\\
		\tr{\left ( \delta + \gcoupling \eta h \right )^2} & = 4 + 2 \gcoupling \tr{\eta h} + \gcoupling^2 \tr{\left ( \eta h \right )^2}\\
		\tr{\left ( \delta + \gcoupling \eta h \right )^3} & = 4 + 3 \gcoupling \tr{\eta h} + 3 \gcoupling^2 \tr{\left ( \eta h \right )^2} + \gcoupling^3 \tr{\left ( \eta h \right )^3}\\
		\tr{\left ( \delta + \gcoupling \eta h \right )^4} & = 4 + 4 \gcoupling \tr{\eta h} + 6 \gcoupling^2 \tr{\left ( \eta h \right )^2} + 4 \gcoupling^3 \tr{\left ( \eta h \right )^3} + \gcoupling^4 \tr{\left ( \eta h \right )^4} \, .
	\end{align}
	Combining these results, we obtain
	\begin{equation}
	\begin{split}
		- \dt{g} & = 1 + \gcoupling \tr{\eta h} + \gcoupling^2 \left ( \frac{1}{2} \tr{\eta h}^2 - \frac{1}{2} \tr{\left ( \eta h \right )^2} \right ) \\
		& \hphantom{ = } + \gcoupling^3 \left ( \frac{1}{6} \tr{\eta h}^3 - \frac{1}{2} \tr{\eta h} \tr{\left ( \eta h \right )^2} + \frac{1}{3} \tr{\left ( \eta h \right )^3} \right ) \\
		& \hphantom{ = } + \gcoupling^4 \left ( \frac{1}{24} \tr{\eta h}^4 - \frac{1}{4} \tr{\eta h}^2 \tr{\eta h^2} + \frac{1}{3} \tr{\eta h} \tr{\left ( \eta h \right )^3} \right . \\ & \hphantom{ =  + \gcoupling^4 (} \left . + \frac{1}{8} \tr{\left ( \eta h \right )^2}^2 - \frac{1}{4} \tr{\left ( \eta h \right )^4} \right ) \, ,
	\end{split}
	\end{equation}
	which, when restricting to the powers in the coupling constant, yields the claimed result.
\end{proof}

\enter

\begin{col} \label{col:determinant_metric_restriction}
	Given the situation of \propref{prop:determinant_metric} and assume furthermore the boundedness condition from \assref{ass:bdns_gf}, the grade-\(m\) part in the gravitational coupling constant \(\gcoupling\) of the square-root of the negative of the determinant of the metric, \(\sqrt{- \dt{g}}\), is given via
	\begin{equation}
	\begin{split}
		\eval{\sqrt{- \dt{g}}}_{\order{\gcoupling^m}} & = \sum_{\substack{i + j + k + l = m\\i \geq j \geq k \geq l \geq 0}} \sum_{p = 0}^{j - k} \sum_{q = 0}^{k - l} \sum_{r = 0}^{q} \sum_{s = 0}^l \sum_{t = 0}^s \sum_{u = 0}^t \sum_{v = 0}^u \\
		& \hphantom{ = } \binom{\frac{1}{2}}{i} \binom{i}{j} \binom{j}{k} \binom{k}{l} \binom{j - k}{p} \binom{k - l}{q} \binom{q}{r} \binom{l}{s} \binom{s}{t} \binom{t}{u} \binom{u}{v} \\
		& \hphantom{ = } \times \left ( - 1 \right )^{p + q - r + s - t + v} 2^{- j + l + r + s + 2t - 3u + v} 3^{- k + q - r + s - t + u} \\
		& \hphantom{ = } \times \mathfrak{a}^{i + j + k + l - 2p - 2q - r - 2s - t - u} \mathfrak{b}^{p + q - r + s - t + 2u - 2v} \mathfrak{c}^{r + t - u} \mathfrak{d}^v
	\end{split}
	\end{equation}
	with
	\begin{subequations} \label{eqns:varfirst-varfourth}
	\begin{align}
		\begin{split}
			\mathfrak{a} & := \gcoupling \tr{\eta h} \\
			& \hphantom{ : } \equiv \gcoupling \eta^{\mu \nu} h_{\mu \nu} \, ,
		\end{split}
		\\
		\begin{split}
			\mathfrak{b} & := \gcoupling^2 \tr{\left ( \eta h \right )^2} \\
			& \hphantom{ : } \equiv \gcoupling^2 \eta^{\mu \sigma} \eta^{\rho \nu} h_{\mu \nu} h_{\rho \sigma} \, ,
		\end{split}
		\\
		\begin{split}
			\mathfrak{c} & := \gcoupling^3 \tr{\left ( \eta h \right )^3} \\
			& \hphantom{ : } \equiv \gcoupling^3 \eta^{\mu \tau} \eta^{\rho \nu} \eta^{\lambda \sigma} h_{\mu \nu} h_{\rho \sigma} h_{\lambda \tau}
		\end{split}
		\intertext{and}
		\begin{split}
			\mathfrak{d} & := \gcoupling^4 \tr{\left ( \eta h \right )^4} \\
			& \hphantom{ : } \equiv \gcoupling^4 \eta^{\mu \varphi} \eta^{\rho \nu} \eta^{\lambda \sigma} \eta^{\vartheta \tau} h_{\mu \nu} h_{\rho \sigma} h_{\lambda \tau} h_{\vartheta \varphi} \, .
		\end{split}
	\end{align}
	\end{subequations}
\end{col}

\begin{proof}
	We use \eqnref{eqn:determinant_metric},
	\begin{equation}
		- \dt{g} = 1 + \first + \second + \third + \fourth \, ,
	\end{equation}
	and plug it into the Taylor series of the square-root around \(x = 0\),\footnote{Here we need the assumption \(\left | \gcoupling \right | \left \| h \right \|_{\max} := \left | \gcoupling \right | \max_{\lambda \in \operatorname{EW} \left ( h \right )} \left | \lambda \right | < 1\), where \(\operatorname{EW} \left ( h \right )\) denotes the set of eigenvalues of \(h\), to assure convergence.}
	\begin{equation}
		\sqrt{x} = \sum_{i = 0}^\infty \binom{\frac{1}{2}}{i} \left ( x - 1 \right )^i \, ,
	\end{equation}
	to obtain
	\begin{equation}
		\sqrt{- \dt{g}} = \sum_{i = 0}^\infty \binom{\frac{1}{2}}{i} \left ( \first + \second + \third + \fourth \right )^i \, .
	\end{equation}
	Applying the binomial theorem iteratively three times, we get
	\begin{equation}
	\begin{split}
		\sqrt{- \dt{g}} & = \sum_{i = 0}^\infty \binom{\frac{1}{2}}{i} \left ( \first + \second + \third + \fourth \right )^i\\
		& = \sum_{i = 0}^\infty \sum_{j = 0}^i \binom{\frac{1}{2}}{i} \binom{i}{j} \first^{i - j} \left ( \second + \third + \fourth \right )^j\\
		& = \sum_{i = 0}^\infty \sum_{j = 0}^i \sum_{k = 0}^j \binom{\frac{1}{2}}{i} \binom{i}{j} \binom{j}{k} \first^{i - j} \second^{j - k} \left ( \third + \fourth \right )^k\\
		& = \sum_{i = 0}^\infty \sum_{j = 0}^i \sum_{k = 0}^j \sum_{l = 0}^k \binom{\frac{1}{2}}{i} \binom{i}{j} \binom{j}{k} \binom{k}{l} \first^{i - j} \second^{j - k} \third^{k - l} \fourth^l \, .
	\end{split}
	\end{equation}
	Observe, that from \eqnsaref{eqn:determinant_metric}{eqns:first-fourth} we have the relations
	\begin{subequations}
	\begin{align}
		\eval{- \dt{g}}_{\order{\gcoupling}} & \equiv \first \\
		\eval{- \dt{g}}_{\order{\gcoupling^2}} & \equiv \second \\
		\eval{- \dt{g}}_{\order{\gcoupling^3}} & \equiv \third
		\intertext{and}
		\eval{- \dt{g}}_{\order{\gcoupling^4}} & \equiv \fourth \, ,
	\end{align}
	\end{subequations}
	and thus the restriction to the grade-\(m\) part in the gravitational coupling constant \(\gcoupling\) is given via the integer solutions to
	\begin{equation}
	\begin{split}
		m & \overset{!}{=} i - j + 2 j - 2 k + 3 k - 3 l + 4 l\\
		& = i + j + k + l
	\end{split}
	\end{equation}
	with \(i \geq j \geq k \geq l\), i.e.\
	\begin{equation} \label{eqn:riemannean_volume_form_grade_m}
		\eval{\sqrt{- \dt{g}}}_{\order{\gcoupling^m}} = \sum_{\substack{i + j + k + l = m\\i \geq j \geq k \geq l \geq 0}} \binom{\frac{1}{2}}{i} \binom{i}{j} \binom{j}{k} \binom{k}{l} \first^{i - j} \second^{j - k} \third^{k - l} \fourth^l \, .
	\end{equation}
	Finally, using Newton's identities, i.e.\ the relations from \eqnsaref{eqns:first-fourth}{eqns:varfirst-varfourth},
	\begin{subequations}
	\begin{align}
		\first & \equiv \mathfrak{a} \, , \\
		\second & \equiv \frac{1}{2} \mathfrak{a}^2 - \frac{1}{2} \mathfrak{b} \, , \\
		\third & \equiv \frac{1}{6} \mathfrak{a}^3 - \frac{1}{2} \mathfrak{a} \mathfrak{b} + \frac{1}{3} \mathfrak{c}
		\intertext{and}
		\fourth & \equiv \frac{1}{24} \mathfrak{a}^4 - \frac{1}{4} \mathfrak{a}^2 \mathfrak{b} + \frac{1}{3} \mathfrak{a} \mathfrak{c} + \frac{1}{8} \mathfrak{b}^2 - \frac{1}{4} \mathfrak{d} \, ,
	\end{align}
	\end{subequations}
	we obtain, using again the Binomial theorem iteratively seven times,
	\begin{equation}
	\begin{split}
		\first^{i - j} \second^{j - k} \third^{k - l} \fourth^l & = \sum_{p = 0}^{j - k} \sum_{q = 0}^{k - l} \sum_{r = 0}^{q} \sum_{s = 0}^l \sum_{t = 0}^s \sum_{u = 0}^t \sum_{v = 0}^u\\
		& \hphantom{ = } \binom{j - k}{p} \binom{k - l}{q} \binom{q}{r} \binom{l}{s} \binom{s}{t} \binom{t}{u} \binom{u}{v}\\
		& \hphantom{ = } \times \left ( - 1 \right )^{p + q - r + s - t + v} 2^{- j + l + r + s + 2t - 3u + v} 3^{- k + q - r + s - t + u} \\
		& \hphantom{ = } \times \mathfrak{a}^{i + j + k + l - 2p - 2q - r - 2s - t - u} \mathfrak{b}^{p + q - r + s - t + 2u - 2v} \mathfrak{c}^{r + t - u} \mathfrak{d}^v \, ,
	\end{split}
	\end{equation}
	and thus finally
	\begin{equation}
	\begin{split}
		\eval{\sqrt{- \dt{g}}}_{\order{\gcoupling^m}} & = \sum_{\substack{i + j + k + l = m\\i \geq j \geq k \geq l \geq 0}} \sum_{p = 0}^{j - k} \sum_{q = 0}^{k - l} \sum_{r = 0}^{q} \sum_{s = 0}^l \sum_{t = 0}^s \sum_{u = 0}^t \sum_{v = 0}^u \\
		& \hphantom{ = } \binom{\frac{1}{2}}{i} \binom{i}{j} \binom{j}{k} \binom{k}{l} \binom{j - k}{p} \binom{k - l}{q} \binom{q}{r} \binom{l}{s} \binom{s}{t} \binom{t}{u} \binom{u}{v} \\
		& \hphantom{ = } \times \left ( - 1 \right )^{p + q - r + s - t + v} 2^{- j + l + r + s + 2t - 3u + v} 3^{- k + q - r + s - t + u} \\
		& \hphantom{ = } \times \mathfrak{a}^{i + j + k + l - 2p - 2q - r - 2s - t - u} \mathfrak{b}^{p + q - r + s - t + 2u - 2v} \mathfrak{c}^{r + t - u} \mathfrak{d}^v \, ,
	\end{split}
	\end{equation}
	as claimed.
\end{proof}

\subsection{Notation and building blocks} \label{ssec:preparations_graviton_graviton_ghost}

Given the Quantum General Relativity Lagrange density
\begin{equation}
	\begin{split}
	\mathcal{L}_\text{QGR} & = - \frac{1}{2 \gcoupling^2} \left ( \sqrt{- \dt{g}} R + \frac{1}{2 \zeta}  \eta^{\mu \nu} \deDonder^{(1)}_\mu \deDonder^{(1)}_\nu \right ) \dif V_\eta \\
	& \phantom{:=} - \frac{1}{2} \eta^{\rho \sigma} \left ( \frac{1}{\zeta} \overline{C}^\mu \left ( \partial_\rho \partial_\sigma C_\mu \right ) + \overline{C}^\mu \left ( \partial_\mu \big ( \tensor{\Gamma}{^\nu _\rho _\sigma} C_\nu \big ) - 2 \partial_\rho \big ( \tensor{\Gamma}{^\nu _\mu _\sigma} C_\nu \big ) \right ) \right ) \dif V_\eta
	\end{split}
\end{equation}
from \conref{con:Lagrange_density} and the decomposition into its powers in the gravitational coupling constant \(\varkappa\) and the ghost field \(C\)
\begin{equation}
	\mathcal{L}_\text{QGR} \equiv \sum_{m = 0}^\infty \sum_{n = 0}^1 \mathcal{L}_\text{QGR}^{m,n}
\end{equation}
from the introduction of \sectionref{sec:expansion_lagrange_density}. Then, we extend the Lagrange densities \(\mathcal{L}_\text{QGR}^{m,n}\) for given \(m \in \mathbb{N}_+\), which were interpreted in the introduction of \sectionref{sec:expansion_lagrange_density} as potential terms for either \(\left ( m + 2 \right )\) gravitons or \(m\) gravitons and a graviton-ghost and graviton-antighost, to either \(\left ( m + 2 \right )\) distinguishable gravitons or \(m\) distinguishable gravitons and a graviton-ghost and graviton-antighost via symmetrization, depending on \(n \in \set{0,1}\). This then reflects the bosonic character of gravitons and allows the calculation of the corresponding Feynman rules as the remaining matrix elements of these potential terms. We start by introducing the notation and then present the Feynman rules.

\enter

\begin{defn} \label{defn:notation_qgr-fr}
	We denote the graviton \(m\)-point vertex Feynman rule with ingoing momenta \(\set{p_1^\sigma, \cdots, p_m^\sigma}\) via \(\gravfr_m^{\mu_1 \nu_1 \vert \cdots \vert \mu_m \nu_m} \left ( p_1^\sigma, \cdots, p_m^\sigma \right )\).\footnote{The vertical bars in \(\mu_1 \nu_1 \vert \cdots \vert \mu_m \nu_m\) are added solely for better readability.} It is defined as follows:
	\begin{equation}
		\gravfr_m^{\mu_1 \nu_1 \vert \cdots \vert \mu_m \nu_m} \left ( p_1^\sigma, \cdots, p_m^\sigma \right ) := \imaginary \left ( \prod_{i = 1}^m \frac{\bar{\delta}}{\bar{\delta} \hat{h}_{\mu_i \nu_i}} \right ) \mathscr{F} \left ( \overline{\mathcal{L}}_\text{QGR}^{(m-2),0} \right ) \, ,
	\end{equation}
	where the prefactor \(\imaginary\) is a convention from the path integral, \(\textfrac{\bar{\delta}}{\bar{\delta} \hat{h}_{\mu_i \nu_i}}\) denotes the symmetrized functional derivative with respect to the Fourier transformed graviton field \(\hat{h}_{\mu_i \nu_i}\) together with the additional agreement (represented by the bar \(\textfrac{\bar{\delta}}{\bar{\delta} \cdot}\)) that the possible preceding momentum is also labelled by the particle number \(i\), e.g.\
	\begin{equation}
		\frac{\bar{\delta}}{\bar{\delta} \hat{h}_{\mu_i \nu_i}} \left ( p_\kappa \hat{h}_{\rho \sigma} \right ) := \frac{1}{2} p_\kappa^i \left ( \hat{\delta}_\rho^{\mu_i} \hat{\delta}_\sigma^{\nu_i} + \hat{\delta}_\sigma^{\mu_i} \hat{\delta}_\rho^{\nu_i} \right ) \, ,
	\end{equation}
	and \(\overline{\mathcal{L}}_\text{QGR}^{(m-2),0}\) is the symmetrized extension of \(\mathcal{L}_\text{QGR}^{(m-2),0}\) to \(m\) distinguishable gravitons. Furthermore, we denote the graviton propagator Feynman rule with momentum \(p^\sigma\), gauge parameter \(\zeta\) and regulator for Landau singularities \(\epsilon\) via \(\gravprop_{\mu_1 \nu_1 \vert \mu_2 \nu_2} \left ( p^\sigma; \zeta; \epsilon \right )\). It is defined as the inverse of the matrix element for the graviton kinetic term:\footnote{We use momentum conservation to set \(p_1^\sigma := p^\sigma\) and \(p_2^\sigma := - p^\sigma\) in the expression \(\gravfr_2^{\mu_2 \nu_2 \vert \mu_3 \nu_3} \left ( p_1^\sigma, p_2^\sigma; \zeta \right )\).}
	\begin{equation}
		\gravprop_{\mu_1 \nu_1 \vert \mu_2 \nu_2} \left ( p^\sigma; \zeta; 0 \right ) \gravfr_2^{\mu_2 \nu_2 \vert \mu_3 \nu_3} \left ( p^\sigma; \zeta \right ) = \frac{1}{2} \left ( \hat{\delta}_{\mu_1}^{\mu_3} \hat{\delta}_{\nu_1}^{\nu_3} + \hat{\delta}_{\mu_1}^{\nu_3} \hat{\delta}_{\nu_1}^{\mu_3} \right ) \, ,
	\end{equation}
	where each tuple \(\mu_i \nu_i\) is treated as one index, which excludes the a priori possible term \(\hat{\eta}_{\mu_1 \nu_1} \hat{\eta}^{\mu_3 \nu_3}\) on the right-hand side. Moreover, we denote the graviton-ghost \(m\)-point vertex Feynman rule with ingoing momenta \(\set{p_1^\sigma, \cdots, p_m^\sigma}\) via \(\gravghostfr_m^{\rho_1 \vert \rho_2 \| \mu_3 \nu_3 \vert \cdots \vert \mu_m \nu_m} \left ( p_1^\sigma, \cdots, p_m^\sigma \right )\), where particle 1 is the graviton-ghost, particle 2 is the graviton-antighost and the rest are gravitons. It is defined as follows:
	\begin{equation}
		\gravghostfr_m^{\rho_1 \vert \rho_2 \| \mu_3 \nu_3 \vert \cdots \vert \mu_m \nu_m} \left ( p_1^\sigma, \cdots, p_m^\sigma \right ) := \imaginary \left ( \frac{\bar{\delta}}{\bar{\delta} \widehat{\gravitonghost}_{\rho_1}} \frac{\bar{\delta}}{\bar{\delta} \widehat{\overline{\gravitonghost}}_{\rho_2}} \prod_{i = 3}^m \frac{\bar{\delta}}{\bar{\delta} \hat{h}_{\mu_i \nu_i}} \right ) \mathscr{F} \left ( \overline{\mathcal{L}}_\text{QGR}^{m,1} \right ) \, ,
	\end{equation}
where, additionally to the above mentioned setting, \(\textfrac{\bar{\delta}}{\bar{\delta} \widehat{\gravitonghost}_{\rho_1}}\) and \(\textfrac{\bar{\delta}}{\bar{\delta} \widehat{\overline{\gravitonghost}}_{\rho_2}}\) denotes the functional derivative with respect to the Fourier transformed graviton-ghost field \(\widehat{\gravitonghost}_{\rho_1}\) and Fourier transformed graviton-antighost field \(\widehat{\overline{\gravitonghost}}_{\rho_2}\), respectively, and \(\overline{\mathcal{L}}_\text{QGR}^{m,1}\) is the symmetrized extension of \(\mathcal{L}_\text{QGR}^{m,1}\) to \(m\) distinguishable gravitons. Additionally, we denote the graviton-ghost propagator Feynman rule with momentum \(p^\sigma\) and regulator for Landau singularities \(\epsilon\) via \(\gravghostprop_{\rho_1 \vert \rho_2} \left ( p^\sigma; \epsilon \right )\). It is defined as the inverse of the matrix element for the graviton-ghost kinetic term:\footnote{Again, we use momentum conservation to set \(p_1^\sigma := p^\sigma\) and \(p_2^\sigma := - p^\sigma\) in the expression \(\gravghostfr_2^{\mu_2 \nu_2 \vert \mu_3 \nu_3} \left ( p_1^\sigma, p_2^\sigma \right )\).}
	\begin{equation}
		\gravghostprop_{\rho_1 \vert \rho_2} \left ( p^\sigma; 0 \right ) \gravghostfr_2^{\rho_2 \vert \rho_3} \left ( p^\sigma \right ) = \hat{\delta}_{\rho_1}^{\rho_3} \, .
	\end{equation}
	Finally, we denote the graviton-matter \(m\)-point vertex Feynman rule of type \(k\) from \lemref{lem:matter-model-Lagrange-densities} with ingoing momenta \(\set{p_1^\sigma, \cdots, p_m^\sigma}\) via \(\matterfrk_m^{\kappa \dots \tau \| o \dots t \triplevert \mu_1 \nu_1 \vert \cdots \vert \mu_m \nu_m} \left ( p_1^\sigma, \cdots, p_m^\sigma \right )\), where we count only graviton particles, as the matter-contributions are condensed into the tensors \(\tensor[_k]{\! T}{}\), whose Feynman rule contributions can be found e.g.\ in \cite{Romao_Silva}. They are defined as follows:
	\begin{multline}
		\matterfrk_m^{\kappa \dots \tau \| o \dots t \triplevert \mu_1 \nu_1 \vert \cdots \vert \mu_m \nu_m} \left ( p_1^\sigma, \cdots, p_m^\sigma \right ) := \\ \imaginary \left ( \frac{\bar{\delta}}{\bar{\delta} \tensor[_k]{\! \widehat{T}}{_\kappa _\dots _\tau _\| _o _\dots _t}} \prod_{i = 1}^m \frac{\bar{\delta}}{\bar{\delta} \hat{h}_{\mu_i \nu_i}} \right ) \mathscr{F} \left ( \tensor[_k]{{\overline{\mathcal{L}}^{m,0}_{\text{QGR-SM}}}}{} \right ) \, ,
	\end{multline}
	where we use again the above mentioned setting.
\end{defn}

\enter

\begin{con}
	We consider all momenta \(\set{p_1^\sigma, \cdots, p_m^\sigma}\) incoming and we assume momentum conservation on quadratic Feynman rules, i.e.\ set \(p_1^\sigma := p^\sigma\) and \(p_2^\sigma := - p^\sigma\).
\end{con}

\enter

\begin{lem} \label{lem:traces_FR}
	Introducing the notation
	\begin{equation}
	\mathfrak{T}_n^{\mu_1 \nu_1 \vert \cdots \vert \mu_n \nu_n} := \left ( \prod_{i = 1}^n \frac{\bar{\delta}}{\bar{\delta} \hat{h}_{\mu_i \nu_i}} \right ) \mathscr{F} \left ( \tr{\left ( \eta h \right )^n} \right ) \, ,
	\end{equation}
	we obtain
	\begin{subequations}
	\begin{align}
		\mathfrak{T}_n^{\mu_1 \nu_1 \vert \cdots \vert \mu_n \nu_n} & = \frac{1}{2^n} \sum_{\mu_i \leftrightarrow \nu_i} \sum_{s \in S_n} \mathfrak{t}_n^{\mu_{s(1)} \nu_{s(1)} \vert \cdots \vert \mu_{s(n)} \nu_{s(n)}}
		\intertext{with}
		\mathfrak{t}_n^{\mu_1 \nu_1 \vert \cdots \vert \mu_n \nu_n} & = \gcoupling^n \left ( \hat{\delta}^{\nu_1}_{\nu_{n+1}} \prod_{a = 1}^n \hat{\eta}^{\mu_a \nu_{a+1}} \right ) \, .
	\end{align}
	\end{subequations}
	Furthermore, introducing the notation
	\begin{equation}
		\mathfrak{H}_n^{\mu \nu \triplevert \mu_1 \nu_1 \vert \cdots \vert \mu_n \nu_n} := \left ( \prod_{i = 1}^n \frac{\bar{\delta}}{\bar{\delta} \hat{h}_{\mu_i \nu_i}} \right ) \mathscr{F} \left ( \left ( h^n \right )^{\mu \nu} \right ) \, ,
	\end{equation}
	we obtain
	\begin{subequations}
	\begin{align}
		\mathfrak{H}_0^{\mu \nu} & = \eta^{\mu \nu}
		\intertext{and for \(n > 0\)}
		\mathfrak{H}_n^{\mu \nu \triplevert \mu_1 \nu_1 \vert \cdots \vert \mu_n \nu_n} & = \frac{1}{2^n} \sum_{\mu_i \leftrightarrow \nu_i} \sum_{s \in S_n} \mathfrak{h}_n^{\mu \nu \vert \mu_{s(1)} \nu_{s(1)} \vert \cdots \vert \mu_{s(n)} \nu_{s(n)}}
		\intertext{with}
		\mathfrak{h}_n^{\mu \nu \triplevert \mu_1 \nu_1 \vert \cdots \vert \mu_n \nu_n} & = \gcoupling^n \left ( \hat{\delta}^{\mu}_{\mu_0} \hat{\delta}^{\nu}_{\nu_{n+1}} \prod_{a = 0}^n \hat{\eta}^{\mu_a \nu_{a+1}} \right ) \, .
	\end{align}
	\end{subequations}
	Moreover, introducing the notation
	\begin{equation}
		\left ( \mathfrak{H}_n^\prime \right )_\rho^{\mu \nu \triplevert \mu_1 \nu_1 \vert \cdots \vert \mu_n \nu_n} \left ( p_1^\sigma, \cdots, p_n^\sigma \right ) := \left ( \prod_{i = 1}^n \frac{\bar{\delta}}{\bar{\delta} \hat{h}_{\mu_i \nu_i}} \right ) \mathscr{F} \left ( \partial_\rho \left ( \left ( h^n \right )^{\mu \nu} \right ) \right ) \, ,
	\end{equation}
	we obtain
	\begin{subequations}
	\begin{align}
		\left ( \mathfrak{H}_0^\prime \right )_\rho^{\mu \nu} & = 0 \\
		\intertext{and for \(n > 0\)}
		\begin{split}
			\left ( \mathfrak{H}_n^\prime \right )_\rho^{\mu \nu \triplevert \mu_1 \nu_1 \vert \cdots \vert \mu_n \nu_n} \left ( p_1^\sigma, \cdots, p_n^\sigma \right ) & = \\ & \hphantom{=} \mkern-44mu \frac{1}{2^n} \sum_{\mu_i \leftrightarrow \nu_i} \sum_{s \in S_n} \left ( \mathfrak{h}_n^\prime \right )_\rho^{\mu \nu \triplevert \mu_{s(1)} \nu_{s(1)} \vert \cdots \vert \mu_{s(n)} \nu_{s(n)}} \left ( p_{s(1)}^\sigma, \cdots, p_{s(n)}^\sigma \right )
		\end{split}
		\intertext{with}
		\left ( \mathfrak{h}_n^\prime \right )_\rho^{\mu \nu \triplevert \mu_1 \nu_1 \vert \cdots \vert \mu_n \nu_n} \left ( p_1^\sigma, \cdots, p_n^\sigma \right ) & = \gcoupling^n \left ( \sum_{m = 1}^n p^n_\rho \right ) \left ( \hat{\delta}^{\mu}_{\mu_0} \hat{\delta}^{\nu}_{\nu_{n+1}} \prod_{a = 0}^n \hat{\eta}^{\mu_a \nu_{a+1}} \right ) \, .
	\end{align}
	\end{subequations}
\end{lem}

\begin{proof}
	This follows from directly from the definition.
\end{proof}

\enter

\begin{col} \label{col:inverse_metric_vielbeins_FR}
	Given the situation of \lemref{lem:traces_FR}, we have
	{\allowdisplaybreaks
	\begin{align}
		\left ( \prod_{i = 1}^n \frac{\bar{\delta}}{\bar{\delta} \hat{h}_{\mu_i \nu_i}} \right ) \mathscr{F} \left ( \eval{g^{\mu \nu}}_{\order{\gcoupling^n}} \right ) & = \left ( - 1 \right )^n \mathfrak{H}_n^{\mu \nu \triplevert \mu_1 \nu_1 \vert \cdots \vert \mu_n \nu_n} \, ,
		\\
		\left ( \prod_{i = 1}^n \frac{\bar{\delta}}{\bar{\delta} \hat{h}_{\mu_i \nu_i}} \right ) \mathscr{F} \left ( \eval{e_\rho^r}_{\order{\gcoupling^n}} \right ) & = \binom{\frac{1}{2}}{n} \left ( \mathfrak{H}_n \right )_\rho^{r \triplevert \mu_1 \nu_1 \vert \cdots \vert \mu_n \nu_n} \, ,
		\\
		\left ( \prod_{i = 1}^n \frac{\bar{\delta}}{\bar{\delta} \hat{h}_{\mu_i \nu_i}} \right ) \mathscr{F} \left ( \eval{e^\rho_r}_{\order{\gcoupling^n}} \right ) & = \binom{- \frac{1}{2}}{n} \left ( \mathfrak{H}_n \right )_r^{\rho \triplevert \mu_1 \nu_1 \vert \cdots \vert \mu_n \nu_n} \, ,
		\\
		\left ( \prod_{i = 1}^n \frac{\bar{\delta}}{\bar{\delta} \hat{h}_{\mu_i \nu_i}} \right ) \mathscr{F} \left ( \eval{\left ( \partial_\sigma e_\rho^r \right )}_{\order{\gcoupling^n}} \right ) & = \binom{\frac{1}{2}}{n} \hat{\eta}_{\mu \rho} \hat{\delta}_\nu^r \left ( \mathfrak{H}_n^\prime \right )_\sigma^{\mu \nu \triplevert \mu_1 \nu_1 \vert \cdots \vert \mu_n \nu_n} \left ( p_1^\sigma, \cdots, p_n^\sigma \right )
		\intertext{and}
		\left ( \prod_{i = 1}^n \frac{\bar{\delta}}{\bar{\delta} \hat{h}_{\mu_i \nu_i}} \right ) \mathscr{F} \left ( \eval{\left ( \partial_\sigma e^\rho_r \right )}_{\order{\gcoupling^n}} \right ) & = \binom{- \frac{1}{2}}{n} \hat{\delta}_\mu^\rho \hat{\eta}_{\nu r} \left ( \mathfrak{H}_n^\prime \right )_\sigma^{\mu \nu \triplevert \mu_1 \nu_1 \vert \cdots \vert \mu_n \nu_n} \left ( p_1^\sigma, \cdots, p_n^\sigma \right ) \, .
	\end{align}
}%
\end{col}

\begin{proof}
	This follows directly from Lemmata \ref{lem:inverse_metric_series}, \ref{lem:vielbeins_series} and \ref{lem:traces_FR}.
\end{proof}

\enter

\begin{lem} \label{lem:Christoffel_FR}
	Introducing the notation
	\begin{equation}
		\boldsymbol{\Gamma}^{\mu_1 \nu_1}_{\mu \nu \rho} \left ( p_1^\sigma \right ) := \frac{\bar{\delta}}{\bar{\delta} \hat{h}_{\mu_1 \nu_1}} \mathscr{F} \left ( \Gamma_{\mu \nu \rho} \right )
	\end{equation}
	with
	\begin{equation}
	\begin{split}
		\Gamma_{\mu \nu \rho} & := g_{\rho \sigma} \tensor{\Gamma}{^\sigma _\mu _\nu} \\
		& \hphantom{:} \equiv \frac{1}{2} \left ( \partial_\mu g_{\nu \rho} + \partial_\nu g_{\rho \mu} - \partial_\rho g_{\mu \nu} \right )
	\end{split}
	\end{equation}
	we obtain
	\begin{equation}
	\begin{split}
		\boldsymbol{\Gamma}^{\mu_1 \nu_1}_{\mu \nu \rho} \left ( p_1^\sigma \right ) & = \frac{\gcoupling}{4} \left ( p^1_\mu \left ( \hat{\delta}_\rho^{\mu_1} \hat{\delta}_\nu^{\nu_1} + \hat{\delta}_\nu^{\mu_1} \hat{\delta}_\rho^{\nu_1} \right ) + p^1_\nu \left ( \hat{\delta}_\mu^{\mu_1} \hat{\delta}_\rho^{\nu_1} + \hat{\delta}_\rho^{\mu_1} \hat{\delta}_\mu^{\nu_1} \right ) \right . \\
		& \hphantom{= \frac{\gcoupling}{4} (} \left . - p^1_\rho \left ( \hat{\delta}_\mu^{\mu_1} \hat{\delta}_\nu^{\nu_1} + \hat{\delta}_\nu^{\mu_1} \hat{\delta}_\mu^{\nu_1} \right ) \right ) \, .
	\end{split}
	\end{equation}
\end{lem}

\begin{proof}
	This follows from directly from the expression
	\begin{equation}
		\widehat{\Gamma}_{\mu \nu \rho} = \frac{\gcoupling}{2} \left ( p_\mu \hat{h}_{\nu \rho} + p_\nu \hat{h}_{\rho \mu} - p_\rho \hat{h}_{\mu \nu} \right ) \, .
	\end{equation}
\end{proof}

\enter

\begin{lem} \label{lem:Ricci_scalar_FR}
	Introducing the notation
	\begin{equation}
		\mathfrak{R}_n^{\mu_1 \nu_1 \vert \cdots \vert \mu_n \nu_n} \left ( p_1^\sigma, \cdots, p_n^\sigma \right ) := \left ( \prod_{i = 1}^n \frac{\bar{\delta}}{\bar{\delta} \hat{h}_{\mu_i \nu_i}} \right ) \mathscr{F} \left ( \eval{R}_{\order{\gcoupling^{n}}} \right ) \, ,
	\end{equation}
	we obtain
	\begin{subequations}
	\begin{align}
		\mathfrak{R}_0 & = 0 \, , \\
		\mathfrak{R}_1^{\mu_1 \nu_1} \left ( p_1^\sigma \right ) & = - \gcoupling \left ( p_1^{\mu_1} p_1^{\nu_1} - p_1^2 \hat{\eta}^{\mu_1 \nu_1} \right )
		\intertext{and for \(n > 1\)}
		\mathfrak{R}_n^{\mu_1 \nu_1 \vert \cdots \vert \mu_n \nu_n} \left ( p_1^\sigma, \cdots, p_n^\sigma \right ) & = \frac{1}{2^n} \sum_{\mu_i \leftrightarrow \nu_i} \sum_{s \in S_n} \mathfrak{r}_n^{\mu_{s(1)} \nu_{s(1)} \vert \cdots \vert \mu_{s(n)} \nu_{s(n)}} \left ( p_{s(1)}^\sigma, \cdots, p_{s(n)}^\sigma \right )
		\intertext{with}
		\begin{split}
			\mathfrak{r}_n^{\mu_1 \nu_1 \vert \cdots \vert \mu_n \nu_n} \left ( p_1^\sigma, \cdots, p_n^\sigma \right ) & = \left ( \mathfrak{r}_n^{\partial \Gamma} \right )^{\mu_1 \nu_1 \vert \cdots \vert \mu_n \nu_n} \left ( p_1^\sigma, \cdots, p_n^\sigma \right ) \\ & \hphantom{=} + \left ( \mathfrak{r}_n^{\Gamma^2} \right )^{\mu_1 \nu_1 \vert \cdots \vert \mu_n \nu_n} \left ( p_1^\sigma, \cdots, p_n^\sigma \right ) \, ,
		\end{split}
		\\
		\begin{split}
			\left ( \mathfrak{r}_n^{\partial \Gamma} \right )^{\mu_1 \nu_1 \vert \cdots \vert \mu_n \nu_n} \left ( p_1^\sigma, \cdots, p_n^\sigma \right ) & = \left ( - \gcoupling \right )^n \sum_{i+j = n-1} \left ( \hat{\delta}^\rho_{\nu_{i+1}} \prod_{a = 0}^i \hat{\eta}^{\mu_a \nu_{a+1}} \right ) \left ( \hat{\delta}^\nu_{\mu_i} \prod_{b = i}^{i+j} \hat{\eta}^{\mu_b \nu_{b+1}} \right ) \\
			& \hphantom{= - \gcoupling^n \sum_{i+j = n-1}} \times \left ( p^n_{\mu_0} p^n_\nu \hat{\delta}_\rho^{\mu_n} - p^n_{\mu_0} p^n_\rho \hat{\delta}_\nu^{\mu_n} \right ) \\
			& \hphantom{=} \mkern-72mu - \left ( - \gcoupling \right )^n \sum_{i+j+k = n-2} \left ( \hat{\delta}^\rho_{\nu_{i+1}} \prod_{a = 0}^i \hat{\eta}^{\mu_a \nu_{a+1}} \right ) \left ( \hat{\delta}^\nu_{\mu_i} \hat{\delta}^\sigma_{\nu_{i+j+1}} \prod_{b = i}^{i+j} \hat{\eta}^{\mu_b \nu_{b+1}} \right ) \\
			 & \hphantom{= - \left ( - \gcoupling \right )^n \sum_{i+j+k = n-2}} \mkern-104mu \times \left ( \hat{\delta}^\kappa_{\mu_{i+j}} \hat{\delta}^\lambda_{\nu_{i+j+k+1}} \prod_{c = i+j}^{i+j+k} \hat{\eta}^{\mu_c \nu_{c+1}} \right ) \\
			 & \hphantom{= - \left ( - \gcoupling \right )^n \sum_{i+j+k = n-2}} \mkern-104mu \times \left ( \left ( p^{n-1}_{\mu_0} \hat{\delta}_\rho^{\mu_{n-1}} \hat{\delta}_\kappa^{\nu_{n-1}} \right ) \left ( \frac{1}{2} p^n_\lambda \hat{\delta}_\nu^{\mu_n} \hat{\delta}_\sigma^{\nu_n} - p^n_\nu \hat{\delta}_\lambda^{\mu_n} \hat{\delta}_\sigma^{\nu_n} \right ) \right . \\
			 & \hphantom{= - \left ( - \gcoupling \right )^n \sum_{i+j+k = n-2} \times (} \mkern-104mu \left . + \frac{1}{2} \left ( p^{n-1}_\nu \hat{\delta}_{\mu_0}^{\mu_{n-1}} \hat{\delta}_\kappa^{\nu_{n-1}} \right ) \left ( p^n_\sigma \hat{\delta}_\rho^{\mu_n} \hat{\delta}_\lambda^{\nu_n} \right ) \vphantom{\left ( \frac{1}{2} \right )} \right )
		\end{split}
		\intertext{and}
		\begin{split}
			\left ( \mathfrak{r}_n^{\Gamma^2} \right )^{\mu_1 \nu_1 \vert \cdots \vert \mu_n \nu_n} \left ( p_1^\sigma, \cdots, p_n^\sigma \right ) & = \\
			& \hphantom{=} \mkern-120mu - \left ( - \gcoupling \right )^n \sum_{i+j+k = n-2} \left ( \hat{\delta}^\rho_{\nu_{i+1}} \prod_{a = 0}^i \hat{\eta}^{\mu_a \nu_{a+1}} \right ) \left ( \hat{\delta}^\nu_{\mu_i} \hat{\delta}^\sigma_{\nu_{i+j+1}} \prod_{b = i}^{i+j} \hat{\eta}^{\mu_b \nu_{b+1}} \right ) \\
			& \hphantom{= - \left ( - \gcoupling \right )^n \sum_{i+j+k = n-2}} \mkern-152mu \times \left ( \hat{\delta}^\kappa_{\mu_{i+j}} \hat{\delta}^\lambda_{\nu_{i+j+k+1}} \prod_{c = i+j}^{i+j+k} \hat{\eta}^{\mu_c \nu_{c+1}} \right ) \\
			& \hphantom{= - \left ( - \gcoupling \right )^n \sum_{i+j+k = n-2}} \mkern-152mu \times \left ( \left ( p^{n-1}_\kappa \hat{\delta}_{\mu_0}^{\mu_{n-1}} \hat{\delta}_\rho^{\nu_{n-1}} \right ) \left ( \frac{1}{2} p^n_\nu \hat{\delta}_\sigma^{\mu_n} \hat{\delta}_\lambda^{\nu_n} - \frac{1}{4} p^n_\lambda \hat{\delta}_\nu^{\mu_n} \hat{\delta}_\sigma^{\nu_n} \right ) \right . \\
			& \hphantom{= - \left ( - \gcoupling \right )^n \sum_{i+j+k = n-2} \times (} \mkern-152mu \left . - \left ( p^{n-1}_\nu \hat{\delta}_{\mu_0}^{\mu_{n-1}} \hat{\delta}_\kappa^{\nu_{n-1}} \right ) \left ( \frac{1}{2} p^n_\rho \hat{\delta}_\sigma^{\mu_n} \hat{\delta}_\lambda^{\nu_n} - \frac{1}{4} p^n_\sigma \hat{\delta}_\rho^{\mu_n} \hat{\delta}_\lambda^{\nu_n} \right ) \vphantom{\left ( \frac{1}{2} \right )} \right ) \, .
		\end{split}
	\end{align}
	\end{subequations}
\end{lem}

\begin{proof}
	This follows directly from Corollaries \ref{col:ricci_scalar_for_the_levi_civita_connection_restriction} and \ref{col:inverse_metric_vielbeins_FR}. Furthermore, we remark the global minus sign due to the Fourier transform and the omission of Kronecker symbols, if possible.
\end{proof}

\enter

\begin{lem} \label{lem:de_Donder_gauge_fixing_FR}
	Introducing the notation
	\begin{equation}
		\deDonderFR_n^{\mu_1 \nu_1 \vert \cdots \vert \mu_n \nu_n} \left ( p_1^\sigma, \cdots, p_n^\sigma \right ) := \left ( \prod_{i = 1}^n \frac{\bar{\delta}}{\bar{\delta} \hat{h}_{\mu_i \nu_i}} \right ) \mathscr{F} \left ( \eval{\deDonder^2}_{\order{\gcoupling^{n}}} \right ) \, ,
	\end{equation}
	we obtain
	\begin{subequations}
	\begin{align}
		\deDonderFR_0 & = 0 \, , \\
		\deDonderFR_1^{\mu_1 \nu_1} \left ( p_1^\sigma \right ) & = 0
		\intertext{and for \(n > 1\)}
		\! \! \! \deDonderFR_n^{\mu_1 \nu_1 \vert \cdots \vert \mu_n \nu_n} \left ( p_1^\sigma, \cdots, p_n^\sigma \right ) & = \frac{1}{2^n} \sum_{\mu_i \leftrightarrow \nu_i} \sum_{s \in S_n} \deDonderpreFR_n^{\mu_{s(1)} \nu_{s(1)} \vert \cdots \vert \mu_{s(n)} \nu_{s(n)}} \left ( p_{s(1)}^\sigma, \cdots, p_{s(n)}^\sigma \right )
		\intertext{with}
		\begin{split}
			\! \! \! \deDonderpreFR_n^{\mu_1 \nu_1 \vert \cdots \vert \mu_n \nu_n} \left ( p_1^\sigma, \cdots, p_n^\sigma \right ) & = - \left ( - \gcoupling \right )^n \sum_{i+j+k = n-2} \left ( \hat{\delta}^\rho_{\nu_{i+1}} \prod_{a = 0}^i \hat{\eta}^{\mu_a \nu_{a+1}} \right ) \\ & \hphantom{=} \times \left ( \hat{\delta}^\nu_{\mu_i} \hat{\delta}^\sigma_{\nu_{i+j+1}} \prod_{b = i}^{i+j} \hat{\eta}^{\mu_b \nu_{b+1}} \right ) \left ( \hat{\delta}^\kappa_{\mu_{i+j}} \hat{\delta}^\lambda_{\nu_{i+j+k+1}} \prod_{c = i+j}^{i+j+k} \hat{\eta}^{\mu_c \nu_{c+1}} \right ) \\
			& \mkern-72mu \hphantom{=} \times \left ( \left ( p^{n-1}_\nu \hat{\delta}_\sigma^{\mu_{n-1}} \hat{\delta}_{\mu_0}^{\nu_{n-1}} \right ) \left ( p^n_\kappa \hat{\delta}_\lambda^{\mu_n} \hat{\delta}_\rho^{\nu_n} \right ) - \left ( p^{n-1}_\nu \hat{\delta}_\sigma^{\mu_{n-1}} \hat{\delta}_{\mu_0}^{\nu_{n-1}} \right ) \left ( p^n_\rho \hat{\delta}_\kappa^{\mu_n} \hat{\delta}_\lambda^{\nu_n} \right ) \right . \\ & \mkern-72mu \hphantom{= \times (} \left . + \frac{1}{4} \left ( p^{n-1}_{\mu_0} \hat{\delta}_\nu^{\mu_{n-1}} \hat{\delta}_\sigma^{\nu_{n-1}} \right ) \left ( p^n_\rho \hat{\delta}_\kappa^{\mu_n} \hat{\delta}_\lambda^{\nu_n} \right ) \right ) \, .
		\end{split}
	\end{align}
	\end{subequations}
	In particular, the quadratic part is given by (using momentum conservation, i.e.\ setting \(p_1^\sigma := p^\sigma\) and \(p_2^\sigma := - p^\sigma\))
	\begin{equation} \label{eqn:de_donder-quadratic}
	\begin{split}
		\deDonderFR_2^{\mu_1 \nu_1 \vert \mu_2 \nu_2} \left ( p^\sigma, - p^\sigma \right ) & = \gcoupling^2 \big ( p^{\mu_1} p^{\nu_1} \hat{\eta}^{\mu_2 \nu_2} + p^{\mu_2} p^{\nu_2} \hat{\eta}^{\mu_1 \nu_1} \big ) \\
		& \hphantom{=} - \frac{1}{2} \gcoupling^2 \big ( p^{\mu_1} p^{\mu_2} \hat{\eta}^{\nu_1 \nu_2} + p^{\mu_1} p^{\nu_2} \hat{\eta}^{\nu_1 \mu_2} + p^{\nu_1} p^{\mu_2} \hat{\eta}^{\mu_1 \nu_2} + p^{\nu_1} p^{\nu_2} \hat{\eta}^{\mu_1 \mu_2} \big ) \\
		& \hphantom{=} - \frac{1}{2} \gcoupling^2 \big ( p^2 \hat{\eta}^{\mu_1 \nu_1} \hat{\eta}^{\mu_2 \nu_2} \big ) \, .
	\end{split}
	\end{equation}
\end{lem}

\begin{proof}
	This follows directly from Corollaries \ref{col:metric_expression_for_de_donder_gauge_fixing_restriction} and \ref{col:inverse_metric_vielbeins_FR}. Furthermore, we remark the global minus sign due to the Fourier transform and the omission of Kronecker symbols, if possible.
\end{proof}

\enter

\begin{lem} \label{lem:Riemannian_volume_form_FR}
	Introducing the notation
	\begin{equation}
		\mathfrak{V}_n^{\mu_1 \nu_1 \vert \cdots \vert \mu_n \nu_n} := \left ( \prod_{i = 1}^n \frac{\delta}{\delta \hat{h}_{\mu_i \nu_i}} \right ) \mathscr{F} \left ( \eval{\sqrt{- \dt{g}}}_{\order{\gcoupling^{n}}} \right ) \, ,
	\end{equation}
	we obtain
	\begin{subequations}
	\begin{align}
		\mathfrak{V}_n^{\mu_1 \nu_1 \vert \cdots \vert \mu_n \nu_n} & = \frac{1}{2^n} \sum_{\mu_i \leftrightarrow \nu_i} \sum_{s \in S_n} \mathfrak{v}_n^{\mu_{s(1)} \nu_{s(1)} \vert \cdots \vert \mu_{s(n)} \nu_{s(n)}}
		\intertext{with}
		\begin{split}
			\mathfrak{v}_n^{\mu_1 \nu_1 \vert \cdots \vert \mu_n \nu_n} & = \gcoupling^n \sum_{\substack{i + j + k + l = m\\i \geq j \geq k \geq l \geq 0}} \sum_{p = 0}^{j - k} \sum_{q = 0}^{k - l} \sum_{r = 0}^{q} \sum_{s = 0}^l \sum_{t = 0}^s \sum_{u = 0}^t \sum_{v = 0}^u \\
			& \hphantom{=} \mkern-54mu \binom{\frac{1}{2}}{i} \binom{i}{j} \binom{j}{k} \binom{k}{l} \binom{j - k}{p} \binom{k - l}{q} \binom{q}{r} \binom{l}{s} \binom{s}{t} \binom{t}{u} \binom{u}{v} \\
			& \hphantom{=} \mkern-54mu \times \left ( - 1 \right )^{p + q - r + s - t + v} 2^{- j + l + r + s + 2t - 3u + v} 3^{- k + q - r + s - t + u} \\
			& \hphantom{=} \mkern-54mu \times \left ( \prod_{a = 1}^{\boldsymbol{a}} \hat{\eta}^{\mu_a \nu_a} \right ) \left ( \prod_{b = {\boldsymbol{a} + 1}}^{\boldsymbol{a} + \boldsymbol{b}} \hat{\eta}^{\mu_{b} \mu_{b + \boldsymbol{b}}} \hat{\eta}^{\nu_{b} \nu_{b + \boldsymbol{b}}} \right ) \left ( \prod_{c = {\boldsymbol{a} + 2 \boldsymbol{b} + 1}}^{\boldsymbol{a} + 2 \boldsymbol{b} + \boldsymbol{c}} \hat{\eta}^{\mu_{c} \nu_{c + \boldsymbol{c}}} \hat{\eta}^{\mu_{c + \boldsymbol{c}} \nu_{c + 2 \boldsymbol{c}}} \hat{\eta}^{\mu_{c + 2 \boldsymbol{c}} \nu_{c}} \right ) \\
			& \hphantom{=} \mkern-54mu \times \left ( \prod_{d = {\boldsymbol{a} + 2 \boldsymbol{b} + 3 \boldsymbol{c} + 1}}^{\boldsymbol{a} + 2 \boldsymbol{b} + 3 \boldsymbol{c} + \boldsymbol{d}} \hat{\eta}^{\mu_{d} \nu_{d + \boldsymbol{d}}} \hat{\eta}^{\mu_{d + \boldsymbol{d}} \nu_{d + 2 \boldsymbol{d}}} \hat{\eta}^{\mu_{d + 2 \boldsymbol{d}} \nu_{d + 3 \boldsymbol{d}}} \hat{\eta}^{\mu_{d + 3 \boldsymbol{d}} \nu_{d}} \right )
		\end{split}
		\intertext{and}
		\begin{split}
			\boldsymbol{a} & := i + j + k + l - 2p - 2q - r - 2s - t - u \\
			\boldsymbol{b} & := p + q - r + s - t + 2u - 2v \\
			\boldsymbol{c} & := r + t - u \\
			\boldsymbol{d} & := v \, .
		\end{split}
	\end{align}
	\end{subequations}
\end{lem}

\begin{proof}
	This follows directly from Corollaries \ref{col:determinant_metric_restriction} and \ref{col:inverse_metric_vielbeins_FR}.
\end{proof}

\section{Feynman rules for any valence} \label{sec:feynman_rules_for_any_valence}

In this section, we present the gravitational Feynman rules for (effective) Quantum General Relativity coupled to the Standard Model. We provide the vertex Feynman rules in any valence and the propagator Feynman rules with a general gauge parameter corresponding to the quadratic de Donder gauge fixing.

\subsection{Gravitons and graviton-ghosts}

Having done all preparations in \ssecref{ssec:preparations_graviton_graviton_ghost}, we now list the corresponding Feynman rules for gravitons and their ghosts.

\enter

\begin{thm} \label{thm:grav-fr}
	Given the metric decomposition \(g_{\mu \nu} = \eta_{\mu \nu} + \gcoupling h_{\mu \nu}\) and assume \(\left | \gcoupling \right | \left \| h \right \|_{\max} := \left | \gcoupling \right | \max_{\lambda \in \operatorname{EW} \left ( h \right )} \left | \lambda \right | < 1\), where \(\operatorname{EW} \left ( h \right )\) denotes the set of eigenvalues of \(h\). Then the graviton \(2\)-point vertex Feynman rule for (effective) Quantum General Relativity reads (where \(\zeta\) denotes the gauge parameter and we use momentum conservation on the quadratic term, i.e.\ set \(p_1^\sigma := p^\sigma\) and \(p_2^\sigma := - p^\sigma\)):
		\begin{equation}
		\begin{split}
			\gravfr_2^{\mu_1 \nu_1 \vert \mu_2 \nu_2} \left ( p^\sigma; \zeta \right ) & = \frac{\imaginary}{4} \left ( 1 - \frac{1}{\zeta} \right ) \left ( p^{\mu_1} p^{\nu_1} \hat{\eta}^{\mu_2 \nu_2} + p^{\mu_2} p^{\nu_2} \hat{\eta}^{\mu_1 \nu_1} \right ) \\
			& \hphantom{=} \mkern-55mu - \frac{\imaginary}{8} \left ( 1 - \frac{1}{\zeta} \right ) \left ( p^{\mu_1} p^{\mu_2} \hat{\eta}^{\nu_1 \nu_2} + p^{\mu_1} p^{\nu_2} \hat{\eta}^{\nu_1 \mu_2} + p^{\nu_1} p^{\mu_2} \hat{\eta}^{\mu_1 \nu_2} + p^{\nu_1} p^{\nu_2} \hat{\eta}^{\mu_1 \mu_2} \right ) \\
			& \hphantom{=} \mkern-55mu - \frac{\imaginary}{4} \left ( 1 - \frac{1}{2 \zeta} \right ) \left ( p^2 \hat{\eta}^{\mu_1 \nu_1} \hat{\eta}^{\mu_2 \nu_2} \right ) \\
			& \hphantom{=} \mkern-55mu + \frac{\imaginary}{8} \left ( p^2 \hat{\eta}^{\mu_1 \mu_2} \hat{\eta}^{\nu_1 \nu_2} + p^2 \hat{\eta}^{\mu_1 \nu_2} \hat{\eta}^{\nu_1 \mu_2} \right )
		\end{split}
		\end{equation}
	Furthermore, the graviton \(n\)-point vertex Feynman rules with \(n > 2\) for (effective) Quantum General Relativity read (where \(\pregravfr_n\) denotes the corresponding unsymmetrized Feynman rules and \(\boldsymbol{\delta}_{m_1 \neq n}\) is set to \(0\) if \(m_1 = n\) and to \(1\) else, eliminating contributions coming from total derivatives):
		\begin{subequations}
		\begin{align}
			\gravfr_n^{\mu_1 \nu_1 \vert \cdots \vert \mu_n \nu_n} \left ( p_1^\sigma, \cdots, p_n^\sigma \right ) & = \frac{\imaginary}{2^n} \sum_{\mu_i \leftrightarrow \nu_i} \sum_{s \in S_n} \pregravfr_n^{\mu_{s(1)} \nu_{s(1)} \vert \cdots \vert \mu_{s(n)} \nu_{s(n)}} \left ( p_{s(1)}^\sigma, \cdots, p_{s(n)}^\sigma \right )
			\intertext{with}
			\begin{split}
				\pregravfr_n^{\mu_1 \nu_1 \vert \cdots \vert \mu_n \nu_n} \left ( p_1^\sigma, \cdots, p_n^\sigma \right ) & = \\
				& \mkern-210mu \frac{\left ( - \gcoupling \right )^{n-2}}{2} \sum_{m_1 + m_2 = n} \left \{ \sum_{i = 0}^{m_1 - 1} \left ( \hat{\delta}^\mu_{\mu_0} \hat{\delta}^\rho_{\nu_{i+1}} \prod_{a = 0}^i \hat{\eta}^{\mu_a \nu_{a+1}} \right ) \left ( \hat{\delta}^\nu_{\mu_i} \hat{\delta}^\sigma_{\nu_{m_1}} \prod_{b = i}^{m_1 - 1} \hat{\eta}^{\mu_b \nu_{b+1}} \right ) \right . \\
				& \mkern-210mu \hphantom{\frac{\left ( - \gcoupling \right )^{n-2}}{2} \sum \{ \sum} \times \boldsymbol{\delta}_{m_1 \neq n} \left [ p^{m_1}_\mu p^{m_1}_\nu \hat{\delta}_\rho^{\mu_{m_1}} \hat{\delta}_\sigma^{\nu_{m_1}} - p^{m_1}_\mu p^{m_1}_\rho \hat{\delta}_\nu^{\mu_{m_1}} \hat{\delta}_\sigma^{\nu_{m_1}} \right ] \\
				& \mkern-210mu \hphantom{\frac{\left ( - \gcoupling \right )^{n-2}}{2} \sum} - \sum_{j + k + l = m_1 - 2} \left ( \hat{\delta}^\mu_{\mu_0} \hat{\delta}^\rho_{\nu_{j+1}} \prod_{a = 0}^j \hat{\eta}^{\mu_a \nu_{a+j}} \right ) \left ( \hat{\delta}^\nu_{\mu_j} \hat{\delta}^\sigma_{\nu_{j+k+1}} \prod_{b = j}^{j+k} \hat{\eta}^{\mu_b \nu_{b+1}} \right ) \\
				& \mkern-210mu \hphantom{\frac{\left ( - \gcoupling \right )^{n-2}}{2} \sum +} \times \left ( \hat{\delta}^\kappa_{\mu_{j+k}} \hat{\delta}^\lambda_{\nu_{m_1 - 1}} \prod_{c = j+k}^{m_1 - 2} \hat{\eta}^{\mu_c \nu_{c+1}} \right ) \\
				& \mkern-210mu \hphantom{\frac{\left ( - \gcoupling \right )^{n-2}}{2} \sum +} \times \left ( \boldsymbol{\delta}_{m_1 \neq n} \left [ \left ( p^{n-1}_{\mu} \hat{\delta}_\rho^{\mu_{n-1}} \hat{\delta}_\kappa^{\nu_{n-1}} \right ) \left ( \frac{1}{2} p^n_\lambda \hat{\delta}_\nu^{\mu_n} \hat{\delta}_\sigma^{\nu_n} - p^n_\nu \hat{\delta}_\lambda^{\mu_n} \hat{\delta}_\sigma^{\nu_n} \right ) \right . \right . \\
				& \mkern-210mu \hphantom{\frac{\left ( - \gcoupling \right )^{n-2}}{2} \sum + \times ( + \boldsymbol{\delta}_{m_1 \neq n} [} \left . + \frac{1}{2} \left ( p^{n-1}_\nu \hat{\delta}_{\mu}^{\mu_{n-1}} \hat{\delta}_\kappa^{\nu_{n-1}} \right ) \left ( p^n_\sigma \hat{\delta}_\rho^{\mu_n} \hat{\delta}_\lambda^{\nu_n} \right ) \vphantom{\left ( \frac{1}{2} \right )} \right ] \\
				& \mkern-210mu \hphantom{\frac{\left ( - \gcoupling \right )^{n-2}}{2} \sum + \times (} + \left ( p^{n-1}_\kappa \hat{\delta}_{\mu}^{\mu_{n-1}} \hat{\delta}_\rho^{\nu_{n-1}} \right ) \left ( \frac{1}{2} p^n_\nu \hat{\delta}_\sigma^{\mu_n} \hat{\delta}_\lambda^{\nu_n} - \frac{1}{4} p^n_\lambda \hat{\delta}_\nu^{\mu_n} \hat{\delta}_\sigma^{\nu_n} \right ) \\
				& \mkern-210mu \hphantom{\frac{\left ( - \gcoupling \right )^{n-2}}{2} \sum + \times (} - \left . \left ( p^{n-1}_\nu \hat{\delta}_{\mu}^{\mu_{n-1}} \hat{\delta}_\kappa^{\nu_{n-1}} \right ) \left ( \frac{1}{2} p^n_\rho \hat{\delta}_\sigma^{\mu_n} \hat{\delta}_\lambda^{\nu_n} - \frac{1}{4} p^n_\sigma \hat{\delta}_\rho^{\mu_n} \hat{\delta}_\lambda^{\nu_n} \right ) \right \} \\
				& \mkern-210mu \hphantom{\frac{\left ( - \gcoupling \right )^{n-2}}{2} \sum \{} \boldsymbol{\times} \left \{ \sum_{\substack{i + j + k + l = m_2\\i \geq j \geq k \geq l \geq 0}} \sum_{p = 0}^{j - k} \sum_{q = 0}^{k - l} \sum_{r = 0}^{q} \sum_{s = 0}^l \sum_{t = 0}^s \sum_{u = 0}^t \sum_{v = 0}^u \right . \\
				& \mkern-210mu \hphantom{\frac{\left ( - \gcoupling \right )^{n-2}}{2} \sum \{ \boldsymbol{\times}} \binom{\frac{1}{2}}{i} \binom{i}{j} \binom{j}{k} \binom{k}{l} \binom{j - k}{p} \binom{k - l}{q} \binom{q}{r} \binom{l}{s} \binom{s}{t} \binom{t}{u} \binom{u}{v} \\
				& \mkern-210mu \hphantom{\frac{\left ( - \gcoupling \right )^{n-2}}{2} \sum \{ \boldsymbol{\times}} \times \left ( - 1 \right )^{p + q - r + s - t + v} 2^{- j + l + r + s + 2t - 3u + v} 3^{- k + q - r + s - t + u} \\
				& \mkern-210mu \hphantom{\frac{\left ( - \gcoupling \right )^{n-2}}{2} \sum \{ \boldsymbol{\times}} \times \left ( \prod_{a = m_1 + 1}^{m_1 + \boldsymbol{a}} \hat{\eta}^{\mu_a \nu_a} \right ) \left ( \prod_{b = m_1 + \boldsymbol{a} + 1}^{m_1 + \boldsymbol{a} + \boldsymbol{b}} \hat{\eta}^{\mu_{b} \mu_{b + \boldsymbol{b}}} \hat{\eta}^{\nu_{b} \nu_{b + \boldsymbol{b}}} \right ) \\
				& \mkern-210mu \hphantom{\frac{\left ( - \gcoupling \right )^{n-2}}{2} \sum \{ \boldsymbol{\times}} \times \left ( \prod_{c = m_1 + \boldsymbol{a} + 2 \boldsymbol{b} + 1}^{m_1 + \boldsymbol{a} + 2 \boldsymbol{b} + \boldsymbol{c}} \hat{\eta}^{\mu_{c} \nu_{c + \boldsymbol{c}}} \hat{\eta}^{\mu_{c + \boldsymbol{c}} \nu_{c + 2 \boldsymbol{c}}} \hat{\eta}^{\mu_{c + 2 \boldsymbol{c}} \nu_{c}} \right ) \\
				& \mkern-210mu \hphantom{\frac{\left ( - \gcoupling \right )^{n-2}}{2} \sum \{ \boldsymbol{\times}} \left . \times \left ( \prod_{d = m_1 + \boldsymbol{a} + 2 \boldsymbol{b} + 3 \boldsymbol{c} + 1}^{m_1 + \boldsymbol{a} + 2 \boldsymbol{b} + 3 \boldsymbol{c} + \boldsymbol{d}} \hat{\eta}^{\mu_{d} \nu_{d + \boldsymbol{d}}} \hat{\eta}^{\mu_{d + \boldsymbol{d}} \nu_{d + 2 \boldsymbol{d}}} \hat{\eta}^{\mu_{d + 2 \boldsymbol{d}} \nu_{d + 3 \boldsymbol{d}}} \hat{\eta}^{\mu_{d + 3 \boldsymbol{d}} \nu_{d}} \right ) \right \}
			\end{split}
			\intertext{and}
			\begin{split}
				\boldsymbol{a} & := i + j + k + l - 2p - 2q - r - 2s - t - u \\
				\boldsymbol{b} & := p + q - r + s - t + 2u - 2v \\
				\boldsymbol{c} & := r + t - u \\
				\boldsymbol{d} & := v
			\end{split}
		\end{align}
		\end{subequations}
\end{thm}

\begin{proof}
	This follows from the combination of Lemmata \ref{lem:traces_FR}, \ref{lem:Ricci_scalar_FR}, \ref{lem:de_Donder_gauge_fixing_FR} and \ref{lem:Riemannian_volume_form_FR}, since we have
	\begin{subequations}
	\begin{align}
		\gravfr_n^{\mu_1 \nu_1 \vert \cdots \vert \mu_n \nu_n} \left ( p_1^\sigma, \cdots, p_n^\sigma \right ) & = \frac{\imaginary}{2^n} \sum_{\mu_i \leftrightarrow \nu_i} \sum_{s \in S_n} \pregravfr_n^{\mu_{s(1)} \nu_{s(1)} \vert \cdots \vert \mu_{s(n)} \nu_{s(n)}} \left ( p_{s(1)}^\sigma, \cdots, p_{s(n)}^\sigma \right )
		\intertext{with}
		\begin{split}
			\pregravfr_n^{\mu_1 \nu_1 \vert \cdots \vert \mu_n \nu_n} \left ( p_1^\sigma, \cdots, p_n^\sigma \right ) & = \\ & \hphantom{=} \mkern-120mu - \frac{1}{2 \gcoupling^2} \sum_{m = 1}^n \left ( \mathfrak{R}_m^{\mu_1 \nu_1 \vert \cdots \vert \mu_m \nu_m} \left ( p_1^\sigma, \cdots, p_m^{\sigma_m} \right ) \times \mathfrak{V}_{n - m}^{\mu_{n - m} \nu_{n - m} \vert \cdots \vert \mu_n \nu_n} \right ) \\ & \hphantom{= - \frac{\imaginary}{\gcoupling^2} \sum_{m = 1}^n} \mkern-120mu + \boldsymbol{\delta}_{n = 2} \frac{1}{2 \zeta} \deDonderFR_2^{\mu_1 \nu_1 \vert \mu_2 \nu_2} \left ( p_1^\sigma, p_2^\sigma \right ) \, ,
		\end{split}
	\end{align}
	\end{subequations}
	where \(\boldsymbol{\delta}_{n = 2}\) is set to 1 for \(n = 2\) and to 0 else, and modulo total derivatives which come from the \(R^{\partial \Gamma}\) contributions of degree \(n\).
\end{proof}

\enter

\begin{rem} \label{rem:one-valent-FR}
	The one-valent Feynman rule actually reads
	\begin{equation}
		\gravfr_1^{\mu_1 \nu_1} \left ( p_1^\sigma \right ) = \frac{\imaginary}{2 \gcoupling} \left ( p_1^{\mu_1} p_1^{\nu_1} - p_1^2 \hat{\eta}^{\mu_1 \nu_1} \right ) \, .
	\end{equation}
	However this term comes from a total derivative in the Lagrange density and can thus be set to zero. Equivalently, on the level of Feynman rules, it vanishes due to momentum conservation.
\end{rem}

\enter

\begin{thm} \label{thm:grav-prop}
	Given the situation of \thmref{thm:grav-fr}, the graviton propagator Feynman rule for (effective) Quantum General Relativity reads:
	\begin{equation} \label{eqn:grav-prop}
	\begin{split}
		\gravprop_{\mu_1 \nu_1 \vert \mu_2 \nu_2} \left ( p^\sigma; \zeta; \epsilon \right ) & = - \frac{2 \imaginary}{p^2 + \imaginary \epsilon} \Bigg [ \left ( \hat{\eta}_{\mu_1 \mu_2} \hat{\eta}_{\nu_1 \nu_2} + \hat{\eta}_{\mu_1 \nu_2} \hat{\eta}_{\nu_1 \mu_2} - \hat{\eta}_{\mu_1 \nu_1} \hat{\eta}_{\mu_2 \nu_2} \right ) \\
		& \! \! \! \! - \left ( \frac{1 - \zeta}{p^2} \right ) \left ( \hat{\eta}_{\mu_1 \mu_2} p_{\nu_1} p_{\nu_2} + \hat{\eta}_{\mu_1 \nu_2} p_{\nu_1} p_{\mu_2} + \hat{\eta}_{\nu_1 \mu_2} p_{\mu_1} p_{\nu_2} + \hat{\eta}_{\nu_1 \nu_2} p_{\mu_1} p_{\mu_2} \right ) \Bigg ]
	\end{split}
	\end{equation}
\end{thm}

\begin{proof}
	To calculate the graviton propagator, we recall
	\begin{equation} \label{eqn:grav-fr-quadratic}
	\begin{split}
		\gravfr_2^{\mu_1 \nu_1 \vert \mu_2 \nu_2} \left ( p^\sigma; \zeta \right ) & = \frac{\imaginary}{4} \left ( 1 - \frac{1}{\zeta} \right ) \left ( p^{\mu_1} p^{\nu_1} \hat{\eta}^{\mu_2 \nu_2} + p^{\mu_2} p^{\nu_2} \hat{\eta}^{\mu_1 \nu_1} \right ) \\
		& \hphantom{=} \mkern-55mu - \frac{\imaginary}{8} \left ( 1 - \frac{1}{\zeta} \right ) \left ( p^{\mu_1} p^{\mu_2} \hat{\eta}^{\nu_1 \nu_2} + p^{\mu_1} p^{\nu_2} \hat{\eta}^{\nu_1 \mu_2} + p^{\nu_1} p^{\mu_2} \hat{\eta}^{\mu_1 \nu_2} + p^{\nu_1} p^{\nu_2} \hat{\eta}^{\mu_1 \mu_2} \right ) \\
		& \hphantom{=} \mkern-55mu - \frac{\imaginary}{4} \left ( 1 - \frac{1}{2 \zeta} \right ) \left ( p^2 \hat{\eta}^{\mu_1 \nu_1} \hat{\eta}^{\mu_2 \nu_2} \right ) \\
		& \hphantom{=} \mkern-55mu + \frac{\imaginary}{8} \left ( p^2 \hat{\eta}^{\mu_1 \mu_2} \hat{\eta}^{\nu_1 \nu_2} + p^2 \hat{\eta}^{\mu_1 \nu_2} \hat{\eta}^{\nu_1 \mu_2} \right )
	\end{split}
	\end{equation}
	from \thmref{thm:grav-fr} and then invert it to obtain the propagator, i.e.\ such that\footnote{Where we treat the tuples of indices \(\mu_i \nu_i\) as one index, i.e.\ exclude the a priori possible term \(\hat{\eta}^{\mu_1 \nu_1} \hat{\eta}_{\mu_3 \nu_3}\) on the right hand side.}
	\begin{equation}
		\gravfr_2^{\mu_1 \nu_1 \vert \mu_2 \nu_2} \left ( p^\sigma; \zeta \right ) \gravprop_{\mu_2 \nu_2 \vert \mu_3 \nu_3} \left ( p^\sigma; \zeta; 0 \right ) = \frac{1}{2} \left ( \hat{\delta}^{\mu_1}_{\mu_3} \hat{\delta}^{\nu_1}_{\nu_3} + \hat{\delta}^{\mu_1}_{\nu_3} \hat{\delta}^{\nu_1}_{\mu_3} \right )
	\end{equation}
	holds, and we obtain \eqnref{eqn:grav-prop}.
\end{proof}

\enter

\begin{thm} \label{thm:ghost-fr}
	Given the situation of \thmref{thm:grav-fr}, the graviton-ghost \(2\)-point vertex Feynman rule for (effective) Quantum General Relativity reads:
	\begin{equation}
		\gravghostfr_2^{\rho_1 \rho_2} \left ( p^\sigma \right ) = \frac{\imaginary}{2 \zeta} p^2 \hat{\eta}^{\rho_1 \rho_2}
	\end{equation}
	Furthermore, the graviton-ghost \(n\)-point vertex Feynman rules with \(n > 2\) for (effective) Quantum General Relativity read   (where \(\preghostfr_n\) denotes the corresponding unsymmetrized Feynman rules):
	\begin{subequations}
	\begin{align}
		\begin{split}
			\gravghostfr_n^{\rho_1 \vert \rho_2 \| \mu_3 \nu_3 \vert \cdots \vert \mu_n \nu_n} \left ( p_1^\sigma, \cdots, p_n^\sigma \right ) & = \\
			& \hphantom{=} \mkern-72mu \frac{\imaginary}{2^{n-2}} \sum_{\mu_i \leftrightarrow \nu_i} \sum_{\substack{s \in S_{n-2}\\\tilde{s}(i) := s(i-2)+2}} \preghostfr_n^{\rho_1 \vert \rho_2 \| \mu_{\tilde{s}(3)} \nu_{\tilde{s}(3)} \vert \cdots \vert \mu_{\tilde{s}(n)} \nu_{\tilde{s}(n)}} \left ( p_1^\sigma, \cdots, p_n^\sigma \right )
		\end{split}
		\intertext{with}
		\begin{split}
			\preghostfr_n^{\rho_1 \vert \rho_2 \| \mu_3 \nu_3 \vert \cdots \vert \mu_n \nu_n} \left ( p_1^\sigma, \cdots, p_n^\sigma \right ) & = \frac{\left ( - \gcoupling \right )^{n-2}}{4} \left \{ \Bigg ( \hat{\delta}^{\rho_1}_{\mu_3} \hat{\delta}^{\mu}_{\nu_{n+1}} \prod_{a = 3}^n \hat{\eta}^{\mu_a \nu_{a+1}} \Bigg ) \hat{\eta}^{\rho_2 \nu} \hat{\eta}^{\rho \sigma} \right . \\
			& \hphantom{=} \times \Bigg [ p^{(2)}_\nu \bigg ( p^{(3)}_\rho \hat{\delta}_\mu^{\mu_3} \hat{\delta}_\sigma^{\nu_3} + p^{(3)}_\sigma \hat{\delta}_\rho^{\mu_3} \hat{\delta}_\mu^{\nu_3} - p^{(3)}_\mu \hat{\delta}_\rho^{\mu_3} \hat{\delta}_\sigma^{\nu_3} \bigg ) \\
			& \hphantom{= \times [} \left . - 2 p^{(2)}_\rho \bigg ( p^{(3)}_\sigma \hat{\delta}_\nu^{\mu_3} \hat{\delta}_\mu^{\nu_3} + p^{(3)}_\nu \hat{\delta}_\mu^{\mu_3} \hat{\delta}_\sigma^{\nu_3} - p^{(3)}_\mu \hat{\delta}_\sigma^{\mu_3} \hat{\delta}_\nu^{\nu_3} \bigg ) \Bigg ] \right \} \, ,
		\end{split}
	\end{align}
	\end{subequations}
	where particle \(1\) is the graviton-ghost, particle \(2\) is the graviton-antighost and the other particles are gravitons.
\end{thm}

\begin{proof}
	The case \(n = 2\) is immediate and the cases \(n > 2\) follow directly from Lemmata \ref{lem:traces_FR} and \ref{lem:Christoffel_FR}, since we have for \(n > 2\)
	\begin{equation}
	\begin{split}
		\gravghostfr_n^{\rho_1 \vert \rho_2 \| \mu_3 \nu_3 \vert \cdots \vert \mu_n \nu_n} \left ( p_1^\sigma, \cdots, p_n^\sigma \right ) & = \\
		& \hphantom{=} \mkern-72mu \frac{\imaginary}{2^{n-2}} \sum_{\mu_i \leftrightarrow \nu_i} \sum_{\substack{s \in S_{n-2}\\\tilde{s}(i) := s(i-2)+2}} \preghostfr_n^{\rho_1 \vert \rho_2 \| \mu_{\tilde{s}(3)} \nu_{\tilde{s}(3)} \vert \cdots \vert \mu_{\tilde{s}(n)} \nu_{\tilde{s}(n)}} \left ( p_1^\sigma, \cdots, p_n^\sigma \right )
	\end{split}
	\end{equation}
	with
	\begin{multline}
		\preghostfr_n^{\rho_1 \vert \rho_2 \| \mu_3 \nu_3 \vert \cdots \vert \mu_n \nu_n} \left ( p_1^\sigma, \cdots, p_n^\sigma \right ) = \\ \frac{\left ( - 1 \right )^{n}}{2} \mathfrak{h}_{n-3}^{\rho_1 \mu \triplevert \mu_4 \nu_4 \vert \cdots \vert \mu_n \nu_n} \hat{\eta}^{\rho_2 \nu} \hat{\eta}^{\rho \sigma} \left [ \left ( \sum_{\substack{k = 1\\k \neq 2}}^n p_\nu^k \right ) \boldsymbol{\Gamma}^{\mu_3 \nu_3}_{\rho \sigma \mu} \left ( p_1^\sigma \right ) - 2 \left ( \sum_{\substack{k = 1\\k \neq 2}}^n p_\rho^k \right ) \boldsymbol{\Gamma}^{\mu_3 \nu_3}_{\nu \sigma \mu} \left ( p_1^\sigma \right ) \right ] \, ,
	\end{multline}
	and then used momentum conservation twice, i.e.\ the relation
	\begin{equation}
		\left ( \sum_{\substack{k = 1\\k \neq 2}}^n p_\tau^{(k)} \right ) = - p^{(2)}_\tau \, .
	\end{equation}
\end{proof}

\enter

\begin{thm} \label{thm:ghost-prop}
	Given the situation of \thmref{thm:grav-fr}, the graviton-ghost propagator Feynman rule for (effective) Quantum General Relativity reads:
	\begin{equation} \label{eqn:grav-ghost-prop}
		\gravghostprop_{\rho_1 \vert \rho_2} \left ( p^2, \epsilon \right ) = - \frac{2 \imaginary \zeta}{p^2 + \imaginary \epsilon} \hat{\eta}_{\rho_1 \rho_2}
	\end{equation}
\end{thm}

\begin{proof}
	To calculate the graviton-ghost propagator, we recall
	\begin{equation}
		\gravghostfr_2^{\rho_1 \vert \rho_2} \left ( p^\sigma \right ) = \frac{\imaginary}{2 \zeta} p^2 \hat{\eta}^{\rho_1 \rho_2}
	\end{equation}
	from \thmref{thm:ghost-fr} and then invert it to obtain the propagator, i.e.\ such that
	\begin{equation}
		\gravghostfr_2^{\rho_1 \vert \rho_2} \left ( p^\sigma \right ) \gravghostprop_{\rho_2 \vert \rho_3} \left ( p^2, 0 \right ) = \hat{\delta}^{\rho_1}_{\rho_3}
	\end{equation}
	holds, and we obtain \eqnref{eqn:grav-ghost-prop}.
\end{proof}

\enter

\begin{exmp} \label{exmp:FR}
	Given the situation of \thmref{thm:grav-fr}, the three- and four-valent graviton vertex Feynman rules read as follows:\footnote{We have used momentum conservation, i.e.\ performed a partial integration on the Lagrange density for General Relativity, to obtain a more compact form.}
	\begin{subequations}
	\begin{align}
		& \gravfr_3^{\mu_1 \nu_1 \vert \mu_2 \nu_2 \vert \mu_3 \nu_3} \left ( p_1^\sigma, p_2^\sigma, p_3^\sigma \right ) = \frac{\imaginary}{8} \sum_{\mu_i \leftrightarrow \nu_i} \sum_{s \in S_3} \pregravfr_3^{\mu_{s(1)} \nu_{s(1)} \vert \mu_{s(2)} \nu_{s(2)} \vert \mu_{s(3)} \nu_{s(3)}} \left ( p_{s(1)}^\sigma, p_{s(2)}^\sigma \right )
		\intertext{with}
		\begin{split}
			& \pregravfr_3^{\mu_1 \nu_1 \vert \mu_2 \nu_2 \vert \mu_3 \nu_3} \left ( p_1^\sigma, p_2^\sigma \right ) = \frac{\gcoupling}{4} \Bigg \{ \frac{1}{2} p_1^{\mu_3} p_2^{\nu_3} \hat{\eta}^{\mu_1 \mu_2} \hat{\eta}^{\nu_1 \nu_2} - p_1^{\mu_3} p_2^{\mu_1} \hat{\eta}^{\nu_1 \mu_2} \hat{\eta}^{\nu_2 \nu_3} \\
			& \phantom{\pregravfr_3^{\mu_1 \nu_1 \vert \mu_2 \nu_2 \vert \mu_3 \nu_3} \left ( p_1^\sigma, p_2^\sigma \right ) = \frac{\gcoupling}{4} \Bigg \{} + \left ( p_1 \cdot p_2 \right ) \bigg ( - \frac{1}{2} \hat{\eta}^{\mu_1 \nu_1 } \hat{\eta}^{\mu_2 \mu_3 } \hat{\eta}^{\nu_2 \nu_3 } + \hat{\eta}^{\mu_1 \nu_2 } \hat{\eta}^{\mu_2 \nu_3 } \hat{\eta}^{\mu_3 \nu_1 } \\ & \phantom{\pregravfr_3^{\mu_1 \nu_1 \vert \mu_2 \nu_2 \vert \mu_3 \nu_3} \left ( p_1^\sigma, p_2^\sigma \right ) = \frac{\gcoupling}{4} \Bigg \{ + \left ( p_1 \cdot p_2 \right ) \bigg (} - \frac{1}{4} \hat{\eta}^{\mu_1 \mu_2 } \hat{\eta}^{\nu_1 \nu_2 } \hat{\eta}^{\mu_3 \nu_3 } + \frac{1}{8} \hat{\eta}^{\mu_1 \nu_1 } \hat{\eta}^{\mu_2 \nu_2 } \hat{\eta}^{\mu_3 \nu_3 } \bigg ) \! \Bigg \}
		\end{split}
	\end{align}
	\end{subequations}
	and
	\begin{subequations}
	\begin{align}
		& \gravfr_4^{\mu_1 \nu_1 \vert \cdots \vert \mu_4 \nu_4} \left ( p_1^\sigma, \cdots , p_4^\sigma \right ) = \frac{\imaginary}{16} \sum_{\mu_i \leftrightarrow \nu_i} \sum_{s \in S_4} \pregravfr_4^{\mu_{s(1)} \nu_{s(1)} \vert \cdots \vert \mu_{s(4)} \nu_{s(4)}} \left ( p_{s(1)}^\sigma, p_{s(2)}^\sigma \right )
		\intertext{with}
		\begin{split}
		& \pregravfr_4^{\mu_1 \nu_1 \vert \cdots \vert \mu_4 \nu_4} \left ( p_1^\sigma, p_2^\sigma \right ) = \frac{\gcoupling}{4} \Bigg \{ - p_1^{\mu_3} p_2^{\nu_3} \hat{\eta}^{\mu_1 \mu_2} \hat{\eta}^{\nu_1 \mu_4} \hat{\eta}^{\nu_2 \nu_4} + 2 p_1^{\mu_3} p_2^{\mu_1} \hat{\eta}^{\nu_1 \mu_2} \hat{\eta}^{\nu_2 \mu_4} \hat{\eta}^{\nu_3 \nu_4} \\
		& \phantom{\pregravfr_4^{\mu_1 \nu_1 \vert \cdots \vert \mu_4 \nu_4} \left ( p_1^\sigma, p_2^\sigma \right ) = \frac{\gcoupling}{4} \Bigg \{} - \frac{1}{2} p_1^{\mu_3} p_2^{\mu_1} \hat{\eta}^{\nu_1 \mu_2} \hat{\eta}^{\nu_2 \nu_3} \hat{\eta}^{\mu_4 \nu_4} + p_1^{\mu_3} p_2^{\nu_3} \hat{\eta}^{\mu_1 \mu_2} \hat{\eta}^{\nu_1 \nu_2} \hat{\eta}^{\mu_4 \nu_4} \\
		& \phantom{\pregravfr_4^{\mu_1 \nu_1 \vert \cdots \vert \mu_4 \nu_4} \left ( p_1^\sigma, p_2^\sigma \right ) = \frac{\gcoupling}{4} \Bigg \{} - \frac{1}{2} p_1^{\mu_3} p_2^{\mu_4} \hat{\eta}^{\mu_1 \mu_2} \hat{\eta}^{ \nu_1 \nu_2} \hat{\eta}^{\nu_3 \nu_4} + p_1^{\mu_3} p_2^{\mu_4} \hat{\eta}^{\mu_1 \mu_2} \hat{\eta}^{\nu_1 \nu_3} \hat{\eta}^{\nu_2 \nu_4} \\
		& \phantom{\pregravfr_4^{\mu_1 \nu_1 \vert \cdots \vert \mu_4 \nu_4} \left ( p_1^\sigma, p_2^\sigma \right ) = \frac{\gcoupling}{4} \Bigg \{} - \frac{1}{2} p_1^{\mu_2} p_2^{\mu_3} \hat{\eta}^{\mu_1 \nu_1} \hat{\eta}^{\nu_2 \mu_4} \hat{\eta}^{\nu_3 \nu_4} + \frac{1}{4} p_1^{\mu_2} p_2^{\mu_1} \hat{\eta}^{\nu_1 \nu_2} \hat{\eta}^{\mu_3 \mu_4} \hat{\eta}^{\nu_3 \nu_4} \\
		& \phantom{\pregravfr_4^{\mu_1 \nu_1 \vert \cdots \vert \mu_4 \nu_4} \left ( p_1^\sigma, p_2^\sigma \right ) =} \! \! \! \! + \left ( p_1 \cdot p_2 \right ) \bigg ( - \frac{1}{16} \hat{\eta}^{\mu_1 \mu_2} \hat{\eta}^{\nu_1 \nu_2} \hat{\eta}^{\mu_3 \nu_3} \hat{\eta}^{\mu_4 \nu_4} + \frac{1}{8} \hat{\eta}^{\mu_1 \mu_2} \hat{\eta}^{\nu_1 \nu_2} \hat{\eta}^{\mu_3 \mu_4} \hat{\eta}^{\nu_3 \nu_4} \\
		& \phantom{\pregravfr_4^{\mu_1 \nu_1 \vert \cdots \vert \mu_4 \nu_4} \left ( p_1^\sigma, p_2^\sigma \right ) = + \left ( p_1 \cdot p_2 \right ) \bigg (} \! \! \! \! + \frac{1}{2} \hat{\eta}^{\mu_1 \mu_2} \hat{\eta}^{\nu_1 \mu_3} \hat{\eta}^{\nu_2 \nu_3} \hat{\eta}^{\mu_4 \nu_4} - \hat{\eta}^{\mu_1 \mu_2} \hat{\eta}^{\nu_1 \mu_3} \hat{\eta}^{\nu_2 \mu_4} \hat{\eta}^{\nu_3 \nu_4} \\
		& \phantom{\pregravfr_4^{\mu_1 \nu_1 \vert \cdots \vert \mu_4 \nu_4} \left ( p_1^\sigma, p_2^\sigma \right ) = + \left ( p_1 \cdot p_2 \right ) \bigg (} \! \! \! \! + \frac{1}{2} \hat{\eta}^{\mu_1 \mu_3} \hat{\eta}^{\nu_1 \nu_3} \hat{\eta}^{\mu_2 \mu_4} \hat{\eta}^{\nu_2 \nu_4} - \frac{1}{2} \hat{\eta}^{\mu_1 \mu_3} \hat{\eta}^{\nu_1 \mu_4} \hat{\eta}^{\mu_2 \nu_3} \hat{\eta}^{\nu_2 \nu_4} \\
		& \phantom{\pregravfr_4^{\mu_1 \nu_1 \vert \cdots \vert \mu_4 \nu_4} \left ( p_1^\sigma, p_2^\sigma \right ) = + \left ( p_1 \cdot p_2 \right ) \bigg (} \! \! \! \! - \frac{1}{4} \hat{\eta}^{\mu_1 \nu_1} \hat{\eta}^{\mu_2 \mu_3} \hat{\eta}^{\nu_2 \nu_3} \hat{\eta}^{\mu_4 \nu_4} + \frac{1}{2} \hat{\eta}^{\mu_1 \nu_1} \hat{\eta}^{\mu_2 \mu_3} \hat{\eta}^{\nu_2 \mu_4} \hat{\eta}^{\nu_3 \nu_4} \\
		& \phantom{\pregravfr_4^{\mu_1 \nu_1 \vert \cdots \vert \mu_4 \nu_4} \left ( p_1^\sigma, p_2^\sigma \right ) = + \left ( p_1 \cdot p_2 \right ) \bigg (} \! \! \! \! + \frac{1}{32} \hat{\eta}^{\mu_1 \nu_1} \hat{\eta}^{\mu_2 \nu_2} \hat{\eta}^{\mu_3 \nu_3} \hat{\eta}^{\mu_4 \nu_4} - \frac{1}{8} \hat{\eta}^{\mu_1 \nu_1} \hat{\eta}^{\mu_2 \nu_2} \hat{\eta}^{\mu_3 \mu_4} \hat{\eta}^{\nu_3 \nu_4} \bigg ) \Bigg \}
		\end{split}
	\end{align}
	\end{subequations}
	We remark that the three- and four-valent graviton vertex Feynman rules agree with the cited literature modulo prefactors and minus signs. Additionally, given the situation of \thmref{thm:ghost-fr}, the three- and four-valent graviton-ghost vertex Feynman rules read as follows:
	\begin{equation} \label{eqn:grav-ghost-fr-three-val}
	\begin{split}
		\gravghostfr_3^{\rho_1 \vert \rho_2 \| \mu_3 \nu_3} \left ( p_2^\sigma, p_3^\sigma \right ) & = \frac{\imaginary \gcoupling}{4} \Bigg \{ - p_2^{\rho_2} \bigg ( p_3^{\mu_3} \hat{\eta}^{\rho_1 \nu_3} + p_3^{\nu_3} \hat{\eta}^{\rho_1 \mu_3} - p_3^{\rho_1} \hat{\eta}^{\mu_3 \nu_3} \bigg ) \\
		& \phantom{= \frac{\imaginary \gcoupling}{8} \Bigg \{} - p_3^{\rho_1} \bigg ( p_2^{\mu_3} \hat{\eta}^{\rho_2 \nu_3} + p_2^{\nu_3} \hat{\eta}^{\rho_2 \mu_3} \bigg ) + p_3^{\rho_2} \bigg ( p_2^{\mu_3} \hat{\eta}^{\rho_1 \nu_3} + p_2^{\nu_3} \hat{\eta}^{\rho_1 \mu_3} \bigg ) \\
		& \phantom{= \frac{\imaginary \gcoupling}{8} \Bigg \{} + \left ( p_2 \cdot p_3 \right ) \bigg ( \hat{\eta}^{\rho_1 \mu_3} \hat{\eta}^{\rho_2 \nu_3} + \hat{\eta}^{\rho_1 \nu_3} \hat{\eta}^{\rho_2 \mu_3} \bigg ) \Bigg \}
	\end{split}
	\end{equation}
	and
	\begin{equation}
	\begin{split}
		& \gravghostfr_4^{\rho_1 \vert \rho_2 \| \mu_3 \nu_3 \vert \mu_4 \nu_4} \left ( p_2^\sigma, p_3^\sigma, p_4^\sigma \right ) = \\
		& \phantom{=} \frac{\imaginary \gcoupling^2}{8} \Bigg \{ p_2^{\rho_2} \bigg ( p_3^{\nu_3} \hat{\eta}^{\rho_1 \mu_4} \hat{\eta}^{\mu_3 \nu_4} + p_3^{\mu_3} \hat{\eta}^{\rho_1 \mu_4} \hat{\eta}^{\nu_3 \nu_4} + p_3^{\nu_3} \hat{\eta}^{\rho_1 \nu_4} \hat{\eta}^{\mu_3 \mu_4} \\
		& \phantom{\frac{\imaginary \gcoupling^2}{16} \Bigg \{ \, p_2^{\rho_2} \bigg (} + p_3^{\mu_3} \hat{\eta}^{\rho_1 \nu_4} \hat{\eta}^{\nu_3 \mu_4}	- p_3^{\mu_4} \hat{\eta}^{\rho_1 \nu_4} \hat{\eta}^{\mu_3 \nu_3} - p_3^{\nu_4} \hat{\eta}^{\rho_1 \mu_4} \hat{\eta}^{\mu_3 \nu_3} \bigg ) \\
		&  \phantom{\frac{\imaginary \gcoupling^2}{16} \Bigg \{} + p_2^{\rho_2} \bigg ( p_4^{\nu_4} \hat{\eta}^{\rho_1 \mu_3} \hat{\eta}^{\mu_4 \nu_3} + p_4^{\mu_4} \hat{\eta}^{\rho_1 \mu_3} \hat{\eta}^{\nu_4 \nu_3} + p_4^{\nu_4} \hat{\eta}^{\rho_1 \nu_3} \hat{\eta}^{\mu_4 \mu_3} \\
		& \phantom{\frac{\imaginary \gcoupling^2}{16} \Bigg \{ \, p_2^{\rho_2} \bigg (} + p_4^{\mu_4} \hat{\eta}^{\rho_1 \nu_3} \hat{\eta}^{\nu_4 \mu_3}	- p_4^{\mu_3} \hat{\eta}^{\rho_1 \nu_3} \hat{\eta}^{\mu_4 \nu_4} - p_4^{\nu_3} \hat{\eta}^{\rho_1 \mu_3} \hat{\eta}^{\mu_4 \nu_4} \bigg ) \\
		& \phantom{\frac{\imaginary \gcoupling^2}{16} \Bigg \{} - p_3^{\rho_2} \bigg ( p_2^{\mu_3} \hat{\eta}^{\rho_1 \mu_4} \hat{\eta}^{\nu_3 \nu_4} + p_2^{\nu_3} \hat{\eta}^{\rho_1 \mu_4} \hat{\eta}^{\mu_3 \nu_4} + p_2^{\mu_3} \hat{\eta}^{\rho_1 \nu_4} \hat{\eta}^{\nu_3 \mu_4} + p_2^{\nu_3} \hat{\eta}^{\rho_1 \nu_4} \hat{\eta}^{\mu_3 \mu_4} \bigg ) \\
		& \phantom{\frac{\imaginary \gcoupling^2}{16} \Bigg \{} - p_4^{\rho_2} \bigg ( p_2^{\mu_4} \hat{\eta}^{\rho_1 \mu_3} \hat{\eta}^{\nu_4 \nu_3} + p_2^{\nu_4} \hat{\eta}^{\rho_1 \mu_3} \hat{\eta}^{\mu_4 \nu_3} + p_2^{\mu_4} \hat{\eta}^{\rho_1 \nu_3} \hat{\eta}^{\nu_4 \mu_3} + p_2^{\nu_4} \hat{\eta}^{\rho_1 \nu_3} \hat{\eta}^{\mu_4 \mu_3} \bigg ) \\
		& \phantom{\frac{\imaginary \gcoupling^2}{16} \Bigg \{} + p_2^{\mu_3} p_3^{\mu_4} \hat{\eta}^{\rho_1 \nu_4} \hat{\eta}^{\rho_2 \nu_3} + p_2^{\nu_3} p_3^{\mu_4} \hat{\eta}^{\rho_1 \nu_4} \hat{\eta}^{\rho_2 \mu_3} + p_2^{\mu_3} p_3^{\nu_4} \hat{\eta}^{\rho_1 \mu_4} \hat{\eta}^{\rho_2 \nu_3} + p_2^{\nu_3} p_3^{\nu_4} \hat{\eta}^{\rho_1 \mu_4} \hat{\eta}^{\rho_2 \mu_3} \\
		& \phantom{\frac{\imaginary \gcoupling^2}{16} \Bigg \{} + p_2^{\mu_4} p_4^{\mu_3} \hat{\eta}^{\rho_1 \nu_3} \hat{\eta}^{\rho_2 \nu_4} + p_2^{\mu_4} p_4^{\nu_3} \hat{\eta}^{\rho_1 \mu_3} \hat{\eta}^{\rho_2 \nu_4} + p_2^{\nu_4} p_4^{\mu_3} \hat{\eta}^{\rho_1 \nu_3} \hat{\eta}^{\rho_2 \mu_4} + p_2^{\nu_4} p_4^{\nu_3} \hat{\eta}^{\rho_1 \mu_3} \hat{\eta}^{\rho_2 \mu_4} \\
		& \phantom{\frac{\imaginary \gcoupling^2}{16} \Bigg \{} - \left ( p_2 \cdot p_3 \right ) \bigg ( \hat{\eta}^{\rho_1 \mu_4} \hat{\eta}^{\rho_2 \mu_3} \hat{\eta}^{\nu_3 \nu_4} +  \hat{\eta}^{\rho_1 \mu_4} \hat{\eta}^{\rho_2 \nu_3} \hat{\eta}^{\mu_3 \nu_4} \\
		& \phantom{\frac{\imaginary \gcoupling^2}{16} \Bigg \{ - \left ( p_2 \cdot p_3 \right ) \bigg (} + \hat{\eta}^{\rho_1 \nu_4} \hat{\eta}^{\rho_2 \mu_3} \hat{\eta}^{\nu_3 \mu_4} +  \hat{\eta}^{\rho_1 \nu_4} \hat{\eta}^{\rho_2 \nu_3} \hat{\eta}^{\mu_3 \mu_4} \bigg ) \\
		& \phantom{\frac{\imaginary \gcoupling^2}{16} \Bigg \{} - \left ( p_2 \cdot p_4 \right ) \bigg ( \hat{\eta}^{\rho_1 \mu_3} \hat{\eta}^{\rho_2 \mu_4} \hat{\eta}^{\nu_4 \nu_3} +  \hat{\eta}^{\rho_1 \mu_3} \hat{\eta}^{\rho_2 \nu_4} \hat{\eta}^{\mu_4 \nu_3} \\
		& \phantom{\frac{\imaginary \gcoupling^2}{16} \Bigg \{ - \left ( p_2 \cdot p_4 \right ) \bigg (} +  \hat{\eta}^{\rho_1 \nu_3} \hat{\eta}^{\rho_2 \mu_4} \hat{\eta}^{\nu_4 \mu_3} +  \hat{\eta}^{\rho_1 \nu_3} \hat{\eta}^{\rho_2 \nu_4} \hat{\eta}^{\mu_4 \mu_3} \bigg ) \Bigg \}
	\end{split}
	\end{equation}
\end{exmp}

\subsection{Gravitons and matter} \label{sec:fr-gravitons_and_matter}

Having done all preparations in \ssecref{ssec:preparations_graviton_graviton_ghost}, we now list the corresponding Feynman rules for the interactions of gravitons with matter from the Standard Model. To this end, we state the Feynman rules for the interactions according to the classification of \lemref{lem:matter-model-Lagrange-densities} and refer for the corresponding matter contributions to \cite{Romao_Silva} in order to keep this dissertation at a reasonable length.

\enter

\begin{thm} \label{thm:matter-fr}
	Given the situation of \thmref{thm:grav-fr} and the matter-model Lagrange densities from \lemref{lem:matter-model-Lagrange-densities}, the graviton-matter \(n\)-point vertex Feynman rule for (effective) Quantum General Relativity coupled to the matter-model Lagrange density of type \(k\) reads:
	\begin{equation}
	\matterfrk_n^{\mu_1 \nu_1 \vert \cdots \vert \mu_n \nu_n} \left ( p_1^\sigma, \cdots, p_n^\sigma \right ) = \frac{\imaginary}{2^n} \sum_{\mu_i \leftrightarrow \nu_i} \sum_{s \in S_n} \prematterfrk_n^{\mu_{s(1)} \nu_{s(1)} \vert \cdots \vert \mu_{s(n)} \nu_{s(n)}} \left ( p_{s(1)}^\sigma, \cdots, p_{s(n)}^\sigma \right )
	\end{equation}
	with
	{\allowdisplaybreaks
	\begin{align}
		\prematterfri_n^{\mu_1 \nu_1 \vert \cdots \vert \mu_n \nu_n} \left ( \tensor[_1]{\! \widehat{T}}{} \right ) & = \tensor[_1]{\! \widehat{T}}{} \mathfrak{v}_n^{\mu_1 \nu_1 \vert \cdots \vert \mu_n \nu_n} \, , \\
		\begin{split}
			\prematterfrii_n^{\mu_1 \nu_1 \vert \cdots \vert \mu_n \nu_n} \left ( \tensor[_2]{\! \widehat{T}}{} \right ) & = \tensor[_2]{\! \widehat{T}}{_\mu _\nu} \sum_{m_1 + m_2 = n} \left ( -1 \right )^{m_1} \mathfrak{h}_{m_1}^{\mu \nu \triplevert \mu_1 \nu_1 \vert \cdots \vert \mu_{m_1} \nu_{m_1}} \\ & \hphantom{=} \times \mathfrak{v}_{m_2}^{\mu_{{m_1} + 1} \nu_{{m_1} + 1} \vert \cdots \vert \mu_n \nu_n} \, ,
		\end{split}
		\\
		\begin{split}
			\prematterfriii_n^{\mu_1 \nu_1 \vert \cdots \vert \mu_n \nu_n} \left ( \tensor[_3]{\! \widehat{T}}{} \right ) & = \tensor[_3]{\! \widehat{T}}{_\mu _\nu _\rho _\sigma} \sum_{\subalign{m_1 & + m_2 \\ & + m_3 = n}} \left ( - 1 \right )^{m_1 + m_2} \mathfrak{h}_{m_1}^{\mu \nu \triplevert \mu_1 \nu_1 \vert \cdots \vert \mu_{m_1} \nu_{m_1}} \\ & \hphantom{=} \times \mathfrak{h}_{m_2}^{\rho \sigma \triplevert \mu_{{m_1} + 1} \nu_{{m_1} + 1} \vert \cdots \vert \mu_{{m_1} + {m_2}} \nu_{{m_1} + {m_2}}} \\ & \hphantom{=} \times \mathfrak{v}_{m_3}^{\mu_{{m_1} + {m_2} + 1} \nu_{{m_1} + {m_2} + 1} \vert \cdots \vert \mu_n \nu_n} \, ,
		\end{split}
		\\
		\begin{split}
			\prematterfriv_n^{\mu_1 \nu_1 \vert \cdots \vert \mu_n \nu_n} \left ( \tensor[_4]{\! \widehat{T}}{}; p_1^\sigma \right ) & = \tensor[_4]{\! \widehat{T}}{_\rho} \boldsymbol{\Gamma}^{\mu_1 \nu_1}_{\mu \nu \sigma} \left ( p_1^\sigma \right ) \sum_{\subalign{m_1 & + m_2 \\ & + m_3 = n - 1}} \left ( - 1 \right )^{m_1 + m_2} \\ & \hphantom{=} \times \mathfrak{h}_{m_1}^{\mu \nu \triplevert \mu_2 \nu_2 \vert \cdots \vert \mu_{m_1 + 1} \nu_{m_1 + 1}} \\ & \hphantom{=} \times \mathfrak{h}_{m_2}^{\rho \sigma \triplevert \mu_{{m_1} + 2} \nu_{{m_1} + 2} \vert \cdots \vert \mu_{{m_1} + {m_2} + 1} \nu_{{m_1} + {m_2} + 1}} \\ & \hphantom{=} \times \mathfrak{v}_{m_3}^{\mu_{{m_1} + {m_2} + 2} \nu_{{m_1} + {m_2} + 2} \vert \cdots \vert \mu_n \nu_n} \, ,
		\end{split}
		\\
		\begin{split}
			\prematterfrv_n^{\mu_1 \nu_1 \vert \cdots \vert \mu_n \nu_n} \left ( \tensor[_5]{\! \widehat{T}}{}; p_1^\sigma \right ) & = \tensor[_5]{\! \widehat{T}}{_\rho _\sigma _\kappa} \boldsymbol{\Gamma}^{\mu_1 \nu_1}_{\mu \nu \lambda} \left ( p_1^\sigma \right ) \sum_{\subalign{m_1 & + m_2 + m_3 \\ & + m_4 = n - 1}} \left ( - 1 \right )^{m_1 + m_2 + m_3} \\ & \hphantom{=} \times \mathfrak{h}_{m_1}^{\mu \nu \triplevert \mu_2 \nu_2 \vert \cdots \vert \mu_{m_1 + 1} \nu_{m_1 + 1}} \\ & \hphantom{=} \times \mathfrak{h}_{m_2}^{\rho \sigma \triplevert \mu_{{m_1} + 2} \nu_{{m_1} + 2} \vert \cdots \vert \mu_{{m_1} + {m_2} + 1} \nu_{{m_1} + {m_2} + 1}} \\ & \hphantom{=} \times \mathfrak{h}_{m_3}^{\kappa \lambda \triplevert \mu_{{m_1} + {m_2} + 2} \nu_{{m_1} + {m_2} + 2} \vert \cdots \vert \mu_{{m_1} + {m_2} + {m_3} + 1} \nu_{{m_1} + {m_2} + {m_3} + 1}} \\ & \hphantom{=} \times \mathfrak{v}_{m_4}^{\mu_{{m_1} + {m_2} + {m_3} + 2} \nu_{{m_1} + {m_2} + {m_3} + 2} \vert \cdots \vert \mu_n \nu_n} \, ,
		\end{split}
		\\
		\begin{split}
			\prematterfrvi_n^{\mu_1 \nu_1 \vert \cdots \vert \mu_n \nu_n} \left ( \tensor[_6]{\! \widehat{T}}{}; p_1^\sigma, p_2^\sigma \right ) & = \tensor[_6]{\! \widehat{T}}{_\kappa _\iota} \boldsymbol{\Gamma}^{\mu_1 \nu_1}_{\mu \nu \lambda} \left ( p_1^\sigma \right ) \boldsymbol{\Gamma}^{\mu_2 \nu_2}_{\rho \sigma \tau} \left ( p_2^\sigma \right ) \sum_{\subalign{m_1 & + m_2 + m_3 \\ & + m_4 + m_5 = n - 2}} \\ & \hphantom{=} \times \left ( - 1 \right )^{m_1 + m_2 + m_3 + m_4} \mathfrak{h}_{m_1}^{\mu \nu \triplevert \mu_3 \nu_3 \vert \cdots \vert \mu_{m_1 + 2} \nu_{m_1 + 2}} \\ & \hphantom{=} \times \mathfrak{h}_{m_2}^{\rho \sigma \triplevert \mu_{{m_1} + {m_2} + 3} \nu_{{m_1} + {m_2} + 3} \vert \cdots \vert \mu_{{m_1} + {m_2} + 2} \nu_{{m_1} + {m_2} + 2}} \\ & \hphantom{=} \times \mathfrak{h}_{m_3}^{\kappa \lambda \triplevert \mu_{{m_1} + 1} \nu_{{m_1} + 1} \vert \cdots \vert \mu_{{m_1} + {m_2}} \nu_{{m_1} + {m_2}}} \\ & \hphantom{=} \times \mathfrak{h}_{m_4}^{\iota \tau \triplevert \mu_{{m_1} + 1} \nu_{{m_1} + 1} \vert \cdots \vert \mu_{{m_1} + {m_2}} \nu_{{m_1} + {m_2}}} \mathfrak{v}_{m_5}^{\mu_{m + 1} \nu_{m + 1} \vert \cdots \vert \mu_n \nu_n} \, ,
		\end{split}
		\\
		\begin{split}
			\prematterfrvii_n^{\mu_1 \nu_1 \vert \cdots \vert \mu_n \nu_n} \left ( \tensor[_7]{\! \widehat{T}}{} \right ) & = \tensor[_7]{\! \widehat{T}}{_o} \sum_{m_1 + m_2 = n} \binom{\frac{1}{2}}{m_1} \hat{\eta}_{0 \upsilon} \mathfrak{h}_{m_1}^{\upsilon o \triplevert \mu_1 \nu_1 \vert \cdots \vert \mu_{m_1} \nu_{m_1}} \\ & \hphantom{=} \times \mathfrak{v}_{m_2}^{\mu_{{m_1} + 1} \nu_{{m_1} + 1} \vert \cdots \vert \mu_n \nu_n} \, ,
		\end{split}
		\\
		\begin{split}
			\prematterfrviii_n^{\mu_1 \nu_1 \vert \cdots \vert \mu_n \nu_n} \left ( \tensor[_8]{\! \widehat{T}}{} \right ) & = \tensor[_8]{\! \widehat{T}}{_o _\rho _r} \sum_{\subalign{m_1 & + m_2 \\ & + m_3 = n}} \binom{\frac{1}{2}}{m_1} \binom{- \frac{1}{2}}{m_2} \hat{\eta}_{0 \upsilon} \mathfrak{h}_{m_1}^{\upsilon o \triplevert \mu_1 \nu_1 \vert \cdots \vert \mu_{m_1} \nu_{m_1}} \\ & \hphantom{=} \times \mathfrak{h}_{m_2}^{\rho r \triplevert \mu_{m_1 + 1} \nu_{m_1 + 1} \vert \cdots \vert \mu_{m_1 + m_2} \nu_{m_1 + m_2}} \\ & \hphantom{=} \times \mathfrak{v}_{m_3}^{\mu_{m_1 + m_2 + 1} \nu_{m_1 + m_2 + 1} \vert \cdots \vert \mu_n \nu_n} \, ,
		\end{split}
		\\
		\begin{split}
			\prematterfrix_n^{\mu_1 \nu_1 \vert \cdots \vert \mu_n \nu_n} \left ( \tensor[_9]{\! \widehat{T}}{}; p_1^\sigma, \cdots, p_n^\sigma \right ) & = \tensor[_9]{\! \widehat{T}}{_o _r _s _t} \sum_{\subalign{m_1 & + m_2 + m_3 \\ & + m_4 + m_5 = n}} \binom{\frac{1}{2}}{m_1} \binom{- \frac{1}{2}}{m_2} \binom{- \frac{1}{2}}{m_3} \binom{\frac{1}{2}}{m_4} \\ & \hphantom{=} \mkern-136mu \times \hat{\eta}_{0 \upsilon} \mathfrak{h}_{m_1}^{\upsilon o \triplevert \mu_1 \nu_1 \vert \cdots \vert \mu_{m_1} \nu_{m_1}} \\ & \hphantom{=} \mkern-136mu \times \mathfrak{h}_{m_2}^{\rho r \triplevert \mu_{m_1 + 1} \nu_{m_1 + 1} \vert \cdots \vert \mu_{m_1 + m_2} \nu_{m_1 + m_2}} \\ & \hphantom{=} \mkern-136mu \times \mathfrak{h}_{m_3}^{\sigma s \triplevert \mu_{m_1 + m_2 + 1} \nu_{m_1 + m_2 + 1} \vert \cdots \vert \mu_{m_1 + m_2 + m_3} \nu_{m_1 + m_2 + m_3}} \\ & \hphantom{=} \mkern-136mu \times \hat{\eta}_{\sigma \tau} \left ( \mathfrak{h}_{m_4}^\prime \right )_\rho^{\tau t \triplevert \mu_{m_1 + m_2 + m_3 + 1} \nu_{m_1 + m_2 + m_3 + 1} \vert \cdots \vert \mu_{m_1 + m_2 + m_3 + m_4} \nu_{m_1 + m_2 + m_3 + m_4}} \\ & \hphantom{\times \hat{\eta}_{\sigma \tau} \left ( \mathfrak{h}_{m_4}^\prime \right )} \mkern-136mu \left ( p_{m_1 + m_2 + m_3 + 1}^{\sigma_{m_1 + m_2 + m_3 + 1}}, \cdots, p_{m_1 + m_2 + m_3 + m_4}^{\sigma_{m_1 + m_2 + m_3 + m_4}} \right ) \\ & \hphantom{=} \mkern-136mu \times \mathfrak{v}_{m_5}^{\mu_{m_1 + m_2 + m_3 + m_4 + 1} \nu_{m_1 + m_2 + m_3 + m_4 + 1} \vert \cdots \vert \mu_n \nu_n} \, ,
		\end{split}
		\intertext{and}
		\begin{split}
			\prematterfrx_n^{\mu_1 \nu_1 \vert \cdots \vert \mu_n \nu_n} \left ( \tensor[_{10}]{\! \widehat{T}}{}; p_1^\sigma \right ) & = \tensor[_{10}]{\! \widehat{T}}{_o _r _s _t}  \boldsymbol{\Gamma}^{\mu_1 \nu_1}_{\rho \sigma \tau} \left ( p_1^\sigma \right ) \sum_{\subalign{m_1 & + m_2 + m_3 \\ & + m_4 + m_5 = n}} \binom{\frac{1}{2}}{m_1} \binom{- \frac{1}{2}}{m_2} \binom{- \frac{1}{2}}{m_3} \\ & \hphantom{=} \mkern-100mu \times \binom{- \frac{1}{2}}{m_4} \hat{\eta}_{0 \upsilon} \mathfrak{h}_{m_1}^{\upsilon o \triplevert \mu_2 \nu_2 \vert \cdots \vert \mu_{m_1 + 1} \nu_{m_1 + 1}} \\ & \hphantom{=} \mkern-100mu \times \mathfrak{h}_{m_2}^{\rho r \triplevert \mu_{m_1 + 2} \nu_{m_1 + 2} \vert \cdots \vert \mu_{m_1 + m_2 + 1} \nu_{m_1 + m_2 + 1}} \\ & \hphantom{=} \mkern-100mu \times \mathfrak{h}_{m_3}^{\sigma s \triplevert \mu_{m_1 + m_2 + 2} \nu_{m_1 + m_2 + 2} \vert \cdots \vert \mu_{m_1 + m_2 + m_3} \nu_{m_1 + m_2 + m_3}} \\ & \hphantom{=} \mkern-100mu \times \mathfrak{h}_{m_4}^{\tau t \triplevert \mu_{m_1 + m_2 + m_3 + 2} \nu_{m_1 + m_2 + m_3 + 2} \vert \cdots \vert \mu_{m_1 + m_2 + m_3 + m_4 + 1} \nu_{m_1 + m_2 + m_3 + m_4 + 1}} \\ & \hphantom{=} \mkern-100mu \times \mathfrak{v}_{m_5}^{\mu_{m_1 + m_2 + m_3 + m_4 + 2} \nu_{m_1 + m_2 + m_3 + m_4 + 2} \vert \cdots \vert \mu_n \nu_n} \, .
		\end{split}
	\end{align}
	}%
\end{thm}

\begin{proof}
	This follows directly from \colref{col:inverse_metric_vielbeins_FR} with Lemmata~\ref{lem:Christoffel_FR}~and~\ref{lem:Riemannian_volume_form_FR}.
\end{proof}

\section{Explicit Feynman rules} \label{sec:explicit_feynman_rules}

After these general results, we additionally provide the concrete gravity-matter Feynman rules for all propagators and three-valent vertices of (effective) Quantum General Relativity coupled to the Standard Model. In this section and the section thereafter, we use the symmetrized fermion Lagrange density and the following symmetric (hermitian) ghost Lagrange densities associated to the Lorenz and de Donder gauge fixing conditions, respectively. For Quantum Yang--Mills theory, we propose the following Lagrange density as a direct generalization of \cite[Equation (3.6)]{Baulieu_Thierry-Mieg} to include the couplings of gauge ghosts and gauge bosons to gravitons:
\begin{align}
	\begin{split}
		\mathcal{L}_{\textup{QYM-Sym-Ghost}} & = - \frac{1}{\xi} g^{\mu \nu} \left ( \big ( \partial_\mu \overline{c}_a \big ) \big ( \partial_\mu c^a \big ) \right ) \dif V_g \\ & \phantom{:=} - \frac{\mathrm{g}}{2} g^{\mu \nu} \tensor{f}{^a _b _c} \left ( \big ( \partial_\mu \overline{c}_a \big ) \big ( c^b A^c_\nu \big ) + \big ( \overline{c}_a A^b_\nu \big ) \big ( \partial_\mu c^c \big ) \right ) \dif V_g \\ & \phantom{:=} - \frac{\mathrm{g}^2}{8} g^{\mu \nu} \tensor{f}{^a ^b _c} \tensor{f}{_a ^d ^e} \overline{c}_b \overline{c}_d c^c c^e
	\end{split} \label{eqn:qym-sym-ghost}
	\intertext{Then, applying the same reasoning to (effective) Quantum General Relativity, we propose the following symmetric Lagrange density for graviton-ghosts:}
	\begin{split}
		\mathcal{L}_{\textup{QGR-Sym-Ghost}} & = - \frac{1}{4 \zeta} \eta^{\rho \sigma} \Big ( \overline{C}^\mu \big ( \partial_\rho \partial_\sigma C_\mu \big ) + \big ( \partial_\rho \partial_\sigma \overline{C}_\mu \big ) C^\mu \Big ) \dif V_\eta \\ & \phantom{=} - \frac{1}{4} \eta^{\rho \sigma} \Big ( \big ( \partial_\mu \overline{C}^\mu \big ) \big ( \tensor{\Gamma}{^\nu _\rho _\sigma} C_\nu \big ) - 2 \big ( \partial_\rho \overline{C}^\mu \big ) \big ( \tensor{\Gamma}{^\nu _\mu _\sigma} C_\nu \big ) \\ & \phantom{= + \frac{1}{4} \eta^{\rho \sigma} \Big (} + \big ( \tensor{\Gamma}{^\nu _\rho _\sigma} \overline{C}_\nu \big ) \big ( \partial_\mu C^\mu \big ) - 2 \big ( \tensor{\Gamma}{^\nu _\mu _\sigma} \overline{C}_\nu \big ) \big ( \partial_\rho C^\mu \big ) \Big ) \dif V_\eta \\ & \phantom{=} - \frac{\varkappa^2}{16} \eta^{\rho \sigma} \bigg ( \overline{C}^\mu \big ( \partial_\mu C_\rho \big ) \overline{C}^\nu \big ( \partial_\nu C_\sigma \big ) - \big ( \partial_\rho \overline{C}^\mu \big ) C_\mu \overline{C}^\nu \big ( \partial_\nu C_\sigma \big ) \\ & \phantom{= + \frac{\varkappa^2}{16} \eta^{\rho \sigma} \bigg (} - \overline{C}^\mu \big ( \partial_\mu C_\rho \big ) \big ( \partial_\sigma \overline{C}^\nu \big ) C_\nu + \big ( \partial_\rho \overline{C}^\mu \big ) C_\mu \big ( \partial_\sigma \overline{C}^\nu \big ) C_\nu \! \bigg ) \dif V_\eta
	\end{split} \label{eqn:qgr-sym-ghost}
\end{align}
We remark that a proper derivation thereof via the diffeomorphism and gauge BRST and anti-BRST operators from \sectionref{sec:diffeomorphism-gauge-brst-double-complex} will be provided in \cite{Prinz_6}. In particular, we emphasize that, compared to the analysis of Quantum Yang--Mills theory in \cite{Baulieu_Thierry-Mieg}, there are now a priori infinitely many possible monomials, which makes this task quite involved. Additionally, we remark the presence of four-valent gauge ghost and graviton-ghost interactions. However, we also highlight the benefit of significantly simpler cancellation identities, which is the reason for choosing them here, cf.\ \thmsaref{thm:three-valent-contraction-identities-qym-with-matter}{thm:three-valent-contraction-identities-qgr-with-matter}. From this section onwards, we also omit the hats to indicate Fourier transformed quantities, in particular on the Minkowski metric \(\eta\) and on the Kronecker symbol \(\delta\), to improve readability:

\subsection{Gravity-matter propagators} \label{ssec:gravity-matter-propagators}

{\allowdisplaybreaks
\begin{align}
	\Phi \left (  \cgreen{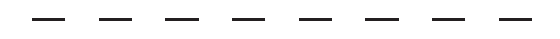} \right ) & = \frac{\imaginary}{p^2 - m^2 + \imaginary \epsilon} \\
	\Phi \left (  \cgreen{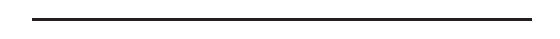} \right ) & = \frac{\imaginary \big ( \slashed{p} + m \big )}{p^2 - m^2 + \imaginary \epsilon} \\
	\Phi \left (  \cgreen{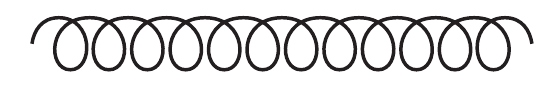} \right ) & = - \frac{\imaginary}{p^2 + \imaginary \epsilon} \delta^{a_1 a_2} \left ( \eta_{\rho_1 \rho_2} - \frac{\left ( 1 - \xi \right )}{p^2} p_{\rho_1} p_{\rho_2}  \right ) \\
	\Phi \left (  \cgreen{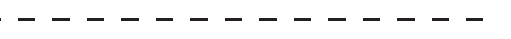} \right ) & = - \frac{\imaginary \xi}{p^2 + \imaginary \epsilon} \delta^{a_1 a_2} \\
	\begin{split}
		\Phi \left (  \cgreen{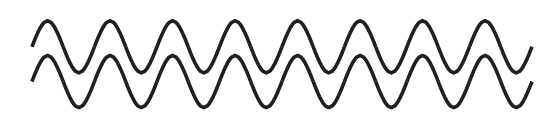} \right ) & = - \frac{2 \imaginary}{p^2 + \imaginary \epsilon} \Bigg [ \left ( \eta_{\mu_1 \mu_2} \eta_{\nu_1 \nu_2} + \eta_{\mu_1 \nu_2} \eta_{\nu_1 \mu_2} - \eta_{\mu_1 \nu_1} \eta_{\mu_2 \nu_2} \right ) \\
		& \phantom{= - [} - \left ( \frac{1 - \zeta}{p^2} \right ) \left ( \eta_{\mu_1 \mu_2} p_{\nu_1} p_{\nu_2} + \eta_{\mu_1 \nu_2} p_{\nu_1} p_{\mu_2} + \eta_{\nu_1 \mu_2} p_{\mu_1} p_{\nu_2} + \eta_{\nu_1 \nu_2} p_{\mu_1} p_{\mu_2} \right ) \Bigg ]
	\end{split} \\
	\Phi \left (  \cgreen{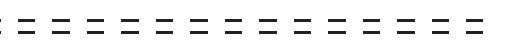} \right ) & = - \frac{2 \imaginary \zeta}{p^2 + \imaginary \epsilon} \eta_{\rho_1 \rho_2}
\end{align}
}%

\subsection{Gravity-matter vertices} \label{ssec:gravity-matter-vertices}

{\allowdisplaybreaks
\begin{align}
	\Phi \left ( \tcgreen{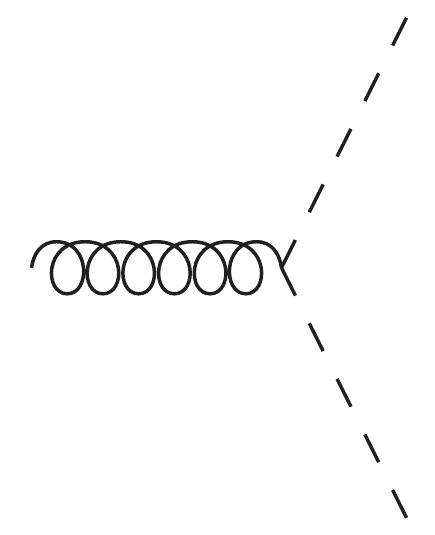} \right ) & = - \frac{\imaginary \mathrm{g}}{2} \left ( q_1 - q_2 \right )^\rho \mathfrak{H}_{a k l} \\
	\Phi \left ( \tcgreen{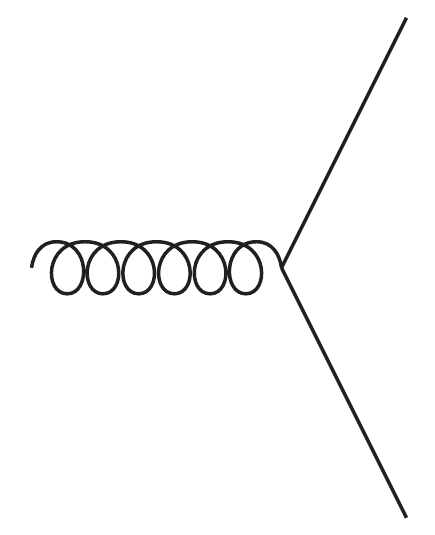} \right ) & = - \imaginary \mathrm{g} \gamma^\rho \mathfrak{S}_{a k l} \\
	\Phi \left ( \tcgreen{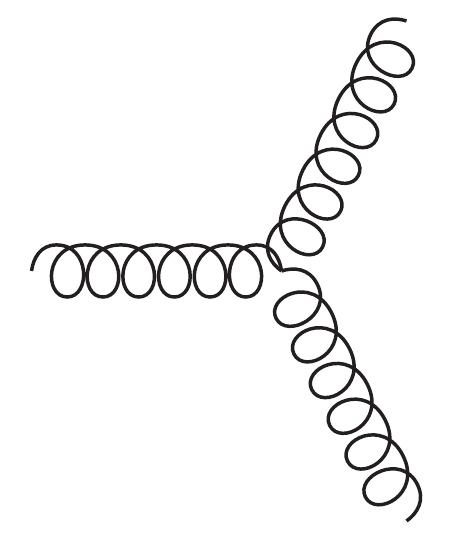} \right ) & = - \mathrm{g} f_{a_1 a_2 a_3} \sum_{s \in S_3} \left ( \eta^{\rho_{s(1)} \rho_{s(2)}} \left ( p_{s(1)} - p_{s(2)} \right )^{\rho_{s(3)}} \right ) \\
	\Phi \left ( \tcgreen{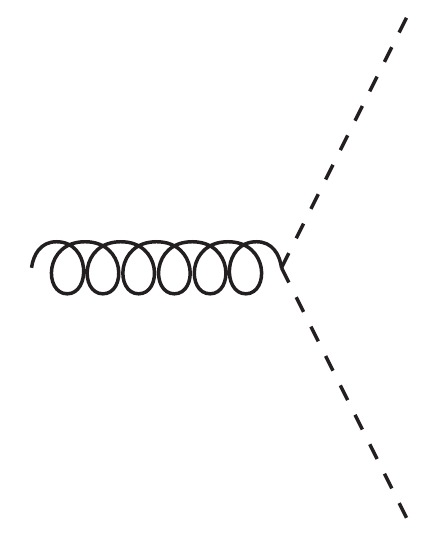} \right ) & = \frac{\mathrm{\imaginary g}}{2} f_{a b_1 b_2} \left ( q_1 - q_2 \right )^\rho \\
	\Phi \left ( \tcgreen{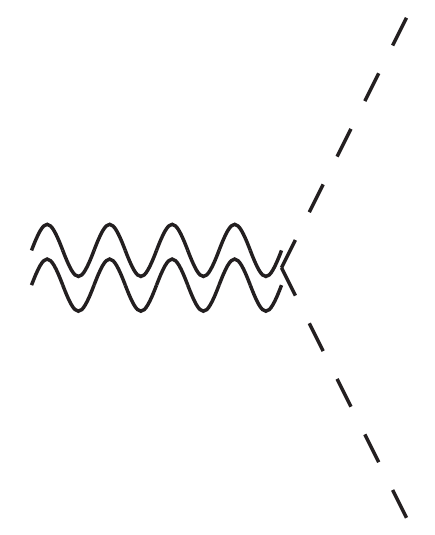} \right ) & = \frac{\imaginary \varkappa}{2} \left ( q_1^\mu q_2^\nu + q_1^\nu q_2^\mu - \eta^{\mu \nu} \left ( q_1 \cdot q_2 + m^2 \right ) \right ) \\
	\Phi \left ( \tcgreen{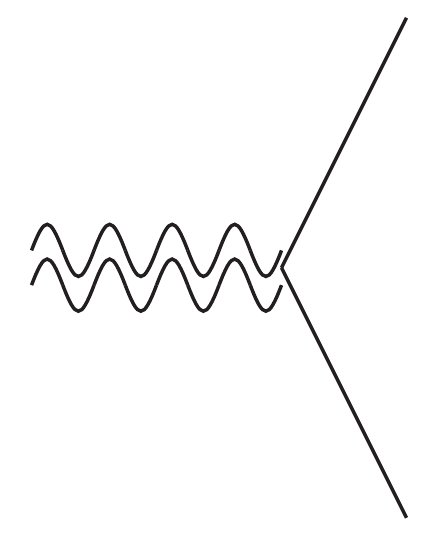} \right ) & = \frac{\imaginary \varkappa}{8} \left ( 2 \eta^{\mu \nu} \left ( \slashed{q}_1 - \slashed{q}_2 - 2m \right ) - \left ( q_1 - q_2 \right )_\mu \gamma_\nu - \left ( q_1 - q_2 \right )_\nu \gamma_\mu \right ) \\
	\begin{split}
		\Phi \left ( \tcgreen{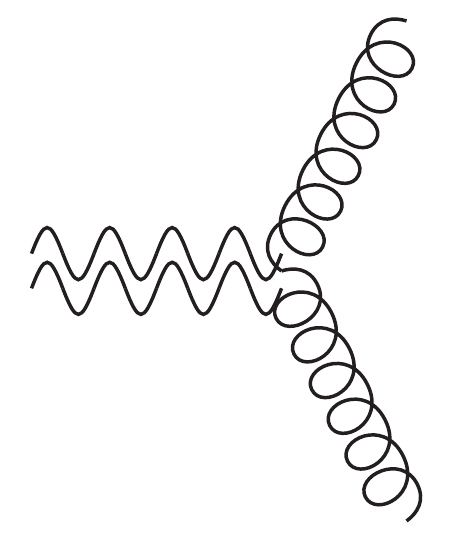} \right ) & = \frac{\imaginary \varkappa}{2} \delta_{a_1 a_2} \Bigg ( \! \left ( q_1 \cdot q_2 \right ) \left ( \eta^{\mu \nu} \eta^{\rho_1 \rho_2} - \eta^{\mu \rho_1} \eta^{\nu \rho_2} - \eta^{\mu \rho_2} \eta^{\nu \rho_1} \right ) \\
		& \phantom{= \frac{\imaginary \varkappa}{2} \delta_{a_1 a_2} \Bigg ( \!} - \eta^{\mu \nu} q_1^{\rho_2} q_2^{\rho_1} - \eta^{\rho_1 \rho_2} \left ( q_1^\mu q_2^\nu + q_2^\mu q_1^\nu \right ) \\
		& \phantom{= \frac{\imaginary \varkappa}{2} \delta_{a_1 a_2} \Bigg ( \!} + q_1^{\rho_2} \left ( \eta^{\mu \rho_1} q_2^\nu + \eta^{\nu \rho_1} q_2^\mu \right ) + q_2^{\rho_1} \left ( \eta^{\mu \rho_2} q_1^\nu + \eta^{\nu \rho_2} q_1^\mu \right ) \\
		& \phantom{= \frac{\imaginary \varkappa}{2} \delta_{a_1 a_2} \Bigg ( \!} - \frac{1}{\xi} \eta^{\mu \nu} \left ( q_1^{\rho_1} q_2^{\rho_2} + p^{\rho_1} q_2^{\rho_2} + p^{\rho_2} q_1^{\rho_1} \right ) \\
		& \phantom{= \frac{\imaginary \varkappa}{2} \delta_{a_1 a_2} \Bigg ( \!} + \frac{1}{\xi} q_1^{\rho_1} \left ( \eta^{\mu \rho_2} q_2^\nu + \eta^{\nu \rho_2} q_2^\mu + \eta^{\mu \rho_2} p^\nu + \eta^{\nu \rho_2} p^\mu \right ) \\
		& \phantom{= \frac{\imaginary \varkappa}{2} \delta_{a_1 a_2} \Bigg ( \!} + \frac{1}{\xi} q_2^{\rho_2} \left ( \eta^{\mu \rho_1} q_1^\nu + \eta^{\nu \rho_1} q_1^\mu + \eta^{\mu \rho_1} p^\nu + \eta^{\nu \rho_1} p^\mu \right ) \! \Bigg )
	\end{split} \\
	\Phi \left ( \tcgreen{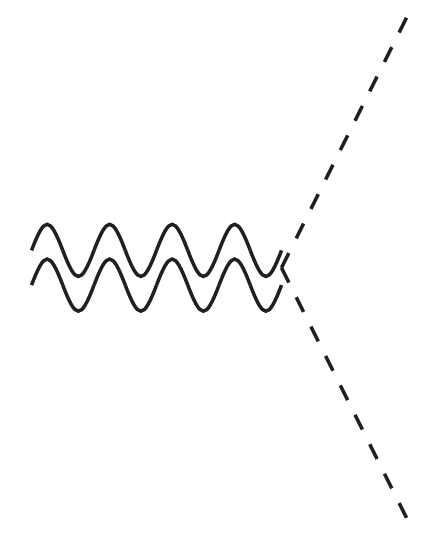} \right ) & = \frac{\imaginary \varkappa}{2 \xi} \Big ( \! \left ( q_1 \cdot q_2 \right ) \eta^{\mu \nu} - q_1^\mu q_2^\nu - q_2^\mu q_1^\nu \Big ) \\
	\begin{split}
		\Phi \left ( \tcgreen{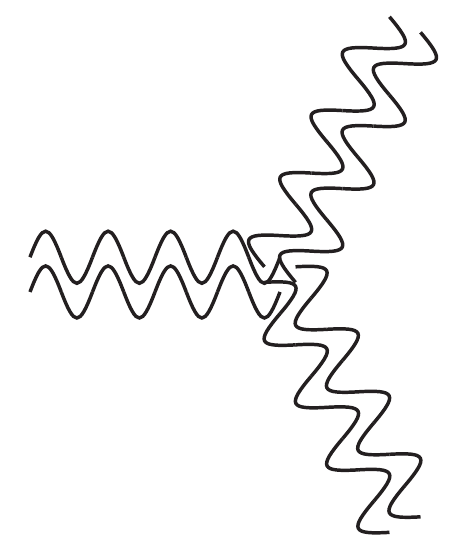} \right ) & = - \frac{\imaginary \gcoupling}{32} \sum_{\mu_i \leftrightarrow \nu_i} \sum_{s \in S_3} \Bigg ( p_{s(1)}^{\mu_{s(3)}} p_{s(2)}^{\mu_{s(1)}} \eta^{\nu_{s(1)} \mu_{s(2)}} \eta^{\nu_{s(2)} \nu_{s(3)}} - \frac{1}{2} p_{s(1)}^{\mu_{s(3)}} p_{s(2)}^{\nu_{s(3)}} \eta^{\mu_{s(1)} \mu_{s(2)}} \eta^{\nu_{s(1)} \nu_{s(2)}} \\
			& + \left ( p_{s(1)} \cdot p_{s(2)} \right ) \bigg ( \, \frac{1}{2} \eta^{\mu_{s(1)} \nu_{s(1)} } \eta^{\mu_{s(2)} \mu_{s(3)} } \eta^{\nu_{s(2)} \nu_{s(3)} } - \eta^{\mu_{s(1)} \nu_{s(2)} } \eta^{\mu_{s(2)} \nu_{s(3)} } \eta^{\mu_{s(3)} \nu_{s(1)} } \\
			& \phantom{+ \left ( p_{s(1)} \cdot p_{s(2)} \right ) \bigg (} + \frac{1}{4} \eta^{\mu_{s(1)} \mu_{s(2)} } \eta^{\nu_{s(1)} \nu_{s(2)} } \eta^{\mu_{s(3)} \nu_{s(3)} } - \frac{1}{8} \eta^{\mu_{s(1)} \nu_{s(1)} } \eta^{\mu_{s(2)} \nu_{s(2)} } \eta^{\mu_{s(3)} \nu_{s(3)} } \bigg ) \! \Bigg )
	\end{split} \\
	\begin{split}
		\Phi \left ( \tcgreen{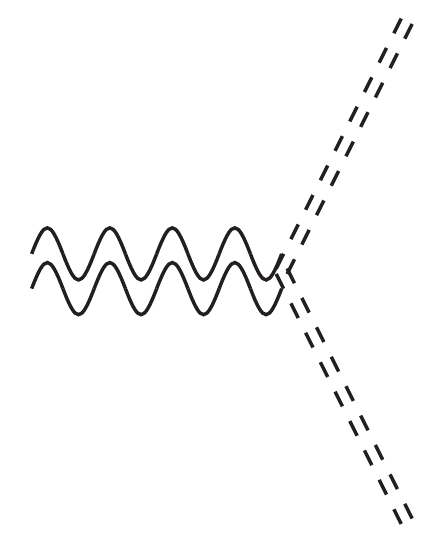} \right ) & = \frac{\imaginary \gcoupling}{8} \Bigg ( \, p^{\rho_1} \bigg ( q_1^{\mu} \eta^{\rho_2 \nu} + q_1^{\nu} \eta^{\rho_2 \mu} - q_2^{\mu} \eta^{\rho_2 \nu} - q_2^{\nu} \eta^{\rho_2 \mu} \bigg ) \\
		& \phantom{= \frac{\imaginary \gcoupling}{8} \Bigg (} + p^{\rho_2} \bigg ( q_2^{\mu} \eta^{\rho_1 \nu} + q_2^{\nu} \eta^{\rho_1 \mu} - q_1^{\mu} \eta^{\rho_1 \nu} - q_1^{\nu} \eta^{\rho_1 \mu} \bigg ) \\
		& \phantom{= \frac{\imaginary \gcoupling}{8} \Bigg (} - q_1^{\rho_1} \bigg ( p^{\mu} \eta^{\rho_2 \nu} + p^{\nu} \eta^{\rho_2 \mu} - p^{\rho_2} \eta^{\mu \nu} \bigg ) \\
		& \phantom{= \frac{\imaginary \gcoupling}{8} \Bigg (} - q_2^{\rho_2} \bigg ( p^{\mu} \eta^{\rho_1 \nu} + p^{\nu} \eta^{\rho_1 \mu} - p^{\rho_1} \eta^{\mu \nu} \bigg ) \\
		& \phantom{= \frac{\imaginary \gcoupling}{8} \Bigg (} - 2 \big ( q_1 \cdot q_2 \big ) \bigg ( \eta^{\rho_1 \mu} \eta^{\rho_2 \nu} + \eta^{\rho_1 \nu} \eta^{\rho_2 \mu} \bigg ) \Bigg )
	\end{split}
\end{align}
}%

\section{Longitudinal and transversal projections} \label{sec:longitudinal_and_transversal_projections}

The gauge fixing of (generalized) gauge theories introduces a fundamental new property, namely the notion of longitudinal and transversal degrees of freedom. These degrees of freedom, which characterize propagating gauge fields, become especially important when Feynman rules are considered: This is due to the fact that only the transversal degrees of freedom are physical. Thus, physically consistent theories should be such that unphysical (longitudinal) degrees of freedom are suppressed in scattering processes. In this section, we want to study the cases of Quantum Yang--Mills theory with a Lorenz gauge fixing (QYM) and (effective) Quantum General Relativity with a de Donder gauge fixing (QGR). In particular, we start with known and immediate identities in QYM and then proceed to their novel and involved counterparts in QGR. To this end, we first introduce the notion of an \emph{optimal gauge fixing} in \defnref{defn:optimal-gauge-fixing}: This is a gauge fixing that, for a given gauge theory, acts only on the vertical (i.e.\ gauge) degrees of freedom and thus complements the Lagrange density of the gauge theory in a unique way. In particular, we show that the Lorenz gauge fixing for Quantum Yang--Mills theory and the de Donder gauge fixing for (effective) Quantum General Relativity are both optimal in \thmsaref{thm:feynman-rule-gluon-propagator-lt-decomposition}{thm:feynman-rule-graviton-propagator-lt-decomposition}, which highlights their special roles. Then we present the respective transversal structures, i.e.\ the sets of longitudinal, identical and transversal projection operators introduced in \defnref{defn:transversal_structure}, together with their corresponding metrics: First we recall the situation of QYM in \defnref{defn:qym-transversal-structure} and then we introduce the corresponding counterpart of QGR in \defnref{defn:qgr-transversal-structure}. Then we study the decomposition of the longitudinal projection operators into the product of the respective gauge transformation and gauge fixing projections in \lemsaref{lem:g_and_l_inverse_decomposition_gl_qym}{lem:g_and_l_inverse_decomposition_gl_qgr}. In addition, we show that the provided longitudinal, identical and transversal projection operators are indeed projectors in \propsaref{prop:qym-transversal-structure}{prop:qgr-transversal-structure}. Furthermore, we show that the gauge transformation and gauge fixing projections are eigentensors of the respective transversal structures in \colsaref{col:gl-eigenvectors-lit-qym}{col:gl-eigentensors-lit-qgr}. Thereafter, we study the action of the corresponding metrics on said tensors in \lemsaref{lem:identities_tensors_qym}{lem:identities_tensors_qgr} and \colsaref{col:l-tensor-gg-ll-qym}{col:l-tensor-gg-ll-qgr}. This allows us then finally to simplify the gluon and graviton propagators and relate them to their ghost propagators via the gauge fixing projections in \thmsaref{thm:feynman-rule-gluon-propagator-lt-decomposition}{thm:feynman-rule-graviton-propagator-lt-decomposition}. Additionally, we provide cancellation identities for the gluon and graviton vertex Feynman rules in \thmsaref{thm:three-valent-contraction-identities-qym}{thm:three-valent-contraction-identities-qgr}, as well as for the corresponding couplings to matter from the Standard Model in \thmsaref{thm:three-valent-contraction-identities-qym-with-matter}{thm:three-valent-contraction-identities-qgr-with-matter}. These identities then suggest the use of \solref{sol:solution_4} for gravity-matter couplings: Higher-valent matter Feynman diagrams are divergent, if the virtual particles of the Feynman graphs are gravitons (e.g.\ scalar, spinor or photon four-point diagrams).\footnote{A similar situation appears for gauge ghosts with the Faddeev--Popov ghost construction.} This is problematic, as there exist no such residues in the respective Lagrange densities, cf.\ \sectionref{sec:lagrange-densities}. However, we suggest the absorption of these divergences via corresponding trees and claim that the corresponding cancellation identities render these a priori non-local operations local. More precisely, this resembles the Slavnov--Taylor identities in Quantum Yang--Mills theory, which also relate local residues (four-valent interaction-vertex) with non-local trees (two three-valent vertices glued together with a propagator); the only difference here is the lack of a four-valent vertex in the first place. Finally, we complete our investigations with a comment on the differences of the two most prominent definitions of the graviton field: The metric decomposition and the metric density decomposition of Goldberg and Capper et al.\ in \remref{rem:md_vs_mdd}.

\enter

\begin{defn}[Optimal gauge fixing] \label{defn:optimal-gauge-fixing}
	Let \(\Q\) be a quantum gauge theory with Lagrange density
	\begin{equation}
		\LQ = \mathcal{L}_\text{Classical} + \mathcal{L}_\text{GF} + \mathcal{L}_\text{Ghost} \, ,
	\end{equation}
	We call a gauge fixing functional \emph{optimal}, if the following three conditions are satisfied:
	\begin{itemize}
		\item The tensor \(\boldsymbol{T}\) is proportional to the Feynman rule of the quadratic term in \(\mathcal{L}_\text{Classical}\)
		\item The tensor \(\boldsymbol{L}\) is proportional to the Feynman rule of the quadratic term in \(\mathcal{L}_\text{GF}\)
		\item The tensors satisfy \(\boldsymbol{T} + \boldsymbol{L} = \boldsymbol{I}\), where \(\boldsymbol{I}\) denotes the corresponding identity tensor
	\end{itemize}
\end{defn}

\subsection{Quantum Yang--Mills theory with matter} \label{ssec:transversality_qym}

We recall known and immediate identities for the transversal structure of Quantum Yang--Mills theory with a Lorenz gauge fixing.

\enter

\begin{defn}[Transversal structure of QYM] \label{defn:qym-transversal-structure}
	Consider Quantum Yang--Mills theory with a Lorenz gauge fixing. Then we set its transversal structure \(\mathcal{T}_\text{QYM} := \setbig{L, I, T}\) as follows:\footnote{We remark that the color indices are implicitly included in the tensors \(L\), \(I\) and \(T\) by considering their tensor product with the identity matrix \(\delta\). We suppress this to simplify the notation.}
	\begin{subequations} \label{eqn:defn_projection_tensors_qym}
	\begin{align}
		L^\nu_\mu & := \frac{1}{p^2} p^\nu p_\mu \, , \\
		I^\nu_\mu & := \delta^\nu_\mu
		\intertext{and}
		T^\nu_\mu & := I^\nu_\mu - L^\nu_\mu \, , \label{eqn:defn-t-qym}
	\end{align}
	\end{subequations}
	where we have set \(p^2 := \eta_{\mu \nu} p^\mu p^\nu\). Lorentz indices on \(L\), \(I\) and \(T\) are raised and lowered with the metric \(G\), defined via
	\begin{subequations}
	\begin{align}
		G_{\mu \nu} & := \frac{1}{p^2} \eta_{\mu \nu}
		\intertext{and its inverse}
		G^{\mu \nu} & := p^2 \eta^{\mu \nu} \, .
	\end{align}
	\end{subequations}
	Finally, we define the following two tensors
	\begin{subequations} \label{eqns:gauge_transformation_and_longitudinal_projection_tensors_qym}
	\begin{align}
		g_\mu & := \frac{1}{p^2} p_\mu
		\intertext{and}
		l^\nu & := p^\nu \, .
	\end{align}
	\end{subequations}
\end{defn}

\enter

\begin{rem}
	The tensor \(g\) corresponds to a gauge transformation and the tensor \(l\) describes the gauge fixing projection. Furthermore, their degree in \(p^2\) is chosen such that the contraction with \(g\) corresponds to the contraction with half of a longitudinal gauge boson propagator.
\end{rem}

\enter

\begin{lem} \label{lem:g_and_l_inverse_decomposition_gl_qym}
	The following identities hold, i.e.\ \(g\) and \(l\) are inverse to each other and \(L\) decomposes into the product of \(g\) and \(l\):
	\begin{subequations}
	\begin{align}
		g_\mu l^\mu & = 1 \label{eqn:g_and_l_inverse_qym} \\
		l^\nu g_\mu & = L^\nu_\mu \label{eqn:decomposition_longitudinal_projection_tensor_qym}
	\end{align}
	\end{subequations}
\end{lem}

\begin{proof}
	This follows immediately from basic tensor calculations.
\end{proof}

\enter

\begin{prop} \label{prop:qym-transversal-structure}
	The following identities hold, i.e.\ the tensors \(L\), \(I\) and \(T\) are projectors:
	\begin{subequations}
	\begin{align}
		L_\mu^\tau L_\tau^\nu & = L_\mu^\nu \\
		I_\mu^\tau I_\tau^\nu & = I_\mu^\nu \\
		T_\mu^\tau T_\tau^\nu & = T_\mu^\nu
	\end{align}
	\end{subequations}
	Additionally, the tensor \(I\) is the identity with respect to the metric \(G\) and its inverse \(G^{-1}\):
	\begin{equation}
		G_{\mu \tau} G^{\tau \nu} = I^\nu_\mu
	\end{equation}
\end{prop}

\begin{proof}
	This follows immediately from \lemref{lem:g_and_l_inverse_decomposition_gl_qym} and basic tensor calculations.
\end{proof}

\enter

\begin{col} \label{col:gl-eigenvectors-lit-qym}
	The two tensors \(g\) and \(l\) are eigenvectors of the tensors \(L\), \(I\) and \(T\) with respective eigenvalues \(1\) and \(0\), i.e.\ we have:
	\begin{subequations}
	\begin{align}
		L^\nu_\mu g_\nu & = g_\nu \label{eqn:g_eigentensor_qym} \\
		L^\nu_\mu l^\mu & = l^\mu \label{eqn:l_eigentensor_qym} \\
		I^\nu_\mu g_\nu & = g_\nu \label{eqn:g_identity_qym} \\
		I^\nu_\mu l^\mu & = l^\mu \label{eqn:l_identity_qym} \\
		T^\nu_\mu g_\nu & = 0 \label{eqn:g_orthogonal_qym} \\
		T^\nu_\mu l^\mu & = 0 \label{eqn:l_orthogonal_qym}
	\end{align}
	\end{subequations}
\end{col}

\begin{proof}
	This follows immediately from \lemref{lem:g_and_l_inverse_decomposition_gl_qym} and basic tensor calculations.
\end{proof}

\enter

\begin{lem} \label{lem:identities_tensors_qym}
	The following identities hold, i.e.\ \(g\) and \(l\) are related via \(G\):
	\begin{subequations}
	\begin{align}
		G_{\mu \nu} l^\nu & = g_\mu \label{eqn:g_and_l_relation_metric_qym} \\
		G^{\mu \nu} g_\mu & = l^\nu \label{eqn:g_and_l_relation_inverse-metric_qym}
	\end{align}
	\end{subequations}
\end{lem}

\begin{proof}
	This follows immediately from basic tensor calculations.
\end{proof}

\enter

\begin{col} \label{col:l-tensor-gg-ll-qym}
	The following identities hold, i.e.\ \(L\) with raised and lowered indices decomposes into products of two \(g\) or \(l\) tensors, respectively:
	\begin{subequations}
	\begin{align}
		L_{\mu \nu} & = g_\mu g_\nu \label{eqn:l_gg_qym} \\
		L^{\mu \nu} & = l^\mu l^\nu \label{eqn:l_ll_inverse_qym}
	\end{align}
	\end{subequations}
\end{col}

\begin{proof}
	This follows immediately from \lemref{lem:g_and_l_inverse_decomposition_gl_qym} and \lemref{lem:identities_tensors_qym}.
\end{proof}

\enter

\begin{thm} \label{thm:feynman-rule-gluon-propagator-lt-decomposition}
	The Feynman rule of the gauge boson propagator can be written as follows:
	\begin{align}
		\Phi \left (  \cgreen{p-gluon} \right ) & = - \frac{\imaginary p^2}{p^2 + \imaginary \epsilon} \delta^{a b} \left ( T_{\mu \nu} + \xi L_{\mu \nu} \right ) \label{eqn:decomposition_gluon_propagator}
		\intertext{Furthermore, the Feynman rules of the gauge boson propagator and the gauge ghost propagator are related as follows:}
		\Phi \big ( \cgreen{p-gluonghost} \big ) & = \Phi \big ( \scriptstyle{l} \cgreen{p-gluon} \scriptstyle{l} \big ) \displaystyle \label{eqn:relation_gluon_gluon-ghost_propagators}
	\end{align}
	In particular, the Lorenz gauge fixing is the optimal gauge fixing condition for Quantum Yang--Mills theory.
\end{thm}

\begin{proof}
	\eqnref{eqn:decomposition_gluon_propagator} follows from the previous results in this subsection together with the Feynman rule
	\begin{align}
		\Phi \left (  \cgreen{p-gluon} \right ) & = - \frac{\imaginary}{p^2 + \imaginary \epsilon} \delta^{a b} \left ( \eta_{\mu \nu} - \frac{\left ( 1 - \xi \right )}{p^2} p_\mu p_\nu \right ) \, .
		\intertext{From this, \eqnref{eqn:relation_gluon_gluon-ghost_propagators} follows from the previous results in this subsection together with the Feynman rule}
		\Phi \left (  \cgreen{p-gluonghost} \right ) & = - \frac{\imaginary \xi}{p^2 + \imaginary \epsilon} \delta^{a b} \, .
	\end{align}
	The final claim follows then directly from \eqnref{eqn:decomposition_gluon_propagator}.
\end{proof}

\enter

\begin{thm} \label{thm:three-valent-contraction-identities-qym}
	The Feynman rule of the three-valent gauge boson vertex satisfies the following contraction identities:
	\begin{align}
		\Phi \left ( \scriptstyle{g} \tcgreen{v-gluontriple} _{\scriptstyle{g}}^{\scriptstyle{g}} \right ) & = \Phi \left ( \scriptstyle{L} \tcgreen{v-gluontriple} _{\scriptstyle{L}}^{\scriptstyle{L}} \right ) = 0 \label{eqn:contraction_v-triple-gluon} \, , \\
		\Phi \left ( \scriptstyle{g} \tcgreen{v-gluontriple} _{\scriptstyle{T}}^{\scriptstyle{T}} \right ) & \simeq_\textup{OS} 0 \label{eqn:contraction_v-triple-gluon-os-tt}
		\intertext{and thus}
		\Phi \left ( \scriptstyle{g} \tcgreen{v-gluontriple} _{\scriptstyle{I}}^{\scriptstyle{I}} \right ) & \simeq_\textup{OS} \Phi \left ( \scriptstyle{g} \tcgreen{v-gluontriple} _{\scriptstyle{L}}^{\scriptstyle{T}} \right ) + \Phi \left ( \scriptstyle{g} \tcgreen{v-gluontriple} _{\scriptstyle{T}}^{\scriptstyle{L}} \right ) \, , \label{eqn:contraction_v-triple-gluon-os-ii}
	\end{align}
	where \(\simeq_\textup{OS}\) indicates equality on-shell, i.e.\ modulo momentum conservation and equations of motion.
\end{thm}

\begin{proof}
	Starting with the first identity, we recall the decomposition \(L_\mu^\nu = l^\nu g_\mu\) of \eqnref{eqn:decomposition_longitudinal_projection_tensor_qym} and therefore only calculate the contraction with the \(g\) tensors:
	\begin{equation}
	\begin{split}
		\Phi \left ( \scriptstyle{g} \tcgreen{v-gluontriple} _{\scriptstyle{g}}^{\scriptstyle{g}} \right ) & = - \mathrm{g} f_{a_1 a_2 a_3} \frac{1}{p_1^2 \, p_2^2 \, p_3^2} p^1_{\rho_1} p^2_{\rho_2} p^3_{\rho_3} \sum_{s \in S_3} \Bigg \{ \eta^{\rho_{s(1)} \rho_{s(2)}} \left ( p_{s(1)} - p_{s(2)} \right )^{\rho_{s(3)}} \! \Bigg \} \\
	& = 0
	\end{split}
	\end{equation}
	For the two remaining identities we recall the decomposition \(I_\mu^\nu = T_\mu^\nu + L_\mu^\nu\) of \eqnref{eqn:defn-t-qym} and calculate
	\begin{equation}
	\begin{split}
		\Phi \left ( \scriptstyle{g} \tcgreen{v-gluontriple} _{\scriptstyle{I}}^{\scriptstyle{I}} \right ) & = - \mathrm{g} f_{a_1 a_2 a_3} \frac{1}{p_1^2} p^1_{\rho_1} \sum_{s \in S_3} \Bigg \{ \eta^{\rho_{s(1)} \rho_{s(2)}} \left ( p_{s(1)} - p_{s(2)} \right )^{\rho_{s(3)}} \! \Bigg \} \\
	& \simeq_\textup{MC} \mathrm{g} f_{a_1 a_2 a_3} \frac{1}{p_1^2} \left ( I^{\rho_1 \rho_2} \left ( p_2^\sigma \right ) - I^{\rho_1 \rho_2} \left ( p_3^\sigma \right ) - L^{\rho_1 \rho_2} \left ( p_2^\sigma \right ) + L^{\rho_1 \rho_2} \left ( p_3^\sigma \right ) \right ) \\
	 & \simeq_\textup{EoM} - \mathrm{g} f_{a_1 a_2 a_3} \frac{1}{p_1^2} \left ( L^{\rho_1 \rho_2} \left ( p_2^\sigma \right ) - L^{\rho_1 \rho_2} \left ( p_3^\sigma \right ) \right )
	\end{split}
	\end{equation}
	by noting \(I^{\mu \nu} = p^2 \eta^{\mu \nu}\) and recalling the identity \(L^{\mu \nu} \, T_\mu^\rho = 0\), where \(\simeq_\text{MC}\) indicates equality modulo momentum conservation and \(\simeq_\textup{EoM}\) indicates equality modulo equations of motion.
\end{proof}

\enter

\begin{rem} \label{rem:contraction-single-gluon}
	\eqnsaref{eqn:contraction_v-triple-gluon-os-tt}{eqn:contraction_v-triple-gluon-os-ii} imply that the longitudinal projection of a single gluon results in both, transversal on-shell cancellations and propagating longitudinal gluon modes. We will find out that this is similar in (effective) Quantum General Relativity, cf.\ \remref{rem:contraction-single-graviton}.
\end{rem}

\enter

\begin{thm} \label{thm:three-valent-contraction-identities-qym-with-matter}
	The Feynman rules of the three-valent interactions of gauge bosons with scalars, spinors and gauge ghosts satisfy the following on-shell contraction identities:
	\begin{align}
		\Phi \left ( \scriptstyle{g} \tcgreen{v-gluonscalartriple} \right ) & \simeq_\textup{OS} \Phi \left ( \scriptstyle{L} \tcgreen{v-gluonscalartriple} \right ) \simeq_\textup{OS} 0 \label{eqn:contraction_v-gluonscalartriple} \\
		\Phi \left ( \scriptstyle{g} \tcgreen{v-gluonspinortriple} \right ) & \simeq_\textup{OS} \Phi \left ( \scriptstyle{L} \tcgreen{v-gluonspinortriple} \right ) \simeq_\textup{OS} 0 \label{eqn:contraction_v-gluonspinortriple} \\
		\Phi \left ( \scriptstyle{g} \tcgreen{v-gluonghosttriple} \right ) & \simeq_\textup{OS} \Phi \left ( \scriptstyle{L} \tcgreen{v-gluonghosttriple} \right ) \simeq_\textup{OS} 0 \, , \label{eqn:contraction_v-gluonghosttriple}
	\end{align}
	where \(\simeq_\textup{OS}\) indicates equality on-shell, i.e.\ modulo momentum conservation and equations of motion.
\end{thm}

\begin{proof}
	Again, we only calculate the contraction with the \(g\) tensors due to the decomposition \(L_\mu^\nu = l^\nu g_\mu\) of \eqnref{eqn:decomposition_longitudinal_projection_tensor_qym}. Furthermore, we consider all momenta incoming and denote the gluon momentum by \(p^\sigma\) and the matter momenta by \(q_1^\sigma\) and \(q_2^\sigma\). Furthermore, we denote the gauge boson Lorentz and color indices by \(\rho\) and \(a\), respectively. Moreover, \(\mathfrak{H}\) and \(\mathfrak{S}\) denote the infinitesimal gauge group actions on the Higgs bundle and spinor bundle, respectively. In \eqnsaref{eqn:proof_contraction_v-gluonspinortriple}{eqn:proof_contraction_v-gluonghosttriple} number 1 denotes the particle and number 2 denotes the anti-particle. In particular, in \eqnref{eqn:proof_contraction_v-gluonspinortriple} this implies that the equations of motion differ in a relative sign. In addition, in \eqnref{eqn:proof_contraction_v-gluonghosttriple} we denote the gauge ghost color indices by \(c_1\) and \(c_2\), respectively. Additionally, in \eqnref{eqn:proof_contraction_v-gluonghosttriple} we use the symmetric (hermitian) gauge ghost Lagrange density of \eqnref{eqn:qym-sym-ghost}. With that, we present the actual calculations:
	\begin{align}
	\begin{split} \label{eqn:proof_contraction_v-gluonscalartriple}
		\Phi \left ( \scriptstyle{g} \tcgreen{v-gluonscalartriple} \right ) & = - \frac{\imaginary \mathrm{g}}{2} \left ( \frac{1}{p^2} p_\rho \right ) \left ( \left ( q_1 - q_2 \right )^\rho \mathfrak{H}_{a k l} \right ) \\
		& \simeq_\text{MC} \frac{\imaginary \mathrm{g}}{2 p^2} \left ( q_1^2 - q_2^2 \right ) \mathfrak{H}_{a k l} \\
		& = \frac{\imaginary \mathrm{g}}{2 p^2} \left ( \left ( q_1^2 - m^2 \right ) - \left ( q_2^2 - m^2 \right ) \right ) \mathfrak{H}_{a k l} \\
		& \simeq_\text{EoM} 0
		\end{split} \\
		\begin{split} \label{eqn:proof_contraction_v-gluonspinortriple}
		\Phi \left ( \scriptstyle{g} \tcgreen{v-gluonspinortriple} \right ) & = - \imaginary \mathrm{g} \left ( \frac{1}{p^2} p_\rho \right ) \left ( \gamma^\rho \mathfrak{S}_{a k l} \right ) \\
		& \simeq_\text{MC} \frac{\imaginary \mathrm{g}}{p^2} \left ( \left ( \slashed{q}_1 - m \right ) + \left ( \slashed{q}_2 + m \right ) \right ) \mathfrak{S}_{a k l} \\
		& \simeq_\text{EoM} 0 \, ,
		\end{split} \\
		\begin{split} \label{eqn:proof_contraction_v-gluonghosttriple}
		\Phi \left ( \scriptstyle{g} \tcgreen{v-gluonghosttriple} \right ) & = \frac{\mathrm{\imaginary g}}{2} \left ( \frac{1}{p^2} p_\rho \right ) \left ( f_{a b_1 b_2} \left ( q_1 - q_2 \right )^\rho \right ) \\
		& \simeq_\text{MC} - \frac{\imaginary \mathrm{g}}{2 p^2} f_{a b_1 b_2} \left ( q_1^2 - q_2^2 \right ) \\
		& \simeq_\text{EoM} 0 \, ,
		\end{split}
	\end{align}
	where \(\simeq_\text{MC}\) indicates equality modulo momentum conservation and \(\simeq_\textup{EoM}\) indicates equality modulo equations of motion.
\end{proof}

\subsection{(Effective) Quantum General Relativity with matter} \label{ssec:transversality_qgr}

We introduce novel and involved identities for the transversal structure of (effective) Quantum General Relativity with a de Donder gauge fixing.

\enter

\begin{defn}[Transversal structure of QGR] \label{defn:qgr-transversal-structure}
	Consider (effective) Quantum General Relativity with a de Donder gauge fixing. Then we set its transversal structure \(\mathcal{T}_\text{QGR} := \setbig{\mathbbit{L}, \mathbbit{I}, \mathbbit{T}}\) as follows:
	\begin{subequations} \label{eqn:defn_projection_tensors_qgr}
	\begin{align}
		\mathbbit{L}^{\rho \sigma}_{\mu \nu} & := \frac{1}{2 p^2} \left ( \delta^\rho_\mu p^\sigma p_\nu + \delta^\sigma_\mu p^\rho p_\nu + \delta^\rho_\nu p^\sigma p_\mu + \delta^\sigma_\nu p^\rho p_\mu - 2 \eta^{\rho \sigma} p_\mu p_\nu \right ) \, , \\
		\mathbbit{I} \mspace{2mu} ^{\rho \sigma}_{\mu \nu} & := \frac{1}{2} \left ( \delta^\rho_\mu \delta^\sigma_\nu + \delta^\sigma_\mu \delta^\rho_\nu \right )
		\intertext{and}
		\mathbbit{T} \mspace{2mu} ^{\rho \sigma}_{\mu \nu} & := \mathbbit{I} \mspace{2mu} ^{\rho \sigma}_{\mu \nu} - \mathbbit{L}^{\rho \sigma}_{\mu \nu} \, , \label{eqn:defn-t-qgr}
	\end{align}
	\end{subequations}
	where we have set \(p^2 := \eta_{\mu \nu} p^\mu p^\nu\). Lorentz indices on \(\mathbbit{L}\), \(\mathbbit{I}\) and \(\mathbbit{T}\) are raised and lowered with the metric \(\mathbbit{G}\), defined via\footnote{The reason for the unsymmetric definition concerning the factor \(\textfrac{1}{4}\) is motivated by \eqnsref{eqns:gauge_transformation_and_longitudinal_projection_tensors_qgr}.}
	\begin{subequations}
	\begin{align}
		\mathbbit{G}_{\mu \nu \rho \sigma} & := \frac{1}{p^2} \left ( \eta_{\mu \rho} \eta_{\nu \sigma} + \eta_{\mu \sigma} \eta_{\nu \rho} - \eta_{\mu \nu} \eta_{\rho \sigma} \right )
		\intertext{and its inverse}
		\mathbbit{G}^{\mu \nu \rho \sigma} & := \frac{p^2}{4} \left ( \eta^{\mu \rho} \eta^{\nu \sigma} + \eta^{\mu \sigma} \eta^{\nu \rho} - \eta^{\mu \nu} \eta^{\rho \sigma} \right ) \, .
	\end{align}
	\end{subequations}
	Finally, we define the following two tensors
	\begin{subequations} \label{eqns:gauge_transformation_and_longitudinal_projection_tensors_qgr}
	\begin{align}
		\mathscr{G}_{\mu \nu}^\kappa & := \frac{1}{p^2} \big ( p_\mu \delta_\nu^\kappa + p_\nu \delta_\mu^\kappa \big )
		\intertext{and}
		\mathscr{L}^{\rho \sigma}_\lambda & := \frac{1}{2} \big ( p^\rho \delta^\sigma_\lambda + p^\sigma \delta^\rho_\lambda - p_\lambda \eta^{\rho \sigma} \big ) \, .
	\end{align}
	\end{subequations}
\end{defn}

\enter

\begin{rem}
	The tensor \(\mathscr{G}\) corresponds to a gauge transformation and the tensor \(\mathscr{L}\) describes the gauge fixing projection. Furthermore, their degree in \(p^2\) is chosen such that the contraction with \(\mathscr{G}\) corresponds to the contraction with half of a longitudinal gauge boson propagator.
\end{rem}

\enter

\begin{lem} \label{lem:g_and_l_inverse_decomposition_gl_qgr}
	The following identities hold, i.e.\ \(\mathscr{G}\) and \(\mathscr{L}\) are inverse to each other and \(\mathbbit{L}\) decomposes into the product of \(\mathscr{G}\) and \(\mathscr{L}\):
	\begin{subequations}
	\begin{align}
		\mathscr{G}_{\mu \nu}^\kappa \mathscr{L}^{\mu \nu}_\lambda & = \delta^\kappa_\lambda \label{eqn:g_and_l_inverse_qgr} \\
		\mathscr{L}^{\rho \sigma}_\tau \mathscr{G}_{\mu \nu}^\tau & = \mathbbit{L}^{\rho \sigma}_{\mu \nu} \label{eqn:decomposition_longitudinal_projection_tensor_qgr}
	\end{align}
	\end{subequations}
\end{lem}

\begin{proof}
	This follows immediately from basic tensor calculations.
\end{proof}

\enter

\begin{prop} \label{prop:qgr-transversal-structure}
	The following identities hold, i.e.\ the tensors \(\mathbbit{L}\), \(\mathbbit{I}\) and \(\mathbbit{T}\) are projectors:
	\begin{subequations}
	\begin{align}
		\mathbbit{L}^{\kappa \lambda}_{\mu \nu} \, \mathbbit{L}^{\rho \sigma}_{\kappa \lambda} & = \mathbbit{L}^{\rho \sigma}_{\mu \nu} \\
		\mathbbit{I}^{\kappa \lambda}_{\mu \nu} \, \mathbbit{I}^{\rho \sigma}_{\kappa \lambda} & = \mathbbit{I}^{\rho \sigma}_{\mu \nu} \\
		\mathbbit{T}^{\kappa \lambda}_{\mu \nu} \, \mathbbit{T}^{\rho \sigma}_{\kappa \lambda} & = \mathbbit{T}^{\rho \sigma}_{\mu \nu}
	\end{align}
	\end{subequations}
	Additionally, the tensor \(\mathbbit{I}\) is the identity with respect to the metric \(\mathbbit{G}\) and its inverse \(\mathbbit{G}^{-1}\):
	\begin{equation}
		\mathbbit{G}_{\mu \nu \kappa \lambda} \mathbbit{G}^{\kappa \lambda \rho \sigma} = \mathbbit{I} \mspace{2mu} ^{\rho \sigma}_{\mu \nu}
	\end{equation}
\end{prop}

\begin{proof}
	This follows immediately from \lemref{lem:g_and_l_inverse_decomposition_gl_qgr} and basic tensor calculations.
\end{proof}

\enter

\begin{col} \label{col:gl-eigentensors-lit-qgr}
	The two tensors \(\mathscr{G}\) and \(\mathscr{L}\) are eigentensors of the tensors \(\mathbbit{L}\), \(\mathbbit{I}\) and \(\mathbbit{T}\) with respective eigenvalues \(1\) and \(0\), i.e.\ we have:
	\begin{subequations}
	\begin{align}
		\mathbbit{L}^{\rho \sigma}_{\mu \nu} \mathscr{G}_{\rho \sigma}^\kappa & = \mathscr{G}_{\mu \nu}^\kappa \label{eqn:g_eigentensor_qgr} \\
		\mathbbit{L}^{\rho \sigma}_{\mu \nu} \mathscr{L}^{\mu \nu}_\lambda & = \mathscr{L}^{\rho \sigma}_\lambda \label{eqn:l_eigentensor_qgr} \\
		\mathbbit{I}^{\rho \sigma}_{\mu \nu} \mathscr{G}_{\rho \sigma}^\kappa & = \mathscr{G}_{\mu \nu}^\kappa \label{eqn:g_identity_qgr} \\
		\mathbbit{I}^{\rho \sigma}_{\mu \nu} \mathscr{L}^{\mu \nu}_\lambda & = \mathscr{L}^{\rho \sigma}_\lambda \label{eqn:l_identity_qgr} \\
		\mathbbit{T}^{\rho \sigma}_{\mu \nu} \mathscr{G}_{\rho \sigma}^\kappa & = 0 \label{eqn:g_orthogonal_qgr} \\
		\mathbbit{T}^{\rho \sigma}_{\mu \nu} \mathscr{L}^{\mu \nu}_\lambda & = 0 \label{eqn:l_orthogonal_qgr}
	\end{align}
	\end{subequations}
\end{col}

\begin{proof}
	This follows immediately from \lemref{lem:g_and_l_inverse_decomposition_gl_qgr} and basic tensor calculations.
\end{proof}

\enter

\begin{lem} \label{lem:identities_tensors_qgr}
	The following identities hold, i.e.\ \(\mathscr{G}\) and \(\mathscr{L}\) are related via \(\mathbbit{G} \otimes \eta\):
	\begin{subequations}
	\begin{align}
		\mathbbit{G}_{\mu \nu \rho \sigma} \eta^{\kappa \lambda} \mathscr{L}^{\rho \sigma}_\lambda & = \mathscr{G}_{\mu \nu}^\kappa \label{eqn:g_and_l_relation_metric_qgr} \\
		\mathbbit{G}^{\mu \nu \rho \sigma} \eta_{\kappa \lambda} \mathscr{G}_{\mu \nu}^\kappa & = \mathscr{L}^{\rho \sigma}_\lambda \label{eqn:g_and_l_relation_inverse-metric_qgr}
	\end{align}
	\end{subequations}
\end{lem}

\begin{proof}
	This follows immediately from basic tensor calculations.
\end{proof}

\enter

\begin{col} \label{col:l-tensor-gg-ll-qgr}
	The following identities hold, i.e.\ \(\mathbbit{L}\) with raised and lowered indices decomposes into products of two \(\mathscr{G}\) or \(\mathscr{L}\) tensors, respectively:
	\begin{subequations}
	\begin{align}
		\mathbbit{L}_{\mu \nu \rho \sigma} & = \eta_{\kappa \lambda} \mathscr{G}_{\mu \nu}^\kappa \mathscr{G}_{\rho \sigma}^\lambda \label{eqn:l_gg_qgr} \\
		\mathbbit{L}^{\mu \nu \rho \sigma} & = \eta^{\kappa \lambda} \mathscr{L}^{\mu \nu}_\kappa \mathscr{L}^{\rho \sigma}_\lambda \label{eqn:l_ll_inverse_qgr}
	\end{align}
	\end{subequations}
\end{col}

\begin{proof}
	This follows immediately from \lemref{lem:g_and_l_inverse_decomposition_gl_qgr} and \lemref{lem:identities_tensors_qgr}.
\end{proof}

\enter

\begin{thm} \label{thm:feynman-rule-graviton-propagator-lt-decomposition}
	The Feynman rule of the graviton propagator can be written as follows:
	\begin{align}
		\Phi \left ( \cgreen{p-graviton} \right ) & = - \frac{2 \imaginary p^2}{p^2 + \imaginary \epsilon} \left ( \mathbbit{T}_{\mu \nu \rho \sigma} + \zeta \mathbbit{L}_{\mu \nu \rho \sigma} \right ) \label{eqn:decomposition_graviton_propagator}
		\intertext{Furthermore, the Feynman rules of the graviton propagator and the graviton-ghost propagator are related as follows:}
		\Phi \big ( \cgreen{p-gravitonghost} \big ) & = \Phi \big ( \scriptstyle{\mathscr{L}} \cgreen{p-graviton} \scriptstyle{\mathscr{L}} \big ) \displaystyle \label{eqn:relation_graviton_graviton-ghost_propagators}
	\end{align}
	In particular, the de Donder gauge fixing is the optimal gauge fixing condition for (effective) Quantum General Relativity.
\end{thm}

\begin{proof}
	\eqnref{eqn:decomposition_graviton_propagator} follows from the previous results in this subsection together with the Feynman rule
	\begin{align}
		\begin{split}
		\Phi \left (  \cgreen{p-graviton} \right ) & = - \frac{2 \imaginary}{p^2 + \imaginary \epsilon} \Bigg [ \left ( \eta_{\mu \rho} \eta_{\nu \sigma} + \eta_{\mu \sigma} \eta_{\nu \rho} - \eta_{\mu \nu} \eta_{\rho \sigma} \right ) \\
		& \phantom{= - \frac{2 \imaginary}{p^2 + \imaginary \epsilon} \Bigg [} - \left ( \frac{1 - \zeta}{p^2} \right ) \left ( \eta_{\mu \rho} p_{\nu} p_{\sigma} + \eta_{\mu \sigma} p_{\nu} p_{\rho} + \eta_{\nu \rho} p_{\mu} p_{\sigma} + \eta_{\nu \sigma} p_{\mu} p_{\rho} \right ) \Bigg ] \, .
		\end{split}
		\intertext{From this, \eqnref{eqn:relation_graviton_graviton-ghost_propagators} follows from the previous results in this subsection together with the Feynman rule}
		\Phi \left (  \cgreen{p-gravitonghost} \right ) & = - \frac{2 \imaginary \zeta}{p^2 + \imaginary \epsilon} \eta_{\rho \sigma} \, .
	\end{align}
	The final claim follows then directly from \eqnref{eqn:decomposition_graviton_propagator}.
\end{proof}

\enter

\begin{rem} \label{rem:md_vs_mdd}
	Given the metric density decomposition of Goldberg and Capper et al.\ \cite{Goldberg,Capper_Leibbrandt_Ramon-Medrano,Capper_Medrano,Capper_Namazie}, i.e.\
	\begin{equation}
		\boldsymbol{\phi}^{\mu \nu} := \frac{1}{\gcoupling} \left ( \sqrt{- \dt{g}} g^{\mu \nu} - \eta^{\mu \nu} \right ) \iff \sqrt{- \dt{g}} g^{\mu \nu} \equiv \eta^{\mu \nu} + \gcoupling \boldsymbol{\phi}^{\mu \nu} \, , \label{eqn:metric_density_decomposition}
	\end{equation}
	together with the gauge fixing functional
	\begin{equation}
		C^\mu := \partial_\nu \boldsymbol{\phi}^{\mu \nu} \equiv 0 \, . \label{eqn:metric_density_decomposition_gauge_fixing}
	\end{equation}
	Then, the corresponding graviton propagator is given via
	\begin{equation}
		\Phi \left ( \cgreen{p-graviton} \right ) = - \frac{2 \imaginary}{p^2 \left ( p^2 + \imaginary \epsilon \right )} \left ( \mathbbit{T}^{\mu \nu \rho \sigma} + \zeta \mathbbit{L}^{\mu \nu \rho \sigma} \right ) \, ,
	\end{equation}
	i.e.\ the roles of \(\mathscr{G}\) and \(\mathscr{L}\) are reversed! In particular, the gauge fixing functional \(C^\mu \left ( \boldsymbol{\phi} \right )\) is the optimal gauge fixing condition for the metric density decomposition, cf.\ \thmref{thm:feynman-rule-graviton-propagator-lt-decomposition}. This is due to the fact that in this case the graviton field \(\boldsymbol{\phi}^{\mu \nu}\) is a tensor density of weight 1, instead of the vertex Feynman rules. This will be studied further in \cite{Prinz_10}, cf.\ \ssecref{ssec:fam_md_lin_grav}.
\end{rem}

\enter

\begin{thm} \label{thm:three-valent-contraction-identities-qgr}
	The Feynman rule of the three-valent graviton vertex satisfies the following contraction identities:
	\begin{align}
		\Phi \left ( \scriptstyle{\mathscr{G}} \tcgreen{v-gravitontriple} _{\scriptstyle{\mathscr{G}}}^{\scriptstyle{\mathscr{G}}} \right ) & \simeq_\textup{MC} \Phi \left ( \bbsL \tcgreen{v-gravitontriple} _{\bbsL}^{\bbsL} \right ) \simeq_\textup{MC} 0 \label{eqn:contraction_v-triple-graviton} \, , \\
		\Phi \left ( \scriptstyle{\mathscr{G}} \tcgreen{v-gravitontriple} _{\bbsT}^{\bbsT} \right ) & \simeq_\textup{OS} 0 \label{eqn:contraction_v-triple-graviton-os-tt}
		\intertext{and thus}
		\Phi \left ( \scriptstyle{\mathscr{G}} \tcgreen{v-gravitontriple} _{\bbsI}^{\bbsI} \right ) & \simeq_\textup{OS} \Phi \left ( \scriptstyle{\mathscr{G}} \tcgreen{v-gravitontriple} _{\bbsL}^{\bbsT} \right ) + \Phi \left ( \scriptstyle{\mathscr{G}} \tcgreen{v-gravitontriple} _{\bbsT}^{\bbsL} \right ) \, , \label{eqn:contraction_v-triple-graviton-os-ii}
	\end{align}
	where \(\simeq_\textup{MC}\) indicates equality modulo momentum conservation and \(\simeq_\textup{OS}\) indicates equality on-shell, i.e.\ modulo momentum conservation and equations of motion.
\end{thm}

\begin{proof}
All three identities are checked with a Python program written by the author, cf.\ \cite{Python}. In addition, we emphasize the relation between the three identities via the decompositions \(\mathbbit{L}^{\rho \sigma}_{\mu \nu} = \mathscr{L}^{\rho \sigma}_\tau \mathscr{G}_{\mu \nu}^\tau\) of \eqnref{eqn:decomposition_longitudinal_projection_tensor_qgr} and \(\mathbbit{I}_{\mu \nu}^{\rho \sigma} = \mathbbit{T}_{\mu \nu}^{\rho \sigma} + \mathbbit{L}_{\mu \nu}^{\rho \sigma}\) of \eqnref{eqn:defn-t-qgr}.
\end{proof}

\enter

\begin{rem} \label{rem:contraction-single-graviton}
	\eqnsaref{eqn:contraction_v-triple-graviton-os-tt}{eqn:contraction_v-triple-graviton-os-ii} imply that the longitudinal projection of a single graviton results in both, transversal on-shell cancellations and propagating longitudinal gluon modes. This is similar to Quantum Yang--Mills theory, cf.\ \remref{rem:contraction-single-gluon}.
\end{rem}

\enter

\begin{thm} \label{thm:three-valent-contraction-identities-qgr-with-matter}
	The Feynman rules of the three-valent interactions of gravitons with scalars, spinors, gauge bosons, gauge ghosts and graviton-ghosts satisfy the following on-shell contraction identities:
	{\allowdisplaybreaks
	\begin{align}
		\Phi \left ( \scriptstyle{\mathscr{G}} \tcgreen{v-gravitonscalartriple} \right ) & \simeq_\textup{OS} \Phi \left ( \mathbbsit{L} \tcgreen{v-gravitonscalartriple} \right ) \simeq_\textup{OS} 0 \label{eqn:contraction_v-gravitonscalartriple} \\
		\Phi \left ( \scriptstyle{\mathscr{G}} \tcgreen{v-gravitonspinortriple} \right ) & \simeq_\textup{OS} \Phi \left ( \mathbbsit{L} \tcgreen{v-gravitonspinortriple} \right ) \simeq_\textup{OS} 0 \label{eqn:contraction_v-gravitonspinortriple} \\
		\Phi \left ( \scriptstyle{\mathscr{G}} \tcgreen{v-gravitongluontriple}^{\scriptstyle{T}}_{\scriptstyle{T}} \right ) & \simeq_\textup{OS} \Phi \left ( \mathbbsit{L} \tcgreen{v-gravitongluontriple}^{\scriptstyle{T}}_{\scriptstyle{T}} \right ) \simeq_\textup{OS} 0 \, , \label{eqn:contraction_v-gravitongluontriple} \\
		\Phi \left ( \scriptstyle{\mathscr{G}} \tcgreen{v-gravitongluonghosttriple} \right ) & \simeq_\textup{OS} \Phi \left ( \mathbbsit{L} \tcgreen{v-gravitongluonghosttriple} \right ) \simeq_\textup{OS} 0 \, , \label{eqn:contraction_v-gravitongluonghosttriple}\\
		\Phi \left ( \scriptstyle{\mathscr{G}} \tcgreen{v-gravitonghosttriple} \right ) & \simeq_\textup{OS} \Phi \left ( \mathbbsit{L} \tcgreen{v-gravitonghosttriple} \right ) \simeq_\textup{OS} 0 \, , \label{eqn:contraction_v-gravitonghosttriple}
	\end{align}
	}%
	where \(\simeq_\textup{OS}\) indicates equality on-shell, i.e.\ modulo momentum conservation and equations of motion.
\end{thm}

\begin{proof}
	Again, we only calculate the contraction with the \(\mathscr{G}\) tensors due to the decomposition \(\mathbbit{L}^{\rho \sigma}_{\mu \nu} = \mathscr{L}^{\rho \sigma}_\tau \mathscr{G}_{\mu \nu}^\tau\) of \eqnref{eqn:decomposition_longitudinal_projection_tensor_qgr}. Furthermore, we consider all momenta incoming and denote the graviton momentum by \(p^\sigma\) and the matter momenta by \(q_1^\sigma\) and \(q_2^\sigma\). Furthermore, we denote the graviton Lorentz indices by \(\mu\) and \(\nu\), respectively. In \eqnssaref{eqn:proof_contraction_v-gravitonspinortriple}{eqn:proof_contraction_v-gravitongluonghosttriple}{eqn:proof_contraction_v-gravitonghosttriple} number 1 denotes the particle and number 2 denotes the anti-particle. In particular, in \eqnref{eqn:proof_contraction_v-gravitonspinortriple} this implies that the equations of motion differ in a relative sign. In addition, in \eqnref{eqn:proof_contraction_v-gravitongluonghosttriple} we denote the gauge ghost color indices by \(c_1\) and \(c_2\), respectively, and in \eqnref{eqn:proof_contraction_v-gravitonghosttriple} we denote the graviton-ghost Lorentz indices by \(\rho_1\) and \(\rho_2\), respectively. Additionally, in \eqnref{eqn:proof_contraction_v-gravitonghosttriple} we use the symmetric (hermitian) graviton-ghost Lagrange density of \eqnref{eqn:qgr-sym-ghost}. With that, we present the actual calculations:
	{\allowdisplaybreaks
	\begin{align}
	\begin{split} \label{eqn:proof_contraction_v-gravitonscalartriple}
		\Phi \left ( \scriptstyle{\mathscr{G}} \tcgreen{v-gravitonscalartriple} \right ) & = \frac{\imaginary \varkappa}{2 p^2} \left ( p_\mu \delta_\nu^\tau + p_\nu \delta_\mu^\tau \right ) \Bigg ( \! - \eta^{\mu \nu} \left ( q_1 \cdot q_2 + m^2 \right ) + q_1^\mu q_2^\nu + q_1^\nu q_2^\mu \Bigg ) \\
		& \simeq_\text{MC} \frac{\imaginary \varkappa}{p^2} \Bigg ( \left ( q_1^2 - m^2 \right ) q_2^\tau + \left ( q_2^2 - m^2 \right ) q_1^\tau \! \Bigg ) \\
		& \simeq_\text{EoM} 0
		\end{split} \\
		\begin{split} \label{eqn:proof_contraction_v-gravitonspinortriple}
		\Phi \left ( \scriptstyle{\mathscr{G}} \tcgreen{v-gravitonspinortriple} \right ) & = \frac{\imaginary \varkappa}{8 p^2} \left ( p_\mu \delta_\nu^\tau + p_\nu \delta_\mu^\tau \right ) \\ & \phantom{=} \times \Bigg ( 2 \eta^{\mu \nu} \left ( \slashed{q}_1 - \slashed{q}_2 - 2m \right ) - \left ( q_1 - q_2 \right )_\mu \gamma_\nu - \left ( q_1 - q_2 \right )_\nu \gamma_\mu \Bigg ) \\
		& \simeq_\text{MC} \frac{\imaginary \varkappa}{4 p^2} \Bigg ( \! - 2 \left ( q_1 + q_2 \right )^\tau \left ( \slashed{q}_1 - \slashed{q}_2 - 2m \right ) \\ & \phantom{\simeq_\text{MC} \frac{\imaginary \varkappa}{4} \Bigg ( \!} + \left ( q_1 - q_2 \right )^\tau \left ( \slashed{q}_1 + \slashed{q}_2 + m - m \right ) + \left ( q_1 ^2 - q_2^2 + m - m \right ) \gamma^\tau \Bigg ) \\
		& \simeq_\text{EoM} \frac{\imaginary \varkappa}{4 p^2} \Bigg ( \! \left ( \slashed{q}_1 - m \right ) \left ( - 2 \left ( q_1 + q_2 \right )^\tau + \left ( q_1 - q_2 \right )^\tau + \left ( \slashed{q}_1 + m \right ) \gamma^\tau \right ) \\ & \phantom{\simeq_\text{MC} \frac{\imaginary \varkappa}{4} \Bigg (} + \left ( \slashed{q}_2 + m \right ) \left ( 2 \left ( q_1 + q_2 \right )^\tau + \left ( q_1 - q_2 \right )^\tau - \left ( \slashed{q}_2 - m \right ) \gamma^\tau \right ) \! \Bigg ) \\
		& \simeq_\text{EoM} 0
		\end{split} \\
		\begin{split} \label{eqn:proof_contraction_v-gravitongluontriple}
		\Phi \left ( \scriptstyle{\mathscr{G}} \tcgreen{v-gravitongluontriple}^{\scriptstyle{T}}_{\scriptstyle{T}} \right ) & = \frac{\imaginary \varkappa}{2 p^2} \left ( p_\mu \delta_\nu^\tau + p_\nu \delta_\mu^\tau \right ) \\
		& \phantom{=} \times \delta_{a_1 a_2} \Bigg ( \! \left ( q_1 \cdot q_2 \right ) \left ( \eta^{\mu \nu} \eta^{\sigma_1 \sigma_2} - \eta^{\mu \sigma_1} \eta^{\nu \sigma_2} - \eta^{\mu \sigma_2} \eta^{\nu \sigma_1} \right ) \\
		& \phantom{= \times \delta_{a_1 a_2} \Bigg ( \!} - \eta^{\mu \nu} q_1^{\sigma_2} q_2^{\sigma_1} - \eta^{\sigma_1 \sigma_2} \left ( q_1^\mu q_2^\nu + q_2^\mu q_1^\nu \right ) \\
		& \phantom{= \times \delta_{a_1 a_2} \Bigg ( \!} + q_1^{\sigma_2} \left ( \eta^{\mu \sigma_1} q_2^\nu + \eta^{\nu \sigma_1} q_2^\mu \right ) + q_2^{\sigma_1} \left ( \eta^{\mu \sigma_2} q_1^\nu + \eta^{\nu \sigma_2} q_1^\mu \right ) \\
		& \phantom{= \times \delta_{a_1 a_2} \Bigg ( \!} - \frac{1}{\xi} \eta^{\mu \nu} \left ( q_1^{\sigma_1} q_2^{\sigma_2} + p^{\sigma_1} q_2^{\sigma_2} + p^{\sigma_2} q_1^{\sigma_1} \right ) \\
		& \phantom{= \times \delta_{a_1 a_2} \Bigg ( \!} + \frac{1}{\xi} q_1^{\sigma_1} \left ( \eta^{\mu \sigma_2} q_2^\nu + \eta^{\nu \sigma_2} q_2^\mu + \eta^{\mu \sigma_2} p^\nu + \eta^{\nu \sigma_2} p^\mu \right ) \\
		& \phantom{= \times \delta_{a_1 a_2} \Bigg ( \!} + \frac{1}{\xi} q_2^{\sigma_2} \left ( \eta^{\mu \sigma_1} q_1^\nu + \eta^{\nu \sigma_1} q_1^\mu + \eta^{\mu \sigma_1} p^\nu + \eta^{\nu \sigma_1} p^\mu \right ) \! \Bigg ) \\
		& \phantom{=} \times T^{\rho_1}_{\sigma_1} \left ( q_1 \right ) \times T^{\rho_2}_{\sigma_2} \left ( q_2 \right ) \\
		& \simeq_\text{MC} \frac{\imaginary \varkappa}{p^2} \delta_{a_1 a_2} \left ( q_1^\tau \left ( q_2^2 \eta^{\sigma_1 \sigma_2} - q_2^{\sigma_1} q_1^{\sigma_2} + \frac{1}{\xi} q_1^{\sigma_1} q_2^{\sigma_2} \right ) \vphantom{\Bigg )} \right . \\
		& \phantom{\simeq_\text{MC} \frac{\imaginary \varkappa}{p^2} \delta_{a_1 a_2} \Bigg ( \!} + q_2^\tau \left ( q_1^2 \eta^{\sigma_1 \sigma_2} - q_2^{\sigma_1} q_1^{\sigma_2} + \frac{1}{\xi} q_1^{\sigma_1} q_2^{\sigma_2} \right ) \\
		& \phantom{\simeq_\text{MC} \frac{\imaginary \varkappa}{p^2} \delta_{a_1 a_2} \Bigg ( \!} + \eta^{\tau \sigma_1} \left ( \left ( q_1 \cdot q_2 \right ) \left ( 1 + \frac{1}{\xi} \right ) q_2^{\sigma_2} - q_2^2 \left ( q_1^{\sigma_2} - \frac{1}{\xi} q_2^{\sigma_2} \right ) \right ) \\
		& \phantom{\simeq_\text{MC} \frac{\imaginary \varkappa}{p^2} \delta_{a_1 a_2} \Bigg ( \!} \left . + \eta^{\tau \sigma_2} \left ( \left ( q_1 \cdot q_2 \right ) \left ( 1 + \frac{1}{\xi} \right ) q_1^{\sigma_1} - q_1^2 \left ( q_2^{\sigma_1} - \frac{1}{\xi} q_1^{\sigma_1} \right ) \right ) \right ) \\
		& \phantom{\simeq_\text{MC}} \times T^{\rho_1}_{\sigma_1} \left ( q_1 \right ) \times T^{\rho_2}_{\sigma_2} \left ( q_2 \right ) \\
		& = \frac{\imaginary \varkappa}{p^2} \delta_{a_1 a_2} \Bigg ( q_1^\tau \left ( q_2^2 \eta^{\sigma_1 \sigma_2} - q_2^{\sigma_1} q_1^{\sigma_2} \right ) + q_2^\tau \left ( q_1^2 \eta^{\sigma_1 \sigma_2} - q_2^{\sigma_1} q_1^{\sigma_2} \right ) \\
		& \phantom{= \frac{\imaginary \varkappa}{p^2} \delta_{a_1 a_2} \Bigg ( \!} - q_2^2 \eta^{\tau \sigma_1} q_1^{\sigma_2} - q_1^2 \eta^{\tau \sigma_2} q_2^{\sigma_1} \! \Bigg ) \times T^{\rho_1}_{\sigma_1} \left ( q_1 \right ) \times T^{\rho_2}_{\sigma_2} \left ( q_2 \right ) \\
		& \simeq_\text{EoM} 0
		\end{split} \\
		\begin{split} \label{eqn:proof_contraction_v-gravitongluonghosttriple}
		\Phi \left ( \scriptstyle{\mathscr{G}} \tcgreen{v-gravitongluonghosttriple} \right ) & = \frac{\imaginary \varkappa}{2 \xi p^2} \left ( p_\mu \delta_\nu^\tau + p_\nu \delta_\mu^\tau \right ) \Bigg ( \! \left ( q_1 \cdot q_2 \right ) \eta^{\mu \nu} - q_1^\mu q_2^\nu - q_2^\mu q_1^\nu \Bigg ) \\
		& \simeq_\text{MC} \frac{\imaginary \varkappa}{\xi p^2} \Bigg ( \! - \left ( q_1 \cdot q_2 \right ) \left ( q_1 + q_2 \right )^\tau + \left ( q_1 \cdot q_2 \right ) \left ( q_1 + q_2 \right )^\tau + q_1^\tau q_2^2 + q_2^\tau q_1^2 \Bigg ) \\
		& = \frac{\imaginary \varkappa}{\xi p^2} \Bigg ( q_1^\tau q_2^2 + q_2^\tau q_1^2 \Bigg ) \\
		& \simeq_\text{EoM} 0 \, ,
		\end{split} \\
		\begin{split} \label{eqn:proof_contraction_v-gravitonghosttriple}
		\Phi \left ( \scriptstyle{\mathscr{G}} \tcgreen{v-gravitonghosttriple} \right ) & = \frac{\imaginary \gcoupling}{8 p^2} \left ( p_\mu \delta_\nu^\tau + p_\nu \delta_\mu^\tau \right ) \\
		& \phantom{= \frac{\imaginary \gcoupling}{8 p^2}} \times \Bigg ( \big ( q_1^2 + q_2^2 - p^2 \big ) \bigg ( \eta^{\rho_1 \mu} \eta^{\rho_2 \nu} + \eta^{\rho_1 \nu} \eta^{\rho_2 \mu} \bigg ) \\
		& \phantom{= \frac{\imaginary \gcoupling}{8} \times \Bigg (} - q_1^{\rho_1} \bigg ( p^{\mu} \eta^{\rho_2 \nu} + p^{\nu} \eta^{\rho_2 \mu} - p^{\rho_2} \eta^{\mu \nu} \bigg ) \\
		& \phantom{= \frac{\imaginary \gcoupling}{8} \times \Bigg (} - q_2^{\rho_2} \bigg ( p^{\mu} \eta^{\rho_1 \nu} + p^{\nu} \eta^{\rho_1 \mu} - p^{\rho_1} \eta^{\mu \nu} \bigg ) \\
		& \phantom{= \frac{\imaginary \gcoupling}{8} \times \Bigg (} + p^{\rho_1} \bigg ( q_1^{\mu} \eta^{\rho_2 \nu} + q_1^{\nu} \eta^{\rho_2 \mu} - q_2^{\mu} \eta^{\rho_2 \nu} - q_2^{\nu} \eta^{\rho_2 \mu} \bigg ) \\
		& \phantom{= \frac{\imaginary \gcoupling}{8} \times \Bigg (} + p^{\rho_2} \bigg ( \! - q_1^{\mu} \eta^{\rho_1 \nu} - q_1^{\nu} \eta^{\rho_1 \mu} + q_2^{\mu} \eta^{\rho_1 \nu} + q_2^{\nu} \eta^{\rho_1 \mu} \bigg ) \Bigg ) \\
		& \simeq_\text{MC} \frac{\imaginary \gcoupling}{4 p^2} \left ( \eta^{\tau \rho_1} \left ( \left ( 2 q_1^2 + p^2 \right ) q_1^{\rho_2} - 2 q_1^2 q_2^{\rho_2} \right ) \right . \\
		& \phantom{\simeq_\text{MC} \frac{\imaginary \gcoupling}{4 p^2} \Bigg (} \left . + \eta^{\tau \rho_2} \left ( \left ( 2 q_2^2 + p^2 \right ) q_2^{\rho_1} - 2 q_2^2 q_1^{\rho_1} \right ) \right ) \\
		& \simeq_\text{EoM} 0 \, ,
		\end{split}
	\end{align}
	}%
	where \(\simeq_\text{MC}\) indicates equality modulo momentum conservation and \(\simeq_\textup{EoM}\) indicates equality modulo equations of motion.
\end{proof}

\enter

\begin{rem}
	We emphasize that the longitudinal projection of the graviton in \eqnref{eqn:contraction_v-gravitongluontriple} also induces longitudinal gluon-modes, cf.\ \eqnref{eqn:proof_contraction_v-gravitongluontriple}. We remove them via transversal gluon projection operators, cf.\ \defnref{defn:qym-transversal-structure}, as we are here interested in physical external particles (which are on-shell and transversal). In general, however, these longitudinal gluon legs are important as they lead to further cancellations, cf.\ \thmsaref{thm:three-valent-contraction-identities-qym}{thm:three-valent-contraction-identities-qym-with-matter}. This is, as we have seen, equivalent to the three-valent gluon and graviton vertex Feynman rules, cf.\ \remsaref{rem:contraction-single-gluon}{rem:contraction-single-graviton}. We will study this in detail in future work.
\end{rem}

\chapter{Conclusion} \label{chp:conclusion}

We end this dissertation with a summary on the achieved results. Then we discuss planned follow-up projects. This includes a generalization of Wigner's classification to linearized gravity, perturbative BRST cohomology via dg-Hopf algebras and a study on the equivalence of the two most prominent definitions of the graviton field.

\section{Summary}

We have studied several renormalization related aspects of gauge theories and gravity:

First, we studied the diffeomorphism-gauge BRST double complex for (effective) Quantum General Relativity coupled to the Standard Model in \sectionref{sec:diffeomorphism-gauge-brst-double-complex}: The main results are \thmref{thm:total_brst_operator}, which states that the diffeomorphism BRST operator and the gauge BRST operator anticommute. In particular, after introducing the corresponding anti-BRST operators, we even found that all BRST and anti-BRST operators mutually anticommute in \colref{col:total_anti-brst_operator}. A further main result is \thmref{thm:no-couplings-grav-ghost-matter-sm}, which states that the graviton-ghosts decouple from matter of the Standard Model if and only if the gauge theory gauge fixing fermion is a tensor density of weight \(w = 1\). Our final main result in this direction is \thmref{thm:total_gauge_fixing_fermion}, which states that we can generate the complete gauge fixing and ghost Lagrange densities for (effective) Quantum General Relativity coupled to the Standard Model via the action of the \emph{total BRST operator} on the \emph{total gauge fixing fermion}.

Then, we analyzed the renormalization of gauge theories and gravity in the Hopf algebra setup of Connes and Kreimer. The main results are \thmref{thm:quantum_gauge_symmetries_induce_hopf_ideals} and \thmref{thm:criterion_ren-hopf-mod}, where the first states that quantum gauge symmetries correspond to Hopf ideals in the renormalization Hopf algebra and the second provides criteria for their validity on the level of renormalized Feynman rules. To this end, we first discussed a particular problem of quantum gauge theories that can lead to an ill-defined renormalization Hopf algebra and then presented four possible solutions therefore in \sectionref{sec:the_associated_renormalization_hopf_algebra}. Next, we introduced the notion of a \emph{predictive Quantum Field Theory} to address non-renormalizable Quantum Field Theories that allow for a well-defined perturbative expansion in \sectionref{sec:predictive-quantum-field-theories}. In particular, building upon the results of \cite{Prinz_3,Kreimer_QG1}, we claim that (effective) Quantum General Relativity is predictive, which constitutes the main motivation for the studies in this dissertation. Further, we studied combinatorial properties of the superficial degree of divergence in \sectionref{sec:a_superficial_argument} and generalized known coproduct and antipode identities to the super- and non-renormalizable cases in \sectionref{sec:coproduct_and_antipode_identities}. Additionally, we extended this framework to theories with multiple vertex residues and coupling constants in \defnref{defn:connectedness_gradings_renormalization_hopf_algebra} and discussed the incorporation of transversal structures in \remref{rem:longitudinal_and_transversal_gauge_fields}. Then we illustrated the developed theory in the cases of Quantum Yang--Mills theory in \exref{exmp:qym} and (effective) Quantum General Relativity in \exref{exmp:qgr}. Finally, we discussed, as a direct consequence of our findings, the well-definedness of the Corolla polynomial without reference to a particular renormalization scheme in \remref{rem:corolla_polynomial}, cf.\ \cite{Kreimer_Yeats,Kreimer_Sars_vSuijlekom,Kreimer_Corolla}.

Next, we derived and presented the Feynman rules for (effective) Quantum General Relativity and the gravitational couplings to the Standard Model. The main results are \thmref{thm:grav-fr} stating the graviton vertex Feynman rules, \thmref{thm:grav-prop} stating the corresponding graviton propagator Feynman rule, \thmref{thm:ghost-fr} stating the graviton-ghost vertex Feynman rules and \thmref{thm:ghost-prop} stating the corresponding graviton-ghost propagator Feynman rule. Additionally, the graviton-matter vertex Feynman rules are stated in \thmref{thm:matter-fr} on the level of 10 generic matter-model Lagrange densities, as classified by \lemref{lem:matter-model-Lagrange-densities}. The complete graviton-matter Feynman rules can then be obtained by adding the corresponding matter contributions, which are listed e.g.\ in \cite{Romao_Silva}. Finally, we displayed the three- and four-valent graviton and graviton-ghost vertex Feynman rules explicitly in \exref{exmp:FR}.

Finally, we presented the explicit Feynman rules for all propagators and three-valent vertices for (effective) Quantum General Relativity coupled to the Standard Model in \sectionref{sec:explicit_feynman_rules}. We then proceeded by studying the corresponding longitudinal, identical and transversal projection operators for gauge bosons and gravitons in \sectionref{sec:longitudinal_and_transversal_projections}: The main results are decompositions of the gluon and graviton propagators with respect to the corresponding transversal and longitudinal projection operators in \thmsaref{thm:feynman-rule-gluon-propagator-lt-decomposition}{thm:feynman-rule-graviton-propagator-lt-decomposition}. Further main results in this direction are the corresponding cancellation identities for all three-valent vertices in \thmsaref{thm:three-valent-contraction-identities-qym}{thm:three-valent-contraction-identities-qym-with-matter} and \thmsaref{thm:three-valent-contraction-identities-qgr}{thm:three-valent-contraction-identities-qgr-with-matter}. Finally, in \remref{rem:md_vs_mdd}, we discussed the duality of the metric density decomposition of Goldberg and Capper et al.\ with respect to the usual metric decomposition used in this dissertation.

\section{Outlook} \label{sec:outlook}

In this final section, we describe our planned follow-up projects that build directly on the results of this dissertation:

\subsubsection{\emph{On a Generalization of Wigner's Classification to Linearized Gravity}} \label{ssec:wigner_lin_grav}

Wigner's classification of elementary particles is a fundamental ingredient for the standard approach to Quantum Field Theories (QFTs) \cite{Wigner}. More precisely, this classification uses the two Casimir operators \(\mathfrak{C}_1\) and \(\mathfrak{C}_2\) of the Poincaré group \(\mathcal{P} \left ( \M \right )\) to classify elementary particles via their mass and helicity or spin. More precisely, we have
\begin{subequations}
\begin{align}
	\mathfrak{C}_1 & := \eta^{\mu \nu} P_\mu P_\nu
	\intertext{and}
	\mathfrak{C}_2 & := \eta^{\mu \nu} W_\mu W_\nu \, ,
\end{align}
\end{subequations}
where \(P_\mu\) is the 4-momentum operator and \(W_\mu := \frac{1}{2} \tensor{\varepsilon}{_\mu ^\nu ^\rho ^\sigma} J_{\nu \rho} P_\sigma\) is the Pauli--Lubanski operator with \(J_{\mu \nu} := X_\mu P_\nu - X_\nu P_\mu\) the angular 4-momentum operator. This classification is then used among other things to construct the Fock space for the theory under consideration. Unfortunately, both these operators rely directly on the linear structure of \(\M\). Thus, there is no direct way to apply this classification to linearized gravity. However, we generalize this classification to simple spacetimes and diffeomorphism invariant theories: To this end, we start with a simple spacetime, cf.\ \defnref{defn:simple_spacetime}, given as the triple \((M,\met,\trivmap)\), where
\begin{equation}
	\trivmap \, : \quad M \to \M
\end{equation}
is a fixed diffeomorphism, called the trivialization map, and \((\M, \eta)\) is the background Minkowski spacetime. We remark that we did not demand that \(\trivmap\) is an isometry, i.e.\ we have in general \(\trivmap_* \gamma \not \equiv \eta\). Instead, we use this construction to pushforward the metric \(g := \trivmap_* \gamma\) and with that define the graviton field on the background Minkowski spacetime \(\M\) via
\begin{equation}
	h := \frac{1}{\varkappa} \left ( g - \eta \right ) \in \sctnbig{\M, \SymMM} \, ,
\end{equation}
cf.\ \defnref{defn:md_and_gf}. We remark that in this picture the choice of \(\trivmap\) is essential to the definition of the graviton field: More precisely, we could equivalently pullback the Minkowski background metric \(b := \trivmap^* \eta\) and with that define the graviton field on the spacetime \(M\) via
\begin{equation}
	\theta := \frac{1}{\varkappa} \left ( \gamma - b \right ) \in \sctnbig{M, \SymM} \, .
\end{equation}
Now, in this picture the choice of \(\trivmap\) is equivalent to the choice of the background metric \(b\). Additionally, we used the construction of a simple spacetime \((M,\met,\trivmap)\) to define the Fourier transformation for sections of particle fields via
\begin{subequations}
\begin{align}
	\mathscr{F} \, : \quad \Gamma \big ( M, E \big ) \to \widehat{\Gamma} \big ( M, E \big ) \, , \quad \particlefield \left ( x^\alpha \right ) \mapsto \hat{\particlefield} \left ( p^\alpha \right )
	\intertext{with}
	\hat{\particlefield} \left ( p^\alpha \right ) := \frac{1}{\left ( 2 \pi \right )^2} \tau^* \left ( \int_\sbbM \left ( \tau_* \particlefield \right ) \big ( y^\beta \big ) e^{-i \eta \left ( y, p \right )} \dif V_\eta \right )
\end{align}
\end{subequations}
in \defnref{defn:fourier_transform}. To generalize Wigner's classification to elementary particles in linearized gravity, we use the trivialization map to define a linear structure on \(M\) as follows:

\enter

\begin{defn}[Linear structure] \label{defn:linear-structure}
Let \((M,\met,\trivmap)\) be a simple spacetime. Then we define an addition via
\begin{subequations}
\begin{align}
+_\trivmap \, & : \quad M \times M \to M \, , \quad (x_1, x_2) \mapsto \trivmap^{-1} \left ( \trivmap \left ( x_1 \right ) + \trivmap \left ( x_2 \right ) \right )
\intertext{and a scalar multiplication by real numbers via}
\cdot_\trivmap \, & : \quad \mathbb{R} \times M \to M \, , \quad (r, x) \mapsto \trivmap^{-1} \left ( r \cdot \trivmap \left ( x \right ) \right ) \, .
\end{align}
\end{subequations}
We refer to \((M,+_\trivmap , \cdot_\trivmap)\) as linear structure of \(M\).
\end{defn}

\enter

In particular, this allows us to pullback the Poincar\'{e} group of \(\M\) to \(M\) via \(\trivmap\) as follows:

\enter

\begin{defn}[Pullback Poincar\'{e} group] \label{defn:pullback-poincare-group}
Let \((M,\met,\trivmap)\) be a simple spacetime. Then we set
\begin{equation}
\mathcal{P}_\tau \left ( M \right ) := \trivmap^* \mathcal{P} \left ( \M \right ) \, .
\end{equation}
Equivalently, this is the Poincaré group obtained with respect to the linear structrure introduced in \defnref{defn:linear-structure}. We refer to \(\mathcal{P}_\tau \left ( M \right )\) as pullback Poincar\'{e} group.
\end{defn}

\enter

Obviously, both definitions depend crucially on \(\trivmap\) and are not invariant under general diffeomorphisms \(\phi \in \operatorname{Diff} \left ( M \right )\). However, for diffeomorphism invariant theories, we obtain the following:

\enter

\begin{lem} \label{lem:trivialization-map-diffeomorphism-invariant-theory}
	Let \((M,\met,\trivmap)\) be a simple spacetime and \(\T\) a classical diffeomorphism invariant theory, given via the Lagrange density \(\LT\). Then the choice of the trivialization map \(\tau\) can be absorbed into the diffeomorphism invariance of \(\LT\).
\end{lem}

\begin{proof}
	Let \(\trivmap_1 \colon M \to \M\) and \(\trivmap_2 \colon M \to \M\) be two trivialization maps. Then, there exists a diffeomorphism \(\phi \in \operatorname{Diff} \left ( M \right )\) such that
	\begin{equation}
		\trivmap_1 = \trivmap_2 \circ \phi \, .
	\end{equation}
	Since \(\T\) is diffeomorphism invariant by assumption, we have
	\begin{equation}
		\LT \simeq_\text{TD} \phi_* \LT \, ,
	\end{equation}
	where \(\simeq_\text{TD}\) means equality modulo total derivatives. Thus, the choice of the trivialization map \(\trivmap\) has no physical relevance.
\end{proof}

\enter

Unfortunately, this statement does not apply to the pullback Poincaré group with its two Casimir operators. Hence, for the corresponding Quantum Field Theory, it is necessary to fix the trivialization map \(\trivmap\). Thus, given a diffeomorphism \(\phi \in \operatorname{Diff} \left ( M \right )\), the following questions arise:

\begin{itemize}
	\item How is the mass of an elementary particle affected by \(\phi\)?
	\item How is the helicity or spin of an elementary particle affected by \(\phi\)?
	\item Is it possible to give a sensible definition of the corresponding Fock space without the need to fix the trivialization map \(\trivmap\)?
\end{itemize}

This project will be continued in \cite{Prinz_8}.

\subsubsection{\emph{Cancellation Identities and Renormalization}} \label{ssec:can_ids_ren}

BRST cohomology is a powerful tool to study several aspects of gauge theories, most notably its gauge invariance, possible gauge fixings and their corresponding ghosts \cite{Becchi_Rouet_Stora_1,Becchi_Rouet_Stora_2,Tyutin,Becchi_Rouet_Stora_3}. Geometrically, in this approach the BRST operator \(D \colon \FQ \to \FQ\) with \(D^2 = 0\) turns the sheaf of particle fields \(\FQ\) into a differential-graded superalgebra, cf.\ \sectionref{sec:diffeomorphism-gauge-brst-double-complex}. An important open question is how this setup, which is based on the path integral, can be incorporated into the perturbative expansion, i.e.\ on the algebra of Feynman graphs. A possible candidate is the Feynman graph cohomology of Kreimer et al.\ \cite{Kreimer_Sars_vSuijlekom,Berghoff_Knispel}. We argue that, analogously to the situation in BRST cohomology, the Feynman graph differential \(\mathscr{D} \colon \HQ \to \HQ\) with \(\mathscr{D}^2 = 0\) should turn the renormalization Hopf algebra \(\HQ\) into a differential-graded Hopf algebra: This means that the Hopf structures are compatible with the differential. More precisely, we have:

\enter

\begin{defn}[Differential-graded renormalization Hopf algebra] \label{defn:differential-graded_renormalization_hopf_algebra}
Let \(\Q\) be a quantum gauge theory with (associated) renormalization Hopf algebra \((\HQ, \mult, \one, \Delta, \coone, S)\) and quantum gauge symmetry ideal \(\iQ\), cf.\ \defnsaref{defn:renormalization_hopf_algebra}{defn:qgs_ideal}. Let \(\mathscr{D} \colon \HQ \to \HQ\) be a Feynman graph differential, e.g.\ \defnref{defn:qgs_differential}, that satisfies the following equations:
\begin{subequations} \label{eqns:dg-hopf-algebra}
\begin{align}
	\mathscr{D} \circ \mult & = \mult \circ \left ( \mathscr{D} \otimes \operatorname{Id} + (-1)^{\operatorname{Deg}_1} \operatorname{Id} \otimes \mathscr{D} \right ) \\
	\mathscr{D} \circ \one & = 0 \\
	\Delta \circ \mathscr{D} & = \left ( \mathscr{D} \otimes \operatorname{Id} + (-1)^{\operatorname{Deg}_1} \operatorname{Id} \otimes \mathscr{D} \right ) \circ \Delta \label{eqn:dg-hopf-algebra-coproduct} \\
	\coone \circ \mathscr{D} & = 0 \\
	\mathscr{D} \circ S & = S \circ \mathscr{D} \label{eqn:dg-hopf-algebra-antipode} \\
	\mathscr{D} \left ( \iQ \right ) & \subseteq \iQ \label{eqn:compatibility-qgs-differential-qgs-ideal}
\end{align}
\end{subequations}
Here, \(\operatorname{Deg}_1\) denotes the degree of the element on the left-hand side of the tensor product, e.g.\ in the corresponding gauge parameter. Then we call \((\HQ, \mathscr{D})\) a differential-graded renormalization Hopf algebra.
\end{defn}

\enter

We remark that the compatibility of the differential \(\mathscr{D}\) with the coalgebra structure of \(\HQ\) ensures the compatibility of the differential with the renormalization operation. Unfortunately, with this setup we immediately obtain the following negative result:

\enter

\begin{lem} \label{lem:edge-and-cycle-markings-not-compatible}
	The edge-marking and cycle-marking graph differentials of \cite{Kreimer_Sars_vSuijlekom,Berghoff_Knispel} are not compatible with the renormalization Hopf algebra \(\HQ\) in the sense of Equations~(\ref{eqns:dg-hopf-algebra}).
\end{lem}

\begin{proof}
	This statement follows immediately, as both marking differentials are in their present form not compatible with the contraction of divergent subgraphs: The edge-markings are not compatible with the renormalization Hopf algebra, as they would require connected Feynman graphs. More precisely, the edge-markings are designed to relate trees of lower-valent vertices with higher-valent vertices. This is not preserved in the contraction of 1PI divergent subgraphs. Additionally, the cycle-markings are not compatible with the renormalization Hopf algebra, as cycles are not preserved in the contraction of divergent subgraphs.
\end{proof}

\enter

Nevertheless, we remark that there exist the following options to achieve compatibility: The edge-markings could be made compatible via redefining the renormalization Hopf algebra so that it is generated by connected Feynman graphs rather than 1PI Feynman graphs. Additionally, the cycle-markings could be made compatible by changing them into path-markings. However, we suggest the following \emph{quantum gauge symmetry differential}: It has the advantage that it is immediately compatible with the renormalization Hopf algebra in the sense of \eqnsref{eqns:dg-hopf-algebra}. Additionally, it is directly related to BRST cohomology, as it projects edges onto their unphysical degrees of freedom:

\enter

\begin{defn}[Quantum gauge symmetry differential] \label{defn:qgs_differential}
	Let \(\Q\) be a quantum gauge theory with Feynman graph set \(\GQ\), renormalization Hopf algebra \(\HQ\) and transversal structure \(\TQ\). We consider the edge set \(E_\Gamma\) of each Feynman graph \(\Gamma \in \GQ\) to be ordered and introduce an additional edge-labeling via the longitudinal projection operators in \(\TQ\). Then we define for each longitudinal projection operator \(\boldsymbol{L} \in \TQ\) a `quantum gauge symmetry differential' \(\mathscr{D}_{\boldsymbol{L}}\) via
	\begin{subequations}
	\begin{align}
		\mathscr{D}_{\boldsymbol{L}} \left ( \Gamma \right ) & := \sum_{\mathfrak{L} \in \wp \left ( \mathfrak{E}_\Gamma \right )} \left ( -1 \right ) ^{\operatorname{Sgn} \left ( \mathfrak{L} \right )} \Gamma_\mathfrak{L} \, ,
		\intertext{where \(\wp \left ( \mathfrak{E}_\Gamma \right )\) denotes the power set of the set of unlabeled edges \(\mathfrak{E}_\Gamma \subseteq E_\Gamma\) and \(\Gamma_\mathfrak{L}\) denotes the Feynman graph where the edges in the set \(\mathfrak{L}\) are marked via \(\boldsymbol{L}\). In addition, \(\left ( -1 \right ) ^{\operatorname{Sgn} \left ( \mathfrak{L} \right )}\) denotes the sign associated to the set \(\mathfrak{L}\) according to the ordering of the edge set \(E_\Gamma\). Then we define the `total quantum gauge symmetry differential' as the sum}
		\mathscr{D} \left ( \Gamma \right ) & := \sum_{\boldsymbol{L} \in \TQ} \mathscr{D}_{\boldsymbol{L}} \left ( \Gamma \right )
		 \, .
	\end{align}
	\end{subequations}
	Additionally, we define the corresponding Feynman rules as follows: On compatible gauge field edges the label \(\boldsymbol{L}\) simply denotes the longitudinal projection via the operator \(\boldsymbol{L}\). Then, the corresponding vertex Feynman rules are given via the respective longitudinal projections, which lead to cancellation identities. Finally, on the remaining edge-types the label \(\boldsymbol{L}\) represents the cancellation of the edge due to the longitudinal projection of a neighboring vertex via \(\boldsymbol{L}\).
\end{defn}

\enter

Thus, with this definition we can represent the corresponding cancellation identities \cite{tHooft_Veltman,Citanovic,Sars_PhD,Kissler_Kreimer,Gracey_Kissler_Kreimer,Kissler} on the algebra of Feynman graphs. In particular, we remark that the cohomology class corresponding to a given Feynman graph is the sum of this graph with all possible labelings and respective signs. Thus, in particular each such element represents a linear combination of graphs, where all edges are projected onto their physical degrees of freedom. With this setup, we arrive at the following:

\enter

\begin{conj} \label{conj:transversality-via-qgs-differential}
	Let \(\Q\) be a quantum gauge theory with Feynman rule \(\Phi\). Let furthermore \(\iQ\) be the quantum gauge symmetry ideal from \defnref{defn:qgs_ideal} and \(\mathscr{D}\) the quantum gauge symmetry differential from \defnref{defn:qgs_differential}. Then the condition
	\begin{equation}
		\mathscr{D} \left ( \iQ \right ) \in \operatorname{Ker} \left ( \Phi \right )
	\end{equation}
	is equivalent to the transversality of the theory.
\end{conj}

\enter

Thus, given this setup, the following questions arise:
\begin{itemize}
	\item How do the graph differentials of \cite{Kreimer_Sars_vSuijlekom,Berghoff_Knispel} and \defnref{defn:qgs_differential} compare physically?
	\item How is the precise relation of the quantum gauge symmetry differential \(\mathscr{D}\) with the BRST differential \(D\)?
	\item Is it useful to define also a quantum gauge symmetry anti-differential \(\overline{\mathscr{D}}\) analogous to the anti-BRST operator \(\overline{D}\)?
\end{itemize}
Finally, we comment on the concrete case of (effective) Quantum General Relativity coupled to the Standard Model: Given the setup of \defnref{defn:qgs_differential}, we obtain the following two differentials: First we obtain \(\mathscr{P} := \mathscr{D}_{\mathbbsit{L}}\) associated to the longitudinal projection operator for gravitons \(\mathbbit{L} \in \mathcal{T}_\text{QGR}\) and second we obtain \(\mathscr{Q} := \mathscr{D}_L\) associated to the longitudinal projection operator for gauge bosons \(L \in \mathcal{T}_\text{QYM}\). Then we obtain the total quantum gauge symmetry differential via
\begin{subequations}
\begin{align}
	\mathscr{D} & := \mathscr{P} + \mathscr{Q}
	\intertext{as above. We remark the similarity to the BRST differentials \(P\), \(Q\) and}
	D & := P + Q \, ,
\end{align}
\end{subequations}
cf.\ \sectionref{sec:diffeomorphism-gauge-brst-double-complex}. In particular, the labeling via \(\mathbbit{L}\) and \(L\) allows to distinguish if an edge was cancelled due to the longitudinal projection of a neighboring graviton or of a neighboring gauge boson. Additionally, this setup also allows the cancellation of a gauge boson edge due to the longitudinal projection of a neighboring graviton and vice versa. Thus, contrary to the setup of \cite{Kreimer_Sars_vSuijlekom,Berghoff_Knispel}, our differentials are not contracting edges to create higher valent vertices or label cycles to create ghost loops. Rather, they incorporate the cancellation identities of the theory into its renormalization Hopf algebra via \emph{Feynman rule cohomology}. This project will be continued in \cite{Prinz_8}.

\subsubsection{\emph{On a Family of Metric Decompositions in Linearized Gravity}} \label{ssec:fam_md_lin_grav}

There are two prominent metric decompositions and definitions of the graviton field in General Relativity: First, there is the metric decomposition with respect to the Minkowski background metric that goes back to M. Fierz, W. Pauli and L. Rosenfeld from the 1930s \cite{Rovelli}, i.e.\
\begin{align}
h_{\mu \nu} := \frac{1}{\gcoupling} \left ( g_{\mu \nu} - \eta_{\mu \nu} \right ) & \iff g_{\mu \nu} \equiv \eta_{\mu \nu} + \gcoupling h_{\mu \nu} \, .
\intertext{Secondly, there is the metric density decomposition of Goldberg and Capper et al.\ \cite{Goldberg,Capper_Leibbrandt_Ramon-Medrano,Capper_Medrano,Capper_Namazie} from the 1950s and 1970s, respectively, i.e.}
	\boldsymbol{\phi}^{\mu \nu} := \frac{1}{\gcoupling} \left ( \sqrt{- \dt{g}} g^{\mu \nu} - \eta^{\mu \nu} \right ) & \iff \sqrt{- \dt{g}} g^{\mu \nu} \equiv \eta^{\mu \nu} + \gcoupling \boldsymbol{\phi}^{\mu \nu} \, .
\end{align}
There are numerous articles discussing perturbative calculations and Ward identities with respect to both graviton fields, \(h_{\mu \nu}\) and \(\boldsymbol{\phi}^{\mu \nu}\). However, there is no work known to the author that studies their relationship. Thus, comparing both approaches, the following questions arise:
\begin{itemize}
	\item Are both perturbative expansions equivalent?
	\item Do both graviton fields satisfy equivalent Ward identities?
	\item If not, which is the physically correct choice?
\end{itemize}
To this end, we introduce the following homotopy:

\enter

\begin{defn}[Graviton density family] \label{defn:graviton-density-family}
	Let \(\mathfrakit{h}^{\{ \omega \}} \in \Gamma \big ( M, \operatorname{Sym}^2 \left ( T^* M \right ) \! \big )\) be the following tensor density with \(\omega \in [0, 1]\) a homotopy in its tensor density weight:
	\begin{equation} \label{eqn:family-graviton-field-density}
		\mathfrakit{h}^{\{ \omega \}}_{\mu \nu} := \frac{1}{\gcoupling} \left ( \left ( - \dt{g} \right )^{\omega / 2} g_{\mu \nu} - \eta_{\mu \nu} \right ) \iff \left ( - \dt{g} \right )^{\omega / 2} g_{\mu \nu} \equiv \eta_{\mu \nu} + \gcoupling \mathfrakit{h}^{\{ \omega \}}_{\mu \nu}
	\end{equation}
	Then we call \(\mathfrakit{h}^{\{ \omega \}}_{\mu \nu}\) the graviton density family.
\end{defn}

\enter

This definition then allows us to study the transition from the graviton field \(h_{\mu \nu}\) to the graviton density \(\boldsymbol{\phi}^{\mu \nu}\). In particular, we have
\begin{align}
	\mathfrakit{h}^{\{ 0 \}}_{\mu \nu} & \equiv h_{\mu \nu}
	\intertext{and}
	\mathfrakit{h}^{\{ 1 \}}_{\mu \nu} & \equiv \boldsymbol{\phi}_{\mu \nu} \, ,
\end{align}
where the indices on the graviton density family \(\mathfrakit{h}^{\{ \omega \}}_{\mu \nu}\) are raised and lowered with the Minkowski background metric \(\eta_{\mu \nu}\) and its inverse \(\eta^{\mu \nu}\). In this project, we plan to calculate the Feynman rules in the sense of \cite{Prinz_4} with \(\omega\) as an additional open parameter. This then allows us to analyze the dependence of tree and loop Feynman amplitudes on \(\omega\). First calculations for the graviton propagator suggest the following:

\enter

\begin{conj} \label{conj:perturbative-expansion-essentially-independent-of-tensordensityweight}
	The perturbative expansion with respect to the graviton density family \(\mathfrakit{h}^{\{ \omega \}}_{\mu \nu}\) depends on \(\omega\) only via the external leg structure. Thus, the corresponding Feynman integrals and Ward identities are equivalent with respect to \(\omega \in [0, 1]\) modulo transformations on the external graviton legs of the amplitudes.
\end{conj}

\enter

In particular, we emphasize the following interesting result that was already mentioned in \remref{rem:md_vs_mdd}: Comparing the propagators of the graviton field \(h_{\mu \nu}\) and the graviton density \(\boldsymbol{\phi}^{\mu \nu}\), the transversal and longitudinal decompositions are dual to each other with respect to the tensors \(\mathscr{G}\) and \(\mathscr{L}\), cf.\ \colref{col:l-tensor-gg-ll-qgr} and \thmref{thm:feynman-rule-graviton-propagator-lt-decomposition}. In addition, the transition is induced via the metric \(\mathbbit{G}\) and its inverse \(\mathbbit{G}^{-1}\), cf.\ \lemref{lem:identities_tensors_qgr}. This can be understood as follows: \(\mathscr{G}\) describes an infinitesimal gauge transformation of proper \((0,2)\)-tensors, i.e.\ \(\omega = 0\), whereas \(\mathscr{L}\) describes an infinitesimal gauge transformation of \((0,2)\)-tensor densities with weight \(\omega = 1\). Thus, a particularly interesting situation appears in the case \(\omega = \textfrac{1}{2}\), i.e.\
\begin{equation} \label{eqn:family-graviton-field-density}
	\mathfrakit{h}^{\{ \textfrac{1}{2} \}}_{\mu \nu} := \frac{1}{\gcoupling} \left ( \left ( - \dt{g} \right )^{1 / 4} g_{\mu \nu} - \eta_{\mu \nu} \right ) \iff \left ( - \dt{g} \right )^{1 / 4} g_{\mu \nu} \equiv \eta_{\mu \nu} + \gcoupling \mathfrakit{h}^{\{ \textfrac{1}{2} \}}_{\mu \nu} \, ,
\end{equation}
as in this case the tensors \(\mathbbit{L}\) and \(\mathbbit{T}\) become symmetric with respect to their upper and lower indices! In particular, this situation resembles the transversal structure of Quantum Yang--Mills theory with a Lorenz gauge fixing, cf.\ \ssecref{ssec:transversality_qym}. This project will be continued in \cite{Prinz_10}.

\newpage
\quad
\newpage

\begin{appendix}

\chapter{Declarations and Bibliography}

\newpage
\quad
\newpage

\section{Declarations} \label{apx:declarations}

I declare that I have completed the thesis independently using only the aids and tools specified. I have not applied for a doctor’s degree in the doctoral subject elsewhere and do not hold a corresponding doctor’s degree. I have taken due note of the Faculty of Mathematics and Natural Sciences PhD Regulations, published in the Official Gazette of Humboldt-Universität zu Berlin no. 42/2018 on 11/07/2018.

\enter

\hfill David Nicolas Prinz

\newpage
\quad
\newpage

\refstepcounter{section}
\label{apx:references}
\setcounter{section}{1}
\bibliography{References.bib}{}

\begin{thebibliography}{100}
\bibitem {Friedrich}
\btxnamefont {\btxlastnamefont {{H. Friedrich}}}\btxauthorcolon\ \btxjtitlefont
  {\btxifchangecase {{P}eeling or not peeling --- is that the
  question?}{{P}eeling or not peeling --- is that the question?}}
\newblock \btxjournalfont {Class. Quantum Grav.}, 35(8):083001, 2018.
\newblock arXiv:1709.07709v2 [gr-qc].
\end{thebibliography}
\bibliographystyle{babunsrt}

\end{appendix}

\end{document}